\documentclass[12pt]{article}
\usepackage{graphicx}
\usepackage{natbib}
\usepackage{tikz}
\usepackage{multirow}
\usepackage{dcolumn,booktabs}
\usepackage{algorithm,algorithmicx,algpseudocode,algcompatible}
\usepackage{mathtools} 
\usepackage{authblk} 
\usepackage{url} 
\usepackage{enumitem}
\usepackage[font=small]{caption} 
\usepackage{amsmath,amssymb,bm,amsthm} 
\newtheorem{theorem}{Theorem}
\newtheorem{proposition}{Proposition}
\newtheorem{lemma}{Lemma}

\algrenewcommand{\algorithmiccomment}[1]{\hfill$\blacktriangleright$ #1}

\graphicspath{{Figures}}

\DeclareMathOperator{\E}{E}

\DeclareMathOperator*{\argmin}{argmin}

\DeclareMathOperator{\diag}{diag}
\DeclareMathOperator{\rank}{rank}
\DeclareMathOperator{\vecmat}{vec}
\DeclareMathOperator{\IFu}{IF}
\DeclareMathOperator{\vect}{vec}
\def\sd{\text{\tiny SD}}

\newcommand{\rk}{k}
\newcommand{\ran}{s}
\newcommand{\rt}{t}
\newcommand{\rtt}{\tilde t}
\newcommand{\st}{q}
\newcommand{\ki}{\ell}
\newcommand{\der}{\ell}

\newcommand{\case}{\text{\scriptsize{case}}}
\newcommand{\cell}{\mbox{\scriptsize{cell}}}
\newcommand{\obs}{\mbox{\scriptsize{obs}}}

\newcommand{\impx}{x^{\mbox{\scriptsize{imp}}}}
\newcommand{\impxsij}{(x^*_{ij})^{\mbox{\scriptsize{imp}}}}
\newcommand{\bimpx}{\boldsymbol{x}^{\mbox{\scriptsize{imp}}}}
\newcommand{\bimpxs}{(\boldsymbol{x^*})^{\mbox{\scriptsize{imp}}}}

\newcommand{\CD}{{\boldsymbol{\cdot}}}
\newcommand{\eps}{\varepsilon}

\newcommand{\bxbar}{\overline{\boldsymbol x}}

\newcommand{\thh}{\tilde{h}}

\newcommand{\tm}{\tilde{m}}

\newcommand{\tr}{\tilde{r}}
\newcommand{\hx}{\widehat{x}}

\newcommand{\hmu}{\widehat{\mu}}

\newcommand{\hsigma}{\widehat{\sigma}}

\newcommand{\bzero}{\boldsymbol 0}
\newcommand{\bone}{\boldsymbol 1}

\newcommand{\ba}{\boldsymbol a}

\newcommand{\be}{\boldsymbol e}
\newcommand{\bof}{\boldsymbol f}
\newcommand{\bol}{\boldsymbol g}
\newcommand{\bL}{\boldsymbol L}
\newcommand{\bp}{\boldsymbol p}
\newcommand{\br}{\boldsymbol r}

\newcommand{\bu}{\boldsymbol u}

\newcommand{\bv}{\boldsymbol v}
\newcommand{\btv}{\boldsymbol{\tilde{v}}}
\newcommand{\bw}{\boldsymbol w}

\newcommand{\bx}{\boldsymbol x}
\newcommand{\tx}{\tilde{x}}
\newcommand{\btx}{\boldsymbol{\tilde{x}}}
\newcommand{\bhx}{\boldsymbol{\widehat{x}}}

\newcommand{\bz}{\boldsymbol z}

\newcommand{\bA}{\boldsymbol A}
\newcommand{\bB}{\boldsymbol B}
\newcommand{\bC}{\boldsymbol C}
\newcommand{\bD}{\boldsymbol D}
\newcommand{\bE}{\boldsymbol E}
\newcommand{\bI}{\boldsymbol I}
\newcommand{\bK}{\boldsymbol K}
\newcommand{\tL}{\widetilde{L}}
\newcommand{\bM}{\boldsymbol M}
\newcommand{\bO}{\boldsymbol O}
\newcommand{\bP}{\boldsymbol P}
\newcommand{\bhP}{\boldsymbol{\widehat{P}}}
\newcommand{\bR}{\boldsymbol R}
\newcommand{\btr}{\boldsymbol{\tilde{r}}}
\newcommand{\btR}{\boldsymbol{\widetilde{R}}}
\newcommand{\bS}{\boldsymbol S}

\newcommand{\bT}{\boldsymbol T}

\newcommand{\bU}{\boldsymbol U}
\newcommand{\btU}{\boldsymbol{\widetilde{U}}}
\newcommand{\bhU}{\boldsymbol{\widehat{U}}}

\newcommand{\bV}{\boldsymbol V}
\newcommand{\btV}{\boldsymbol{\widetilde{V}}}
\newcommand{\bh}{\boldsymbol h}
\newcommand{\bhV}{\boldsymbol{\widehat{V}}}
\newcommand{\bW}{\boldsymbol W}
\newcommand{\btW}{\boldsymbol{\widetilde{W}}}
\newcommand{\bX}{\boldsymbol X}
\newcommand{\bhX}{\boldsymbol{\widehat{X}}}

\newcommand{\btX}{\boldsymbol{\widetilde{X}}}

\newcommand{\bg}{\boldsymbol g}

\newcommand{\btheta}{\bm \theta}
\newcommand{\bTheta}{\boldsymbol \Theta}
\newcommand{\bPsi}{\boldsymbol \Psi}

\newcommand{\bPi}{\boldsymbol \Pi}

\newcommand{\bmu}{\boldsymbol \mu}
\newcommand{\bhmu}{\boldsymbol{\hat{\mu}}}
\newcommand{\btmu}{\boldsymbol{\widetilde{\mu}}}
\newcommand{\bsigma}{\boldsymbol \sigma}
\newcommand{\bSigma}{\boldsymbol \Sigma}
\newcommand{\bhSigma}{\boldsymbol{\widehat{\Sigma}}}

\newcommand{\hlambda}{\hat{\lambda}}

\newcolumntype{M}[1]{>{\centering\arraybackslash}m{#1}}

\definecolor{orange1}{RGB}{255,128,0}
\definecolor{purple2}{RGB}{102,0,204}
\definecolor{blue}{RGB}{0,0,255}
\definecolor{red}{RGB}{255,0,0}

\begin{document}

\def\spacingset#1{\renewcommand{\baselinestretch}
{#1}\small\normalsize} \spacingset{1}


\title{\bf Robust Principal Components\\ 
           by Casewise and Cellwise Weighting}
\author[1,2]{Fabio Centofanti}
\author[2]{Mia Hubert}
\author[2]{Peter J. Rousseeuw}

\affil[1]{Department of Industrial Engineering, University 
          of Naples Federico II, Naples, Italy}
\affil[2]{Section of Statistics and Data Science, Department 
          of Mathematics, KU Leuven, Belgium}

\setcounter{Maxaffil}{0}
\renewcommand\Affilfont{\itshape\small}
\date{October 31, 2025} 
  \maketitle

\bigskip
\begin{abstract}
Principal component analysis (PCA) is a fundamental tool 
for analyzing multivariate data. Here the focus is on
dimension reduction to the principal subspace, characterized 
by its projection matrix. The  classical principal subspace 
can be strongly affected by the presence of outliers.
Traditional robust approaches consider casewise outliers, 
that is, cases generated by an unspecified outlier 
distribution that differs from that of the clean cases. 
But there may also be cellwise outliers, which are
suspicious entries that can occur anywhere in the data
matrix. Another common issue is that some cells may be 
missing. This paper proposes a new robust PCA method, 
called cellPCA, that can simultaneously deal with 
casewise outliers, cellwise outliers, and missing cells. 
Its single objective function combines two robust loss
functions, that together mitigate the effect of casewise 
and cellwise outliers. The objective function is minimized 
by an iteratively reweighted least squares (IRLS) algorithm.
Residual cellmaps and enhanced outlier maps are proposed for 
outlier detection. The casewise and cellwise influence 
functions of the principal subspace are derived, and its
asymptotic distribution is obtained. Extensive 
simulations and two real data examples 
illustrate the performance of cellPCA.
\end{abstract}


\newpage
\spacingset{1.5}
\section{Introduction} \label{sec:intro}

The prevalence of ever larger datasets poses substantial
challenges for statistical analysis. A common issue is 
the presence of outliers and missing data, caused by a 
variety of factors such as measurement errors
and rare and unexpected events. 
Multivariate data are typically represented by a 
rectangular matrix in which the $n$ rows are the cases 
(objects) and the $p$ columns are the variables (measurements). 
Outliers or anomalies are pieces of data that behave differently from the overall pattern. Depending on the situation, outliers may be undesirable errors, or valuable nuggets of unexpected information. Either way, first we want our statistical analysis to summarize the information contained in the regular values and don't want it to be adversely affected by the outliers. Afterward we want to detect the outliers.

Diagnostic approaches first apply a classical fitting method to the data, next they detect the anomalies using diagnostic tools (such as standardized residuals), and finally they run a standard estimation procedure on the outlier-free data set.  However, classical methods can be affected
by outliers so strongly that the resulting fitted model
may not allow to detect the outliers.
This is called the masking effect. Additionally, some
regular values might even appear to be outlying,
which is known as swamping. A real-data example of how diagnostic approaches can be affected by masking and swamping is provided in Section~\ref{sec:masking} of the Supplementary Material. 
To avoid these effects, robust approaches aim to find an estimate that is
close to the one we would have found without the outliers \citep{huber1981, hampel1986, RL1987, maronna2019robust}. They construct a fit that is affected little by
outliers without the need for searching and explicitly removing them.
Outlier detection is then performed by looking at large deviations from the robustly estimated model.

For a long time, the term ``outlier" meant an outlying case. 
These {\it casewise outliers} are assumed not to be generated by the same mechanism as the majority of the cases. 
Consider for example the TopGear dataset from the R package \texttt{robustHD} \citep{Alfons:robustHD}. It contains technical information of 297 cars such as their height, weight, acceleration, fuel consumption, horsepower, miles per gallon (MPG), etc. Most cars in this dataset have an internal combustion engine that runs on petrol or diesel. Only a few are hybrid or electric cars with different technical properties. In a sense they do not belong to the same population, and as such can be considered as casewise outliers.

Formally, the \textit{casewise contamination model} assumes that the observed $n \times p$ data matrix $\bX$ is a random sample from a $p$-variate random variable $X_\eps$ distributed as\linebreak $(1-\eps^{\case})H_0 + \eps^{\case} H_Z$,  where $0 \leqslant \eps^{\case}<0.5$, $H_0$ is the distribution generating the clean cases, and $H_Z$ is an unspecified outlier-generating distribution. No conditions regarding support or symmetry are imposed on $H_Z$. The variable $X_\eps$ can equivalently be written as
\begin{equation} \label{eq:cont_case}
X_{\eps}=A^{\case} \odot X + (\bone_p-A^{\case}) \odot Z 
\end{equation}
where the Hadamard product $\odot$ multiplies vectors (and matrices) entry by entry. Here $X \sim H_0$, $Z \sim H_Z$, and $\bone_p$ is a column vector with all 
$p$ components equal to $1$. The \mbox{$p$-variate} variable $A^{\case}$ has Bernoulli distributed marginals $A_j^{\case}$ for $j=1,\ldots,p$ with success parameter $1-\eps^{\case}$, and jointly they are fully dependent in the sense that $P(A_1^{\case} = \ldots = A_p^{\case})=1$.\linebreak
When $X$, $A^{\case}$ and $Z$ are independent from each other, model~\eqref{eq:cont_case} corresponds to the classical $\eps$-contamination model of \cite{huber1981}, also called the \textit{fully dependent contamination model} (FDCM) by 
\citet{alqallaf2009}. 
It implies that on average, $(1-\eps^{\case})100\%$ of the cases are clean. The left panel of Figure~\ref{fig:casecellmixed} visualizes this setting for a toy data set with 15 cases and 10 variables. Here 3 out of the 15 cases (20\%) are casewise outliers. 

Casewise robust methods require that fewer than half of the cases are contaminated. This assumption is often realistic in low-dimensional datasets, but it becomes harder in high dimensions because then there are many variables in which something can go wrong. Moreover casewise robust methods work by downweighting or deleting all variables of the outlying cases, whereas their outlying behavior might occur in only a few measurements. 
This has motivated the study of {\it cellwise outliers} in recent years. These are deviating measurements (cells) 
 that can occur anywhere in the data matrix.
For example, in the TopGear data the weight of the Peugeot 207 was reported as only 210 kg, which is clearly wrong. That outlying cell was easy to spot, as other cars are much heavier. But not all cellwise outliers deviate marginally. The MPG of the Suzuki Jimny is 39, which is not an uncommon value by itself, but it is very low given its other characteristics such as its small size.    
The \textit{cellwise contamination model} assumes that the data are generated according to
\begin{equation} \label{eq:cont_cell}
X_{\eps}=A^{\cell} \odot X + (\bone_p-A^{\cell}) \odot Z 
\end{equation}
with $X$ and $Z$ as in \eqref{eq:cont_case}. The $p$-variate variable $A^{\cell}$ has  Bernoulli components $A_j^{\cell}$ with success probability $1-\eps^{\cell}_j$. If the components of $A^{\cell}$ are independent, we obtain the \textit{fully independent contamination model} (FICM) of \citet{alqallaf2009}. Then on average each variable has $(1-\eps^{\cell}_j)100\%$ clean values, but even a relatively small proportion of outlying cells can contaminate over half the cases, which may cause casewise robust methods to fail. 
The middle panel of Figure~\ref{fig:casecellmixed} illustrates how a dataset with $22/150 \approx 15\%$ of outlying cells 
yields only $4/15 \approx 27\%$ entirely uncontaminated cases. It also illustrates that a case can have zero, one, a few, or many outlying cells. In general, when $\eps_j^{\cell}=\eps$ the probability that a case contains at least one outlying cell is $1-(1-\eps)^p$ which grows quickly with the dimension $p$. For example, when $\eps=0.05$ and $p=14$ this probability is 51\%. It increases to 97\% when $p=70$, thus impacting the vast majority of cases. 

To cope with data from model~\eqref{eq:cont_cell} one needs cellwise robust methods. They can differ markedly from casewise robust methods. For instance, cellwise robust fits cannot be equivariant under orthogonal transformations, because rotating data destroys its cells. 
In the last decade cellwise robust estimators have been constructed for multivariate location and covariance matrices \citep{agostinelli2015robust,cellMCD, puchhammer2025multi}, for compositional data \citep{rieser2023comp}, regression \citep{Ollerer:ShootingS} and clustering of low-dimensional data \citep{zaccaria2025cluster}. For a recent review of the properties and challenges of cellwise methods see \cite{raymaekers2025challenges}. 

\begin{figure}[!ht]
\centering
\vspace{2mm}
\includegraphics[width=0.99\textwidth]
   {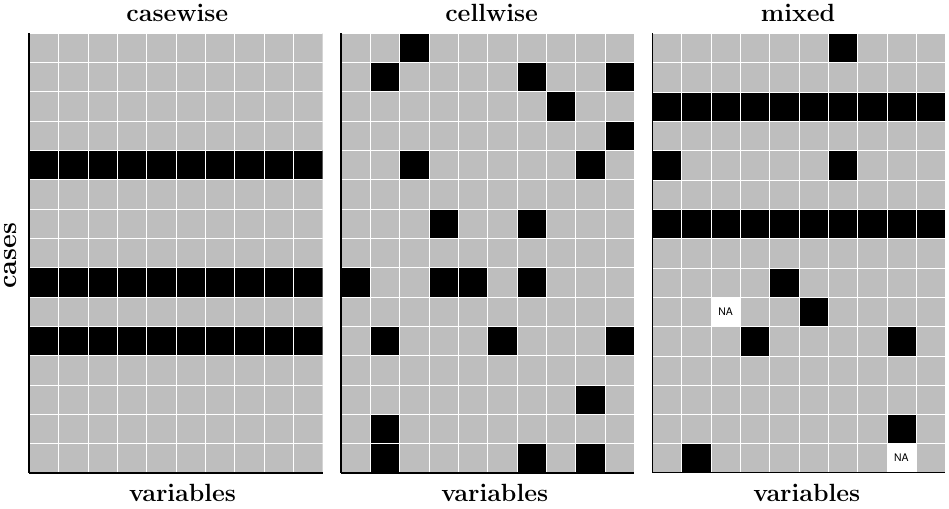}\\
\caption{Illustration of the casewise, 
cellwise, and mixed contamination models.}
\label{fig:casecellmixed}
\end{figure}

In real data both types of outliers often occur simultaneously, and some measurements might be missing. We define the \textit{mixed contaminated and partially observed contamination model} (MCPO) as
\begin{equation} \label{eq:cont_both}
X_{\eps} = A \odot X + (\bone_p-A) \odot Z 
\end{equation}
where $A = A^{\case} \odot A^{\cell} \odot A^{\obs}$, the variable $A^{\case}$ is defined as in \eqref{eq:cont_case}, and $A^{\cell}$ as in \eqref{eq:cont_cell}. The entries $A_j^{\obs}$ in $A^{\obs}$ are binary variables with possible outcomes 1 and NA and $P(A_j^{\obs} = 1) = 1-\eps^{\obs}_j$. The right panel of Figure~\ref{fig:casecellmixed} shows a toy example. 
Different assumptions on the dependence structure of $X$, $Z$,  $A^{\case}$, $A^{\cell}$ and $A^{\obs}$ lead to different contamination models. For instance, when $A^{\obs}$ is independent of the other variables, the values in the dataset are missing completely at random.

In this paper we focus on Principal Component Analysis (PCA), a popular dimension reduction method. Cellwise outliers cause problems in this setting as PCA projects cases orthogonally on a lower dimensional subspace, so outlying cells can propagate to all cells. Therefore, whenever a projection is made, outlying cells should be handled. Classical PCA is strongly affected by both casewise and 
cellwise outliers because it is a least squares method. Also the PCA method of \cite{kiers1997weighted} that can
handle incomplete data is not robust to outliers.
Several casewise robust PCA methods have been proposed, 
such as \cite{locantore1999robust}, 
\cite{hubert2005robpca}, and \cite{she2016robust}. 
Some were designed for structured data matrices
\citep{Engelen2011,deklerk2015}, and some for sparsity
\citep{Rospca2016}. \cite{serneels2008principal} constructed a casewise robust method that can also cope 
with missing values. None of these approaches was designed to handle cellwise outliers. On the other hand, the cellwise robust PCA approaches developed in \citep{de2003framework,maronna2008robust,candes2011robust} are not robust to casewise outliers, as will be illustrated in our simulations.

The more recent MacroPCA
method \citep{hubert2019macropca} was the first to
address the three issues of casewise outliers, cellwise
outliers, and missing values simultaneously. However,
it is a combination of elements from earlier 
methods and lacks a unifying underlying principle,
so it was not possible to derive statistical properties.      

In this paper we propose the cellPCA method that also
handles data generated from the MCPO model~\eqref{eq:cont_both} 
and offers several major improvements over MacroPCA:
\vspace{-2mm}
\begin{enumerate}
\setlength\itemsep{-0.5em}
\item It is the first cellwise and casewise robust PCA method that
minimizes a single objective function. It combines two robust 
losses that effectively mitigate the effect of casewise and 
cellwise outliers.
\item It uses the hyperbolic tangent loss function, which yields
cellwise and casewise weights between 0 and 1, reflecting the 
degree of outlyingness of each entry and each case. Regular cases 
and cells are not downweighted. This makes cellPCA more efficient 
than MacroPCA.
\item The optimization of the objective function is performed by 
an iteratively reweighted least squares algorithm, that is 
proved to converge. 
\item We derive the casewise and cellwise influence functions of 
cellPCA. So far only casewise influence functions were obtained 
for robust PCA methods. 
\item We prove asymptotic normality of the estimated 
principal subspace.
\item We construct imputed cases by modifying 
suspicious and missing cells, in such a way that 
their projection on the principal subspace corresponds 
to the fitted values. 
\item Predictions are constructed for new cases, even
when they are themselves incomplete and/or contain 
cellwise outliers.
\item Enhanced graphical displays are introduced that 
combine information about the \mbox{cellwise} and 
casewise outlyingness in the data. 
\end{enumerate}

Section~\ref{sec:method} presents 
the cellPCA objective and its 
iteratively reweighted least squares (IRLS) algorithm. 
Section~\ref{sec:theo} provides the casewise 
and cellwise influence functions of the estimator,
as well as its asymptotic distribution.
Section~\ref{sec:extensions} contains some practical
extensions such as imputations and predictions, and 
Section~\ref{sec:ionosphere} illustrates our 
enhanced graphical displays of outliers. 
The performance of cellPCA is assessed by Monte 
Carlo in Section~\ref{sec:simulation}, 
and Section~\ref{sec:realdata} illustrates it on 
real data. Section~\ref{sec:conc} concludes.

\section{Methodology}
\label{sec:method}

This section contains the core of the proposed 
methodology, starting with the principal subspace 
model and continuing with the objective function and 
the algorithm.

\subsection{The principal subspace model}
\label{sec:PCAmodel}
The $p$ coordinates of the $n$ cases are stored in
an $n \times p$ data matrix $\bX$.
In the absence of outliers and missing values, the 
goal is to represent the data in a lower 
dimensional space:
\begin{equation}\label{eq:model}
    \bX= \bX^0 + \bone_n \bmu^T +  \bE
\end{equation}
where $\bX^0$ is an $n \times p$ matrix of rank 
$\rk<p$, $\bone_n$ is a column vector with all 
$n$ components equal to $1$, the center 
$\bmu=\left(\mu_1,\dots,\mu_p\right)^T$ is a 
column vector of size $p$, and $\bE$ is the 
error term. We can also write this as
\begin{equation} \label{eq:model2}
  \bX = \bX^0\bP + \bone_n \bmu^T + \bE
\end{equation}
where $\bP$ is a $p \times p$ orthogonal 
projection matrix of rank $\rk$, that is,
$\bP^T=\bP$, $\bP^2=\bP$ and $\rank(\bP)=\rk$,
which projects $\bX^0$ on itself, i.e.
$\bX^0\bP = \bX^0$. We denote 
the rows of $\bX^0$ as $\bx^0_i$\,. 
The image of $\bP$ is a $\rk$-dimensional linear
subspace $\bPi_0$ through the origin. The 
predicted datapoints $\bhx_i = \bx^0_i + \bmu$
lie on the affine subspace
$\bPi = \bPi_0 + \bmu$, which is called the
{\it principal subspace}. 
Note that any $\rk$-dimensional
linear subspace $\bPi_0$ determines a 
unique $\bP$ satisfying the 
constraints $\rank(\bP)=\rk$, $\bP^T=\bP$ 
and $\bP^2=\bP$, and that
any such $\bP$ determines a unique 
subspace $\bPi_0$ of dimension $\rk$. 

In actual computations it may be unwieldy
to work with the matrix $\bP$ because it
might not fit in memory.
For instance, the example in 
Section~\ref{sec:realdata} has $p=40,000$
so $\bP$ has 1.6 billion entries. 
Therefore we parametrize $\bP$ more 
economically. We take an orthonormal basis of 
$\bPi_0$ and form a $p \times \rk$ matrix 
$\bV=\lbrace v_{j\ell}\rbrace=
\left[\bv^1,\dots,\bv^\rk\right]=
\left[\bv_1,\dots,\bv_p\right]^T$ whose
columns are the basis vectors. Therefore
$\bV$ is orthonormal too, that is,
$\bV^T\bV = \bI_\rk$\,. 
We can then write $\bP = \bV \bV^T$. We call 
$\bV$ a loadings matrix, and define the 
corresponding scores matrix as 
$\bU := \bX^0 \bV =\lbrace u_{i\ell}
\rbrace=\left[\bu^1,\dots,\bu^\rk\right]=
\left[\bu_1,\dots,\bu_n\right]^T$ which is 
$n \times \rk$. This way we can carry out 
the computations with the smaller 
matrices $\bV$ and $\bU$ instead 
of $\bP$ and $\bX^0$. For more on 
this reformulation see 
Section~\ref{sec:parametrizations} 
of the Supplementary Material.

\subsection{The objective function}
\label{sec:obj}

Classical PCA approximates $\bX$ by 
$\bhX = \bX^0\bP + \bone_n \bmu^T$ 
which minimizes 
\begin{equation} \label{eq:CPCA_P1}
  ||\bX - \bX^0\bP- \bone_n \bmu^T ||^2_{F}
\end{equation}
under the same constraints on $\bX^0$
and $\bP$,
where $||\CD||_F$ is the Frobenius norm. 
Note that~\eqref{eq:CPCA_P1} estimates the 
principal subspace $\bPi$ determined by 
$\bP$ and $\bmu$, and not (yet) any 
principal directions inside $\bPi$.
Minimizing \eqref{eq:CPCA_P1} is 
equivalent to minimizing 
\begin{equation} \label{eq:appP1}
   \sum_{i=1}^{n}\sum_{j=1}^{p}
   \left(x_{i j}- \hx_{ij}\right)^2
   =\sum_{i=1}^{n}\sum_{j=1}^{p} r_{i j}^2
\end{equation}
with $\hx_{ij} := \bp_j^T\bx_i^0 + \mu_j $\,, 
where  $\bp_{1},\dots,\bp_{p}$  are the columns 
of $\bP$, and  $r_{i j} := x_{i j}-\hx_{ij}$\,. 
The solution is easily obtained. First carry 
out a singular value decomposition (SVD) of 
rank $\rk$ as $\bX-\bone_n\overline{\bx}^T
\approx \bU_\rk\bD_\rk\bV_\rk^T$
where $\overline{\bx}$ is the sample mean,
the $\rk \times \rk$ diagonal matrix $\bD_\rk$ 
contains the $\rk$ leading singular values, 
and the columns of $\bV_\rk$ are the right 
singular vectors. Then the solution is 
$\bP = \bV_\rk\bV_\rk^T$, 
$\bX^0 = \bU_\rk \bD_\rk\bV_\rk^T$, and 
$\bmu = \overline{\bx}-\bP\overline{\bx}$.
But the quadratic loss function 
in~\eqref{eq:appP1} makes this a least squares 
fit, which is very sensitive to casewise as 
well as cellwise outliers. Moreover, 
some $x_{ij}$ may be missing. 

To deal with data generated according to 
the MCPO model~\eqref{eq:cont_both}, we propose 
the cellPCA method which approximates
$\bX$ by $\bhX = \bX^0+\bone_n\bmu^T$ 
obtained by minimizing 
\begin{equation} \label{eq:objP}
  L_{\rho_1,\rho_2}(\bX,\bP,\bX^0,\bmu) := 
  \frac{\hsigma_2^2}{m}\sum_{i=1}^{n}m_i \rho_2\! 
  \left(\frac{1}{\hsigma_2}\sqrt{\frac{1}{m_i}
  \sum_{j=1}^{p} m_{ij}\, \hsigma_{1,j}^2\,
  \rho_1\!\left(\frac{x_{ij}-\hx_{ij}}
  {\hsigma_{1,j}}\right)}\, \right)
\end{equation}
with respect to $(\bP,\bX^0,\bmu)$, under the
same constraints $\bP^T=\bP$, $\bP^2=\bP$, 
$\rank(\bP)=\rk$, and $\bX^0\bP = \bX^0$.
Here $m_{ij}$ is 0 if $x_{ij}$ is missing 
and 1 otherwise, 
$m_i=\sum_{j=1}^{p} m_{ij}$\,,
and $m=\sum_{i=1}^{n} m_i$\,.
The objective~\eqref{eq:objP} can be 
interpreted as follows. The 
{\it cellwise residuals} of variable $j$
are given by 
\begin{equation}\label{eq:r_ij}
  r_{ij} := x_{ij} - \hx_{ij}
\end{equation} 
and divided by a scale estimate 
$\hsigma_{1,j}$. The 
{\it casewise total deviation} of case $i$ 
is defined as
\begin{equation}\label{eq:rt_i}
  \rt_i := \sqrt{\frac{1}{m_i} \sum_{j=1}^{p}
  m_{ij}\, \hsigma_{1,j}^2\, \rho_1\!\left(
  \frac{r_{ij}}{\hsigma_{1,j}} \right)} 
\end{equation} 
and standardized by $\hsigma_2$\,.
For $\rho_1(z) = \rho_2(z) = z^2$ the 
objective~\eqref{eq:objP} would become 
the objective~\eqref{eq:appP1} of classical 
PCA. But we use bounded functions $\rho_1$ 
and $\rho_2$ instead.
The combination of $\rho_1$ and $\rho_2$ in 
\eqref{eq:objP} makes cellPCA robust against 
both cellwise and casewise outliers. 
Indeed, a cellwise outlier in the cell $(i,j)$ 
yields a cellwise residual $r_{ij}$ with a 
large absolute value, but the boundedness of
$\rho_1$ reduces its effect on the estimates.
Similarly, a casewise outlier results in a large 
casewise total deviation $\rt_i$ but its effect
is reduced by $\rho_2$\,.
Note that in the computation of $\rt_i$ the effect 
of cellwise outliers is tempered by  
$\rho_1$\,. This avoids that a single cellwise 
outlier would always give its case a 
large $\rt_i$\,.

The algorithm starts from an initial estimate, 
the MacroPCA fit \citep{hubert2019macropca} 
which is robust against both cellwise and 
casewise outliers and can deal with NAs, 
but is less efficient. 
MacroPCA begins by imputing the NAs by the 
DDC algorithm \citep{DDC2018}.
It yields an initial fit 
$\bhX_{(0)}$
from which we compute cellwise residuals 
$r_{ij}^{(0)}$ as in~\eqref{eq:r_ij}.
For every coordinate $j=1,\ldots,p$ we 
then compute $\hsigma_{1,j}$ as 
an M-scale of the cellwise 
residuals $r_{ij}^{(0)}$\,. An 
M-scale of a univariate sample 
$\left(z_1,\dots,z_n\right)$ is the 
solution $\hsigma$ of the equation
\begin{equation}\label{eq:Mscale}
  \frac{1}{n} \sum_{i=1}^n \rho\left(
  \frac{z_i}{\sigma}\right)=\delta
\end{equation}
for some $\delta$. The classical
standard deviation corresponds to
$\rho(z)=z^2$, but for robust methods it
is important to use a bounded function
$\rho$, otherwise the M-scale can 
become arbitrarily large due to even
a single outlier. Also the choice 
of $\delta$ has an impact on the robustness 
properties of the scale estimator. For 
more on this topic see 
Section~\ref{sec:M-estimation} of the 
Supplementary Material. We use an M-scale
that can resist up to 50\% of outlying values.
We then construct casewise total deviations 
$\rt^{(0)}_i$ by~\eqref{eq:rt_i}.
Next, we compute $\hsigma_2$ as the 
M-scale of those $\rt^{(0)}_i$\,.

Since the DDC algorithm robustly estimates
all pairwise correlations between the variables, 
every pair of variables should have at least 
50\% uncontaminated observations so that a 
casewise robust correlation estimator can 
be applied. This condition is satisfied if each 
variable contains at most 25\% of spoiled cells. 
This assumption was also made by 
\cite{Raymaekers:cellMCD} for cellwise robust
covariance estimation.

Translation equivariance means that if we shift 
the data set $\bX$ by a vector $\ba$ yielding
$\bX+\bone_n\ba^T$, then the fitted $\bhX$ is
transformed in the same way to $\bhX+\bone_n\ba^T$.
This is true for the initial estimator, and 
$\hsigma_{1,j}$ and $\hsigma_2$ do not change.
Therefore also cellPCA is translation equivariant,
as can be seen from its objective. On the other
hand cellPCA is not orthogonally equivariant. 
Orthogonal equivariance would mean that when
the data are rotated, the principal subspace
and the fitted points in it would rotate in 
the same way. This property holds for classical
PCA, but it cannot hold for cellwise robust
methods. Indeed, the cells of the data, that is,
the coordinates of the data points, are tied
with the coordinate system. If the data is
rotated, or equivalently the coordinate
system is rotated, the cells change. The effect 
of one outlying cell could be smeared out over 
all cells in its case. The loss of
orthogonal equivariance is thus an unavoidable 
but necessary trade-off to attain robustness 
against cellwise contamination.

\subsection{Description of the algorithm} 
\label{sec:algo}

We now address the minimization of our 
objective \eqref{eq:objP}.
The solution must satisfy the first-order 
necessary conditions for optimality
which are derived in 
Section~\ref{app:firstorder} of the 
Supplementary Material. For instance, the
first one is obtained by setting the gradients 
of the objective function with respect to 
$\bv_1,\dots,\bv_p$ to zero. The second and 
third one use the gradients with respect to 
$\bu_1,\dots,\bu_n$ and to $\bmu$, yielding  
\begin{align} 
  \label{eq:condi1}
  (\bU^T\bW^j\bU)\bv_j&=\bU^T\bW^j
  (\bx^j-\mu_j\bone_n),  
  \quad j=1,\dots,p\\
  \label{eq:condi2}  
  (\bV^T\bW_i\bV)\bu_i &=
  \bV^T\bW_i\left(\bx_i-\bmu\right),  
  \quad i=1,\dots,n\\
  \label{eq:condi3}
  \sum_{i=1}^n \bW_i\, \bV \bu_i &=  
    \sum_{i=1}^n \bW_i(\bx_i - \bmu)
\end{align}
where $\bx_1^T,\dots,\bx_n^T$ 
and $\bx^1,\dots,\bx^p$ 
are the rows and columns of $\bX$. Here 
$\bW_i$ is a $p \times p$ diagonal matrix, 
whose diagonal entries are equal to the 
$i$th row of the $n\times p$ weight matrix 
\begin{equation} \label{eq:updateW}
  \bW=\lbrace w_{i j}\rbrace=
  \bW^{\case} \odot \bW^{\cell} \odot \bM\;.
\end{equation}
Analogously, $\bW^j$ is an $n \times n$ 
diagonal matrix, whose diagonal entries are 
the $j$th column of the matrix $\bW$.
In expression~\eqref{eq:updateW} for $\bW$, 
the $n \times p$ matrix 
$\bW^{\cell}=\lbrace w_{ij}^{\cell}\rbrace$ 
contains the {\it cellwise weights}
\begin{equation}\label{eq:cellweight}
   w_{ij}^{\cell}=\psi_1\!\left(
   \frac{r_{ij}}{\hsigma_{1,j}}\right)
   \Big/ \frac{r_{ij}}{\hsigma_{1,j}}\; , 
   \quad i=1,\dots,n,\quad j=1\dots,p
\end{equation}
where $\psi_1=\rho_1'$ with the 
convention $w_{ij}^{\cell}(0) = 1$.
The $n \times p$ matrix $\bW^{\case}$ has constant
rows, where each entry of row $i$ is the 
{\it casewise weight} $w_i^{\case}$ given by
\begin{equation}\label{eq:caseweight}
   w_{i}^{\case}=\psi_2\!\left(
   \frac{\rt_i}{\hsigma_2}\right)
   \Big/ \frac{\rt_i}{\hsigma_2}\; , 
   \quad i=1,\dots,n\,
\end{equation}
with $\psi_2=\rho_2'$, and the $n \times p$ 
matrix $\bM$ contains the missingness 
indicators $m_{ij}$\,. Note the similarity 
between the weight matrix $\bW$
in \eqref{eq:updateW} and the variable 
$A = A^{\case} \odot A^{\cell} \odot A^{\obs}$ 
in our model~\eqref{eq:cont_both}.

We now have to choose appropriate functions
$\rho_1$ and $\rho_2$\,. We take the hyperbolic 
tangent (\textit{tanh}) function $\rho_{b,c}$ 
introduced by~\cite{tanh1981}, 
which is defined piecewise by
\begin{equation}\label{eq:rhotanh}
\rho_{b,c}(z) = 
\begin{cases}
 z^2/2 &\mbox{ if } 0 \leqslant |z| \leqslant b\\
  d - (q_1/q_2) \ln(\cosh(q_2(c - |z|)))    
    &\mbox{ if } b \leqslant |z| \leqslant c\\
  d &\mbox{ if } c \leqslant |z|\\
\end{cases}
\end{equation}
where $d = (b^2/2) + (q1/q2)\ln(\cosh(q_2(c - b)))$.
Its first derivative $\psi_{b,c} = \rho'_{b,c}$ has 
been used as the \textit{wrapping function}
\citep{FROC2021} and equals
\begin{equation}\label{eq:psitanh}
\psi_{b,c}(z) = 
\begin{cases}
  z &\mbox{ if } 0 \leqslant |z| \leqslant b\\
  q_1\tanh(q_2(c - |z|))\,\mbox{sign}(z)    
  &\mbox{ if } b \leqslant |z| \leqslant c\\
  0 &\mbox{ if } c \leqslant |z|\,.\\
\end{cases}
\end{equation}
The function $\psi_{b,c}$ is continuous, which 
implies certain constraints on $q_1$ and $q_2$. 
We use the default wrapping 
function shown in Figure~\ref{fig:rho_psi},
which has $b=1.5$ and $c=4$ with 
$q_1=1.54$ and $q_2=0.86$.
\begin{figure}[!ht]
\centering
\includegraphics[width=0.8\textwidth]
  {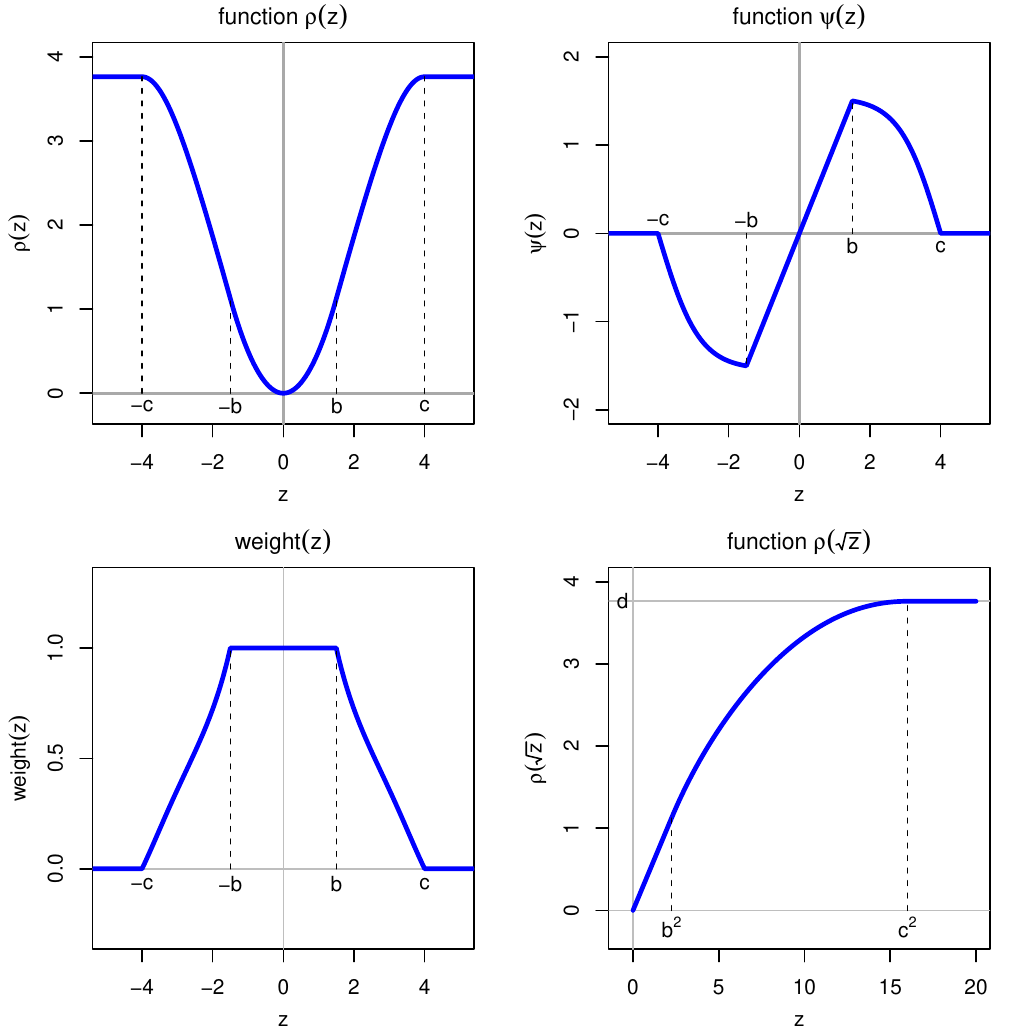}
\vspace{-4mm}
\caption{The function $\rho_{b,c}$ with 
  $b=1.5$ and $c=4$ (top left), its derivative
  $\psi_{b,c}$ (top right), its weight function 
  used in~\eqref{eq:cellweight}
  and~\eqref{eq:caseweight} (bottom left), and 
  the function $\rho(\sqrt{z})$ (bottom right).}
\label{fig:rho_psi}
\end{figure}

The bottom left panel of 
Figure~\ref{fig:rho_psi}
shows the weight function $w(z)$ used
in~\eqref{eq:cellweight} 
and~\eqref{eq:caseweight}. 
An advantage of the function $\psi_{b,c}$ 
is that it is linear in the central region
$[-b,b]$, which yields a higher statistical 
efficiency than competing $\psi$ functions
as shown in~\cite{tanh1981}. This linearity
also makes the weight exactly 1 in that central 
region, so inlying cells will not be 
downweighted, which is an advantage over other 
$\psi$ functions that could have been used.
Moreover, far outliers get weight zero, 
which aids the robustness of cellPCA.

To address \eqref{eq:condi1}--\eqref{eq:condi3} 
we look at a different objective function,
given by
\begin{equation} \label{eq:WPCA}
  \sum_{i=1}^n\sum_{j=1}^p w_{ij}\left(x_{ij}
  - \mu_j - (\bU\bV^T)_{ij}\right)^2
\end{equation}
where $(\bU\bV^T)_{ij} = \bu_i^T\bv_j =
\sum_{\ell=1}^\rk u_{i\ell}v_{j\ell}$ and the 
weight matrix $\bW$ is assumed fixed for now.
What does this have to do with the objective
\eqref{eq:objP} we are trying to minimize?
Well, it is shown in Section~\ref{app:firstorder}
of the Supplementary Material that the first order 
conditions on the weighted PCA of \eqref{eq:WPCA} 
are exactly the same as the first order conditions
\eqref{eq:condi1}--\eqref{eq:condi3} on the
original objective~\eqref{eq:objP}.

The system of equations 
\eqref{eq:condi1}--\eqref{eq:condi3} is 
nonlinear because the weight matrices depend on the 
estimates, and the estimates depend on the weight 
matrices. The optimization of~\eqref{eq:WPCA} can 
be performed by alternating least squares 
\citep{gabriel1978least}. The IRLS algorithm 
starts from our initial estimate 
$(\bV_{(0)}\,,\bU_{(0)}\,,\bmu_{(0)})$ and the
corresponding $\bW_{(0)}$ obtained from 
\eqref{eq:updateW}, \eqref{eq:cellweight}, 
and \eqref{eq:caseweight}.
Then, for each $\st=0,1,2,\dots$, we obtain
$(\bV_{(\st+1)}\,,\bU_{(\st+1)}\,,\bmu_{(\st+1)})$ from 
$(\bV_{(\st)}\,,\bU_{(\st)}\,,\bmu_{(\st)})$ by the 
following four-step procedure, which is described
in more detail in Section~\ref{app:algo} of the 
Supplementary Material.
\begin{itemize}
\item[\textbf{(a)}] Minimize~\eqref{eq:WPCA} with 
  respect to $\bV$ by applying \eqref{eq:condi1} with 
  $\bU_{(\st)}$\,, $\bmu_{(\st)}$\,, and $\bW_{(\st)}$\,. 
  This is done by computing
\begin{equation} \label{eq:updateV}
  (\bv_{(\st+1)})_j = 
  \big(\bU_{(\st)}^T\bW_{(\st)}^j\bU_{(\st)}
  \big)^{\dagger}\;
  \bU_{(\st)}^T\bW_{(\st)}^j\big(\bx^j-
  (\bmu_{(\st)})_j\bone_n\big) 
   \quad j=1,\dots,p
\end{equation}
where $^{\dagger}$ denotes the generalized 
inverse of a matrix.
It is proved that starting from a different
parametrization of the initial $\bP_{(0)}$ 
by $\btV_{(0)} =\bV_{(0)} \bO$ with 
corresponding $\btU_{(0)} = \bU_{(0)}\bO$
yields $\btV_{(\st+1)} = \bV_{(\st+1)} \bO$
and hence the same $\bP_{(\st+1)}$\,.

\item[\textbf{(b)}] To obtain a new $\bU$ from the
  new $\bV_{(\st+1)}$ and the old $\bmu_{(\st)}$ 
  and $\bW_{(\st)}$ we apply
  \begin{equation}\label{eq:condi2cellwise}
  (\bV^T\btW_i\bV)\bu_i =
  \bV^T\btW_i\left(\bx_i-\bmu\right),  
  \quad i=1,\dots,n
  \end{equation}
  where $\btW_i$ is a diagonal matrix whose diagonal
  is the $i$th row of $\bW^{\cell} \odot \bM$. This
  implies~\eqref{eq:condi2} because each case $i$
  has a constant case weight $w_i^{\case}$\,. We compute
  \begin{equation}\label{eq:updateU}
  (\bu_{(\st+1)})_i = \big(\bV_{(\st+1)}^T
  (\btW_{(\st)})_i\bV_{(\st+1)}\big)^{\dagger}\;
  \bV_{(\st+1)}^T(\btW_{(\st)})_i\big(\bx_i-\bmu_{(\st)}
  \big) \quad i=1,\dots,n\,.
  \end{equation}    
  It is proved that this minimizes~\eqref{eq:WPCA},
  and that starting from the alternative parametrization 
  yields $\btU_{(\st+1)} = \bU_{(\st+1)} \bO$ and therefore
  the same
  $\bX^0_{(\st+1)} = \bU_{(\st+1)} \bV_{(\st+1)}^T$\,.

\item[\textbf{(c)}] Minimize~\eqref{eq:WPCA} with respect to
    $\bmu$ by applying \eqref{eq:condi3} with the new 
    $\bV_{(\st+1)}$ and $\bU_{(\st+1)}$ and the old $\bW_{(\st)}$ 
    by setting
\begin{equation}\label{eq:updatemu}
  \bmu_{(\st+1)} = \Big(\sum_{i=1}^n (\bW_{(\st)})_i\Big)^{-1}
  \sum_{i=1}^n (\bW_{(\st)})_i\big(\bx_i - \bV_{(\st+1)}
  (\bu_{(\st+1)})_i\big).
\end{equation}
\item[\textbf{(d)}] Update $\bW_{(\st)}$ according to 
  \eqref{eq:updateW},
  \eqref{eq:cellweight}, and \eqref{eq:caseweight}
  with the new $\bV_{(\st+1)}$, $\bU_{(\st+1)}$ and $\bmu_{(\st+1)}$.
\end{itemize}
The pseudocode of the algorithm is in 
Section~\ref{app:pseudocode}
of the Supplementary Material.

\begin{proposition} \label{the_1P}
Each iteration step of the algorithm decreases 
the objective function~\eqref{eq:objP}, 
that is, 
 $L_{\rho_1,\rho_2}(\bX,\bP_{(\st+1)},
 \bX_{(\st+1)}^0,\bmu_{(\st+1)}) \leqslant 
 L_{\rho_1,\rho_2}
 (\bX,\bP_{(\st)},\bX_{(\st)}^0,\bmu_{(\st)})$\,.
\end{proposition}
This monotonicity result says that going from 
$(\bP_{(\st)},\bX_{(\st)}^0,\bmu_{(\st)})$ to
$(\bP_{(\st+1)},\bX_{(\st+1)}^0,\bmu_{(\st+1)})$ reduces
the variability around the PCA subspace.
Therefore it is a {\it concentration step} in 
the terminology of \cite{FastMCD}.
The proof is given in Section~\ref{app:monotone} 
of the Supplementary Material, and uses the 
fact that for $z \geqslant 0$ the functions 
$z \rightarrow \rho_1(\sqrt{z})$ 
and $z \rightarrow \rho_2(\sqrt{z})$ are
differentiable and concave, as we can see
in Figure~\ref{fig:rho_psi}.
Since the objective function is decreasing 
and it has a lower bound of zero, the algorithm
must converge.

The complexity of the cellPCA algorithm is 
derived in Section~\ref{app:complexity} of the 
Supplementary Material. Its space complexity is
$O(np)$, which equals that of the dataset.
Its time complexity is that of the initial 
estimator MacroPCA, namely 
$O(np(\min(n,p) + \log(n) + \log(p))$. 
This is not much higher than the complexity 
$O(np\min(n,p))$ of classical PCA.

\section{Large-sample properties}
\label{sec:theo}
In the following, the influence function and 
asymptotic normality of cellPCA are presented. All 
the proofs are provided in Section~\ref{app:proofs} 
of the Supplementary Material.

The influence function (IF) is a key robustness tool. 
It reveals how an estimating functional, i.e., a 
mapping from a space of probability measures to a 
parameter space, changes due to a small amount of 
contamination. 
Consider a $p$-variate random variable $X$
with distribution $H_0$. We then contaminate it 
as in \eqref{eq:cont_both} where $H_Z= \Delta_{\bz}$ 
is the distribution that puts all of its mass in a 
fixed $p$-variate vector 
$\bz=\left(z_1, \ldots, z_p\right)^T$, yielding 
\begin{equation} \label{eq:cont_cell_z}
X_{\eps}=A \odot X + (\bone_p-A) \odot \bz 
\end{equation}
with $X \sim H_0$, $A \sim G_\eps$ and 
$A = A^{\case} \odot A^{\cell}$.

The usual casewise IF of \cite{hampel1986} 
is used for casewise contamination under the FDCM 
model (as defined in the introduction). It has 
$A = A^{\case}$ with independent $X$ and 
$A^{\case}$. In that situation the distribution of 
$X_{\eps}$ simplifies to $(1-\eps^{\case})H_0 + 
\eps^{\case}\Delta_{\bz}$. We denote $G_\eps$ as 
$G_\eps^D$ and the distribution of $X_{\eps}$ as 
$H(G_\eps^D,\bz)$. For a functional $\bT$ with 
values in $\mathbb R^p$, the casewise influence 
function is then defined as
\begin{equation}
\label{eq:caseIF}
\IFu_{\case}(\bz, \bT, H_0)=\left.\frac{\partial}{\partial \eps} \bT(H(G_\eps^D,\bz))\right|_{\eps=0}=\,\lim _{\eps \downarrow 0}\frac{\bT\left(H(G_\eps^D,\bz)\right)-\bT(H_0)}{ \eps}\;.
\end{equation} 

\cite{alqallaf2009} proposed a cellwise version of the IF as well.
It considers the contaminated variable~\eqref{eq:cont_cell_z} and the FICM model with $P(A_j^{\cell} = 1) = 1-\eps^{\cell}$ for all $j=1,\ldots,p$. We now denote $G_\eps$ as $G_\eps^I$ and the distribution of 
$X_{\eps}$ as $H(G_\eps^I,\bz)$.
The cellwise influence 
function $\IFu_{\cell}(\bz, \bT, H_0)$ is then defined as
\begin{equation}
\label{eq:gIF}
\IFu_{\cell}(\bz, \bT, H_0)=\left.\frac{\partial}{\partial \eps} \bT(H(G_\eps^I,\bz))\right|_{\eps=0}=\,\lim _{\eps \downarrow 0}\frac{\bT\left(H(G_\eps^I,\bz)\right)-\bT(H_0)}{ \eps}
\end{equation} 

So far influence functions of principal components 
have only been computed under the FDCM and for 
casewise robust PCA methods 
\citep{debruyne2009influence,croux2017robust}. In 
the following, we derive the casewise and cellwise 
IF of cellPCA. We aim to study the robustness 
properties of $\bPi$ characterized by $\bP$.
When there are no missing values we can write the 
functional version $(\bP(H),\bmu(H))$ of 
the minimizer of \eqref{eq:objP} as
\begin{align} \label{eq:ux}
\hspace{-6mm}
  \left(\bP(H),\bmu(H)\right)&=\argmin_{\bP,\bmu}
  \E_H\left[\rho_2 \left(\frac{1}{\sigma_2(H)}\sqrt{\frac{1}{p}
  \sum_{j=1}^{p} \sigma_{1,j}^2(H)\rho_1\left(\frac{
  x_{j}-\mu_j-\bp_j^T\bx^0}{\sigma_{1,j}(H)}\right)}
  \right)\right]\nonumber\\
  \text{such that}\quad \bx^0&= \argmin_{\bx^0}
  \rho_2 
  \left(\frac{1}{\sigma_2(H)}\sqrt{\frac{1}{p}\sum_{j=1}^{p} 
  \sigma_{1,j}^2(H)\rho_1 \left(\frac
  {x_{j}-\mu_j-\bp_j^T\bx^0}{\sigma_{1,j}(H)}
  \right)} \right)
\end{align}
where $\bx^0=\left(x^0_{1},\dots,x^0_{p}\right)^T$ 
and $\bP$ satisfy the usual constraints and 
$\bx=\left(x_1,\dots,x_p\right)^T\sim H$ with $H$ 
a distribution on $\mathbb R^p$, $\bp_j$ and $\mu_j$ 
stay as before, and $\sigma_{1,j}(H)$ and 
$\sigma_2(H)$ are 
the initial scale estimators of 
$r_{j} :=x_{j}-\mu_j- \bp_j^T\bx^0$ and 
$\rt :=\sqrt{\frac{1}{p}\sum_{j=1}^{p} 
\sigma_{1,j}^2(H)\rho_1(r_{j}/\sigma_{1,j}(H))}$. 
If we parametrize $\bP$ by an orthonormal matrix 
$\bV$ with $\bV\bV^T = \bP$ and set the corresponding
$\bU = \bX^0 \bV$, the functional versions of the
first-order conditions 
\eqref{eq:condi1}--\eqref{eq:condi3} must hold:
\begin{align} 
  \label{eq:C1}
  \E_{H}\left[ \bW_{\bx}\bV\bu_{\bx}\bu_{\bx}^T\right]
  &= \E_{H}\left[\bW_{\bx}(\bx-\bmu)\bu_{\bx}^T\right]\\
  \label{eq:C2} 
  \left(\bV^T\bW_{\bx}\bV\right)\bu_{\bx} &=
     \bV^T\bW_{\bx}(\bx-\bmu)\\
  \label{eq:C3}
  \E_{H}\left[\bW_{\bx}\bV\bu_{\bx}\right] &=
   \E_{H}\left[\bW_{\bx} 
   \left(\bx-\bmu\right)\right]\,.
\end{align}
Here $\bW_{\bx}=\diag(\bw_{\bx})$ for 
    $\bw_{\bx}=\left(w_1,\dots,w_p\right)^T$. 
The components of $\bw_{\bx}$ are 
$w_j=w^{\cell}_jw^{\case}$ with  
$ w_{j}^{\cell}=\psi_1\left(\frac{r_{j}}{\sigma_{1,j}}
   \right)\Big/\frac{r_{j}}{\sigma_{1,j}}$
and $w^{\case}=\psi_2\left(\frac{\rt}{\sigma_2}\right)
\Big/\frac{\rt}{\sigma_2}$.
For simplicity we will assume that 
$\bmu=\bzero$. 
\begin{proposition}\label{prop2}
The influence functions of $\bP$ under FDCM and FICM are
\begin{equation}
\label{eq:IFPFDCM}
\IFu_{\textnormal{\scriptsize case}}(\bz,\bP,H_0)= -\bD 
  \Big(\bS\IFu_{\textnormal{\scriptsize case}}(\bz,\bsigma,H_0) +
  \bg(\Delta_{\bz},\bV_0,\bsigma_0)\Big)
\end{equation}
and
\begin{equation}
\label{eq:IFPFICM}
\IFu_{\textnormal{\scriptsize cell}}(\bz,\bP,H_0)= -\bD
  \Big(\bS\IFu_{\textnormal{\scriptsize cell}}(\bz,\bsigma,H_0) + p\sum_{j=1}^p 
  \bg(H(G_1^j,\bz),\bV_0,\bsigma_0)\Big)
\end{equation}
with $\bsigma(H)=\left(\sigma_{1,1}(H), \dots, 
\sigma_{1,p}(H), \sigma_{2}(H)\right)^T$, 
$\bP(H_0)=\bV_0\bV_0^T$, $\bsigma_0=\bsigma(H_0)$ and
\begin{equation}
\label{eq:gP}
   \bg(H,\bV,\bsigma) = \vect\big(\E_{H}[ 
   \bW_{\bx}(\bV\bu_x-\bx)\bu_x^T]\big)
\end{equation}
where $\vect(\CD)$ is the vectorization operator that 
converts a matrix to a vector by stacking its columns on 
top of each other. The matrices $\bD$ and $\bS$ are 
computed in  Sections~\ref{app:proofs} and \ref{app:DS}
of the Supplementary Material, and 
$\IFu_{\textnormal{\scriptsize case}}(\bz,\bsigma)$ and $\IFu_{\textnormal{\scriptsize cell}}(\bz,\bsigma)$ 
are the casewise and cellwise influence functions of 
$\bsigma$. The $G_1^j$ in \eqref{eq:IFPFICM} puts all its 
mass in $(1,\ldots,0,\ldots,1)$ with the single $0$ in 
position $j$. Moreover, choosing a different 
parametrization $\bV_0$ of $\bP(H_0)$ yields the same
$\IFu_{\textnormal{\scriptsize case}}(\bz,\bP)$ and $\IFu_{\textnormal{\scriptsize cell}}(\bz,\bP)$.
\end{proposition}

Note that \eqref{eq:gP} expresses one of the
first order conditions, but the other first-order 
condition $\left(\bV^T\bW_{\bx}\bV\right)\bu_{\bx} =
\bV^T\bW_{\bx}(\bx-\bmu)$ must hold as well, and acts 
as a constraint. Moreover, $\bg$ depends on $\bsigma$ 
through $\bW_{\bx}$ and $\bu_x$\,. 
Also note that $H(G_1^j,\bz)$ in \eqref{eq:IFPFICM} is 
the distribution of $\bX \sim H_0$ but with its 
$j$th component fixed at the constant $z_j$\,. 
It is thus a degenerate distribution concentrated on 
the hyperplane $X_j = z_j$\,.

\begin{figure}[ht]
\centering
\vspace{3mm}
\includegraphics[width=0.45\textwidth]
   {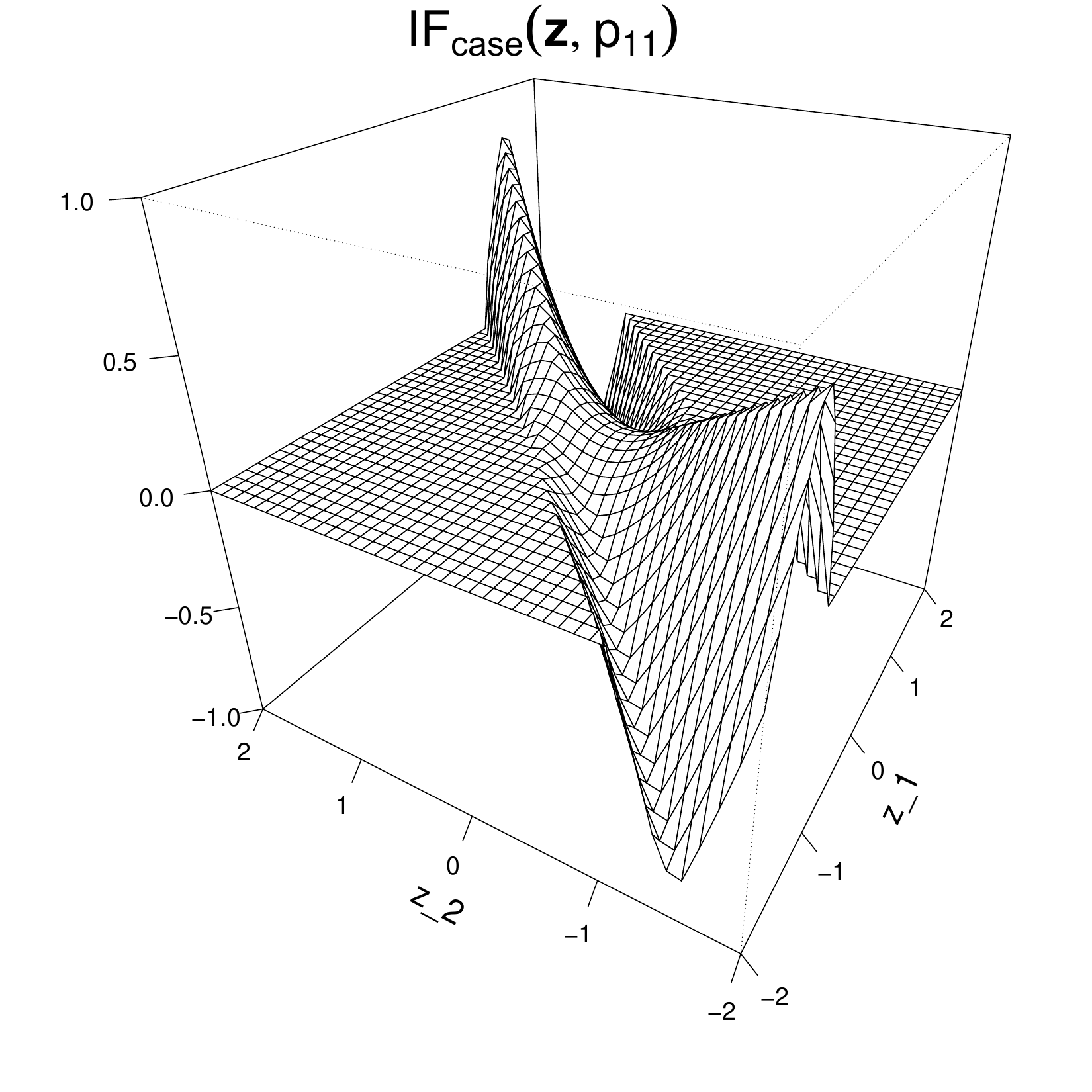}
\includegraphics[width=0.45\textwidth]
   {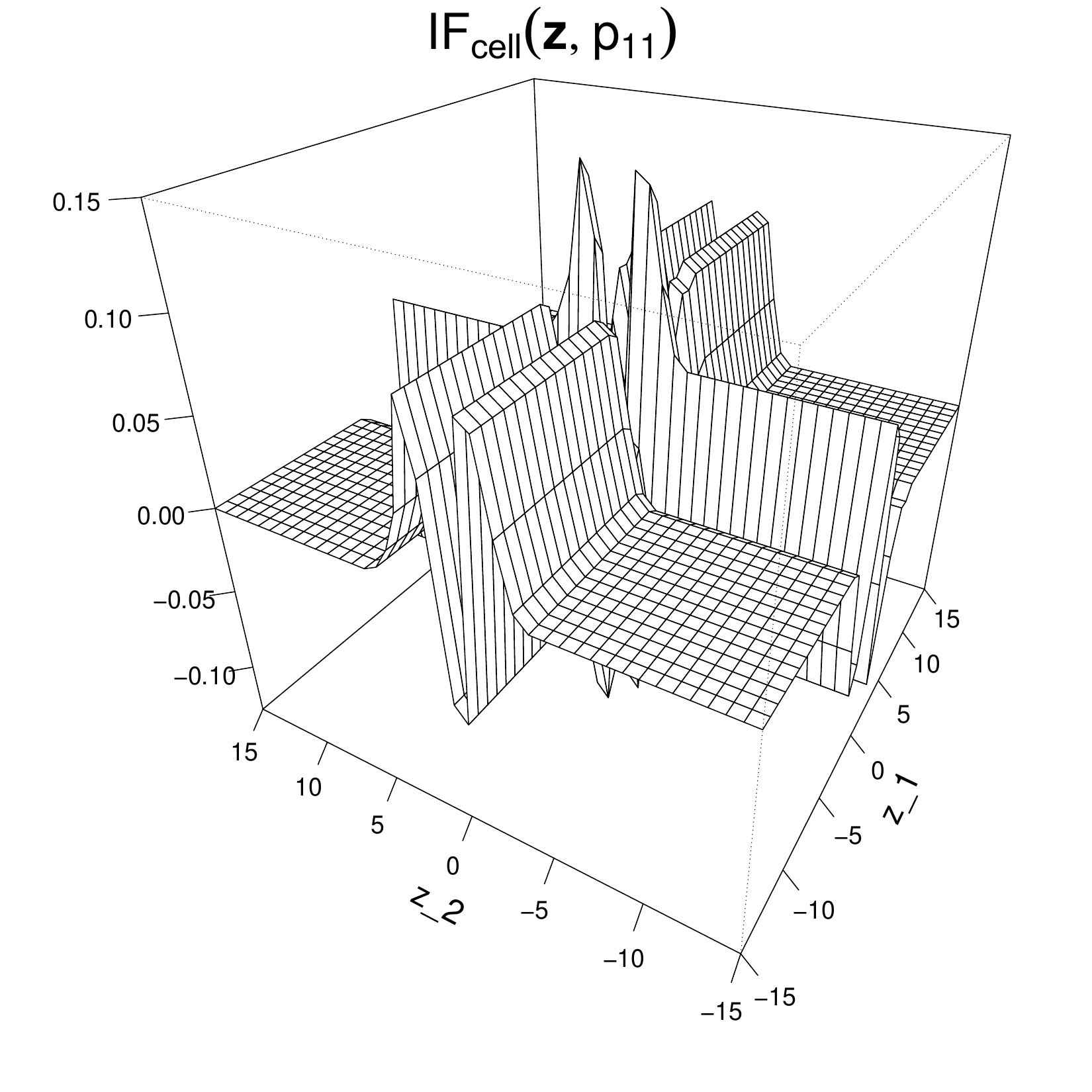}\\
\vspace{-4mm}
\caption{Casewise (left) and cellwise (right)
$\IFu(\bz,p_{11})$ for bivariate normal $H_0$.}
\label{fig_IFP11}
\end{figure}

Let us look at a special case to get a feel for 
these results. For our model distribution we choose the
bivariate normal 
$H_0 = N\Big(\bzero, \begin{bsmallmatrix} 1 & 0.9\\ & \\ 
0.9 & 1 \end{bsmallmatrix}\Big)$ and we want to fit a 
one-dimensional PCA subspace ($\rk=1$). The PCA subspace
of the population is thus the $45^{\circ}$ line $x_2 = x_1$\,.
The left panel of Figure~\ref{fig_IFP11} shows the casewise 
influence function of the entry $p_{11}$ of the estimated 
projection matrix $\bP$. 
If we look along the line $z_2=-z_1$ we see the shape of the 
$\psi$-function in the right panel of Figure~\ref{fig:rho_psi} 
which is bounded and redescending, from which we conclude that 
cellPCA is rather insensitive to outliers orthogonal to the 
fitted subspace. Along lines of the type 
$z_2 = z_1 + \mbox{constant}$ with small nonzero constant, the 
IF is unbounded. This is harmless because those $(z_1,z_2)$ 
are so-called good leverage points, which are closely aligned 
with the fitted subspace. 
For the bad leverage points, i.e.\ with large constant, the IF 
is zero. Also note that $\bP$ itself is always
bounded, since any rank-$\rk$ projection matrix has Frobenius
norm $||\bP||_F=\sqrt{\rk}$\,.

The right panel of Figure~\ref{fig_IFP11} shows the 
cellwise IF in the same setting, which looks quite 
different. As a function of $z_1$ for a large fixed $z_2$ 
it again has the shape of the 
$\psi$-function, and for large $|z_1|$ the cellwise weight 
becomes zero. The situation is the same in the other direction. 
Also note that this IF is bounded. The exact same behavior 
was found by \cite{alqallaf2009} for a redescending 
M-estimator of a bivariate location vector $[\mu_1\;\mu_2]^T$.

It is harder to visualize the IF of the entire 
$2 \times 2$ matrix $\bP$. 
Figure~\ref{fig_IFbivPV} shows the norm of the 
casewise and cellwise IF of $\bP$, in the same 
setting as in Figure~\ref{fig_IFP11}. The shape 
is similar to before, bearing in mind that the 
norm is nonnegative.

\begin{figure}[ht]
\centering
\vspace{3mm}
\includegraphics[width=0.45\textwidth]
    {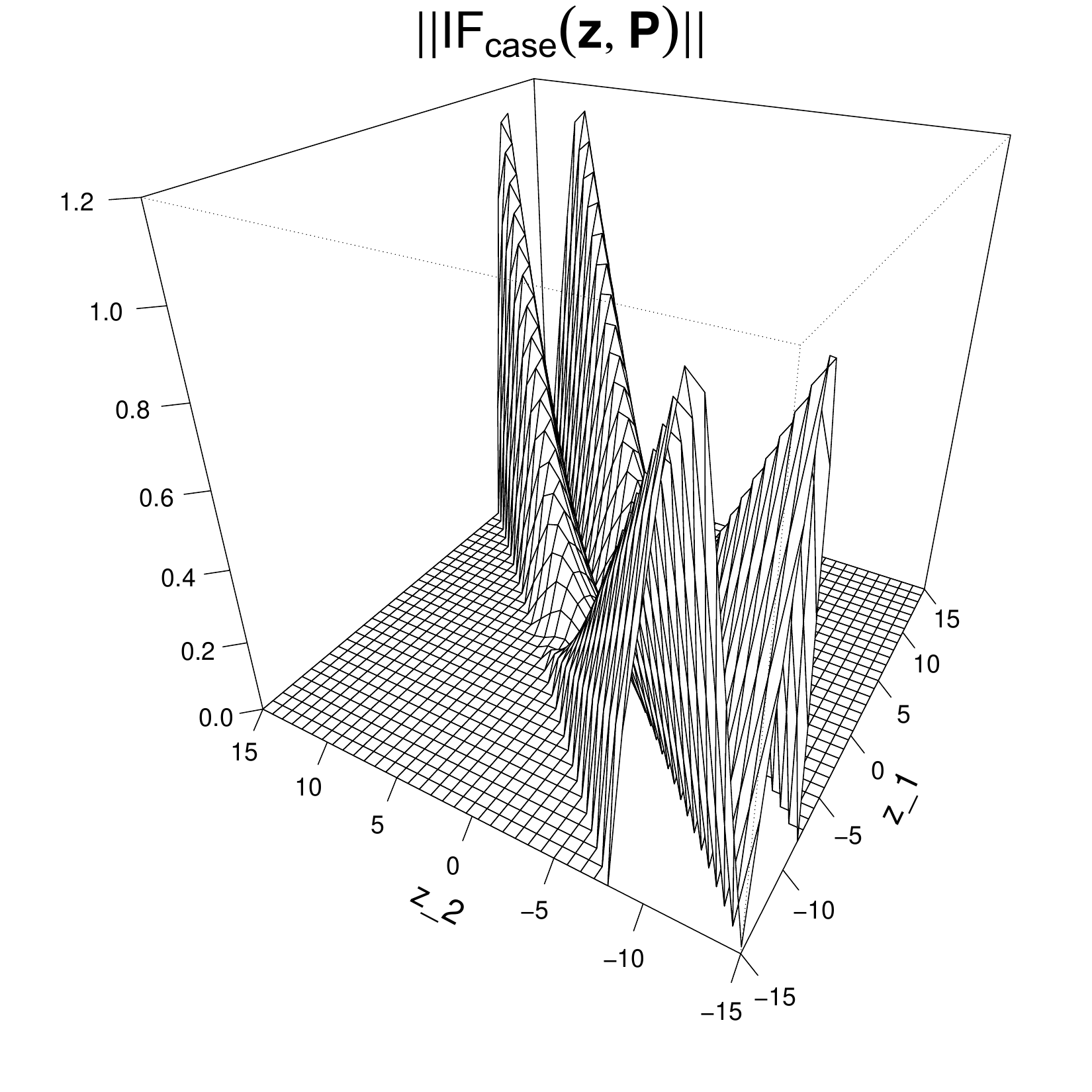}
\includegraphics[width=0.45\textwidth]
    {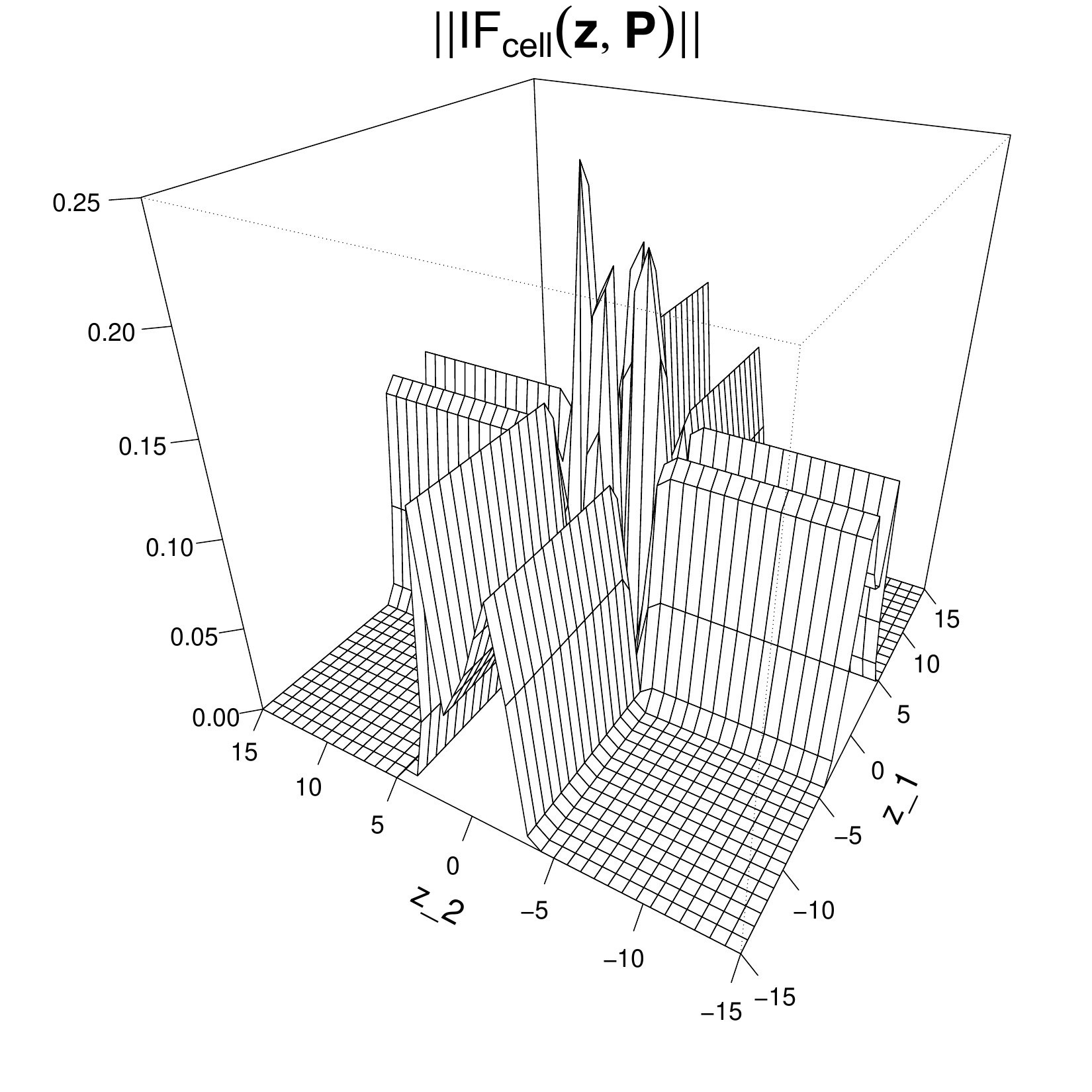}\\
\vspace{-4mm}
\caption{Norm of casewise (left) and cellwise 
(right) $\IFu(\bz,\bP)$ for bivariate normal $H_0$.}
\label{fig_IFbivPV}
\end{figure}

Let us now look at the asymptotic distribution 
of cellPCA. Suppose we obtain i.i.d. observations
$\bx_1,\bx_2,....$ from $H_0$\,.
For simplicity we assume that $\bmu=\bzero$ and that
$\hsigma_{1,j}$ and $\hsigma_2$ are fixed at 
$\sigma_{1,j}(H)$ and $\sigma_{2}(H)$. We use a 
fixed initial estimate $\bP_{(0)}$ obtained as in 
Section~\ref{sec:obj}. We then choose an orthonormal
basis in the principal subspace, that is, a $p\times \rk$
matrix $\bV_{(0)}$ with orthonormal columns such that 
$\bV_{(0)} \bV_{(0)}^T = \bP_{(0)}$\,. We start the 
algorithm from $\bV_{(0)}$\,, yielding a $\bhV_n$ for 
each sample $\{\bx_1,\bx_2,....,\bx_n\}$. We know that 
$\bV_{(0)}$ is not unique, but it is shown in 
Section~\ref{app:algo} of the Supplementary Material that the resulting estimates 
$\bhP_n$ are the same no matter which $\bV_{(0)}$ was 
chosen. Let us assume that $\bhV_n$ converges in 
probability to the population version $\bV(H)$. We 
now denote
\begin{equation*}
  \widehat{\Lambda}_n(\bV):= \frac{1}{n}
  \sum_{i=1}^{n}\bPsi(\bx_i,\bV)\qquad \mbox{and} 
  \qquad\Lambda(\bV):= \E_{H} \bPsi(\bx, \bV)
\end{equation*}
where $\bPsi(\bx_i,\bV)=\left(\bPsi_1(\bx_i,\bV),\dots, 
\bPsi_{p\rk}(\bx_i,\bV) \right)^T=\vect\left(\bW_{i}
\left(\bx_i-\bV\bu_{i} \right)\bu_{i}^T\right)$, 
$\bu_i=\bV^T\bx_i^0$\,, and $\bx_i^0$ satisfies 
$\left(\bP\bW_i\bP\right)\bx_i^0 = \bP\bW_i\bx_i^0$\,.
Then $\widehat{\Lambda}_n(\bhV_n)=\bzero$ and 
$\Lambda(\bV(H))=\bzero$.
\begin{proposition} \label{the_3}
Assume that $\bPsi$ is twice differentiable with respect 
to $\bV$ with bounded second derivatives, and that 
$\dot{\bPsi}_\der =\frac{\partial \bPsi_\der}
{\partial \vect\left(\bV\right)}$ exists and satisfies  
\begin{equation*}
    \left|\dot{\bPsi}_\der(\bx, \bV)\right| \leqslant K(\bx)
    \text { with } \E_{H}[K(\bx)]<\infty
\end{equation*}
for any $\der = 1,\ldots,p\rk$\,, $\bV$ and $\bx$. 
Then $\bhP_n \rightarrow \bP(H)$ in probability, and
\begin{equation*}
  \sqrt{n}\vect\left(\bhP_n-\bP(H)\right) 
  \rightarrow 
  N_{p^2}(\bzero,\bTheta)
\end{equation*}
in distribution, where 
$\bTheta=\E_{H}\left[\IFu_{\case}(\bx,\bP)
\IFu_{\case}(\bx,\bP)^T\right]$
does not depend on the parametrization of $\bP$.
\end{proposition}

\section{Practical extensions of the methodology}
\label{sec:extensions}

\subsection{Selecting the rank $\rk$}
\label{sec:rank}

The rank $\rk$ of the PCA model~\eqref{eq:model} 
is rarely given in advance, one typically 
needs to select it based on the data.
In classical PCA one defines the 
proportion of explained variance, given by
$1 - \nu_\ran/\nu_0$ where 
$\nu_\ran := ||\bX-\bhX_\ran||_F^2$   
in which $\bhX_\ran$ is the best 
approximation of $\bX$ of rank $\ran$. 
In particular 
$\nu_0 = ||\bX-\bone_n\bxbar^T||_F^2$
where $\bxbar$ is the sample mean. 
This $\nu_\ran$ is computed for a range 
of $\ran$ values. One can then make a
plot of $\nu_\ran$ versus $\ran$, which
is called a scree plot. The rank $\rk$ may 
be selected by looking for an `elbow' in 
the scree plot \citep{jolliffe2011principal}.
The elbow can be selected visually, or in an
automatic way by means of the Kneedle 
algorithm \citep{satopaa2011finding}.

Here we define $\nu_\ran$ as the 
objective~\eqref{eq:objP} of the cellPCA 
fit of rank $\ran$. For $\nu_0$ we compute 
\eqref{eq:objP} on the cellwise residuals 
$x_{ij} - \mbox{median}_{\ell=1}^n 
(x_{\ell j})$\,. Afterward we can again 
search for an elbow in the scree plot of 
$\nu_\ran$ versus $\ran$ by the Kneedle
algorithm. 

\subsection{Imputation} 
\label{sec:imput}

When one or more cells of a data point 
$\bx_i$ get cell weights below 1, it means that 
the method considers those cells as contaminated.
In such a situation it may be useful to produce 
an imputed version of that point, in which  
suspicious cells are cleaned and missing cells
are filled in, whereas the other 
cells are kept as they are. So we want to
obtain a point $\bimpx_i$ whose cells are 
$\impx_{ij} = x_{ij}$ for all $j$ with 
$w_{ij}=1$, and with different cells $\impx_{ij}$ 
where $w_{ij}<1$. We would like the modified cells 
to be such that $\bimpx_i$ is closer to the PCA 
subspace, and such that the orthogonal projection 
of $\bimpx_i$ coincides with the fitted $\bhx_i$.

We do this as follows. 
From \eqref{eq:condi2cellwise} 
and $\bhx_i = \bV\bu_i + \bmu$ we know that
$(\btW_i(\bx_i - \bhx_i))^T \bV =\bzero$,
so $\btW_i(\bx_i - \bhx_i)$ is orthogonal
to the PCA subspace. We then construct 
the imputed point 
\begin{equation} \label{eq:impxi}
   \bimpx_i := \bhx_i + \btW_i(\bx_i - \bhx_i)
\end{equation}
so that its orthogonal projection on the PCA
subspace equals $\bhx_i$. Note that every
imputed cell 
$\impx_{ij} = \hx_{ij} + w^{\cell}_{ij} m_{ij}
(x_{ij} - \hx_{ij})$ lies between the original 
$x_{ij}$ and the fitted cell $\hx_{ij}$.

\begin{figure}[t]
\centering
\includegraphics[width=0.7\textwidth]
  {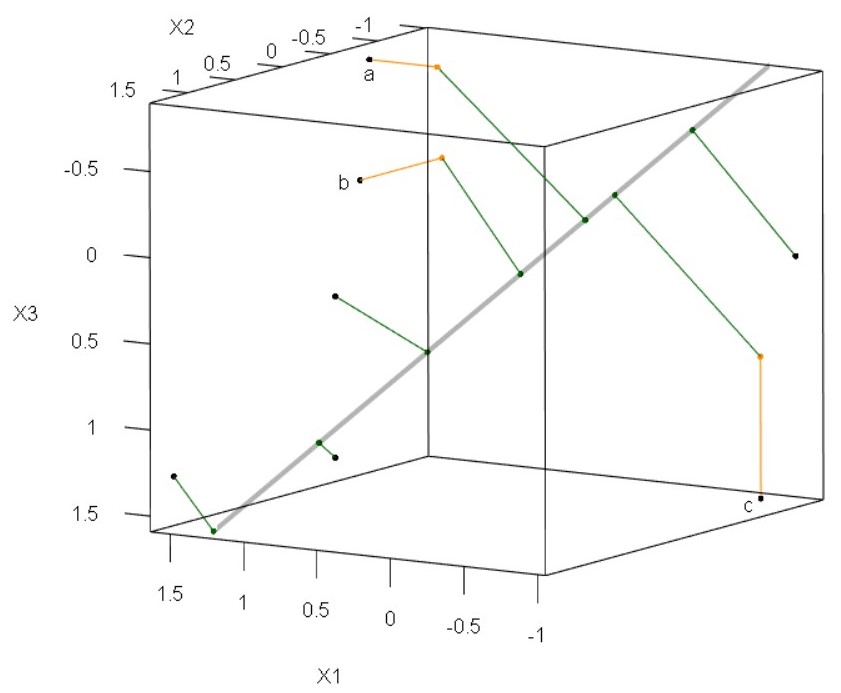}
\caption{Illustration of imputation for $p=3$ 
and $\rk=1$. The first cell of point $a$ was 
imputed before projecting it on the 
principal subspace, the second cell of $b$, 
and the third cell of $c$.}
\label{fig_imp}
\end{figure}

Figure~\ref{fig_imp} illustrates imputation for
a 1-dimensional principal subspace in 3-dimensional
space. Most points are projected orthogonally on the
fitted subspace, but the first cell of point $a$ is
imputed before projecting it, and the second 
cell of $b$, and the third cell of $c$. None 
of these cells would be considered outlying in 
their individual variables.

\subsection{Estimating a center and 
   principal directions}
\label{sec:est_pd}
So far our estimation targets were the 
principal subspace $\bPi$ and the fit 
$\bhX$. This is often sufficient, e.g. 
for face recognition, computer 
vision, signal processing, and data 
compression \citep{vaswani2018robust}.
However, in many other applications one 
may wish to obtain major  
directions in the principal subspace, as 
well as a good estimate of the center of the 
fitted data. These may facilitate 
interpretation of the PCA scores in $\bU$.

Note that the matrix $\bV$ obtained by the
algorithm in Section~\ref{sec:algo} need
not have orthonormal columns. But this can be
fixed by carrying out an SVD of $\bV$ and
putting its right singular vectors in the new
$p \times \rk$ matrix $\btV$\,.
From this we also immediately obtain a set of 
$\rk$-variate scores $\btU = \bU\bV^T\btV$,
and we denote $\btmu:=\bmu$.

This initial parametrization 
$(\btV,\btU,\btmu)$ 
is not yet satisfactory, because the columns
of $\btV$ do not reflect the shape of the
point cloud, and $\btmu$ does not have to
lie in its center. In order to obtain
principal axes in the PCA subspace, 
the algorithm carries out an additional 
step which estimates a center and a scatter 
matrix of the scores in $\btU$\,.
This needs to be done by a robust
estimation method to avoid that 
outlying score vectors have a large
effect on the result, as illustrated in
Appendix A.1 of \cite{hubert2019macropca}. 
For this estimation we use 
the fast and robust deterministic algorithm
DetMCD of \cite{hubert2012deterministic}, 
yielding $\bhmu_{\btU}$ and 
$\bhSigma_{\btU}$\,. 
The spectral decomposition of 
$\bhSigma_{\btU}$ yields a $\rk \times \rk$ 
loading matrix $\bhV_{\btU}$ and eigenvalues 
$\hat{\lambda}_1 \geqslant \hat{\lambda}_2 
\geqslant\ldots \geqslant \hat{\lambda}_\rk$\,. 
We set the final parameter estimates to 
$\bhmu := \btmu+\bhV_{\btU}\bhmu_{\btU}$, 
$\bhV:=\btV\bhV_{\btU}$, and $\bhU:=
(\btU-\bone_n \bhmu_{\btU}^T)\bhV_{\btU}$\,.

\subsection{Out-of-sample prediction}
\label{sec:newx}

The cellPCA estimation in Section~\ref{sec:method}
can be seen as a training stage. There can also
be an out-of-sample stage, where a new datapoint 
$\bx^*$ arrives and we wish to predict its
$\bhx^*$\,. For instance, in chemometrics 
one often carries out a principal component 
regression when there are more regressors 
(say, intensities at many wavelengths) than
cases. In the training stage, also called 
calibration, a PCA model of low rank $\rk$
is fitted to the regressors, after which the 
response is regressed on the low-dimensional 
scores. Afterward, when a new data point arrives, 
it is an $\bx^*$ without response. If $\bx^*$ 
is clean we can simply project it on the principal 
subspace to obtain its scores, and predict the 
response from them.

However, this task becomes nontrivial when 
$\bx^*$ contains NAs and/or cellwise outliers. 
Then we cannot just project $\bx^*$\,. 
Instead we need to estimate its scores vector 
$\bu^*$ for the given $\bV$ and $\bmu$.  
For this purpose we minimize the inner 
part of the objective \eqref{eq:objP} for a 
single $\bx_i = \bx^*$\,. That is, we minimize
\begin{equation}\label{eq:xstar}
\sum_{j=1}^{p}  m^*_{j}\, \hsigma_{1,j}^2\,\rho_1
  \!\left(\frac{x^*_j-\mu_j-\bv_j^T\bu}
  {\hsigma_{1,j}}\right)
\end{equation}
where $\bV$ and $\bmu$ are now known and
$x_{ij}$ was replaced by $x^*_j$ and $m_{ij}$ 
by $m^*_j$\,. When all $m^*_j$ are zero we set $\bu^*$ 
and $\bhx^*$ to NA. If not, \eqref{eq:xstar} 
becomes a sum over $J = \{j\;;\;m^*_j = 1\}$ so
\begin{equation} \label{eq:newu}
\bu^* = \argmin_{\bu}
  \sum_{j \in J} \hsigma_{1,j}^2\,\rho_1
  \!\left(
  \frac{\tx_j - \bv_j^T\bu}{\hsigma_{1,j}}
  \right)
\end{equation}
where $\tx_j := x^*_j - \mu_j$\,. So $\bu^*$ 
is the slope vector of a robust regression
without intercept of the column vector  
$\btx_J = (\tx_j\;;\; j \in J)$ on the 
matrix $\bV_J = [\bv_j^T \;;\; j \in J]$. 
Note that $\bV_J$ has no outliers since all 
$|v_{j\ell}| \leqslant 1$.
The IRLS algorithm uses the first order
condition \eqref{eq:condi2cellwise} and
alternates updating $\bu^*$ according to
\eqref{eq:updateU} with updating the 
cellwise weights~\eqref{eq:cellweight} by 
\begin{equation}\label{eq:cellweightnew}
   w^*_j=\psi_1\!\left(
   \frac{r_j}{\hsigma_{1,j}}\right)
   \Big/ \frac{r_j}{\hsigma_{1,j}} 
   \qquad \mbox{for $j$ in $J$}
\end{equation}
and $w^*_j= 0$ for $j$ not in $J$.
The iterations continue until 
convergence (the pseudocode is in 
Section~\ref{app:pseudocode}
of the Supplementary Material).
Upon convergence we put 
$\hx^*_j := \mu_j + \bv_j^T\bu^*$\,.
When $\bx^*$ happens to be equal to 
some in-sample $\bx_i$ it follows that 
$\bhx^* = \bhx_i$\,.

\section{Displaying outliers graphically}
\label{sec:ionosphere}

We will now construct some graphical displays to
visualize outlying cells or cases when they occur.
As an illustration we consider the 
\texttt{ionosphere} data in the \texttt{R} 
package \texttt{rrcov} \citep{rrcov}. This real 
dataset contains $n=351$ cases with 
$p=32$ numerical variables. The data were collected 
by the Space Physics Group of Johns Hopkins 
as described by \cite{sigillito1989classification}. 
The dataset contains two classes, and we restrict
attention to the 225 cases labeled ``good''. We 
applied cellPCA. The procedure in Section~\ref{sec:rank}
yielded $\rk=2$, with both components together explaining 
84\% of the total variability.

\begin{figure}[!ht]

\vspace{3mm}
\centering
\includegraphics[width=0.95\linewidth]{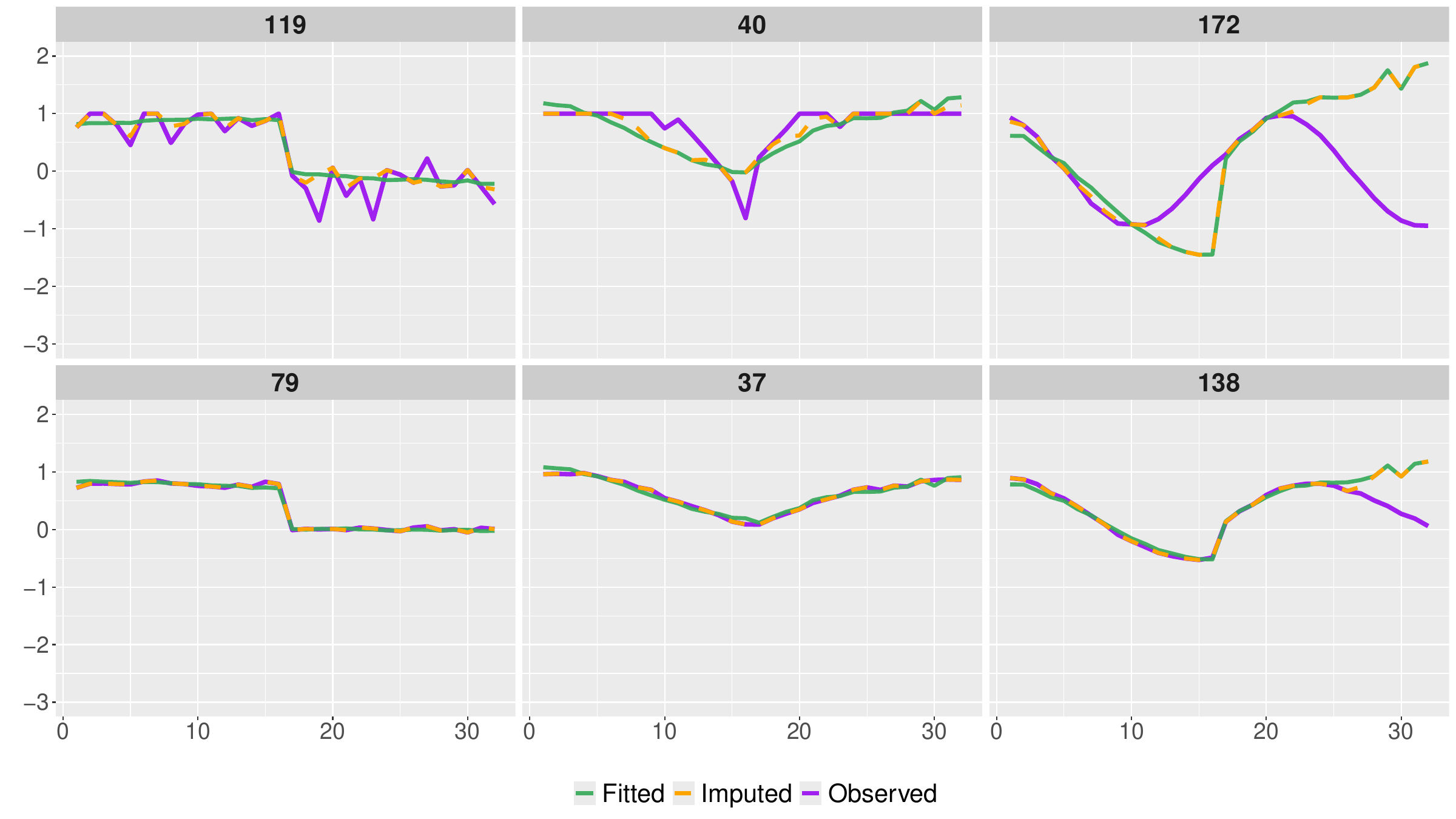}

\vspace{-1mm}
\caption{Observed (purple), fitted (green) and 
imputed (orange dashed) curves of six cases 
in the ionosphere data.}
\vspace{-2mm}
\label{fig:Ionos_curves}
\end{figure}

Figure~\ref{fig:Ionos_curves} shows six cases 
(purple curves) together with their fitted (green) 
and imputed (dashed orange) curves. 
The bottom row shows cases that received a 
large casewise weight 
($w^{\case}_{79} = w^{\case}_{37} = w^{\case}_{138}  = 1$) 
whereas the top row plots curves with low 
($w^{\case}_{119} = 0.44$) and zero casewise weight 
($w^{\case}_{40} = w^{\case}_{172} = 0)$. The fitted values of
the latter differ strongly from their observed values 
in most cells. According to~\eqref{eq:impxi} the 
imputed values agree with the observed ones in cells 
with cell weight $w^{\cell}_{ij} = 1$, whereas they align 
with the fitted values when $w^{\cell}_{ij}=0$. Cell weights
between 0 and 1 yield intermediate imputed values.
Supplementary Material \ref{app:ionos} illustrates 
this interpolation in detail.

In the final residual matrix $\bR = \bX - \bhX = 
\bX - (\bhU \bhV^T + \bone_n \bhmu^T)$ we estimate 
the scale of each column by the 
M-scale~\eqref{eq:Mscale} with the function
$\rho_{b,c}$ of~\eqref{eq:rhotanh}.
Dividing each column of $\bR$ by its scale yields 
the standardized residuals 
$\btR = \{\tr_{ij}\} = [\btr_1,\dots, \btr_n]^T$.
We can then visualize $\btR$, or some of its rows 
and columns, by a \textit{residual cellmap} as in 
Figure~\ref{fig:Ionos_residmap}. Cells with 
$|\tr_{ij}| < \sqrt{\chi^2_{1, 0.99}} = 2.57$ are 
considered regular and colored yellow, whereas any 
missing values would be white. Outlying positive 
residuals receive a color which ranges from light 
orange to dark red (here, when $\tr_{ij} > 6$) and 
outlying negative residuals from light purple to 
dark blue (when $\tr_{ij} < -6$). 

\begin{figure}[!ht]

\vspace{-2mm}
\centering
\includegraphics[width=0.98\linewidth]
  {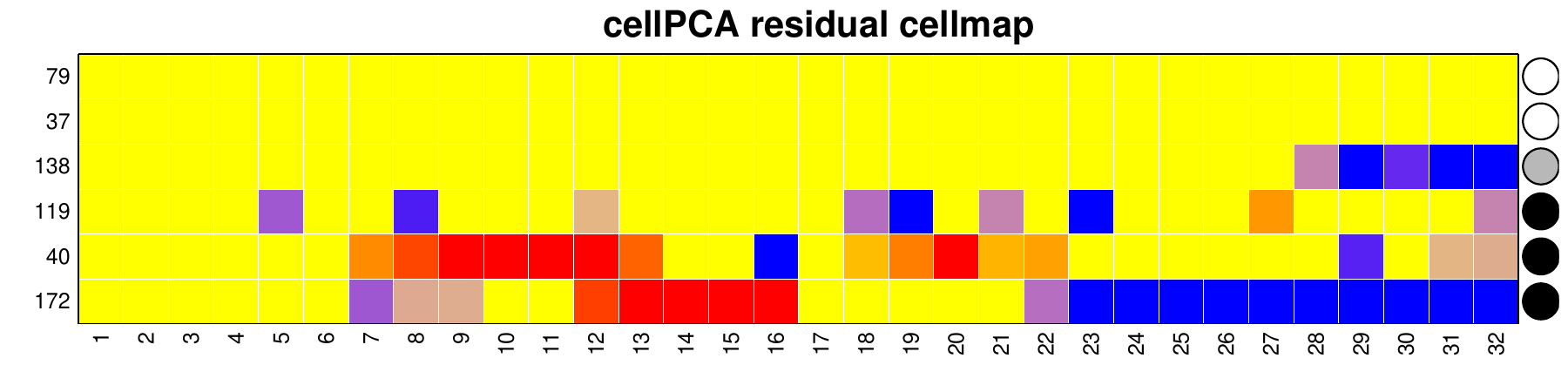}\\

\vspace{-2mm}
\caption{cellPCA residual cellmap of the six cases in 
         Figure~\ref{fig:Ionos_curves}.}
\label{fig:Ionos_residmap}
\end{figure}

We see at a glance that cases 79 and 37 have no outlying 
cells, and that the last cells of case 138 have much 
lower values than expected. The cases at the bottom have 
cells with unexpectedly high values as well as cells with
unexpectedly low values. 
To this residual cellmap we add information about 
casewise outlyingness, by coloring a circle to the 
right of each row according to its  
standardized casewise total deviation $\rtt_i$\,.
For this we compute $\rt_i$ from $\btr_i$ 
using~\eqref{eq:rt_i}, and divide it by its 
M-scale~\eqref{eq:Mscale} with the function
$\rho_{b,c}$ of~\eqref{eq:rhotanh}. The circle is 
colored black when $\rtt_i > c_{\rtt,0.999}$\,, 
white when $\rtt_i < c_{\rtt,0.99}$\,, and has 
an interpolated grayscale in between. The cutoff 
$c_{\rtt,\alpha}$\, is the $\alpha$-quantile of 
the distribution of $\rtt$ simulated for 
uncontaminated data. These colored circles indicate 
casewise outliers. 

To focus more on the casewise outlyingness we propose 
an \textit{enhanced outlier map,} inspired by the 
outlier map of \cite{hubert2019macropca}.
This plot displays for each observation the norm of 
its standardized residual $\|\btr_i\|$ versus its 
score distance $\text{SD}_i$ which is defined as 
the norm of $\big(\bhV^T(\bx_i - \bhmu)\big)\odot
\big[1/\sqrt{\hlambda_{1\phantom{k}}},
\ldots,1/\sqrt{\hlambda_\rk}\big]^T$. 
Figure~\ref{fig:Ionos_outliermap} shows the enhanced 
outlier map obtained by applying cellPCA to the 
ionosphere data. The vertical dotted line indicates 
the cutoff $c_{\sd}=\sqrt{\chi_{\rk, 0.99}^2}$ and the 
horizontal dotted line is at the cutoff $c_{r}$\,, 
the 0.99-quantile of the distribution of 
$\|\tilde{\br}\|$ simulated at uncontaminated data. 
Regular cases have a small 
$\text{SD}_i \leqslant c_{\sd}$ and a small 
$\|\btr_i\| \leqslant c_{r}$\,. Cases with large 
$\text{SD}_i$ and small $\|\btr_i\|$ are called 
good leverage points. The cases with large 
$\|\btr_i\|$ can be divided into orthogonal 
outliers when their $\text{SD}_i$ is small, and bad 
leverage points when their $\text{SD}_i$ is large.

\begin{figure}[!ht]
\centering
\includegraphics[width=0.7\linewidth]
  {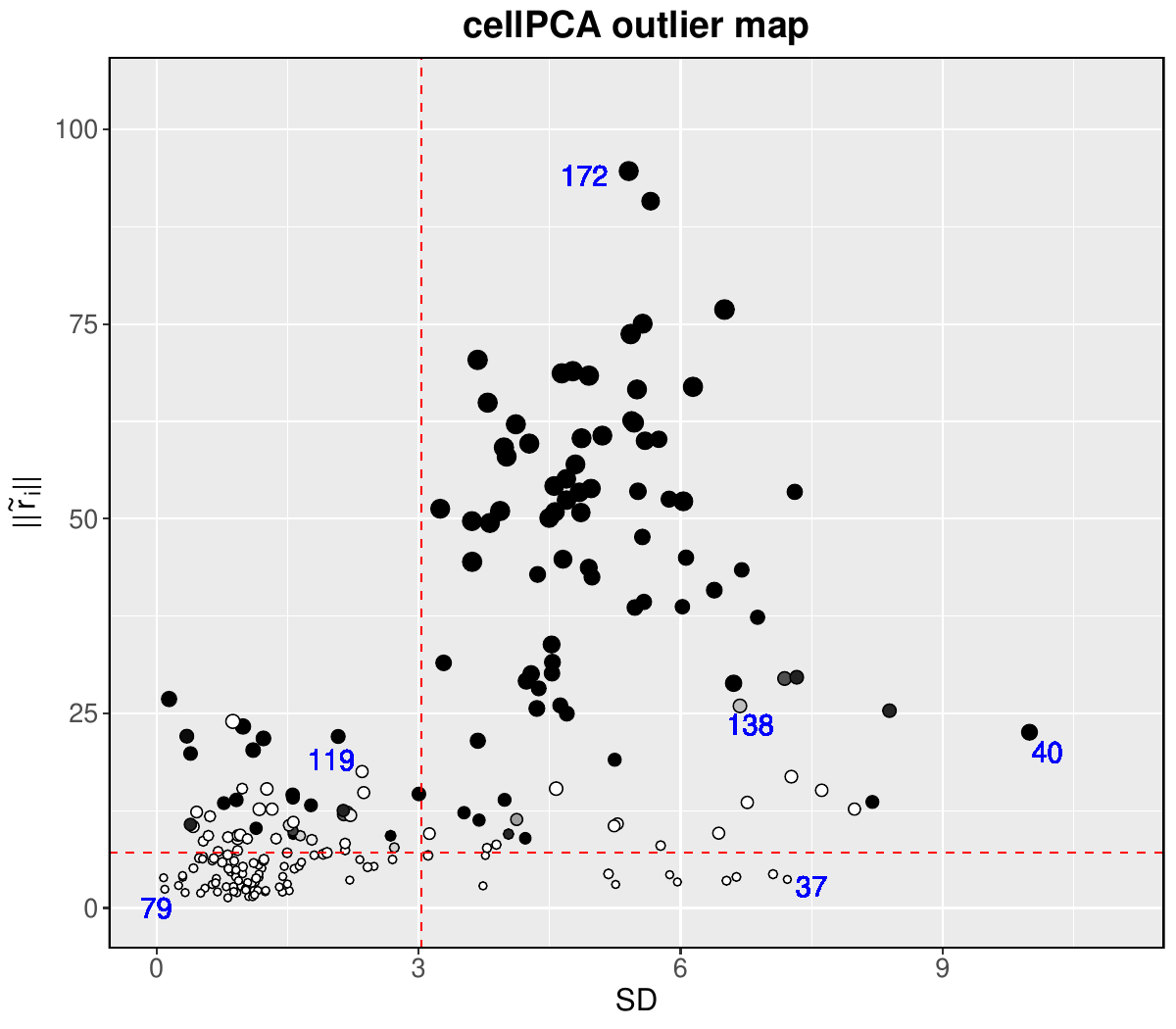}
\caption{Enhanced outlier map of the ionosphere data.}
\label{fig:Ionos_outliermap}
\end{figure}

The size of each point $i$ in 
Figure~\ref{fig:Ionos_outliermap} is proportional to 
one minus the average of its cellwise weights, i.e.\ 
$1-\overline{w}^{\cell}_i=1-\frac{1}{p}\sum_{j=1}^{p}
m_{ij}w_{ij}^{\cell}$. Larger points thus correspond to 
cases with many outlying cells. The points get
the same color as the circles in the 
residual cellmap. This enhanced outlier map thus 
combines information about the cellwise and casewise 
outlyingness of observations, and their position with 
respect to and within the fitted subspace. 
The curves in Figure~\ref{fig:Ionos_curves} 
occupy different positions in this outlier map. 

\section{Simulation study}
\label{sec:simulation}
We study the performance of cellPCA by 
Monte Carlo, with setup similar to  
\citep{hubert2019macropca}. The clean data are 
generated from a multivariate Gaussian with 
$\bmu=\bzero$ and the covariance matrix $\bSigma$
constructed as follows. Its eigenvectors 
are those of the matrix called A09 with entries 
$\Sigma_{j\ell}=(-0.9)^{|j-\ell\,|}$.
For the eigenvalues we take numbers such that
the first component explains 53\% of the total
variance, the first two together explain 90\%, 
and all subsequent eigenvalues are tiny. We 
generate $n=100$ 
data points in dimension $p=20$ so that $n>p$, 
and in dimension $p=200$ for which $n<p$. 

Three contamination types are considered, 
according to 
models~\eqref{eq:cont_case}-\eqref{eq:cont_both}. 
In the cellwise outlier scenario we randomly 
replace $\eps^{\cell}=20\%$ of the 
cells $x_{ij}$ with $\gamma_{\cell}\sigma_j$, 
where $\gamma_{\cell}$ varies from 0 to 6 and
$\sigma_j^2$ is the $j$th diagonal element of 
$\bSigma$. In the casewise outlier setting 
$\eps^{\case}=20\%$ of the cases are generated 
from $N\left(\gamma_{\case} (\be_1+\be_3), 
\bSigma/1.5\right)$, where  $\be_1$ and $\be_3$ 
are the first and third eigenvectors of $\bSigma$, 
and $\gamma_{\case}$ varies from 0 to 9 when 
$p=20$, and from 0 to 24 when $p=200$.
In the third scenario, the data is contaminated 
by $\eps^{\cell}=10\%$ of cellwise outliers as 
well as $\eps^{\case}=10\%$ of 
casewise outliers. Here 
$\gamma_{\case}=1.5\gamma_{\cell}$ when $p=20$ and 
$\gamma_{\case}=4\gamma_{\cell}$ when $p=200$, where 
$\gamma_{\cell}$ again varies from 0 to 6.

We compare cellPCA with $\rk = 2$ to competing 
approaches
that are robust to either cellwise outliers, 
casewise outliers, or both. We run the cellwise 
robust PCA method of \cite{candes2011robust}, 
called CANDES, the special case of cellPCA with 
$\rho_2(z) = z^2$ denoted as Only-cell, 
the casewise robust method of 
\cite{hubert2005robpca} called ROBPCA, and the 
special case of cellPCA with $\rho_1(z) = z^2$\,, 
called Only-case. We also run the MacroPCA  
method \citep{hubert2019macropca}, and include  
classical PCA (denoted as CPCA). 

We measure performance by the angle between the 
estimated principal subspace and the true principal 
subspace. This angle is computed by the function 
\texttt{subspace()} of the \textsf{R} package 
\texttt{pracma} \citep{pracma}.
We also compute the mean squared error given by 
\begin{equation} \label{eq:MSE}
  \text{MSE}=\frac{1}{|\mathcal{U}|}
  \sum_{i\in \mathcal{U}}\sum_{j=1}^p
  \left(x^0_{ij}- \hat{x}_{ij}\right)^2
\end{equation}
where $\mathcal{U}$ is the set of uncontaminated
cases, $\hat{x}_{ij}$ is the prediction of 
$x_{ij}$\,, and $x^0_{ij}$ is the original value 
of the cell before any contamination took place. 
We report the median angle and MSE over 1000 
replications.

\begin{figure}[!ht]
\centering
\begin{tabular}{ccc}
   \large \textbf{Cellwise}  & \large \textbf{Casewise} &\large{\textbf{Casewise \& Cellwise}} \\
   [-4mm]
  \includegraphics[width=.3\textwidth]
  {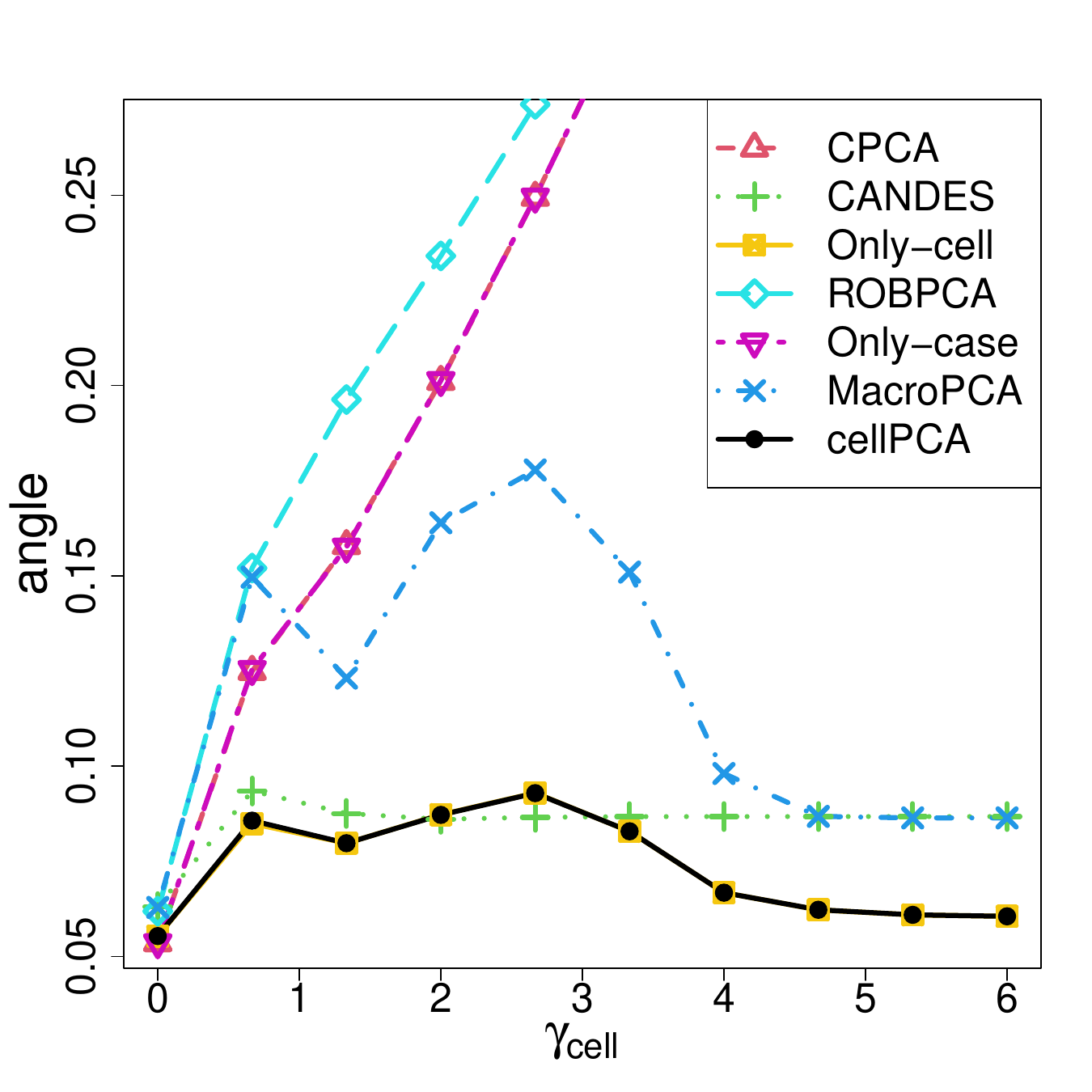} &\includegraphics[width=.3\textwidth]
  {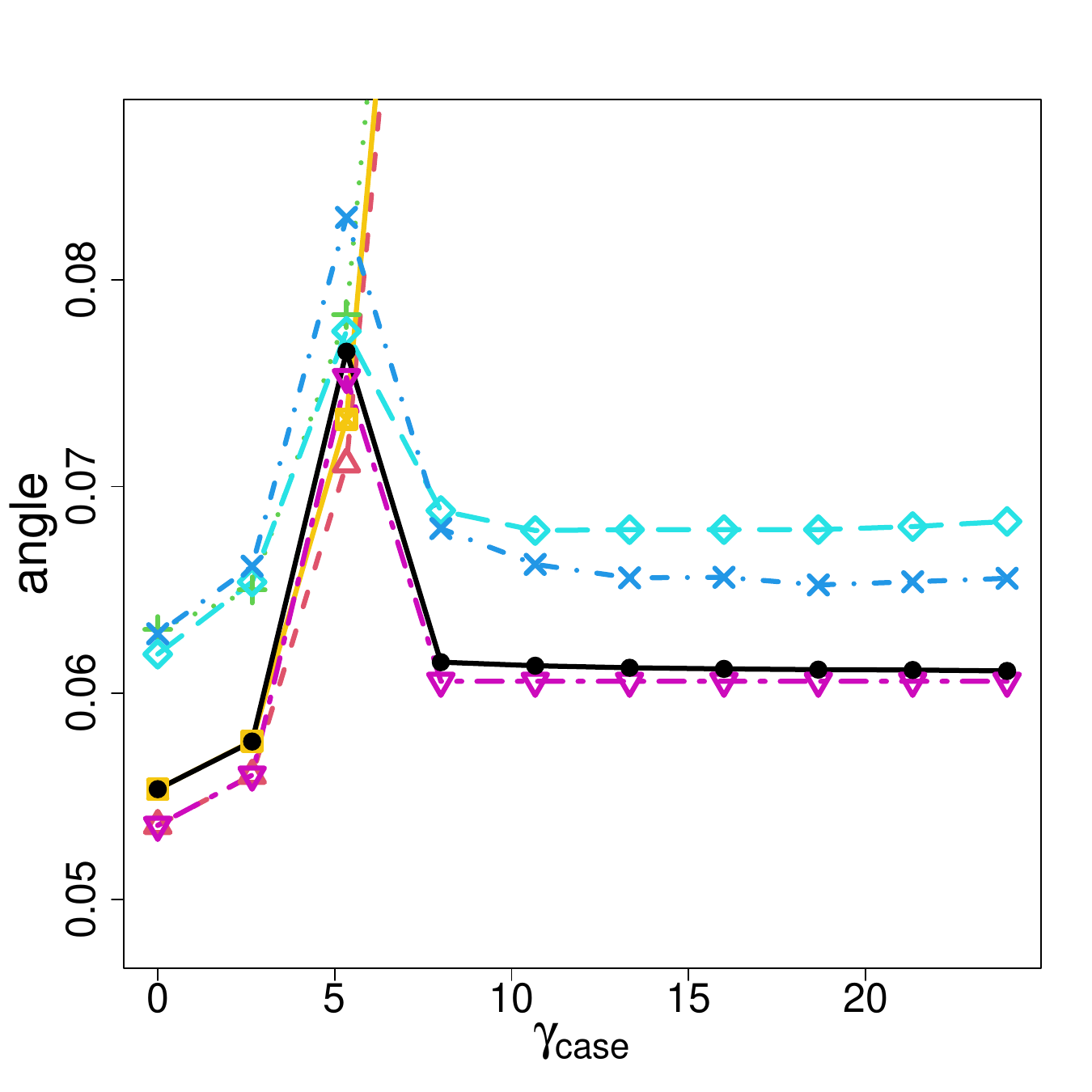} &\includegraphics[width=.3\textwidth]
  {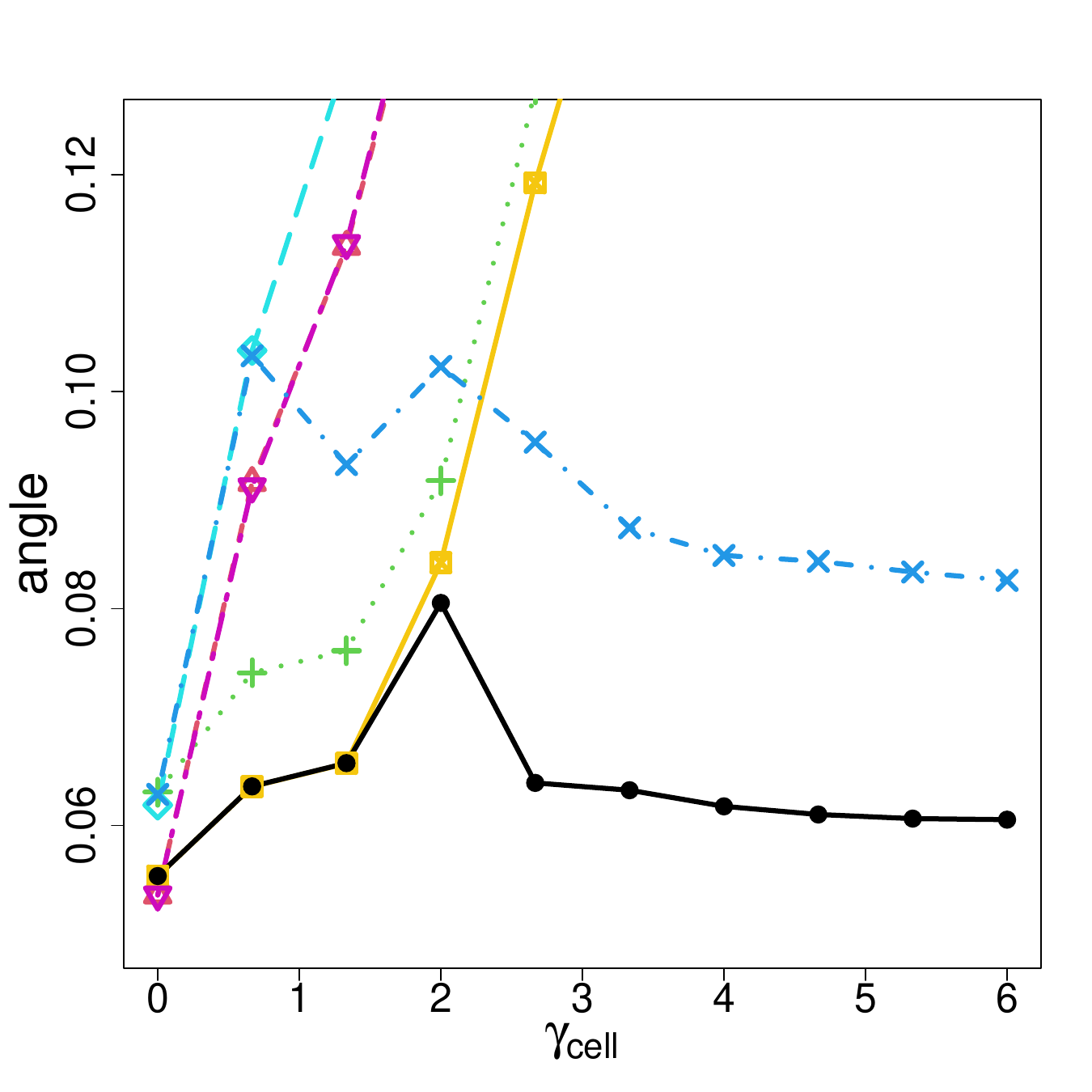}  \\
   [-4mm]
  \includegraphics[width=.3\textwidth]
  {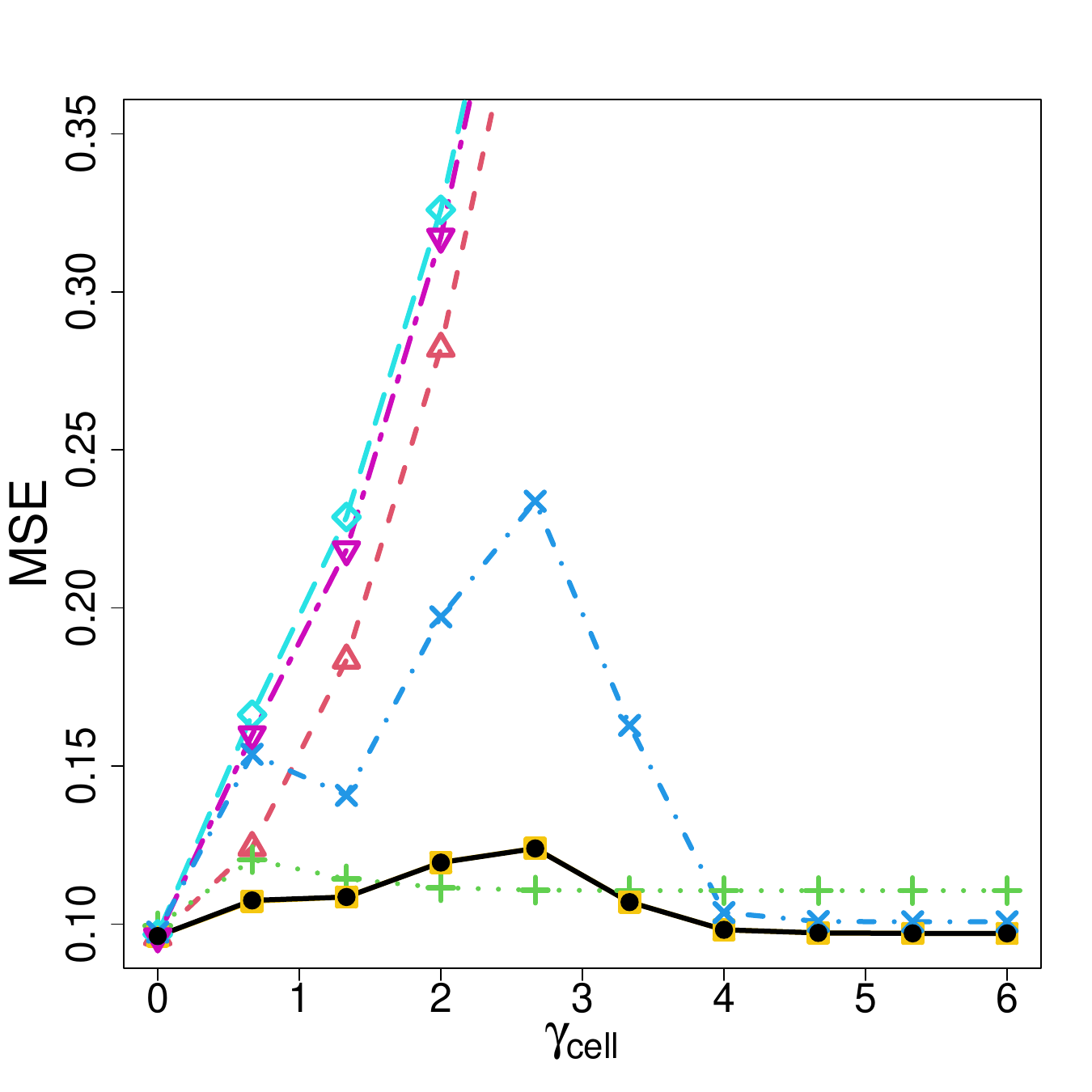} &\includegraphics[width=.3\textwidth]
  {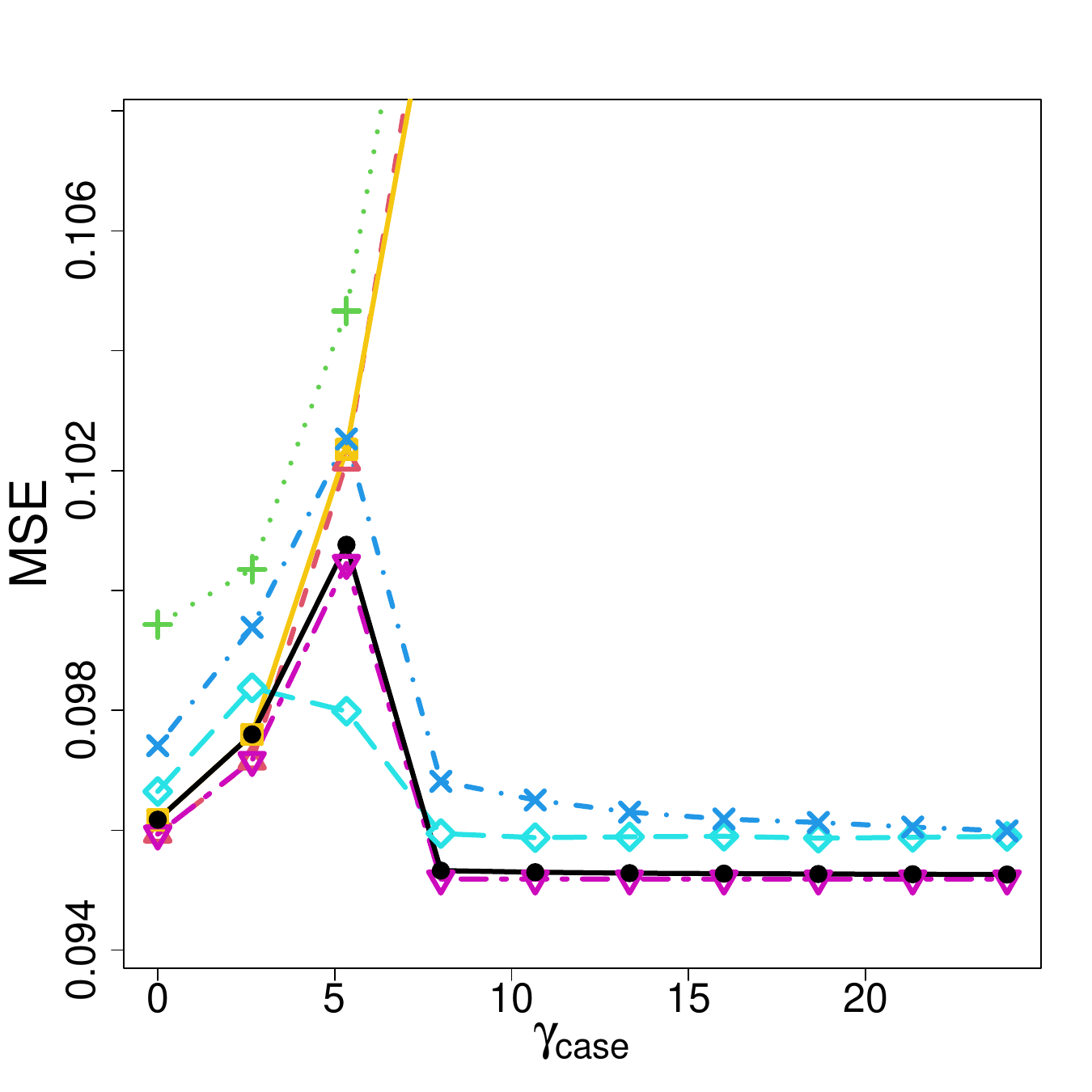} &\includegraphics[width=.3\textwidth]{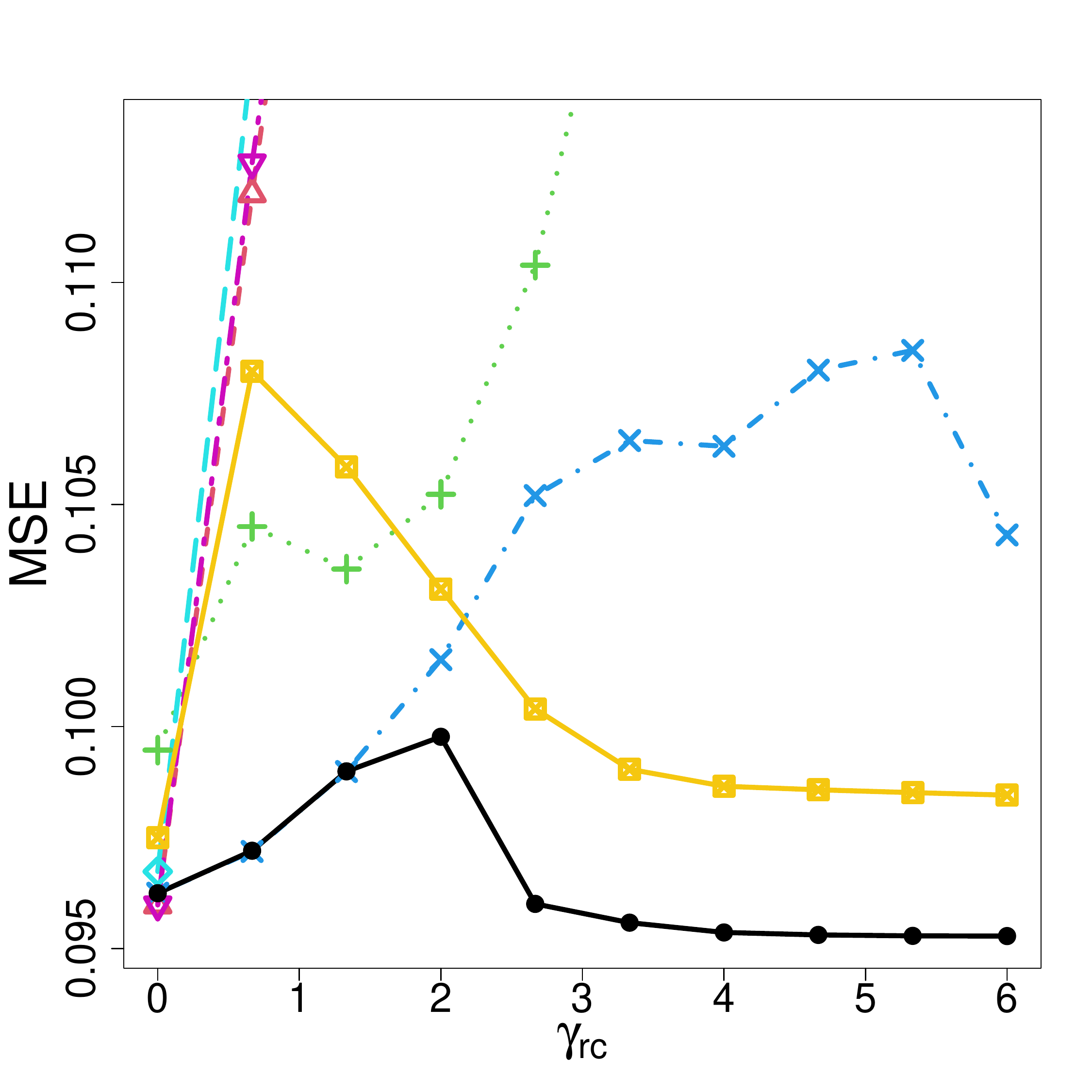} 
\end{tabular}
\caption{Median angle (top) and MSE (bottom) 
attained by CPCA, CANDES, Only-cell, ROBPCA, 
Only-case, MacroPCA, and cellPCA in the presence 
of either cellwise outliers, casewise outliers, 
or both. The covariance model was A09 with 
$n=100$ and $p=200$, without NAs.}
\label{fig:results_p200_NA0}
\end{figure}

Figure~\ref{fig:results_p200_NA0} shows the median 
angle and MSE in the presence of either cellwise
outliers, casewise outliers, or both, for $p=200$.
The plots for $p=20$ are very similar, see 
Figure~\ref{fig:results_p20_NA0} in
Section~\ref{app:addsim} 
of the Supplementary material.
As expected, CPCA did best on clean data ($\gamma=0$),
but outperformed cellPCA by a small margin only.
The results with outliers are more interesting.
When there are only cellwise outliers, CPCA, ROBPCA, 
and Only-case break down, because they were not 
designed for cellwise outliers. CANDES did well
in this setting. Here cellPCA did best, and Only-cell
almost coincided with it. 

In the presence of casewise outliers, CPCA, CANDES and 
Only-cell break down, because they were not created 
for that situation. Also here cellPCA does well, only
slightly outperformed by Only-case. Note that cellPCA
also outperforms the casewise robust method ROBPCA,
as well as MacroPCA, but the latter do not break down.

When cellwise outliers and casewise outliers are
combined, cellPCA outperforms overall.
It naturally beats the methods that are not robust
to cellwise outliers (CPCA, Only-case, ROBPCA) or
not robust to casewise outliers (CPCA, CANDES,
Only-cell).

In all three settings cellPCA outperforms its 
predecessor MacroPCA, because it minimizes an 
objective function in which the cellwise and 
casewise weights adapt to the data.

To assess the performance in the presence of 
missing data, we repeated the three scenarios
but also randomly set $\eps^{\obs}=20\%$ of 
the cells to NA. 
In this situation we cannot run CANDES or ROBPCA 
which are unable to deal with missing data, and 
for CPCA we use the ICPCA method of 
\cite{kiers1997weighted} that can.
The resulting Figure~\ref{fig:results_p200_NA0.2}
looks quite similar to 
Figure~\ref{fig:results_p200_NA0}. Again cellPCA
performs best overall, and outperforms MacroPCA.
The remaining methods break down under the
combination of cellwise outliers, casewise
outliers and NAs.

\begin{figure}[!ht]
\centering 
\begin{tabular}{ccc}
   \large \textbf{Cellwise}  & \large \textbf{Casewise} &\large{\textbf{Casewise \& Cellwise}} \\
   [-4mm]
  \includegraphics[width=.3\textwidth]
  {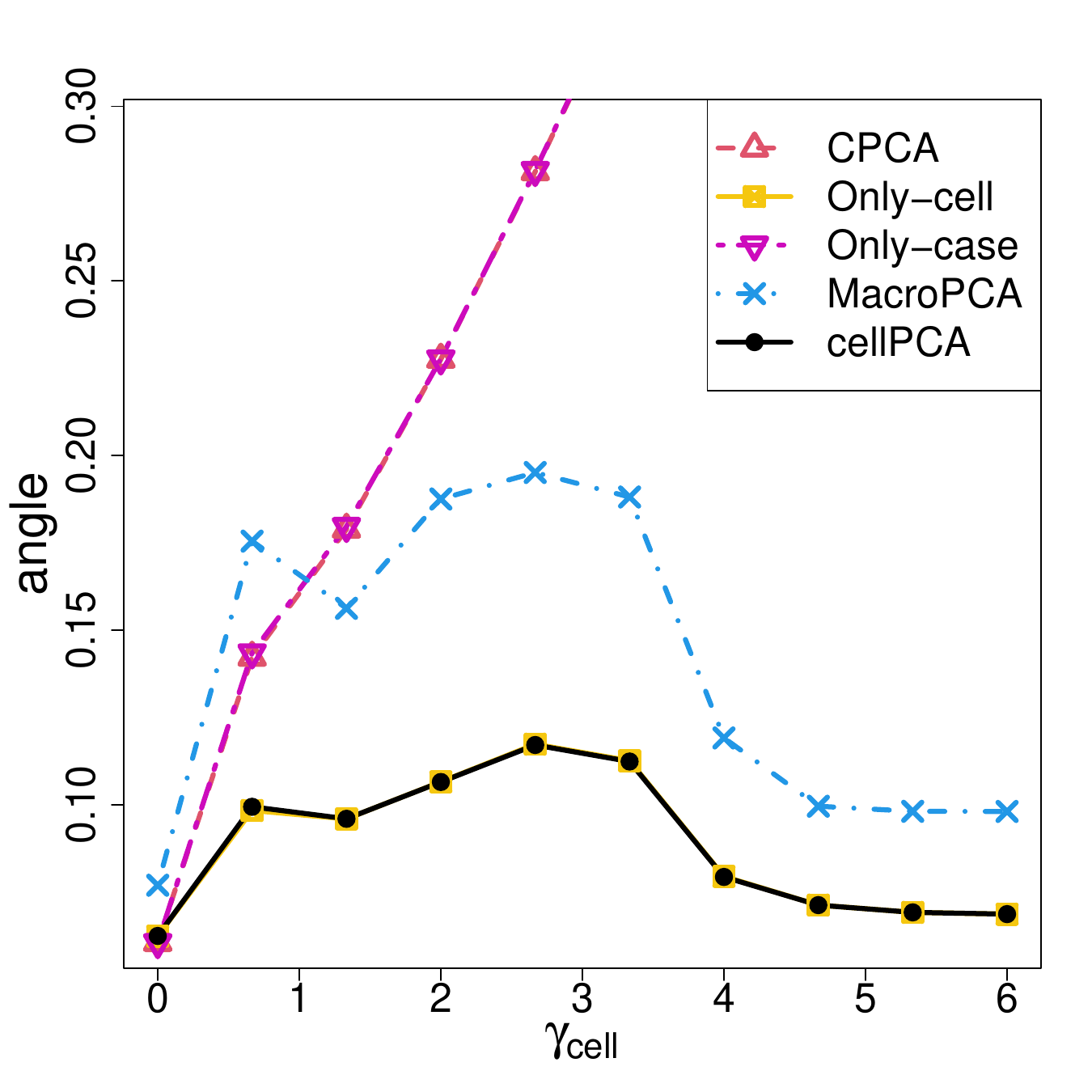} &\includegraphics[width=.3\textwidth]
  {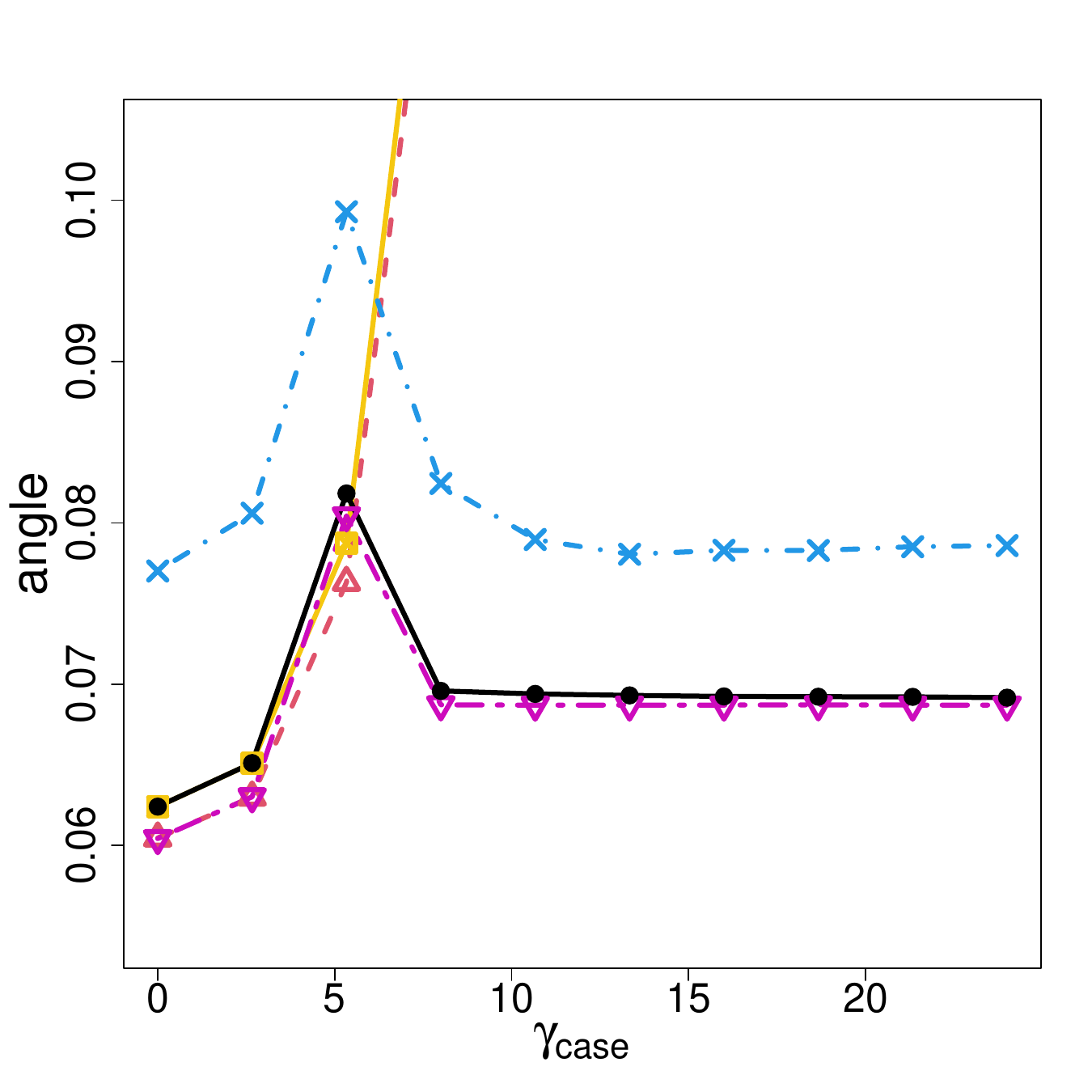} &\includegraphics[width=.3\textwidth]
  {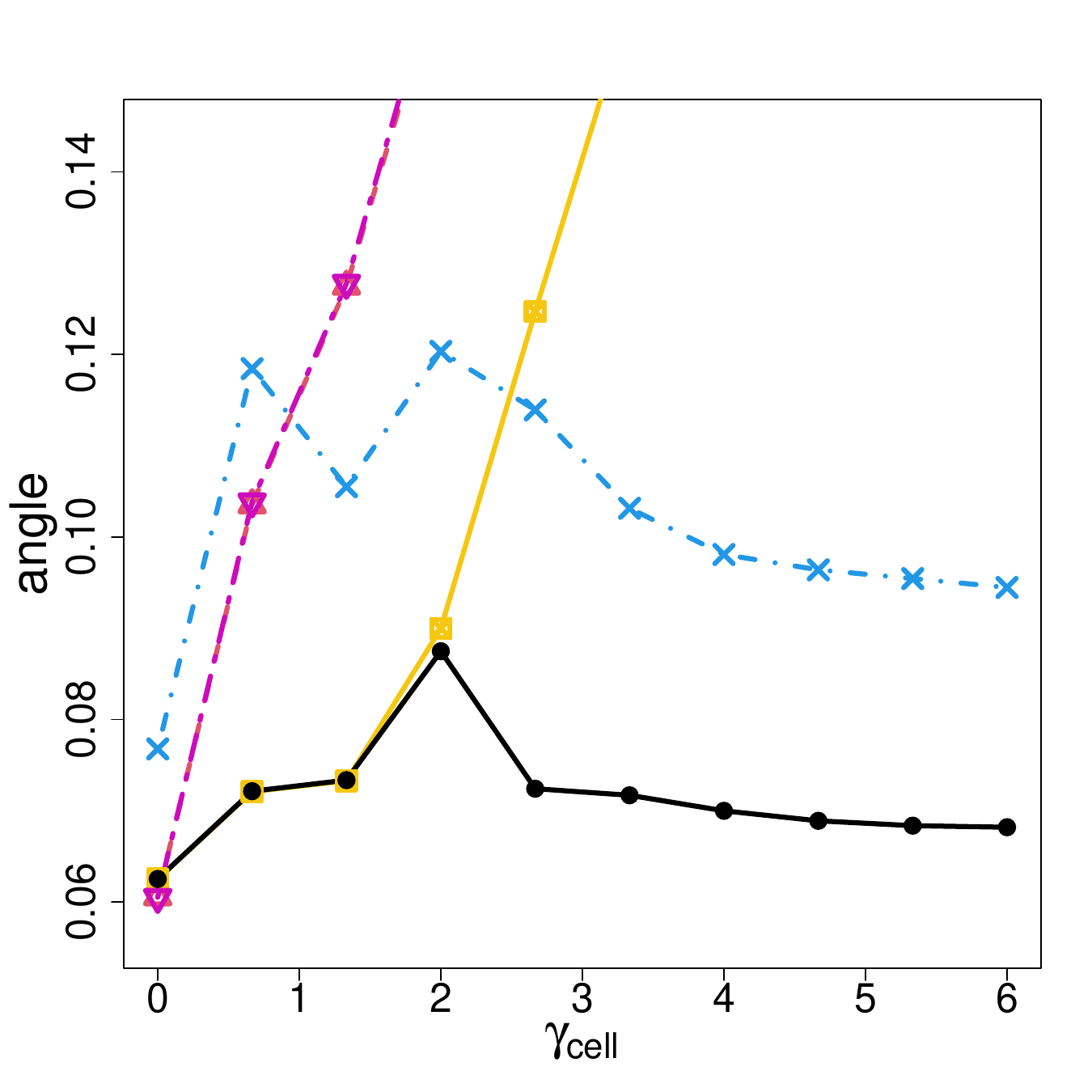}  \\
   [-4mm]
  \includegraphics[width=.3\textwidth]
  {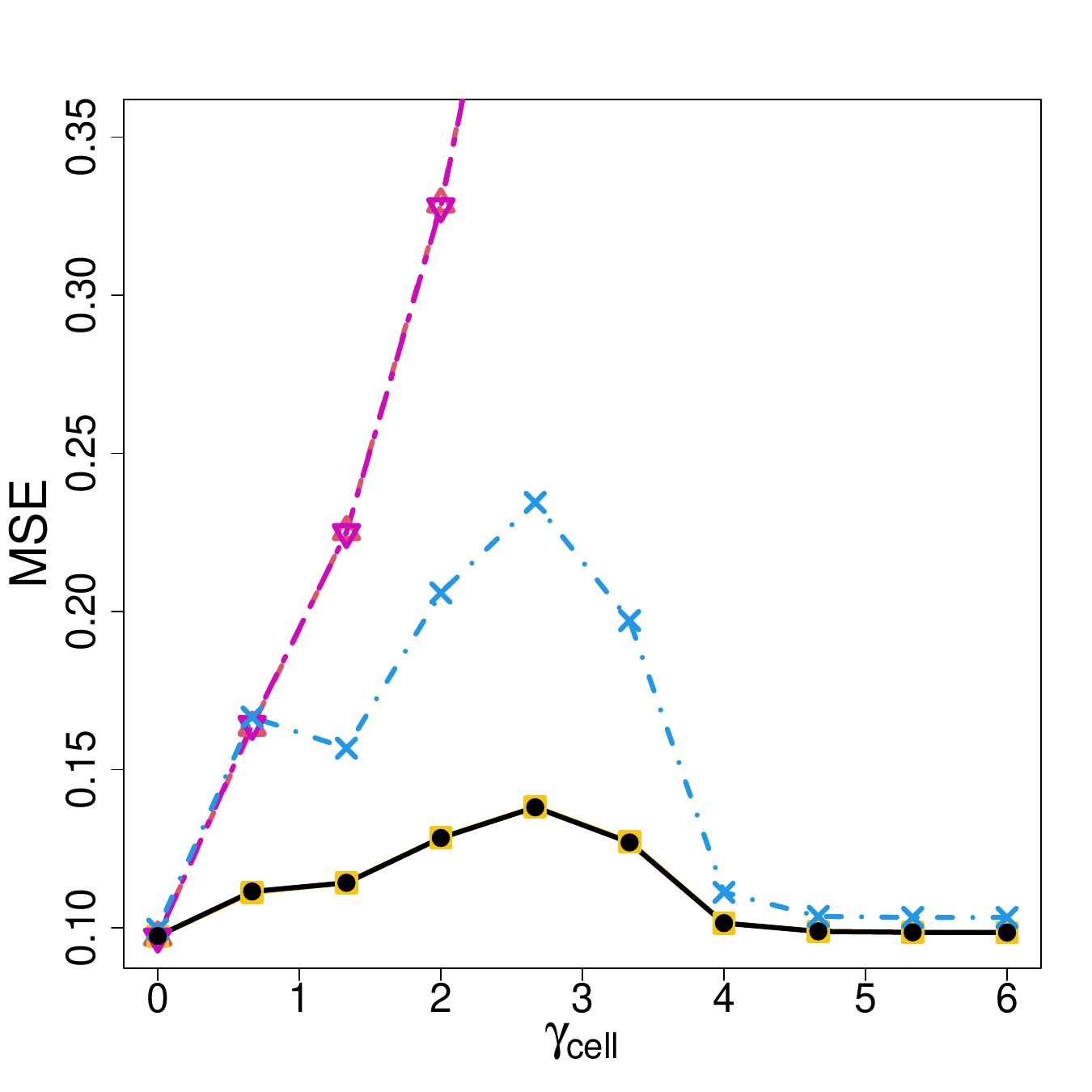} &\includegraphics[width=.3\textwidth]
  {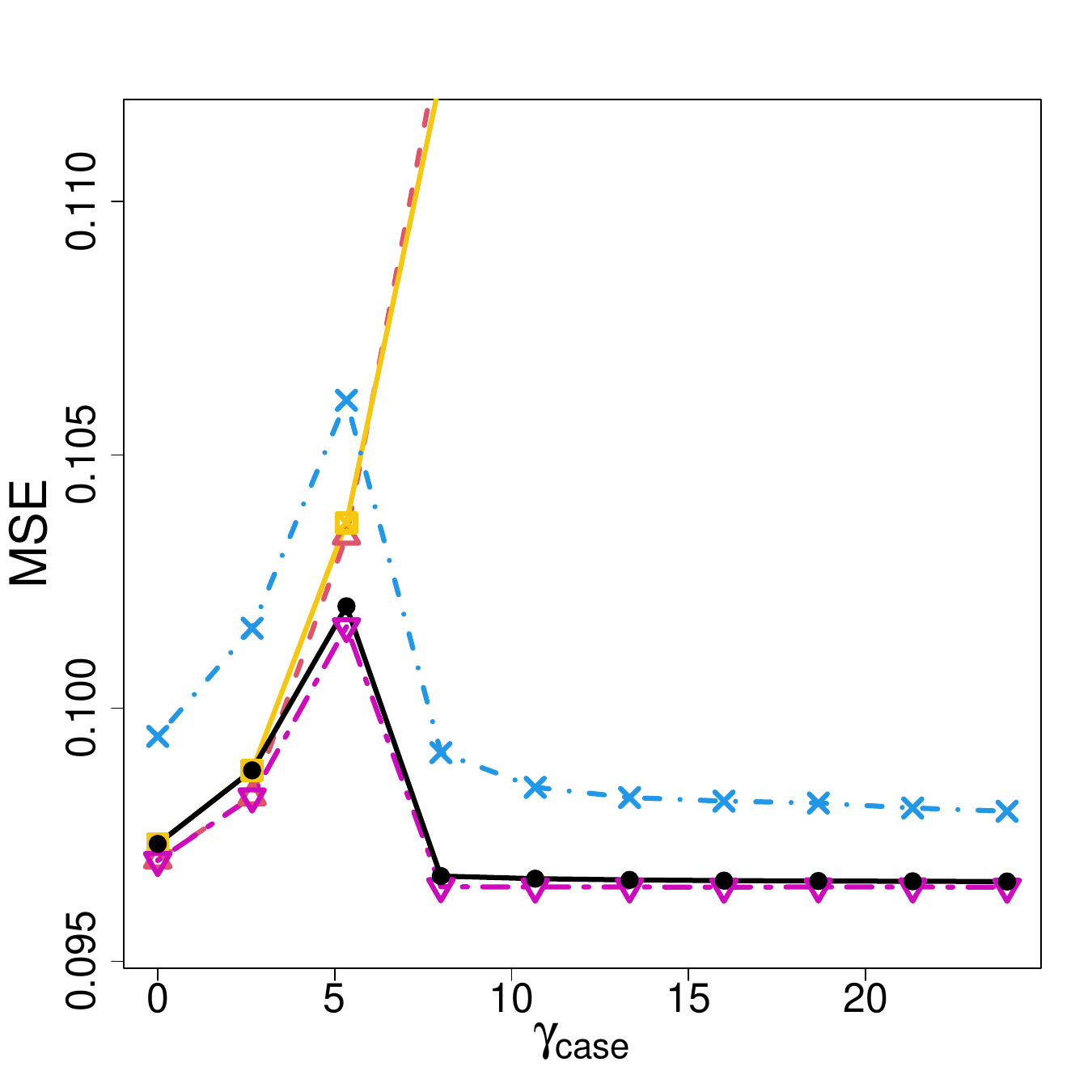} &\includegraphics[width=.3\textwidth]
  {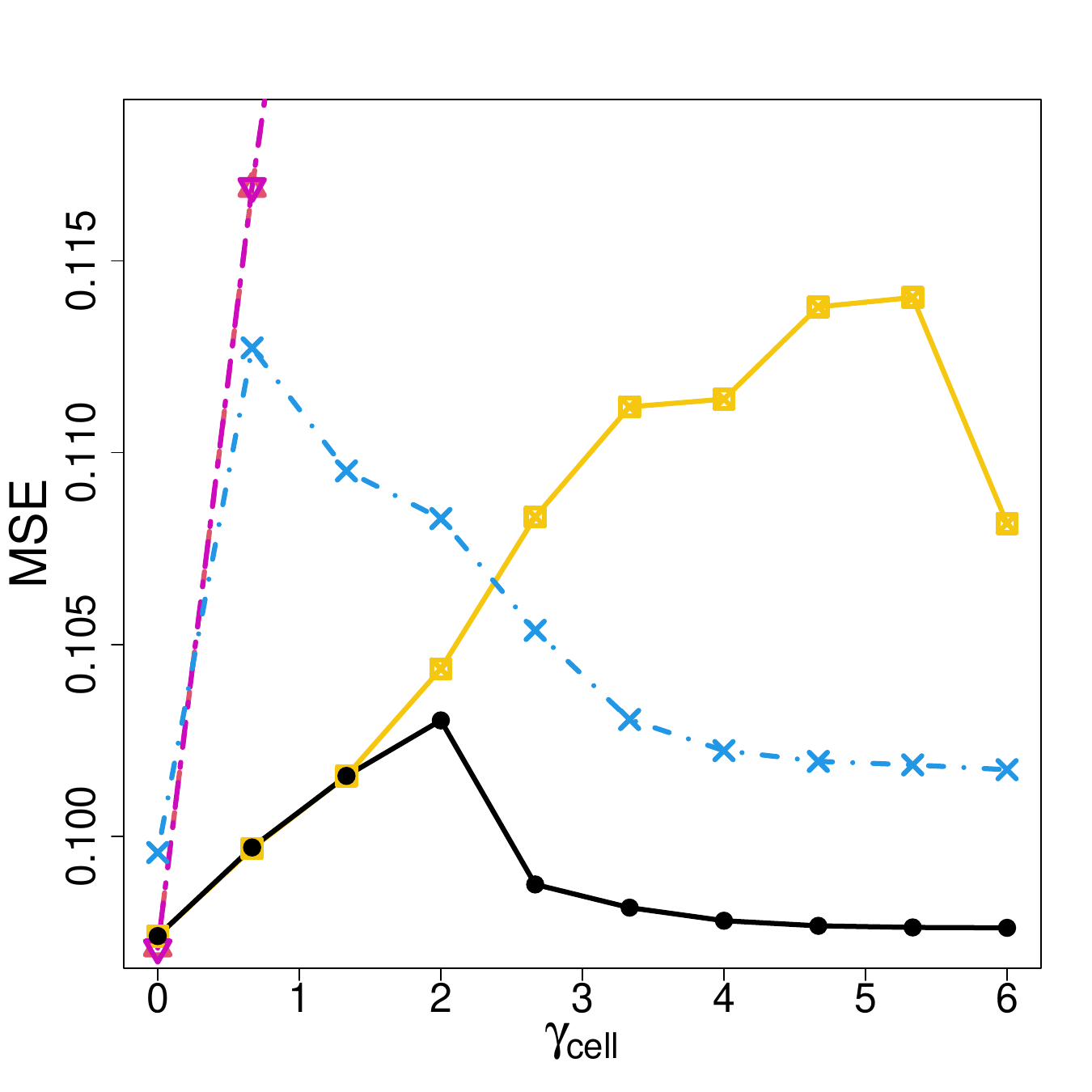} 
\end{tabular}
\caption{Median angle (top) and MSE (bottom) attained 
by CPCA,  Only-cell, Only-case, MacroPCA, and cellPCA 
in the presence of either cellwise outliers, casewise 
outliers, or both. The covariance model was A09 with 
$n=100$ and $p=200$, and $20\%$ of randomly selected 
cells were set to NA.}
\label{fig:results_p200_NA0.2}
\end{figure}

Section~\ref{app:addsim} of the Supplementary material
repeats the entire simulation where the covariance 
matrix A09 is replaced by more diverse random covariance 
matrices based on \cite{agostinelli2015robust}.
The results in Figures~\ref{fig:results_NA0_ALYZ}
and~\ref{fig:results_NA0.2_ALYZ} are very similar to 
those shown here, with the same conclusions.
The performance of the methods is further compared in 
the presence of \textit{structured cellwise outliers} 
as introduced in \cite{Raymaekers:cellHandler}. The 
results are shown in 
Figure~\ref{fig:results_p20_NA0_struc} and 
Figure~\ref{fig:results_p200_NA0_struc} and show that 
cellPCA significantly outperforms its competitors in 
this more challenging setting.
In Figures~\ref{fig:results_varying_pou1} 
and~\ref{fig:results_varying_pou2} we compare the methods 
under varying percentages of contamination, and in
Figures~\ref{fig:results_varying_NA1}--\ref{fig:results_varying_NA3}
under varying percentages of NAs. In these situations 
cellPCA continues to outperform.
Figure~\ref{fig:results_oos_pred} shows that cellPCA 
provides a more accurate prediction of $\bhx$ than
MacroPCA, and Figure~\ref{fig:results_imputation} 
does the same for imputations.
In the simulation summarized
in Figure~\ref{fig:results_rank}, cellPCA typically
found the natural rank $\rk$.

We now investigate how well cellPCA and MacroPCA
detect outliers. For data generated with only cellwise 
outliers, we consider the absolute cellwise residuals 
obtained by both. 
The left panel of Figure~\ref{fig:results_outdet} shows 
the total area under the receiver operating 
characteristic (ROC) curve, referred to as AUC, which 
measures how well the absolute cellwise residual predicts 
whether the cell is outlying. The AUC is a well-known 
measure of the overall performance of a binary 
classification method.
The middle panel of Figure~\ref{fig:results_outdet} 
shows the AUC for the casewise total deviation in 
cellPCA and the orthogonal distance of MacroPCA, which 
does not have the concept of casewise total deviation. 
In the right panel the data are contaminated by both 
cellwise and casewise outliers, and then the AUCs for 
both types of outliers are averaged. In all three 
panels we see that cellPCA outperforms MacroPCA, 
especially when both types of outliers occur together.

\vspace{5mm}
\begin{figure}[!ht]
\centering
\begin{tabular}{ccc}
   \large \textbf{Cellwise}& \large \textbf{Casewise} &\large{\textbf{Casewise \& Cellwise}} \\
   [-4mm]
  \includegraphics[width=.3\textwidth]
  {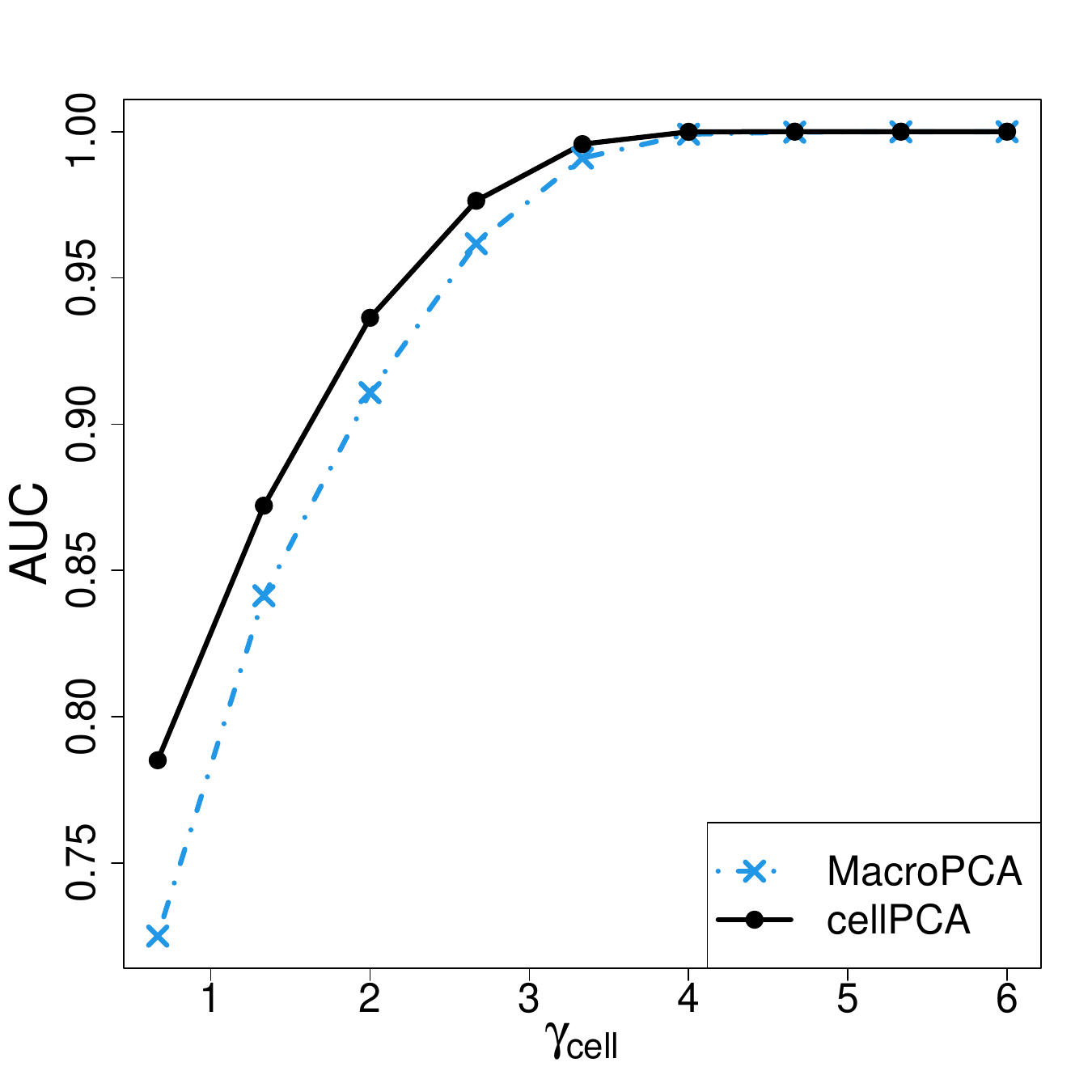} &\includegraphics[width=.3\textwidth]
  {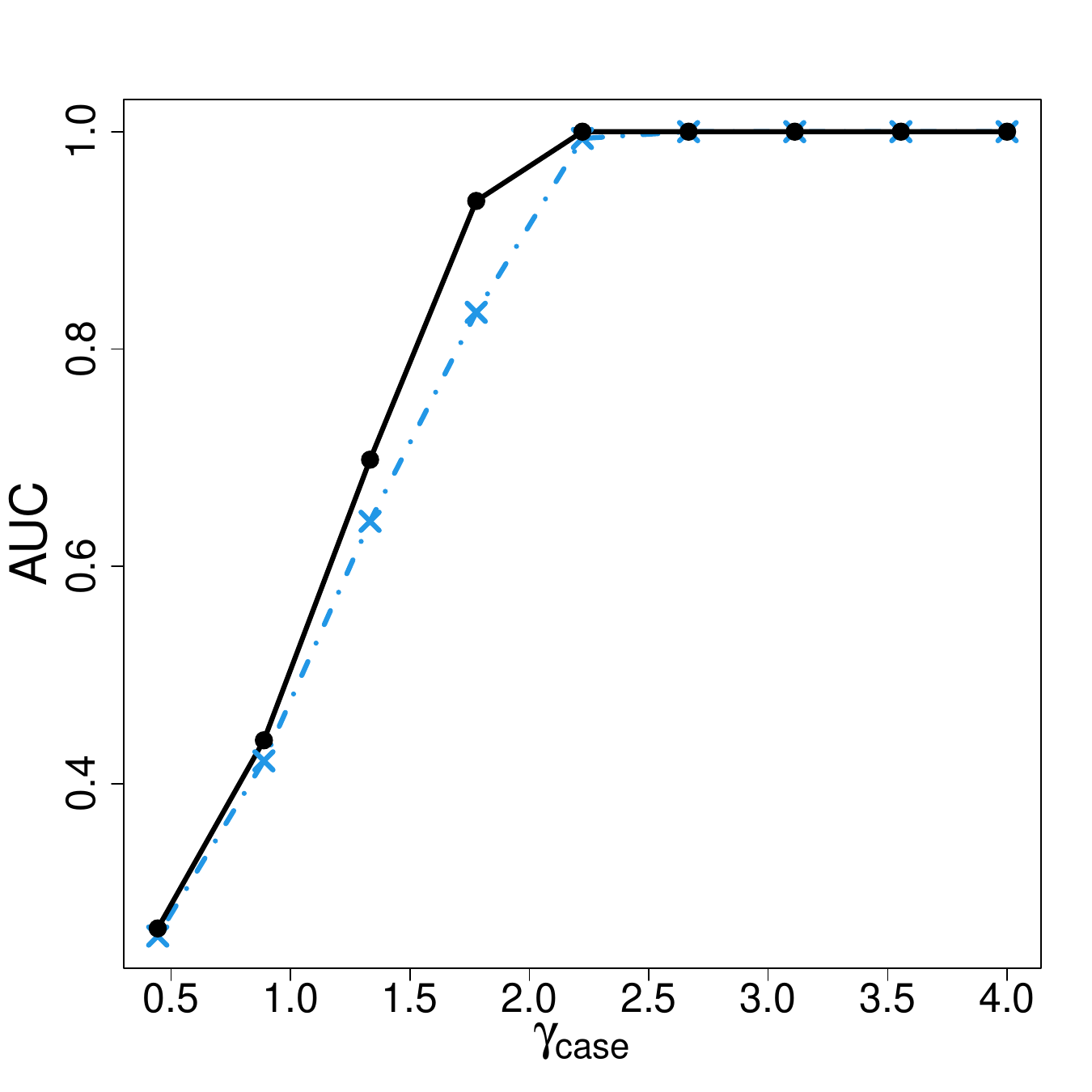} &\includegraphics[width=.3\textwidth]
  {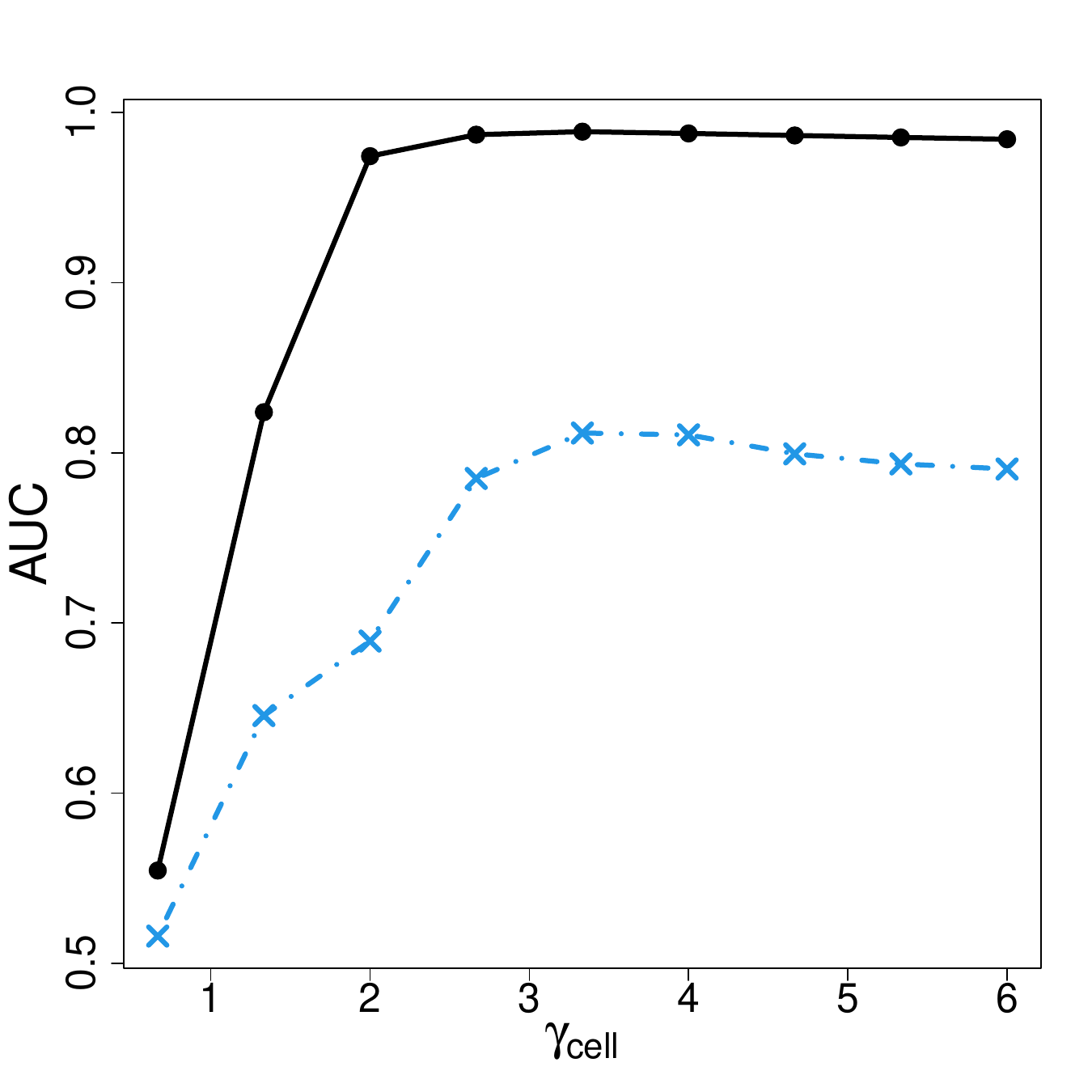} 
\end{tabular}
\caption{Median AUC of outlier detection by MacroPCA 
and cellPCA for either cellwise outliers, casewise outliers, 
or both. The covariance model was A09 with $n=100$ and 
$p=20$.} 
\label{fig:results_outdet}
\end{figure}

\section{Real data example}
\label{sec:realdata}
Campi Flegrei is an active volcanic field partly underlying
the city of Naples, Italy. It is monitored from six permanent 
ground stations, one of which is located at the Vesuvius 
crater \citep{sansivero2022ground}. They record thermal infrared 
(TIR) images to investigate volcanic plumes and gases, lava 
flows, lava lakes, and fumaroles, which are vents of hot gas.
The goal is to track surface thermal anomalies that may reveal 
a renewal of eruptive activity \citep{vilardo2015long}.
The Solfatara data are available at 
\url{https://figshare.com/s/82bcfb64d5130712aeef}\,.
They consist of TIR images of $200 \times 200$ pixels acquired 
from May to November 2022 by the remote station of Solfatara~1. 
Vectorizing each frame yields an ultra-high dimensional 
data matrix with $n=205$ and $p=40,000$. 
We have applied cellPCA, with $\rk=2$ obtained
from Section~\ref{sec:rank}.

\begin{figure}[ht]
\centering
\vspace{2mm}
\includegraphics[width=0.7\textwidth]
   {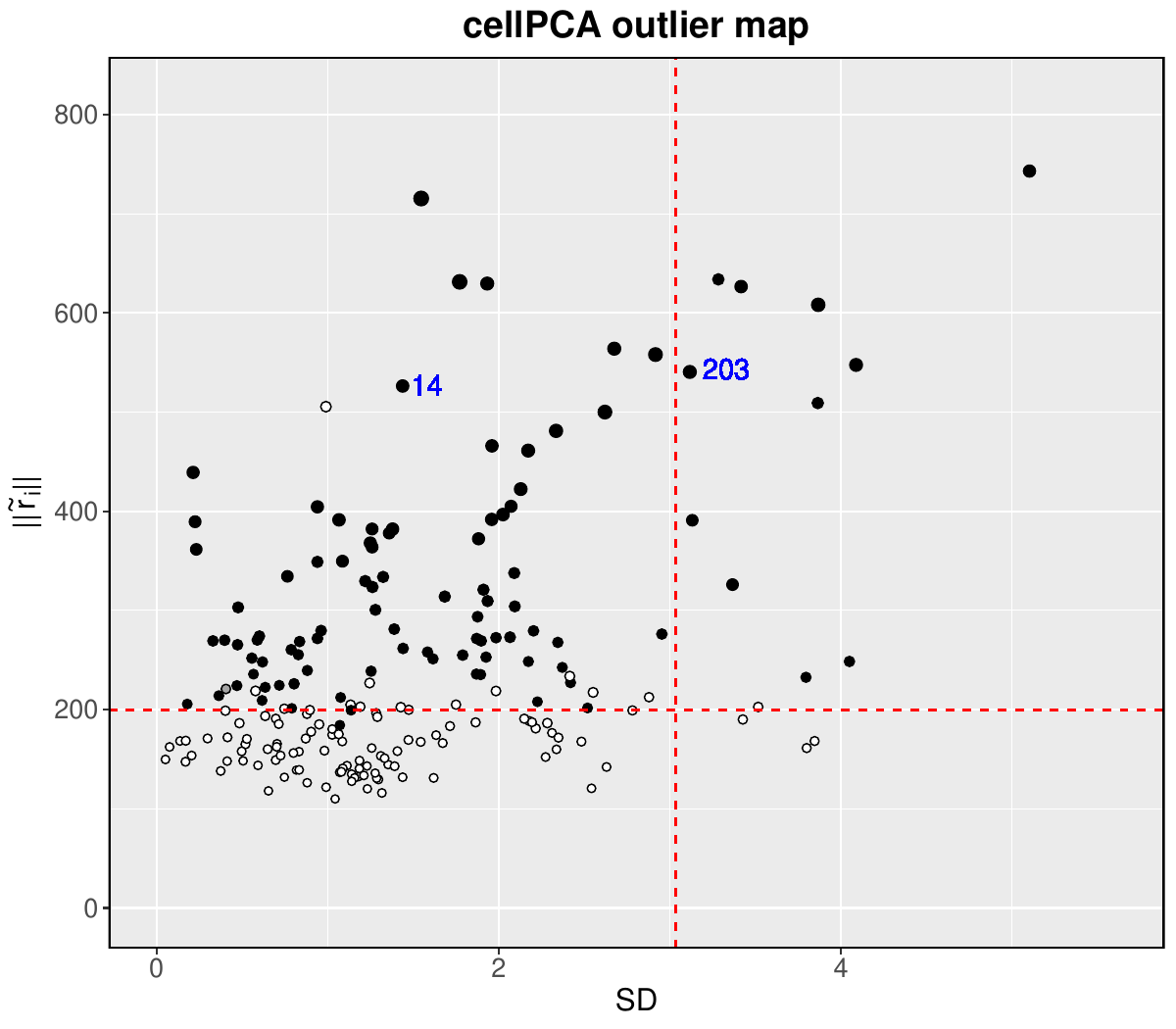}\\
\vspace{-4mm}
\caption{Enhanced outlier map of the Solfatara data.}
\label{fig_solfa_outliermap}
\end{figure}

Figure~\ref{fig_solfa_outliermap} shows the resulting
enhanced outlier map. We see that many residuals have 
$||\btr_i||$ far above the horizontal cutoff line, and
there are quite many casewise outliers (black points).
The size of these points indicates that many of their
cells have low weight. By way of illustration we look
at one of them, case 14.

\begin{figure}[!ht]
\centering
\includegraphics[width=0.49\textwidth]
  {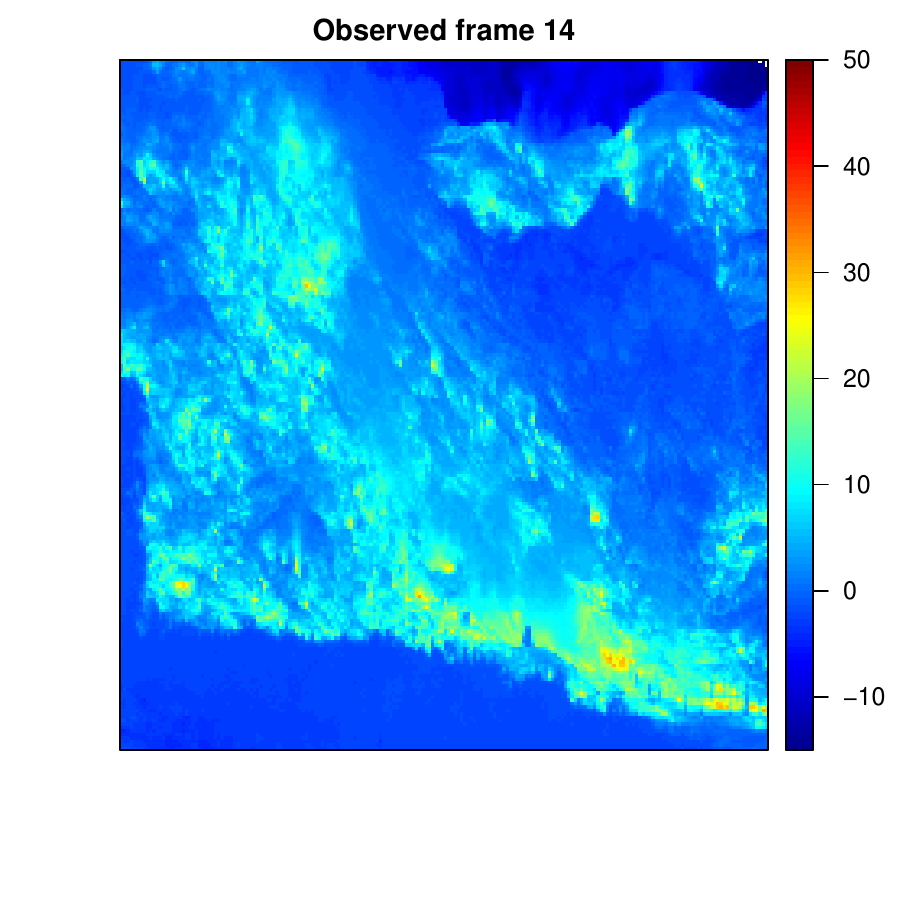}
\includegraphics[width=0.49\textwidth]
  {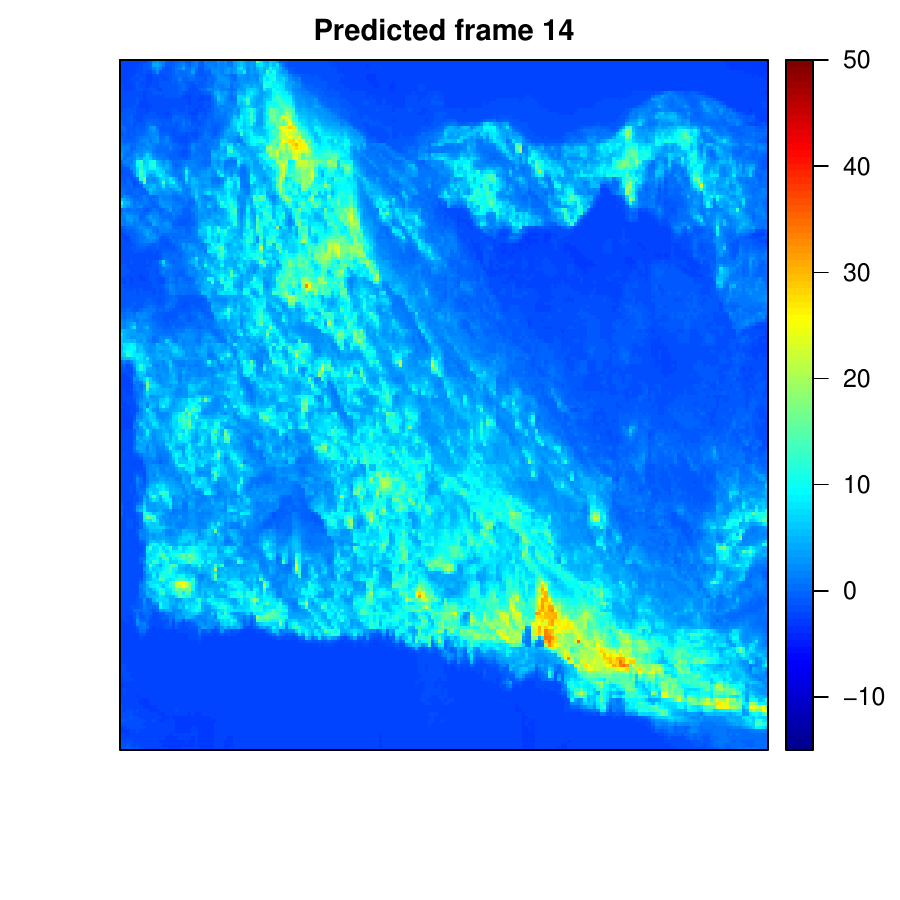}\\

\vspace{-7mm}
\includegraphics[width=0.5\textwidth]
  {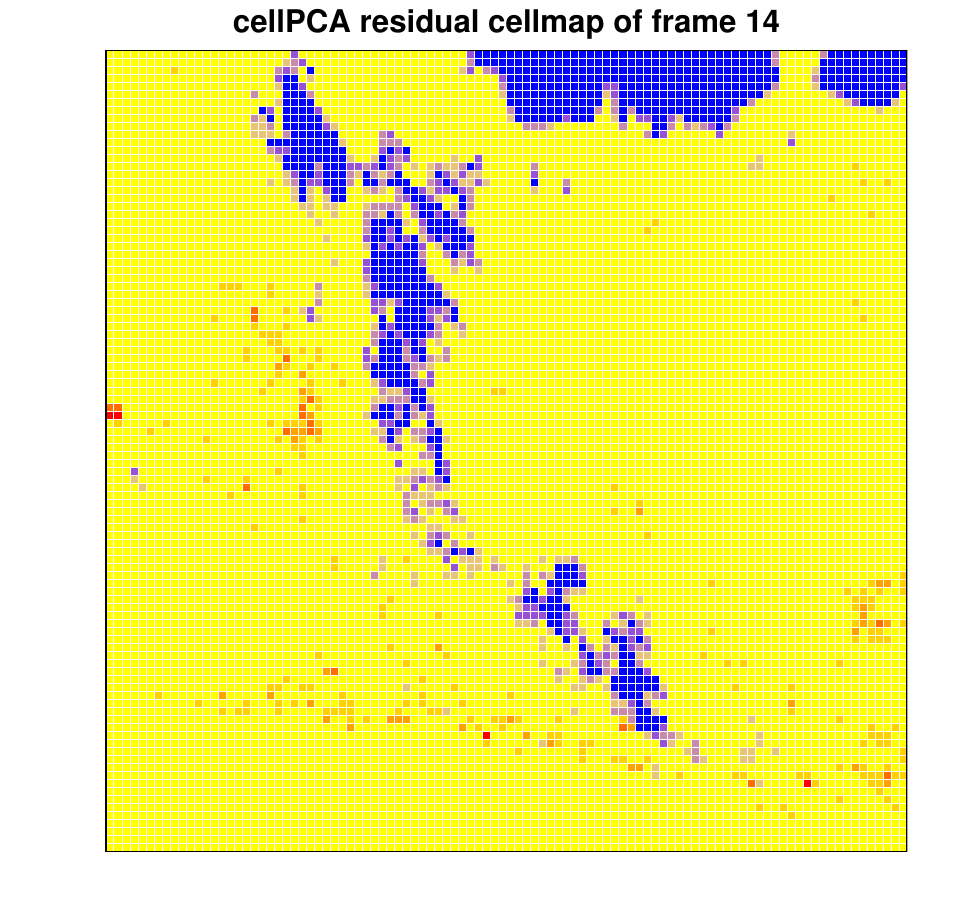}\\
\caption{Solfatara data: (top left) observed 
frame 14, (top right) its prediction, and
(bottom) its residual cellmap.}
\label{fig_frame14}
\end{figure}

Figure~\ref{fig_frame14} shows some results for frame
14. The top left panel is the observed frame, with
cooler regions in blue and warmer regions in yellow
to red. The predicted frame is slightly different,
and overall a bit less cool. The standardized residuals
are shown in the bottom panel. This is the part of the
residual cellmap of the Solfatara data belonging to
frame 14. The entire residual cellmap is much bigger,
and has a row with 40000 cells for case 14. These
cells are more easily visualized in this 200 by 200 
square form corresponding to frame 14 itself.
The blue region in the cellmap indicates where the
observed temperature was lower than expected. It 
points to the condensation of hot water vapor in 
plumes from the volcanic fumaroles, which partly
hides the heat underneath. This behavior is not visible
in the raw observed frame, but it deviates from the
overall linear relations described by the principal
subspace.

The outlier map of Figure~\ref{fig_solfa_outliermap}
shows where case 14 is located, and the results for
case 203 are shown in section~\ref{app:addreal} of 
the Supplementary Material. Inspection showed that 
the casewise outliers in the outlier map were mainly 
among the first 23 and the last 65 cases, that 
correspond to TIR frames acquired in May, October, 
and November. By visually examining the frames we 
observed that the majority of the anomalous images 
exhibit temperature patterns distinctly different 
from those of the regular frames. As in frame 14
the residual frames often show extensive blurred 
regions, which have been attributed to the condensation 
of water vapor. This effect is most pronounced 
during the winter season, due to higher air humidity 
levels. 

Moreover, we assess the performance of cellPCA 
to deal with missing data by replacing a random 
subset of $20 \%$ of the cells with NAs. Then 
the median absolute error (MAE) is calculated, 
i.e., the median absolute difference between 
the known values of the missing cells and the 
corresponding predicted values. This is repeated 
50 times for cellPCA, MacroPCA,  and CPCA, and 
results are shown in 
Figure~\ref{fig_box_solfa_NA}. The proposed 
method clearly outperforms the competing ones 
in prediction accuracy. 

\begin{figure}[t]
\centering
\includegraphics[width=0.4\textwidth]
{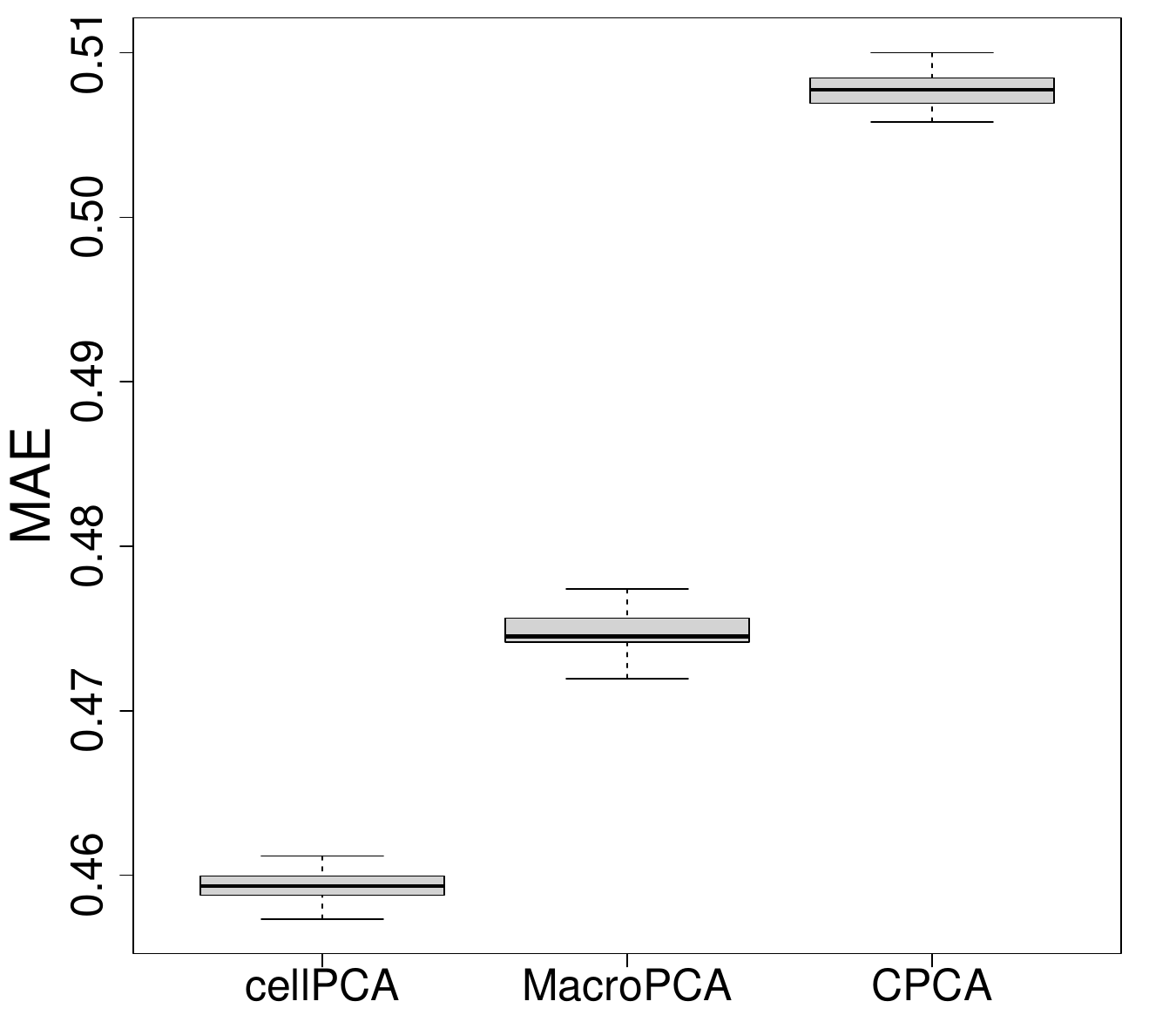}
\caption{Boxplots of the MAE for cellPCA, MacroPCA, 
and CPCA that are obtained by replacing a random 
subset of $20 \%$ of the Solfatara data cells 
with NAs.}
\label{fig_box_solfa_NA}
\end{figure}

\section{Conclusions}
\label{sec:conc}

We have introduced a new contamination model 
that includes casewise outliers, cellwise outliers 
and missing values simultaneously. We proposed the 
cellPCA method, that can handle data generated by
this model. The main novelty of this method is
that it minimizes a single objective function. Its 
algorithm assigns a weight to each cell in the data, 
as well as to each case. The unifying objective
function allowed us to derive both the cellwise and the
casewise influence function of the projection matrix,
as well as the asymptotic distribution of the latter.
The cellwise and casewise weights allowed us to
better visualize outliers in residual cellmaps and 
outlier maps. The method also provides
imputed data that can be used in further analyses. 
The performance of cellPCA was showcased in a 
simulation study, and the method was illustrated on 
interesting datasets.

This work opens several directions for further
research. One of those is an extension to principal
component regression, and related methods
like partial least squares that are often used in
chemometrics. Another avenue is tensor data.

\noindent{\bf Software availability.} \textsf{R} 
code for the proposed method and a script that 
reproduces the examples is available at
\mbox{\url{https://wis.kuleuven.be/statdatascience/robust}}\,.
The rather large Solfatara dataset that was analyzed in Section 
\ref{sec:realdata} can be downloaded from\linebreak 
\url{https://figshare.com/s/82bcfb64d5130712aeef}\,.\\


\spacingset{1}


\clearpage
\pagenumbering{arabic}
\appendix
\begin{center}
\phantom{abc}\\

\Large{Supplementary Material to: 
  Robust Principal Components\\
           by Casewise and Cellwise Weighting}\\
\end{center}

\setcounter{equation}{0} 
\renewcommand{\theequation}
  {A.\arabic{equation}} 

\spacingset{1.45} 
\section{The masking effect in the diagnostic approach}
\label{sec:masking}
To illustrate how diagnostic methods for handling 
outliers can be affected by the masking problem, we 
consider the \texttt{octane} dataset available in the 
\texttt{R} package \texttt{rrcov} \citep{rrcov}. It 
contains near infrared (NIR) absorbance spectra of 
$n = 39$ gasoline samples over $p = 226$ wavelengths 
ranging from 1102 nm to 1552 nm, with measurements 
every two nanometers. The six gasoline samples 
25, 26, and 36--39 contain added ethanol, resulting 
in highly deviating absorbance values in the last 80 
wavelengths, as can be seen in 
Figure~\ref{fig_octane_raw}.

\begin{figure}[ht]
\centering
\includegraphics[width=0.45\textwidth]
   {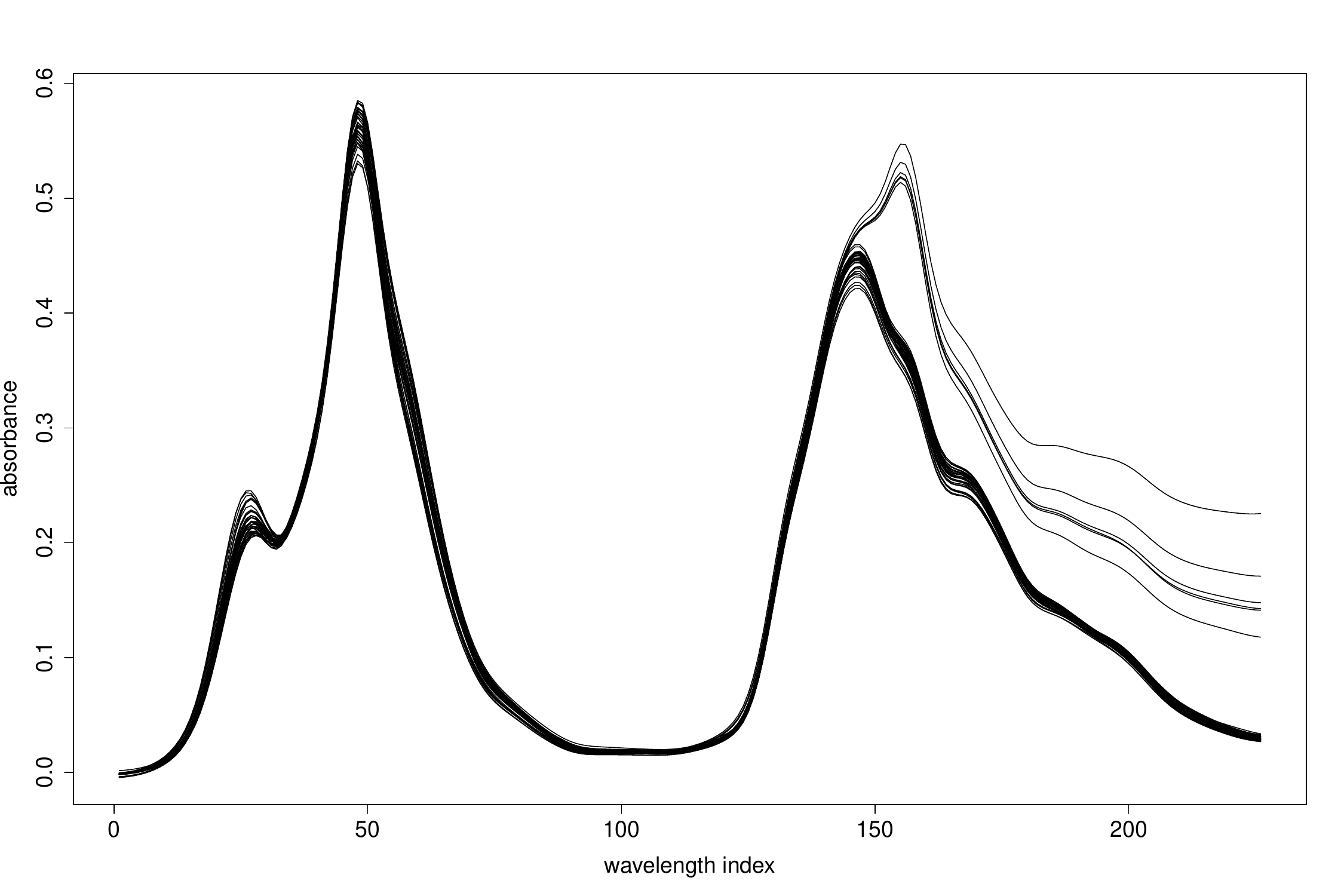}\\
\vspace{-2mm}
\caption{Raw absorbance spectra of the octane data.}
\label{fig_octane_raw}
\end{figure}

The dimensionality of this dataset 
can be reduced by PCA. To account for the possible 
presence of outliers, a typical diagnostic approach 
repeatedly applies classical PCA to the data, and flags 
and removes anomalies by identifying cases with either
orthogonal distance (OD) or score distance (SD) 
exceeding specific thresholds. The OD measures how far 
a point lies from the PCA estimated subspace, whereas 
the SD is the Mahalanobis distance between the projection 
of the point onto the subspace and the center. 
Figure~\ref{fig_octane_outliermap_CPCA} displays the outlier maps, which plot the OD against the SD with the corresponding thresholds, obtained using the diagnostic approach. 
After two iterations, no additional observations are flagged, and observations 18, 26, 32 and 38 are identified as outliers. Therefore, this method fails to detect four of the known outliers, thereby suffering from the masking effect. Moreover, some regular observations are incorrectly flagged as outliers, so this diagnostic approach also suffers from the swamping effect.

\begin{figure}[ht]
\centering
\vspace{3mm}
\includegraphics[width=0.31\textwidth]
   {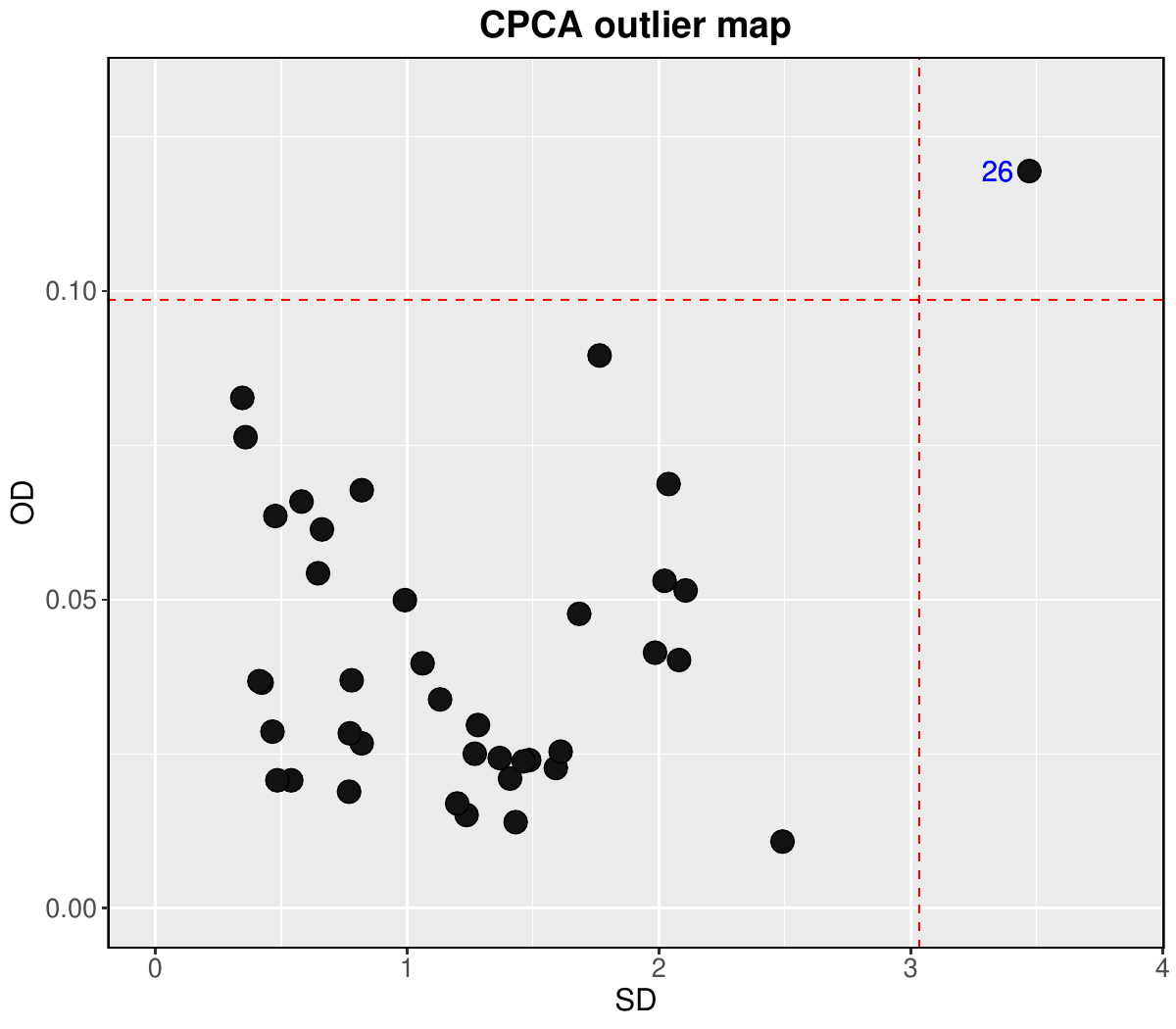}
\includegraphics[width=0.31\textwidth]
   {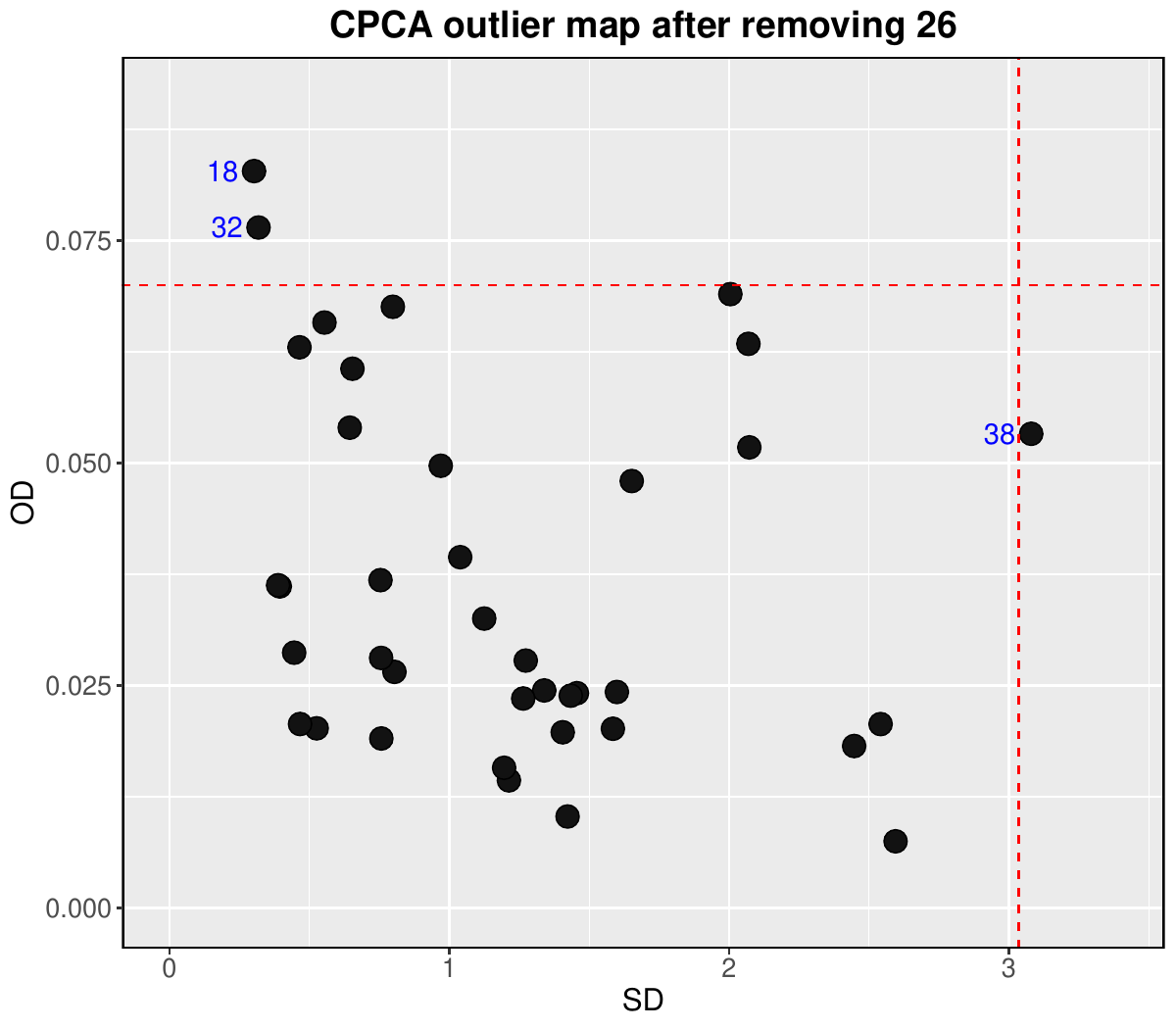}
\includegraphics[width=0.31\textwidth]
   {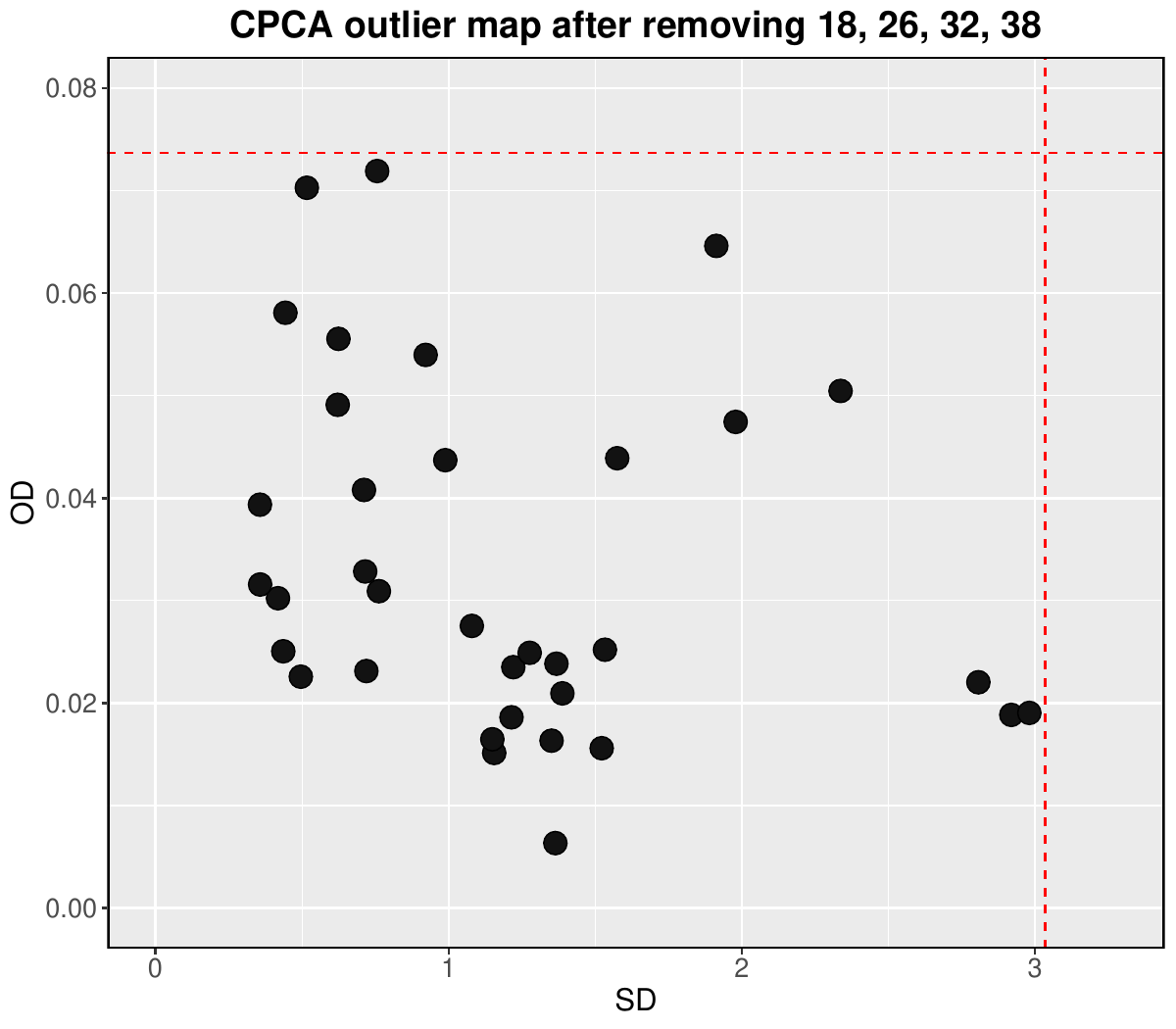}\\
\vspace{-2mm}
\caption{Outlier maps of the octane data 
obtained with the diagnostic approach.}
\label{fig_octane_outliermap_CPCA}
\end{figure}

Figure~\ref{fig_octane_outliermap_cellPCA} shows 
the enhanced outlier map obtained by our proposed 
cellPCA method. Similarly to the classical outlier map, this plot displays, for each observation, the norm of the standardized residuals $\|\btr_i\|$ against its SD, with the corresponding thresholds indicated by horizontal and vertical lines. Additional details about the enhanced outlier map are provided in Section~\ref{sec:ionosphere}.
All the outliers are clearly identified and have no impact on the final estimates, demonstrating the advantage of the robust approach over the diagnostic approach.
\begin{figure}[!ht]
\centering
\vspace{2mm}
\includegraphics[width=0.45\textwidth]
   {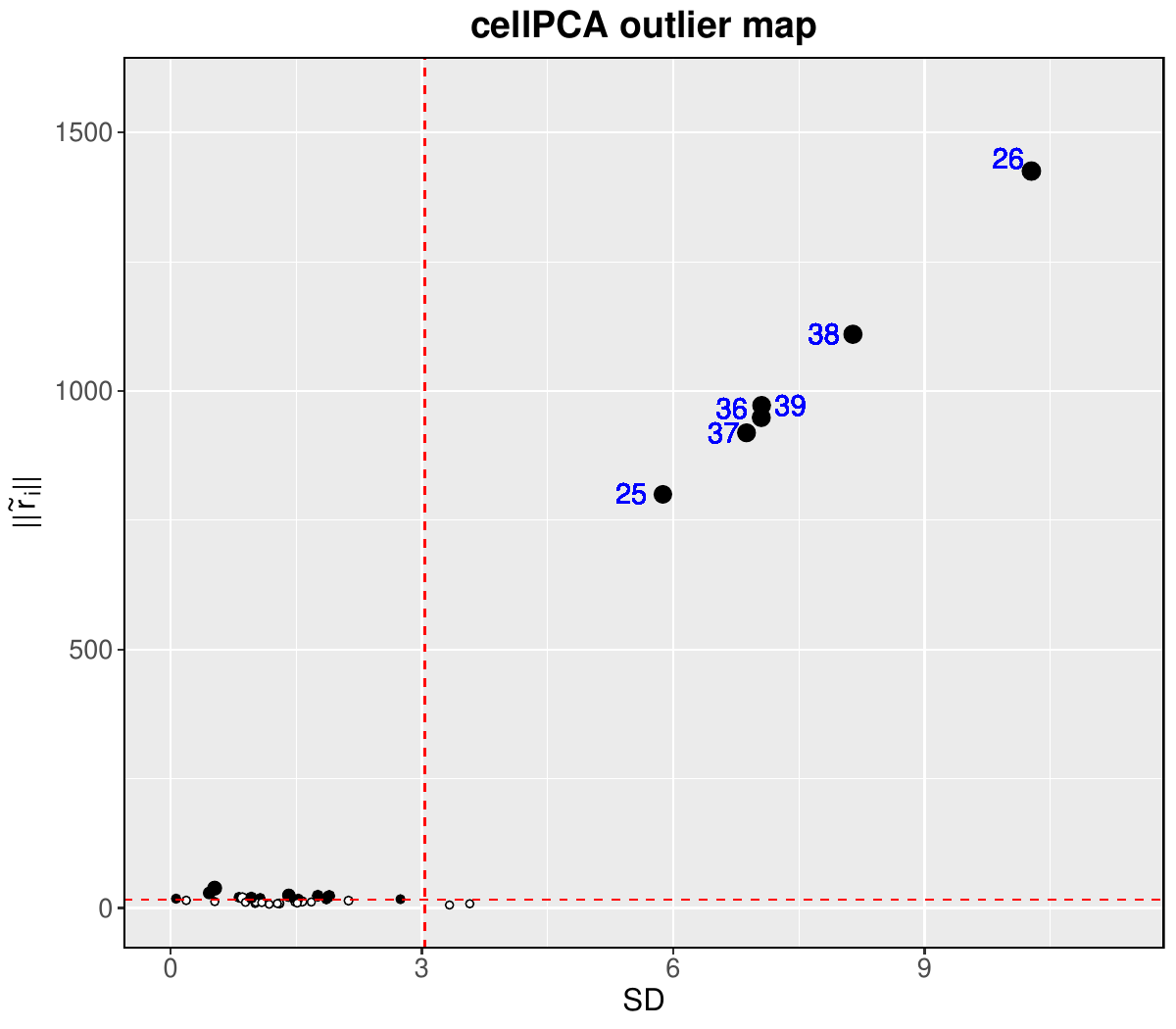}\\
\vspace{-2mm}
\caption{cellPCA enhanced outlier map of the 
octane data.}
\label{fig_octane_outliermap_cellPCA}
\end{figure}

\section{Equivalent parametrizations of $\bV$}
\label{sec:parametrizations}

The principal subspace is of the form
$\bPi = \bPi_0 + \bmu$ where $\bP_0$ is a linear
subspace of dimension $\rk$ which is characterized
by the $p \times p$ projection matrix $\bP$.
We saw in Section~\ref{sec:PCAmodel} that 
$\bP$ can be written as $\bP =\bV\bV^T$ where the 
$p \times \rk$ matrix $\bV$ has orthonormal 
columns, which form an orthonormal basis of 
$\bPi_0$. 

The only disadvantage of working with $\bV$
is that, unlike $\bP$, it is not unique.
Indeed, we could take any other orthonormal 
basis of $\bPi_0$ corresponding to a matrix 
$\btV$, and then $\btV \btV^T = \bP$ as well.
But then there must be a nonsingular 
$\rk \times \rk$ matrix $\bA$ such that 
$\btV = \bV\bA$. Left multiplying by $\bV^T$
yields $\bV^T\btV = \bV^T\bV\bA = \bA$.
We can also go in the opposite direction by
$\bV = \btV\bA^{-1}$ and then left multiply
by $\btV^T$ yielding $\btV^T\bV = 
\btV^T\btV\bA^{-1} = \bA^{-1}$. Then 
$\bA^{-1} = \btV^T\bV = (\bV^T\btV)^T = \bA^T$
so $\bA$ is an orthogonal matrix. Therefore
we can write $\btV = \bV\bO$ where $\bO = \bA$
is an orthogonal $\rk \times \rk$ matrix.
This shows that $\bV$ is only determined up to 
right multiplication by an orthogonal matrix.

However, we will see in the sequel that 
whichever parametrization $\bV$ is chosen, the 
final algorithmic and theoretical results on 
$\bP$ are the same.

\section{M-estimation and robustness}
\label{sec:M-estimation}

The M-scale $\hsigma$ given 
by~\eqref{eq:Mscale} as
the solution $\hsigma$ of the equation
\begin{equation*}
  \frac{1}{n} \sum_{i=1}^n \rho\left(
  \frac{z_i}{\sigma}\right)=\delta
\end{equation*}
belongs to the class of M-estimators, see e.g. 
Section 2.5 of \cite{maronna2019robust}.
Here $z_1,\ldots,z_n$ is a univariate dataset
and $\rho(0)$ is zero, $\rho$ is an even 
function in the sense that $\rho(-z)=\rho(z)$, 
and $\rho(z)$ is nondecreasing for 
$z \geqslant 0$.

The robustness of $\hsigma$ depends on 
the boundedness of $\rho$.  
When $\rho$ is unbounded, a single
far outlier $z_1$ can make $\hsigma$
arbitrarily large, which is a form of
breakdown. For instance, the $L^1$ function
$\rho(z) = |z|$ yields the mean deviation
$\hsigma = (n\delta)^{-1}\sum_{i=1}^n |z_i|$
which goes to infinity when $|z_1|$ grows
while $z_2,\ldots,z_n$ remain the same.
The same is true for the $L^2$ function
$\rho(z) = z^2$ yielding the root mean 
square. Moreover, the influence function 
of $\hsigma$ has the shape of $\rho$, so 
$\rho$ needs to be bounded in order for 
the influence function to be bounded.

The number of outliers that $\hsigma$
can withstand depends on the value of
$\delta$. The {\it explosion breakdown
value} of $\hsigma$ is the smallest 
percentage of outliers in the dataset
that can make $\hsigma$ arbitrarily large,
and equals $\delta/\max(\rho)$. On the
other hand, the {\it implosion breakdown 
value} of $\hsigma$ is the smallest 
percentage of outliers that can bring
$\hsigma$ arbitrarily close to zero,
and equals 
$(\max(\rho) - \delta)/\max(\rho)$.
Since we don't want either type of 
breakdown to occur, the best choice is
thus to take $\delta = \max(\rho)/2$.
Then it would take almost 50\% of outliers 
to destroy $\hsigma$, so we are quite
safe.

We will often want our scale 
estimator to be consistent for the scale 
parameter of a Gaussian distribution. In
order to achieve that we include a
consistency factor $a$ 
in~\eqref{eq:Mscale}, yielding
\begin{equation}\label{eq:consistentM}
  \frac{1}{n} \sum_{i=1}^n \rho\left(
  \frac{z_i}{a\,\sigma}\right)=\delta
\end{equation}
which is a special case of~\eqref{eq:Mscale}
in which $\rho$ is rescaled in the 
horizontal direction. We choose $a$ such 
that $E[\rho(z/a)] = \delta$ for 
$z \sim N(0,1)$, the standard Gaussian 
distribution. In the cellPCA algorithm we 
use the M-scale given by the function
$\rho_{1.5,4}$ shown in the top left
panel of Figure~\ref{fig:rho_psi},
for which $\max(\rho) = 3.7622$ so we
put $\delta = 1.8811$ and obtain
$a = 0.3431$\,.

In the cellPCA algorithm of
Section~\ref{sec:algo} both the cellwise
weights~\eqref{eq:cellweight} and the
casewise weights~\eqref{eq:caseweight} use
the derivative of $\rho_{b,c}$ which is
denoted as $\psi_{b,c}$ and shown in the
top right panel of
Figure~\ref{fig:rho_psi}. Since $\rho_{b,c}$
is bounded it follows that $\psi_{b,c}(z)$ 
goes to zero for large positive and large
negative values of $z$. Such functions 
$\psi$ are called redescending. They
were first used in M-estimators of
location for univariate data, given by 
minimizing the objective function
\begin{equation} \label{eq:locM_rho}
  \sum_{i=1}^n \rho\Big(\frac{z_i - \mu}
  {\hsigma}\Big)
\end{equation}
in $\mu$, where we assume that $\hsigma$ 
is given. The first order condition of 
this minimization states that the 
derivative of the objective function 
is zero, so
\begin{equation} \label{eq:locM_psi}
  \sum_{i=1}^n \psi\Big(\frac{z_i - \hmu}
  {\hsigma}\Big) = 0\,.
\end{equation}
It turns out that the influence function 
of $\hmu$ has the shape of $\psi$. This
implies that the effect of a far 
outlier goes to zero, which aids the
robustness of the estimator. Note that
the cellPCA objective~\eqref{eq:objP} 
contains parts similar 
to~\eqref{eq:locM_rho}.
Moreover, univariate M-estimators of
location can be computed by an algorithm
that uses casewise weights based on their
function $\psi$, as described in Section 
2.8.1 of \cite{maronna2019robust}.

\section{First-order Conditions}
\label{app:firstorder}
In the following, the derivation of the first-order
necessary conditions from \eqref{eq:objP} is presented.
Rewriting~\eqref{eq:objP} in terms of the parameters
$\bV$ and $\bU$ yields 
\begin{equation*}
\hspace{-0.7cm}
  L_{\rho_1,\rho_2}\left(\bX,\bV,\bU,\bmu\right) := 
  \frac{\hsigma_2^2}{m} \sum_{i=1}^{n}m_i \rho_2 
  \left(\frac{1}{\hsigma_2}\sqrt{\frac{1}{m_i}
  \sum_{j=1}^{p}m_{ij} \hsigma_{1,j}^2\rho_1\left(\frac{
  x_{ij} - \mu_j - (\bU\bV^T)_{ij}}
  {\hsigma_{1,j}}\right)} \right)\,.
\end{equation*}
We start with the gradient of $L_{\rho_1,\rho_2}$ 
with respect to $\bv_j$ which is 
\begin{align*} 
\hspace{-0.7cm}
\bzero = \frac{\partial}{\partial \bv_j}
    L_{\rho_1,\rho_2} 
  &= \frac{\hsigma_2^2}{m} \sum_{i=1}^{n}
     m_i \rho'_2(\rt_i/\hsigma_2) 
     \frac{1}{\hsigma_2}
     \frac{\partial}{\partial \bv_j} \sqrt{\frac{1}{m_i}
     \sum_{j=1}^{p} m_{ij} \hsigma_{1,j}^2\rho_1\left(
     \frac{x_{ij}-\hx_{ij}}{\hsigma_{1,j}}\right)}
     \nonumber\\
  &= \frac{1}{m} \sum_{i=1}^{n} m_i
     \frac{\psi_2(\rt_i/\hsigma_2)}
     {\rt_i/\hsigma_2}\, \frac{1}{2m_i}
     \frac{\partial}{\partial \bv_j}\left(\sum_{j=1}^{p} 
     m_{ij} \hsigma_{1,j}^2\rho_1\left(\frac{x_{ij}-\hx_{ij}}
     {\hsigma_{1,j}}\right)\right)\nonumber\\
  &= \frac{1}{2m}\sum_{i=1}^{n} w_i^{\case}\, 
     \hsigma_{1,j}^2\, \psi_1(r_{ij}/\hsigma_{1,j})
     \frac{m_{ij}}{\hsigma_{1,j}}
     \frac{\partial}{\partial \bv_j}
     (x_{ij} - \mu_j - \bu_i^T\bv_j)\nonumber\\
  &= \frac{-1}{2m}\sum_{i=1}^{n} w_i^{\case}\, 
     \hsigma_{1,j}\, \frac{\psi_1(r_{ij}/
     \hsigma_{1,j})}{r_{ij}/\hsigma_{1,j}}
     \frac{ m_{ij}}{\hsigma_{1j}}r_{ij}\bu_i\nonumber\\
  &= \frac{-1}{2m}\sum_{i=1}^{n} 
     \bu_i w_i^{\case} w_{ij}^{\cell} m_{ij} 
     (x_{ij}-\mu_j-\bu_i^T\bv_j)\nonumber\\
  &= \frac{-1}{2m}\sum_{i=1}^{n} 
     \bu_i w_{ij} 
     (x_{ij}-\mu_j-\bu_i^T\bv_j)\nonumber\\
  &= \frac{-1}{2m}
     \left(\bU^T\bW^j(\bx^j-\mu_j \bone_n)-
     \bU^T\bW^j\bU\bv_j\right).
\end{align*}
The gradient of $L_{\rho_1,\rho_2}$ with respect to $\bu_i$ is 
\begin{align*} 
\hspace{-0.7cm}
\bzero = \frac{\partial}{\partial \bu_i}
    L_{\rho_1,\rho_2} 
  &= \frac{\hsigma_2^2}{m} m_i 
     \rho'_2(\rt_i/\hsigma_2) \frac{1}{\hsigma_2}
     \frac{\partial}{\partial \bu_i} \sqrt{\frac{1}{m_i}
     \sum_{j=1}^{p} m_{ij} \hsigma_{1,j}^2\rho_1\left(
     \frac{x_{ij}-\hx_{ij}}{\hsigma_{1,j}}\right)}
     \nonumber\\
  &= \frac{1}{2m} 
     \frac{\psi_2(\rt_i/\hsigma_2)}{\rt_i/\hsigma_2}\,
     \frac{\partial}{\partial \bu_i}\left(\sum_{j=1}^{p} 
     m_{ij} \hsigma_{1,j}^2\rho_1\left(\frac{x_{ij}-\hx_{ij}}
     {\hsigma_{1,j}}\right)\right)\nonumber\\
  &= \frac{1}{2m} \sum_{j=1}^{p} w_i^{\case}\, 
     \hsigma_{1,j}^2\, \psi_1(r_{ij}/\hsigma_{1,j})
     \frac{m_{ij}}{\hsigma_{1,j}}
     \frac{\partial}{\partial \bu_i}
     (x_{ij} - \mu_j - \bv_j^T\bu_i)\nonumber\\
  &= \frac{-1}{2m} \sum_{j=1}^{p} 
     \bv_j w_i^{\case}\, \hsigma_{1,j}\, 
     \frac{\psi_1(r_{ij}/\hsigma_{1,j})}
     {r_{ij}/\hsigma_{1,j}}
     \frac{m_{ij}}{\hsigma_{1,j}}r_{ij}\nonumber\\
  &= \frac{-1}{2m} \sum_{j=1}^{p} 
     \bv_j w_i^{\case} w_{ij}^{\cell} m_{ij} 
     (x_{ij}-\mu_j- \bv_j^T\bu_i)\nonumber\\
  &= \frac{-1}{2m} \left(\bV^T\bW_i
     (\bx_i-\bmu)-\bV^T\bW_i\bV\bu_i\right).
\end{align*}
Finally, the derivative of $L_{\rho_1,\rho_2}$ with 
respect to $\mu_j$ is 
\begin{align*}
\hspace{-0.7cm}
  \frac{\partial}{\partial \mu_j} L_{\rho_1,\rho_2} 
  &= \frac{\hsigma_2^2}{m}
     \sum_{i=1}^{n} m_i \rho'_2(\rt_i/\hsigma_2) 
     \frac{1}{\hsigma_2}
     \frac{\partial}{\partial \mu_j} \sqrt{\frac{1}{m_i}
     \sum_{j=1}^{p} m_{ij}  \hsigma_{1,j}^2\rho_1\left(
     \frac{x_{ij}-\hx_{ij}}{\hsigma_{1,j}}\right)}
     \nonumber\\
  &= \frac{1}{2m} \sum_{i=1}^{n} 
     \frac{\psi_2(\rt_i/\hsigma_2)}{\rt_i/\hsigma_2}\,
     \frac{\partial}{\partial \mu_j}\left(\sum_{j=1}^{p} 
      m_{ij} \hsigma_{1,j}^2\rho_1\left(\frac{x_{ij}-\hx_{ij}}
     {\hsigma_{1,j}}\right)\right)\nonumber\\
  &= \frac{1}{2m}\sum_{i=1}^{n} w_i^{\case}\, 
     \hsigma_{1,j}^2\, \psi_1(r_{ij}/\hsigma_{1,j})
     \frac{m_{ij} }{\hsigma_{1,j}}
     \frac{\partial}{\partial \mu_j}
     (x_{ij} - \mu_j - \bu_i^T\bv_j)\nonumber\\
  &= \frac{-1}{2m}\sum_{i=1}^{n} w_i^{\case}\, 
     \hsigma_{1,j}\, \frac{\psi_1(r_{ij}/
     \hsigma_{1,j})}{r_{ij}/\hsigma_{1,j}}
     \frac{m_{ij}}{\hsigma_{1,j}}r_{ij}\nonumber\\
  &= \frac{-1}{2m}\sum_{i=1}^{n} 
     w_i^{\case} w_{ij}^{\cell} m_{ij} 
     (x_{ij}-\mu_j-\bu_i^T\bv_j)\,
\end{align*}
so the gradient with respect to the vector $\bmu$ becomes
\begin{align*} 
\hspace{-0.7cm}
\bzero = \frac{\partial}{\partial \bmu}L_{\rho_1,\rho_2} 
  &= \frac{-1}{2m}\sum_{i=1}^{n}
     \bW_i (\bx_{i}-\bmu-\bV\bu_i)\,.
\end{align*}
\vskip6mm

We also derive the first order conditions of the 
weighted PCA objective function
\begin{equation*} 
  \tL = \sum_{i=1}^n\sum_{j=1}^p w_{ij}
    \left(x_{ij}-\mu_j-\bu_i^T\bv_j\right)^2\;.
\end{equation*}
We start with the  gradient of $\tL$ 
with respect to $\bv_j$ which is 
\begin{align*} 
\bzero = \frac{\partial}{\partial \bv_j} \tL 
  &= \sum_{i=1}^{n} w_{ij}\,
  \frac{\partial }{\partial \bv_j}
  \left(x_{ij}-\mu_j-\bu_i^T\bv_j\right)^2\nonumber\\
  &= -2\sum_{i=1}^{n} \bu_i w_{ij} 
   (x_{ij}-\mu_j-\bu_i^T\bv_j)\nonumber\\
  &= -2\left(\bU^T\bW^j(\bx^j-\mu_j \bone_n)-
     \bU^T\bW^j\bU\bv_j\right).
\end{align*}
The gradient of $\tL$ with respect to $\bu_i$ is 
\begin{align*} 
\bzero = \frac{\partial}{\partial \bu_i} \tL 
  &= \sum_{j=1}^{p} w_{ij}\,
  \frac{\partial}{\partial \bu_i}
  \left(x_{ij}-\mu_j-\bu_i^T\bv_j\right)^2\nonumber\\
  &= -2\sum_{j=1}^{p} \bv_j w_{ij} 
   (x_{ij}-\mu_j-\bv_j^T\bu_i)\nonumber\\
  &= -2\left(\bV^T\bW_i(\bx_i-\bmu)-
     \bV^T\bW_i\bV\bu_i\right).
\end{align*}
Finally, the derivative of $\tL$ with 
with respect to $\mu_j$ is 
\begin{align*}
  \frac{\partial}{\partial \mu_j} \tL 
  &= \sum_{i=1}^{n} w_{ij}\,
  \frac{\partial}{\partial \mu_j}
  \left(x_{ij}-\mu_j-\bu_i^T\bv_j\right)^2\nonumber\\
  &= -2\sum_{i=1}^{n} w_{ij} 
   (x_{ij}-\mu_j-\bv_j^T\bu_i)\,
\end{align*}
so the gradient with respect to the vector $\bmu$ 
becomes
\begin{align*}
\bzero = \frac{\partial}{\partial \bmu}\tL &= 
  -2\sum_{i=1}^{n}\bW_i (\bx_{i}-\bmu-\bV\bu_i)\,.
\end{align*}

\newpage
\section{Description of the Algorithm}
\label{app:algo}

In each step of the algorithm we optimize the
objective function of the weighted principal
subspace objective~\eqref{eq:WPCA} 
by adjusting the components 
$\bV$, $\bU$ and $\bmu$ one after the other,
by means of the first-order optimality conditions
\eqref{eq:condi1}--\eqref{eq:condi3}. 
This is all done for a fixed weight function. 
Afterward, the weight function is updated too.

The IRLS algorithm starts from our initial 
estimate $(\bV_{(0)}, \bU_{(0)}, \bmu_{(0)})$. 
Then, for each $\st=0, 1,2,\dots$, we obtain
$(\bV_{(\st+1)}, \bU_{(\st+1)}, \bmu_{(\st+1)})$ from 
$(\bV_{(\st)}, \bU_{(\st)}, \bmu_{(\st)})$ by a 
four-step procedure.

\vspace{3mm}
{\bf (a) Minimize~\eqref{eq:WPCA} with 
respect to $\bV$} by applying \eqref{eq:condi1} 
with $\bU_{(\st)}$\,, $\bmu_{(\st)}$\,, and $\bW_{(\st)}$\,, 
the weight matrix based on 
$(\bV_{(\st)}\,,\bU_{(\st)}\,,\bmu_{(\st)})$. 
Condition~\eqref{eq:condi1} says
\begin{equation}\label{eq:WUV}
 (\bU_{(\st)}^T\bW_{(\st)}^j\bU_{(\st)})\bv_j=
 \bU_{(\st)}^T\bW_{(\st)}^j
 \left(\bx^j-(\bmu_{(\st)})_j\bone_n\right)
  \qquad j=1,\dots,p\,.
\end{equation}
For a fixed $j$ we now want to find the best 
column vector $\bv_j$ in the sense of minimizing 
the objective function~\eqref{eq:WPCA}, in 
particular the term
\begin{align}\label{eq:normWUV}
 \sum_{i=1}^n \big(\bW^j_{(\st)}\big)_{ii}
 \big(x_{ij} - (\bmu_{(\st)})_j - \hx_{ij}\big)^2
 &= \sum_{i=1}^n \big(\bW^j_{(\st)}\big)_{ii} 
  \big( x_{ij} - (\bmu_{(\st)})_j -
  (\bU_{(\st)} \bV^T)_{ij}\big)^2 \nonumber\\
 &= \big|\big|
  (\bW^j_{(\st)})^{1/2}\big(\bx^j-(\bmu_{(\st)})_j \bone_n\big)
  - (\bW^j_{(\st)})^{1/2}\bU_{(\st)} \bv_j
  \big|\big|^2
\end{align}
that depends on $\bv_j$\,. This is the
objective of the least squares regression of 
$(\bW^j_{(\st)})^{1/2}\big(\bx^j-(\bmu_{(\st)})_j\bone_n\big)$
on $(\bW^j_{(\st)})^{1/2}\bU_{(\st)}$\,,
so the optimum is reached at 
\begin{align}\label{eq:WUdagger}
\bv_j&=\Big(\big[(\bW^j_{(\st)})^{1/2}\bU_{(\st)}\big]^T
  (\bW^j_{(\st)})^{1/2}\bU_{(\st)}\Big)^{\!\dagger}\,
  \big[(\bW^j_{(\st)})^{1/2}\bU_{(\st)}\big]^T
  (\bW^j_{(\st)})^{1/2} \left(\bx^j-(\bmu_{(\st)})_j
  \bone_n\right)\nonumber \\
  &=\big(\bU_{(\st)}^T\bW_{(\st)}^j\bU_{(\st)}\big)^\dagger
  \bU_{(\st)}^T\bW_{(\st)}^j
  \left(\bx^j-(\bmu_{(\st)})_j\bone_n\right)
\end{align}
where $^\dagger$ stands for the Moore-Penrose
generalized inverse.
Repeating this for all $j=1,\ldots,p$ yields
a new matrix $\bV_{(\st+1)}$ with rows 
$\btv_1,\dots,\btv_p$, which attains the lowest 
objective when everything else remains fixed.

Note that the initial $\bV_{(0)}$ that we 
started the algorithm from was not unique, 
and could be replaced by $\btV_{(0)} = 
\bV_{(0)}\bO$ for any $\rk \times \rk$ orthogonal 
matrix $\bO$. It may seem that the\linebreak
resulting 
$\bP_{(1)},\ldots,\bP_{(\st+1)},\ldots$ need not 
be unique, but we will show that they are.
We\linebreak will show this by induction. 
For $\st=0$ we note that the score matrix 
$\bU_{(0)}$ becomes\linebreak 
$\btU_{(0)} = (\bU_{(0)} \bV_{(0)}^T)
\btV_{(0)} = \bU_{(0)} \bV_{(0)}^T
\bV_{(0)} \bO = \bU_{(0)} \bO$. The right-hand 
side of~\eqref{eq:WUdagger} becomes
\begin{align*}
&(\bO^T\bU_{(0)}^{T}\bW_{(0)}^j\bU_{(0)}
\bO)^\dagger\bO^T\bU_{(0)}^{T}\bW_{(0)}^j
\left(\bx^j-(\bmu_{(0)})_j\bone_n\right)\\
&= \bO^T(\bU_{(0)}^{T}\bW_{(0)}^j
\bU_{(0)})^\dagger\bU_{(0)}^{T}\bW_{(0)}^j
\left(\bx^j-(\bmu_{(0)})_j\bone_n\right)
\end{align*}
because $\bO$ has orthonormal columns and  
$\bO^\dagger = \bO^{-1}= \bO^T$.
So $\btV_{(1)} = \bV_{(1)} \bO$.
The step from $\bV_{(\st)}$ to $\bV_{(\st+1)}$ is
analogous, so by induction we know that for
all $\st$ it holds that 
$\btV_{(\st+1)} = \bV_{(\st+1)} \bO$ and 
$\btV_{(\st+1)} \btV_{(\st+1)}^T = 
\bV_{(\st+1)} \bV_{(\st+1)}^T)$
hence $\bP^*_{(\st+1)} = \bP_{(\st+1)}$\,.

\vspace{3mm}
\textbf{(b) Minimize~\eqref{eq:WPCA} with 
respect to $\bU$} by \eqref{eq:condi2} with 
the new $\bV_{(\st+1)}$ and the old $\bmu_{(\st)}$ 
and $\bW_{(\st)}$. For this we use the first 
order condition~\eqref{eq:condi2}, which says
\begin{equation} \label{eq:WVU}
 \left(\bV_{(\st+1)}^T(\bW_{(\st)})_i\bV_{(\st+1)}\right)\bu_i =
  \bV_{(\st+1)}^T(\bW_{(\st)})_i\left(\bx_i-\bmu_{(\st)}\right) , 
  \quad i=1,\dots,n.
\end{equation}
For a fixed $i$ our goal is find the $\bu_i$ 
that minimizes the corresponding term
\begin{equation} \label{eq:normWVU}
 \big|\big|
 (\bW_{(\st)})_i^{1/2}(\bx_i - \bmu_{(\st)})
 - (\bW_{(\st)})_i^{1/2} \bV_{(\st+1)} \bu_i  
 \big|\big|^2
\end{equation}
of the objective~\eqref{eq:WPCA}. We instead 
minimize
\begin{equation} \label{eq:normWVUcellwise}
 \big|\big|
 (\btW_{(\st)})_i^{1/2}(\bx_i - \bmu_{(\st)})
 - (\btW_{(\st)})_i^{1/2} \bV_{(\st+1)} \bu_i  
 \big|\big|^2\;.
\end{equation}
This is the least squares regression of 
$(\btW_{(\st)})_i^{1/2}(\bx_i-\bmu_{(\st)})$ on
$(\btW_{(\st)})_i^{1/2}\bV_{(\st+1)}$ with solution
\begin{align} \label{eq:WVdagger}
(\bu_{(\st+1)})_i&=
  \Big( \big[(\btW_{(\st)})_i^{1/2}\bV_{(\st+1)}\big]^T 
  (\btW_{(\st)})_i^{1/2}\bV_{(\st+1)} \Big)^{\!\dagger}\,
  \big[(\btW_{(\st)})_i^{1/2}\bV_{(\st+1)}\big]^T
  (\btW_{(\st)})_i^{1/2}
  \left(\bx_i-\bmu_{(\st)}\right)\nonumber \\
  &=\Big(\bV_{(\st+1)}^T(\btW_{(\st)})_i
  \bV_{(\st+1)}\Big)^{\!\dagger}
  \bV_{(\st+1)}^T(\btW_{(\st)})_i
  \left(\bx_i-\bmu_{(\st)}\right).
\end{align}
Note that this solution also 
minimizes~\eqref{eq:normWVU} which equals
$w_i^{\case}$ times~\eqref{eq:normWVUcellwise}.
When the casewise weight $w_i^{\case}$ is strictly
positive these minimizations are equivalent,
and when $w_i^{\case}=0$ the norm~\eqref{eq:normWVU}
attains its lower bound of zero anyway.
Repeating this for all $i$ yields
$\bU_{(\st+1)}$ and therefore 
$\bX^0_{(\st+1)} := \bU_{(\st+1)}\bV_{(\st+1)}^T$\,.

As in step (a) we note that $\bV_{(\st+1)}$ is 
not unique so neither is $\bU_{(\st+1)}$\,, but their
product $\bX^0_{(\st+1)} := \bU_{(\st+1)}\bV_{(\st+1)}^T$
is unique. We show this as follows. Writing
$\btV_{(\st+1)} =\bV_{(\st+1)}\bO$ and 
$\btU_{(\st+1)} =\bU_{(\st+1)}\bO$ the right-hand 
side of~\eqref{eq:WVdagger} becomes
$$\bO^\dagger\!\left(\bV_{(\st+1)}^T(\btW_{(\st)})_i
\bV_{(\st+1)}\right)^\dagger\bV_{(\st+1)}^T
(\btW_{(\st)})_i\left(\bx_i-\bmu_{(\st)}\right)
\qquad \mbox{so} \qquad 
\btU_{(\st+1)} = \bU_{(\st+1)}\bO,$$
and $\btX^0_{(\st+1)} = 
\btU_{(\st+1)}\btV_{(\st+1)}^T =
\bU_{(\st+1)}\bO \bO^T_{(\st)} \bV_{(\st+1)}^T=
\bU_{(\st+1)}\bV_{(\st+1)}^T= \bX^0_{(\st+1)}$\,.
 
\vspace{3mm}
\textbf{(c) Minimize~\eqref{eq:WPCA} with 
respect to $\bmu$} with the new $\bV_{(\st+1)}$ 
and $\bU_{(\st+1)}$ and the old $\bW_{(\st)}$.
For each $j=1,\ldots,p$ the term of the objective 
function~\eqref{eq:WPCA} involving $\mu_j$ is
the weighted sum of squares
\begin{align} \label{eq:columndist}
  \big|\big|
  (\bW_{(\st)})_j^{1/2} &\big(\bx^j
  - (\bmu_{(\st)})_j\bone_n
  - (\bU_{(\st+1)}\bV_{(\st+1)}^T)^j\big)
  \big|\big|^2 \nonumber \\
  &=\sum_{i=1}^n (\bW_{(\st)})_{ij} \left(
  x_{ij}-\mu_j -\big(\bU_{(\st+1)}\bV_{(\st+1)}^T
  \big)_{ij}\right)^2
\end{align}
where $\bx^j$ is the $j$th column of $\bX$,
which is minimized by the weighted mean 
\begin{equation}\label{eq:tmu}
 (\bmu_{(\st+1)})_j = \frac{\sum_{i=1}^n 
 (\bW_{(\st)})_{ij} \big(x_{ij}-\big(
 \bU_{(\st+1)}\bV_{(\st+1)}^T\big)_{ij}\big)}
  {\sum_{i=1}^n (w_{(\st)})_{ij}}\;.
\end{equation}
Repeating this for all $j$ yields the column
vector $\bmu_{(\st+1)}$\,. 

\vspace{3mm}
\textbf{(d) Update $\bW_{(\st)}$} according 
to \eqref{eq:updateW}, \eqref{eq:cellweight}, 
and \eqref{eq:caseweight} with the new 
$\bV_{(\st+1)}$, $\bU_{(\st+1)}$ and $\bmu_{(\st+1)}$.

\vspace{3mm}
The use of the generalized inverse in steps 
(a) and (b) of the algorithm deserves some 
explanation. In~\eqref{eq:WUdagger} it could 
happen that the matrix 
$\bU_{(\st)}^T\bW_{(\st)}^j\bU_{(\st)}$ is 
singular, especially when $\rk$ is not much 
smaller than $n$. In that case we cannot invert 
$\bU_{(\st)}^T\bW_{(\st)}^j\bU_{(\st)}$\,, 
but its generalized inverse still exists. 
The matrix 
$\bV_{(\st+1)}^T(\bW_{(\st)})_i\bV_{(\st+1)}$ 
in~\eqref{eq:WVU} can also be singular. 
This occurs for instance when $\bx_i$ is 
considered as a casewise outlier so $w_i^{\case} = 0$. 
Then the matrix $(\bW_{(\st)})_i$ becomes 
zero so that 
$(\bV_{(\st+1)}^T(\bW_{(\st)})_i
\bV_{(\st+1)})^{\dagger}$ is zero 
too, yielding $(\bu_{(\st+1)})_{i} =\bzero$ so
$\bx^0_i = \bzero$ hence 
$\bhx_i = \bmu_{(\st+1)}$. During the course 
of the algorithm we monitor the 
fraction of zero weights in each variable $j$. 
We do not want this fraction to be too high, 
because this could make it hard to identify the
robust correlation between variables, making 
the estimation of $\bP$ imprecise or even
impossible. So we impose a maximal fraction of
zero weights per column, which is 25\% by default.
In case this fraction is exceeded, the iteration
stops and the code returns the results of the
previous iteration step, together with a warning
that it may be better to remove that variable.

The update formula \eqref{eq:tmu} does 
not compute the generalized inverse of 
$\sum_{i=1}^n(w_{(\st)})_{ij}$ because it is 
strictly positive under the same condition. 
If $\sum_{i=1}^n(w_{(\st)})_{ij}$ were zero 
this would mean all cells $x_{ij}$ of 
variable $j$ received zero weight.

\section{\large Proof of Proposition~\ref{the_1P}}
\label{app:monotone}
In this section Proposition~\ref{the_1P} is proved, 
which ensures that each step of the algorithm 
decreases the objective function \eqref{eq:objP}. 
In order to prove Proposition~\ref{the_1P} we need 
three lemmas.
\begin{lemma} \label{le_1P}
For a given weight matrix $\bW_{(\st)}$\,, each of the 
update steps (a), (b), and (c) of the algorithm in 
Section B decreases the weighted PCA objective 
function~\eqref{eq:WPCA}.
\end{lemma}
\begin{proof}
The objective function~\eqref{eq:WPCA} 
we want to minimize is 
\begin{equation} 
  \sum_{i=1}^n\sum_{j=1}^p w_{ij}
  \big(x_{ij}-\mu_j-(\bU\bV^T)_{ij}\big)^2
\end{equation}
where $w_{ij} = W_{(\st)_{ij}}$ is fixed. 
We start from the triplet $(\bV_{(\st)}, 
\bU_{(\st)}, \bmu_{(\st)})$ with objective
\begin{equation} 
  \sum_{i=1}^n\sum_{j=1}^p w_{ij}
  \big(x_{ij}-(\bmu_{(\st)})_j-
  (\bU_{(\st)}\bV_{(\st)}^T)_{ij}\big)^2.
\end{equation}
Step (a) minimizes the squared 
norm~\eqref{eq:normWUV}, which after summing
over $j=1,\ldots,p$ becomes the 
objective~\eqref{eq:WPCA} in the new triplet
$(\bV_{(\st+1)}, \bU_{(\st)}, \bmu_{(\st)})$.\\
\indent Next, step (b) minimizes the squared 
norm~\eqref{eq:normWVU}, which after summing
over $i=1,\ldots,n$ becomes the 
objective~\eqref{eq:WPCA} in the new triplet
$(\bV_{(\st+1)}, \bU_{(\st+1)}, \bmu_{(\st)})$.\\ 
\indent Finally, step (c) minimizes the squared 
norm~\eqref{eq:columndist}, whose sum
over $j=1,\ldots,p$ is the 
objective~\eqref{eq:WPCA} in the new triplet
$(\bV_{(\st+1)}, \bU_{(\st+1)}, \bmu_{(\st+1)})$.
\end{proof}

We will denote a potential fit to~\eqref{eq:objP} 
as $\btheta = \bhX = \bX^0\bP+\bmu$ where the
matrices have the appropriate dimensions.
We introduce the notation $\bof(\btheta)$ for 
$\vecmat((\bX - \btheta)\odot(\bX - \btheta))$
where $\vecmat(.)$ turns a matrix into a column 
vector. The vector $\bof(\btheta)$ has $np$ 
entries, which are the values 
$(\bX_{ij} - \btheta_{ij})^2$ for 
$i=1,\ldots,n$ and $j=1,\ldots,p$.
We can then write the cellPCA objective function 
\eqref{eq:objP} as $L(\bof(\btheta)) := 
L_{\rho_1,\rho_2}(\bX,\bV,\bU,\bmu)$. 

\begin{lemma}
\label{le_2}
The function $\bof \rightarrow L(\bof)$ is concave.
\end{lemma}
\begin{proof}
We first show that the univariate function
$h: \mathbb R_{+} \rightarrow \mathbb R_{+}:
z \rightarrow \rho_{b,c}(\sqrt{z})$, in which 
$\rho_{b,c}$ is the wrapping $\rho$-function, 
is concave as suggested in 
Figure~\ref{fig:rho_psi}. Note that the 
derivative $h'$ is continuous, and 
differentiable except in the 
points $b^2$ and $c^2$.
Its derivative $h''(z)$ is zero on the 
intervals $(0,b^2)$ and $(c^2,\infty)$, and 
on $(b^2,c^2)$ we obtain
$$h''(z) = \frac{-z^{3/2}}{4}q_1 
\tanh(q_2(c-\sqrt{z})) - \frac{q_1 q_2}{4z}
\mbox{sech}^2(q_2(c-\sqrt{z})) \leqslant 0\,.$$
Since the function $h'(z)$ is continuous
and on each of the three open intervals its
derivative $h''(z)$ is nonpositive, it is
nonincreasing everywhere. 
Therefore $h$ is concave.

By the definition of concavity of a multivariate 
function, we now need to prove that for 
any column vectors $\bof,\bol$ in 
$\mathbb R_{+} ^{np}$ and any $\lambda$ in 
$(0,1)$ it holds that  
$L(\lambda\bof+(1-\lambda)\bol) \geqslant 
\lambda L(\bof)+(1-\lambda)L(\bol)$.
This works out as
\begin{align*}
 L(\lambda\bof+(1-\lambda)\bol)&=
 \frac{\hsigma_2^2}{m} 
 \sum_{i=1}^n h_2\left(\frac{1}{m_i\hsigma_2^2}
 \sum_{j=1}^p m_{ij}\hsigma_{1,j}^2
  h_1\left(
  \frac{\lambda f_{ij}+(1-\lambda)g_{ij}}
  {\hsigma_{1,j}^2}\right)\,\right)\\ & \geqslant
  \frac{\hsigma_2^2}{m}
  \sum_{i=1}^n h_2\left(\frac{1}{m_i\hsigma_2^2}
  \sum_{j=1}^{p} 
  m_{ij}\hsigma_{1,j}^2\left[\lambda h_1(
  f_{ij}/\hsigma_{1,j}^2)+
  (1-\lambda)h_1(g_{ij}/\hsigma_{1,j}^2)
  \right]\,\right) \\
 & = \frac{\hsigma_2^2}{m}
  \sum_{i=1}^n h_2\left(\lambda\frac{1}
  {m_i\hsigma_2^2} \sum_{j=1}^{p} 
  m_{ij}\hsigma_{1,j}^2 h_1(f_{ij}/\hsigma_{1,j}^2)
   \right.\\
   & \left.\hspace{5cm} +\, (1-\lambda)\frac{1}{m_i\hsigma_2^2}
   \sum_{j=1}^{p} m_{ij}\hsigma_{1,j}^2 
   h_1(g_{ij}/\hsigma_{1,j}^2)\right)\\
   & \geqslant
  \frac{\hsigma_2^2}{m}
  \sum_{i=1}^n \left[\lambda h_2\left(
  \frac{1}{m_i\hsigma_2^2} \sum_{j=1}^{p} 
  m_{ij}\hsigma_{1,j}^2h_1(f_{ij}/\hsigma_{1,j}^2)
  \right)\right.\\
  &\left.\hspace{5cm} +\,(1-\lambda)h_2\left(
  \frac{1}{m_i\hsigma_2^2}\sum_{j=1}^{p} 
  m_{ij}\hsigma_{1,j}^2 h_1(g_{ij}/\hsigma_{1,j}^2)
  \right)\right]\\
  & = \lambda L(\bof)+(1-\lambda)L(\bol).
\end{align*}
The first inequality derives from the concavity of $h_1$
and the fact that $h_2$ is nondecreasing. The second 
inequality is due to the concavity of $h_2$\,. 
Therefore $L$ is a concave function.
\end{proof}

We can also write the weighted PCA objective
\eqref{eq:WPCA} as a function of $\bof$. We will
denote it as $L_{\bW}(\bof):=(\vecmat(\bW))^T\bof$,
so $\bW$ was turned into a vector in the same
way as was done to obtain the column vector $\bof$. 
The next lemma makes a connection between the 
weighted PCA objective $L_{\bW}$ and the
original objective $L$.

\begin{lemma} \label{le_3}
If two column vectors $\bof,\bol$ in 
$\mathbb R_{+}^{np}$ satisfy
$L_{\bW}(\bof) \leqslant L_{\bW}(\bol)$, then 
$L(\bof) \leqslant L(\bol)$. 
\end{lemma}

\begin{proof}
From Lemma \ref{le_2} we know that $L(\bof)$ is 
concave as a function of $\bof$, and it is also 
differentiable because $h_1$ and $h_2$ are.
Therefore
\begin{equation*}
    L(\bof) \leqslant L(\bol) +
    (\nabla L(\bol))^T(\bof-\bol)
\end{equation*}
where the column vector $\nabla L(\bol)$ is the 
gradient of $L$ in $\bol$. But $\nabla L(\bol)$
is proportional to $\nabla L_W(\bol)$ which
equals $\vecmat(\bW)$ by construction, 
so $(\nabla L(\bol))^T(\bof-\bol)$ is proportional
to $L_{\bW}(\bof) - L_{\bW}(\bol) \leqslant 0$.
Therefore $L(\bof) \leqslant L(\bol)$.
\end{proof}
With this preparation we can prove
Proposition \ref{the_1P}.

\begin{proof}[Proof of Proposition \ref{the_1P}]
When we go from $\btheta_{(\st)}$ to $\btheta_{(\st+1)}$, 
Lemma~\ref{le_1P} ensures that\linebreak
$L_{\bW_{(\st)}}(\bof(\btheta_{(\st+1)})) \leqslant 
L_{\bW_{(\st)}}(\bof(\btheta_{(\st)}))$, so 
Lemma~\ref{le_3} implies
$L(\bof(\btheta_{(\st+1)})) \leqslant 
L(\bof(\btheta_{(\st)}))$. 
\end{proof}

\section{Proofs of influence functions}
\label{app:proofs}

We consider the contamination model 
\eqref{eq:cont_cell_z} 
where $H_Z= \Delta_{\bz}$ is the distribution that 
puts all of its mass in a fixed $p$-variate vector 
$\bz=\left(z_1, \ldots, z_p\right)^T$, given by 
\begin{equation*}
X_{\eps}=A \odot X + (\bone_p-A) \odot \bz 
\end{equation*}
with $A=(A_1,\dots,A_p)^T \sim G_\eps$.

Under both the dependent and independent 
contamination models with 
$P(A_j^{\cell}=1)=1-\eps^{\cell}$ for all 
$j=1,\ldots,p$, the distribution of $A$ satisfies
$P\left(A_j=1\right)=1-\eps$, $j=1, \dots, p$, 
and (ii) for any sequence 
$\left(j_1, j_2, \ldots, j_p\right)$ of zeroes 
and ones with $p-\ki$ ones and $\ki$ zeroes, 
$P(A_1=j_1, \ldots, A_p=j_p)$ has the same 
value, denoted as $\delta_\ki(\eps)$. Obviously, 
$\eps=\eps^{\case}$ 
under the fully dependent contamination model 
(FDCM) and $\eps=\eps^{\cell}$ under the fully 
independent contamination model (FICM). 
Under FDCM we have that $P(A_1=\cdots=A_p)=1$, 
and then $\delta_0(\eps)=(1-\eps),\; 
\delta_1(\eps)=\cdots=\delta_{p-1}(\eps)=0$, 
and $\delta_p(\eps)=\eps$. In that situation 
the distribution of $X_{\eps}$ simplifies to 
$(1-\eps)H_0 + \eps\Delta_{\bz}$ where 
$\Delta_{\bz}$ is the distribution which puts 
all of its mass in the point $\bz$. FICM instead 
assumes that $A_1, \dots, A_p$ are independent, 
hence
$$\delta_\ki(\eps)= \binom{p}{\ki}
   (1-\eps)^{p-\ki} \eps^{\ki}\,, 
  \quad \ki=0,1, \ldots, p\;. $$
The distribution $G_\eps$ is denoted as 
$G_\eps^D$ in the dependent model, and as 
$G_\eps^I$ in the independent model. 

The proof of Proposition \ref{prop2} is based on the 
implicit function theorem, see e.g.\ Rio Branco
de Oliveira (2012):
\begin{theorem}[Implicit Function Theorem]
\label{the_impl}
Let $\bof(x, \btheta)=(f_1, \ldots, f_p)$ be a function from $\mathbb{R} \times \mathbb{R}^p$ to $\mathbb{R}^p$ that is continuous in $(x_0,\tilde{\btheta})\in \mathbb{R} \times \mathbb{R}^p$ with $\bof(x_0, \tilde{\btheta})=\bzero$. Suppose the derivative of $\bof$ exists in a neighbourhood $N$ of $(x_0,\tilde{\btheta})$ and is continuous at $(x_0,\tilde{\btheta})$,  and that the derivative matrix $\partial \bof/\partial \btheta$ is nonsingular at $(x_0, \tilde{\btheta})$. Then there are neighbourhoods $N_1$ of $x_0$ and $N_p$ of $\tilde{\btheta}$ with $N_1\times N_p \subset N$, such that for every $x$ in $N_1$ there is a unique $\btheta=\bT(x)$ in $N_p$ for which $\bof(x, \bT(x))=\bzero$. In addition, $\bT$ is differentiable in $x_0$ with derivative matrix given by
\begin{equation*}
  \frac{\partial \bT(x)}{\partial x}\Big|_{x=x_0} 
  = -\Big(\frac{\partial \bof(x_0,\btheta)}
    {\partial \btheta}\Big|_{\btheta=
    \tilde{\btheta}}\Big)^{-1}
    \frac{\partial \bof(x,\tilde{\btheta)}}
    {\partial x}\Big|_{x=x_0}.
\end{equation*}
\end{theorem}

\vspace{1mm}
\begin{proof}[Proof of Proposition \ref{prop2}]
From \eqref{eq:gIF}, the IF of $\bP$ at a distribution $H_0$ is  defined as 
\begin{equation*}
    \IFu_{H}(\bz,\bP)=\vect\left(\left.\frac{\partial}{\partial \eps} \bP(H(G_\eps,\bz))\right|_{\eps=0}\right).
\end{equation*}
For $\eps=0$ we obtain $\bP(H(G_0,\bz)) = \bP(H_0) = \bP_0$.
We can then parametrize $\bP_0 = \bV_0\bV_0^T$ where $\bV_0$
has orthonormal columns, corresponding to an orthonormal
basis of the linear subspace $\bPi_0$\,. Note that $\bV_0$ is
not unique, but we will see later that different choices
lead to the same influence function of $\bP$.
When the distribution $H_0$ is contaminated by FDCM or FICM
to $H(G_\eps,\bz)$ we define $\bV(H(G_\eps,\bz))$ as the
result of the algorithm of Section~\ref{app:algo} above, 
translated from finite samples to population distributions
and starting from $\bV_0$. This construction makes
$\bV(H(G_\eps,\bz))$ unique, and it has to satisfy the
first-order condition \eqref{eq:C1} saying
\begin{equation} \label{eq:reparamcond}
  E_H[\bW_{\bx}(\bV\bu_{\bx}-\bx)
  \bu_{\bx}^T] = \bzero\,
\end{equation} 
hence
\begin{equation}\label{eq_newg}
  \bg(H(G_\eps,\bz),\bT(\eps),\bsigma(H(G_\eps,\bz)))= \vect\left(\E_{H(G_\eps,\bz)}\big[\bW_{\bx}(\bV\bu_{\bx}-\bx)\bu_{\bx}^T\big]\right)=\bzero
\end{equation}
where the $p\rk$-variate column vector $\bg$ is written
as a function of the $p\rk$-variate column vector 
$\bT(\eps):=\vect(\bV(H(G_\eps,\bz)))$.
In order to compute $\vect\left(\frac{\partial}{\partial \eps} \bV(H(G_\eps,\bz))\left.\right|_{\eps=0}\right)$ we would like to apply the implicit function theorem in the point $\eps = 0$, but the contaminated distribution $H(G_\eps,\bz)$ is only defined for $\eps > 0$. To circumvent this issue we extend the definition of $\bg$ to negative $\eps$ by defining a function
$\bof$ from $\mathbb{R} \times \mathbb{R}^{p\rk}$ 
to $\mathbb{R}^{p\rk}$ as 
\begin{equation*}
\bof(\eps, \btheta)= \begin{cases}
\bg(H(G_\eps,\bz),\btheta,\bsigma(H(G_\eps,\bz)))
& \mbox{for } \eps \geqslant 0 \\ 
2 \bg(H_0,\btheta,\bsigma(H_0))-\bg(H(|\eps|,\bz),\btheta,\bsigma(H(|\eps|,\bz))) & \mbox{for } \eps<0.\end{cases}
\end{equation*}
We now put $\eps_0=0$ and $\tilde{\btheta}=\bT(0)=\vect(\bV_0)$. 
Then $\bof(\eps_0,\btheta)=\bof(0,\bT(0))=\bzero$, and assuming that $\bg$ is sufficiently smooth for the differentiability requirements of the implicit function theorem, we can conclude that $\bT(\eps)$ is uniquely defined for small $\eps$ and that
\begin{align*} \label{eq_der}
  \frac{\partial \bT(\eps)}{\partial \eps}\Big|_{\eps=0} 
  &= -\Big(\frac{\partial \bof(0,\bT)}
    {\partial \bT}\Big|_{\bT=\bT(0)}\Big)^{-1}
    \;\frac{\partial \bof(\eps,\bT(0))}
    {\partial \eps}\Big|_{\eps=0}\\
  &= -\Big(\frac{\partial \bg(0,\bT)}
    {\partial \bT}\Big|_{\bT=\bT(0)}\Big)^{-1}
    \;\frac{\partial \bg(\eps,\bT(0))}
    {\partial \eps}\Big|_{\eps=0}\;.
\end{align*}
Note that the left hand side is $\vect\left(\frac{\partial}{\partial \eps} \bV(H(G_\eps,\bz))\left.\right|_{\eps=0}\right)$.
We now have to work out the right hand side. For 
the first factor we denote the matrix
\begin{equation*}
  \bB:= \frac{\partial \bg(H_0,\bT,\sigma_0)}
    {\partial \bT}\Big|_{\bT=\vect(\bV_0)} 
\end{equation*}
which does not depend on $\bz$ and will be computed 
numerically in Section \ref{app:DS}. For the second factor, 
from \eqref{eq_newg}
we know that $\bg(H(G_\eps,\bz),\btheta_\eps,\bsigma_\eps)$ 
is an expectation over the mixture distribution
$H(G_\eps,\bz)$, so we can write $\bg$ as a linear 
combination with coefficients $\delta_\ki(\eps)$ for 
$k=0,1,\ldots,p$\,.
For the FDCM model we know that $G_\eps^D$ has 
$\delta_0(\eps)=(1-\eps),\; 
\delta_1(\eps)=\cdots=\delta_{p-1}(\eps)=0$ and 
$\delta_p(\eps)=\eps$, so $\bg$ can be written as
\begin{align*}
\bg(H(G_\eps^D,\bz)&,\vect(\bV_0),
\bsigma(H(G_\eps^D,\bz)))\\
&= \delta_0(\eps) \bg(H_0,\vect(\bV_0),
\bsigma(H(G_\eps^D,\bz))) + 
\delta_p(\eps) \bg(H(G_\eps^D,\bz),\vect(\bV_0),
\bsigma(H(G_\eps^D,\bz)))\\
&= (1-\eps)\bg(H_0,\vect(\bV_0),\bsigma_\eps) + 
\eps\bg(\Delta_{\bz},\vect(\bV_0),\bsigma_\eps)
\end{align*}
which yields the derivative
\begin{align}
  &-\bg(H_0,\vect(\bV_0),\bsigma_0) 
  + \frac{\partial }{\partial \eps}\bg(H_0,
  \vect(\bV_0),\bsigma_\eps)\Big|_{\eps=0} 
  + \bg(\Delta_{\bz},\vect(\bV_0),\bsigma_0) \nonumber\\
  &= \bzero + \frac{\partial \bg}{\partial \bsigma} 
  (H_0, \vect(\bV_0), \bsigma)\Big|_{\bsigma=\bsigma_0}
   \frac{\partial \bsigma_\eps}{\partial \eps} 
   \Big|_{\eps=0} + 
   \bg(\Delta_{\bz},\vect(\bV_0),\bsigma_0) \nonumber\\
   \label{eq:IFcaseV}
  &= \bS\IFu_{\case}(\bz,\bsigma)+
  \bg(\Delta_{\bz},\vect(\bV_0),\bsigma_0)
\end{align}
where $\bS:=\left.\frac{\partial}{\partial \bsigma} 
\bg(H_0, \vect(\bV_0), \bsigma)\right|_{\bsigma=\bsigma_0}$ 
and  $\IFu_{\case}(\bz,\bsigma)$ is the influence
function of $\bsigma$ under FDCM. 

Under the FICM model the second factor is different. 
We have $\delta_0(\eps)=(1-\eps)^p$, $\delta_0(0)=1$, 
$\delta_1(\eps) = p (1-\eps)^{p-1}\eps$ so $\delta_1(0)=0$ 
and $\delta_1^{\prime}(0)=p$, and $\delta_\ki(0)=
\delta_\ki^{\prime}(0)=0$ for $\ki \geqslant 2$. Therefore 
$\bg$ can be written as the sum
\begin{align*}
\bg(H(G_\eps^I,\bz)&,\vect(\bV_0),\bsigma(H(G_\eps^I,\bz)))\\
&= \delta_0(\eps) \bg(H_0,\vect(\bV_0),\bsigma(H(G_\eps^I,\bz)))
+ \delta_1(\eps) 
\sum_{j=1}^p 
\bg(H(G_1^j,\bz), \vect(\bV_0),\bsigma(G_\eps^I,\bz)))\\
&= (1-\eps)^p\bg(H_0,\vect(\bV_0),\bsigma_\eps) 
+ p (1-\eps)^{p-1}\eps \sum_{j=1}^p 
\bg(H(G_1^j,\bz), \vect(\bV_0),\bsigma(G_\eps^I,\bz)))
\end{align*}
where $H(G_1^j,\bz)$ is the distribution of 
$X \sim H_0$ but with its $j$th component fixed 
at the constant $z_j$\,.  It is thus a degenerate 
distribution concentrated on the hyperplane 
$X_j = z_j$\,. 
The derivative now becomes
\begin{equation}\label{eq:IFcellV}
  \bS\IFu_{\cell}(\bz,\bsigma)+p\sum_{j=1}^p 
  \bg(H\left(G_1^j,\bz\right),\vect(\bV_0),\bsigma_0)
\end{equation}
where $\bS$ is the same as before but 
$\IFu_{\cell}(\bz,\bsigma)$ is now the cellwise 
influence function of $\bsigma$. 

Now that we have an expression for 
$\frac{\partial}{\partial \eps}\bV(H(G_\eps,\bz))$ 
we can use it to derive the IF of $\bP$. Note that 
the columns of $\bV_0$ were orthonormal, but the 
columns of $\bV(H(G_\eps,\bz))$ do not have to be.
Therefore the projection matrix $\bV(H(G_\eps,\bz))$ 
is given by
$$\bP(H(G_\eps,\bz))=\bV(H(G_\eps,\bz))
(\bV(H(G_\eps,\bz))^T\bV(H(G_\eps,\bz)))^{-1}
\bV(H(G_\eps,\bz))^T\,.$$ Differentiating yields
\begin{align*}
 \frac{\partial}{\partial \eps}\bP & 
  (H(G_\eps,\bz))\Big|_{\eps=0}=
  \frac{\partial}{\partial \eps}\bV(H(G_\eps,\bz))
  (\bV(H(G_\eps,\bz))^T\bV(H(G_\eps,\bz)))^{-1}
  \bV(H(G_\eps,\bz))^T\Big|_{\eps=0} \\
 =&\left.\frac{\partial}{\partial \eps} 
 \bV(H(G_\eps,\bz))\right|_{\eps=0}
 (\bV_0^T\bV_0)^{-1}\bV_0^T\\
 &\;\;-\bV_0(\bV_0^T\bV_0)^{-1}\left.
 \frac{\partial}{\partial \eps} 
 \bV(H(G_\eps,\bz))^T\right|_{\eps=0}
 \bV_0(\bV_0^T\bV_0)^{-1}\bV_0^T \\
 &\;\;-\bV_0(\bV_0^T\bV_0)^{-1}\bV_0^T\left.
 \frac{\partial}{\partial \eps} 
 \bV(H(G_\eps,\bz))\right|_{\eps=0}
 (\bV_0^T\bV_0)^{-1}\bV_0^T\\
 &\;\;+\bV_0(\bV_0^T\bV_0)^{-1}\left.
 \frac{\partial}{\partial \eps} 
 \bV(H(G_\eps,\bz))^T\right|_{\eps=0}\;
\end{align*}
where the derivative of 
$(\bV(H(G_\eps,\bz))^T\bV(H(G_\eps,\bz)))^{-1}$ 
comes form the identity 
$\frac{\partial \mathbf{Y}^{-1}}{\partial x} =
-\mathbf{Y}^{-1} \frac{\partial \mathbf{Y}}
{\partial x} \mathbf{Y}^{-1}$ (Magnus and
Neudecker, 2019). Since $\bV_0^T\bV_0 = \bI_\rk$ 
is the identity matrix and
$\bV_0\bV_0^T=\bP$, the expression simplifies to 
\begin{align}\label{eq:simple}
 \frac{\partial}{\partial \eps}&\bP(H(G_\eps,\bz))
 \Big|_{\eps=0}
  =\left.\frac{\partial}{\partial \eps}
  \bV(H(G_\eps,\bz))\right|_{\eps=0}\bV_0^T 
 -\bV_0\left.\frac{\partial}{\partial \eps} 
 \bV(H(G_\eps,\bz))^T\right|_{\eps=0}
 \bP_0 \nonumber \\
 &\;\;-\bP_0\left.\frac{\partial}{\partial \eps} 
  \bV(H(G_\eps,\bz))\right|_{\eps=0}\bV_0^T
 +\bV_0\left.\frac{\partial}{\partial \eps} 
 \bV(H(G_\eps,\bz))^T\right|_{\eps=0} \nonumber \\
 &=(\bI_p-\bP_0)\left.\frac{\partial}{\partial \eps}
  \bV(H(G_\eps,\bz))\right|_{\eps=0}\bV_0^T 
 +\bV_0\left.\frac{\partial}{\partial \eps} 
 \bV(H(G_\eps,\bz))^T\right|_{\eps=0}
 (\bI_p - \bP_0)\;.
\end{align}
Since the second term is the transpose of the first
we see that the derivative of $\bP(H(G_\eps,\bz))$ 
is symmetric, as it should be.

Now suppose we had chosen a different orthonormal 
basis of $\bPi_0$, corresponding to a matrix 
$\btV_0$. We need to verify that this would yield 
the same result. First compute the 
$\rk \times \rk$ matrix $\bO=(\btV_0^{-1}\bV_0)$. 
This matrix is orthogonal because $\bO^T\bO =
\bV_0^T (\btV_0 \btV_0^T)^{-1} \bV_0 = 
\bV_0^T \bV_0 = \bI_\rk$\,, and $\btV_0 = \bV_0\bO$.
Then construct $\btV(H(G_\eps,\bz))$ by running the 
algorithm starting from $\btV_0$ instead of $\bV_0$.
In Section~\ref{app:algo} we saw that every step will
have $\btV_{(\st+1)} = \bV_{(\st+1)}\bO$, so this holds 
in the limit as well, hence 
$\btV(H(G_\eps,\bz)) = \bV(H(G_\eps,\bz))\bO$.
Writing~\eqref{eq:simple} with $\btV(H(G_\eps,\bz))$
and $\btV_0$ yields factors $\bO\bO^T$ that cancel,
yielding~\eqref{eq:simple} again. 

Applying the $\vect$ operation to \eqref{eq:simple}
gives the IF. Applying it to the first term yields
$$\vect\Big((\bI_p-\bP_0)\left.
  \frac{\partial}{\partial \eps}
  \bV(H(G_\eps,\bz))\right|_{\eps=0}\bV_0^T\Big)
  = (\bV_0\otimes(\bI_p-\bP_0))\vect\Big(
  \left.\frac{\partial}{\partial \eps}
  \bV(H(G_\eps,\bz))\right|_{\eps=0}\Big)$$
by the rule 
$\vect(\bA\bB\bC) = (\bC^T\otimes\bA)\vect(\bB)$.
For the second term we find
$$\vect\Big(\bV_0\left.\frac{\partial}{\partial \eps} 
  \bV(H(G_\eps,\bz))^T\right|_{\eps=0}(\bI_p - \bP_0)\Big)
  = ((\bI_p-\bP_0)\otimes\bV_0)\vect\Big(\Big(
  \left.\frac{\partial}{\partial \eps}
  \bV(H(G_\eps,\bz))\right|_{\eps=0}\Big)^T\Big)$$
by the same rule. The last factor is the $\vect$ of a 
transposed matrix, which can be written as
$$\vect\Big(\Big(
  \left.\frac{\partial}{\partial \eps}
  \bV(H(G_\eps,\bz))\right|_{\eps=0}\Big)^T\Big) =
  \bK_{p,\rk}\vect\Big(
  \left.\frac{\partial}{\partial \eps}
  \bV(H(G_\eps,\bz))\right|_{\eps=0}\Big)$$
where $\bK_{p,\rk}$ is a $p\rk \times p\rk$ 
permutation 
matrix that rearranges the entries of the column 
vector $\vect(\left.\frac{\partial}{\partial \eps} 
\bV(H(G_\eps,\bz))\right|_{\eps=0})$ to become those 
of $\vect(\left.\frac{\partial}{\partial \eps} 
\bV(H(G_\eps,\bz))^T\right|_{\eps=0})$.
In all we can write
\begin{equation}\label{eq:long_IF_PV}
\vect\left(\left.\frac{\partial}{\partial \eps} 
\bP(H(G_\eps,\bz))\right|_{\eps=0}\right)=
\bR_0\vect\left(\left.\frac{\partial}{\partial \eps} 
\bV(H(G_\eps,\bz))\right|_{\eps=0}\right)
\end{equation}
where $\bR_0$ is the $p^2 \times p\rk$ matrix
\begin{equation}\label{eq:R0}
\bR_0 = \bV_0\otimes(\bI_p-\bP_0) +
((\bI_p-\bP_0)\otimes\bV_0)\bK_{(p,\rk)}\,.
\end{equation}
Combining \eqref{eq:IFcaseV} with $\bD:=\bR_0\bB^{-1}$ 
yields \eqref{eq:IFPFDCM}, and left multiplying 
\eqref{eq:IFcellV} by $\bD$ yields 
\eqref{eq:IFPFICM}.
\end{proof}

\begin{proof}[Proof of Proposition \ref{the_3}]
For each $\der = 1,\ldots,p\rk$\,, call 
$\ddot{\bPsi}_\der=\frac{\partial^2 \bPsi_\der}{
\partial \vect(\bV)\partial \vect(\bV)}$ the 
matrix of second derivatives of $\bPsi_\der$ 
with respect to the entries of $\bV$, and 
$\bC_n(\bx,\bV)$ the matrix with $\der$th 
row equal to $\vect(\bhV_n-\bV(H))^{T} 
\ddot{\bPsi}_\der(\bx, \bV)$. 
A Taylor expansion yields
\begin{align*}
    \bzero=\widehat{\Lambda}_n(\bhV_n)
    = \frac{1}{n}\sum_{i=1}^n \Big\{ &
      \bPsi(\bx_i, \bV(H))+\dot{\bPsi}(\bx_i, \bV(H))
      \vect(\bhV_n-\bV(H))\\
    & + \frac{1}{2} \bC_n(\bx_i, \bV(H))
    \vect(\bhV_n-\bV(H)) \Big\} \,.
\end{align*}
In other words
\begin{equation}
\label{eq:VsumABC}
\mathbf{0}=\bA_n+\left(\bB_n+\overline{\bC}_n\right)\vect(\bhV_n-\bV(H))
\end{equation}
with
\begin{equation*}
\bA_n=\frac{1}{n} \sum_{i=1}^n \bPsi(\bx_i,\bV(H)), \quad \bB_n=\frac{1}{n} \sum_{i=1}^n \dot{\bPsi}(\bx_i, \bV(H)), \quad \overline{\bC}_n=\frac{1}{2n} \sum_{i=1}^n \bC_n(\bx_i, \bV(H)).
\end{equation*}
The $\der$th row of the matrix $\overline{\bC}_n$ equals $\vect(\bhV_n-\bV(H))^{T} \overline{\ddot{\bPsi}_\der}$ where
\begin{equation*}
\overline{\ddot{\boldsymbol{\Psi}}_\der} =\frac{1}{n} \sum_{i=1}^n \ddot{\bPsi}_\der(\bx_i, \bV(H))
\end{equation*}
which is bounded. Since $\bhV_n \rightarrow \bV(H)$ in probability, this implies that  $\overline{\bC}_n \rightarrow \bzero$ in probability. 

Note that for $i=1,2, \ldots$ the vectors $\bPsi(\bx_i, \bV(H))$ are i.i.d.\ with mean $\bzero$ (since $\Lambda(\bV(H))=\bzero$) and covariance matrix $\bA$, where $\bA=\E_{H}[\bPsi(\bx, \bV(H)) \bPsi(\bx, \bV(H))^{T}]$, and the matrices $\dot{\bPsi}(\bx_i, \bV(H))$ are i.i.d.\ with mean $\bB$, where $\bB=\E_{H} [\dot{\bPsi}(\bx,\bV(H))]$. Hence when $n \rightarrow \infty$, the law of large numbers implies $\bB_n \rightarrow \bB$ in probability, which implies $\bB_n+\overline{\bC}_n \rightarrow \bB$ in probability, and we assume that $\bB$ is nonsingular. The central limit theorem implies $\sqrt{n} \bA_n \rightarrow N_p(\bzero,\bA)$ in distribution.
Then from \eqref{eq:VsumABC} and Slutsky's lemma we have that
\begin{equation*}
    \sqrt{n}\vect\Big(\bhV_n-\bV(H)\Big) 
    \rightarrow_d N_{p\rk}\Big(\bzero,
    \bB^{-1}\bA(\bB^{-1})^T\Big).
\end{equation*}

From $\bhV_n \rightarrow \bV(H)$ in probability it 
follows that $\bhP_n \rightarrow \bP(H)$ in 
probability. Consider the differentiable mapping 
$\bh(\bM) = \bM(\bM^T\bM)^{-1}\bM^T$ on the set of
$p \times \rk$ matrices $\bM$ of rank $\rk$. From 
the multivariate delta method, see e.g.\ 
Casella and Berger (2002), it follows that 
\begin{equation*}
  \sqrt{n}\vect\Big(\bh(\bhV_n)-\bh(\bV(H))\Big) 
  \rightarrow_d 
  N_{p^2}\!\Big(\bzero, 
  \bR_0\bB^{-1}\bA\bB^{-T}\bR_0^T\Big)
\end{equation*}
where $\bR_0=\frac{\partial \vect(\bh(\bV))}
{\partial \vect(\bV)}\left.\right|_{\bV=\vect(\bV(H))}$ 
was defined in \eqref{eq:R0}.
Moreover $\bR_0\bB^{-1}\bA\bB^{-T}\bR_0^T=\bTheta$ 
because $\IFu_{\case}(\bx,\bP)=\bR_0\bB^{-1}\bPsi
\left(\bx, \bV(H)\right)$ by Proposition~\ref{prop2}
with fixed $\bsigma$.
Since $\IFu_{\case}(\bx,\bP)$ does not depend on the
parametrization of $\bP$, neither does\linebreak
$\bTheta =
E_H[\IFu_{\case}(\bx,\bP)\IFu_{\case}(\bx,\bP)^T]$.
\end{proof}

\section{Derivation of $\bB$ and $\bS$}
\label{app:DS}
We now compute the $p\rk \times p\rk$ 
matrix $\bB =\left.\frac{\partial}{\partial 
\vect\left(\bV\right)} g
\left(H_0, \vect(\bV),\bsigma_0\right)
\right|_{\vect(\bV)=\vect(\bV(H))}$ 
and the $p\rk \times (p+1)$ matrix 
$\bS=\left.\frac{\partial}{\partial \bsigma} 
g\left(H_0, \vect(\bV(H)), \bsigma\right)
\right|_{\bsigma=\bsigma_0}$.

Recall that $\bmu$ is known and equal to $\bzero$.  
Then
\begin{equation*} 
   w_{j}^{\cell}=\psi_1\left(
   \frac{r_{j}}{\sigma_{1,j}}\right)
   /\frac{r_{j}}{\sigma_{1,j}}\; 
    \quad  r_{j} := x_{j}- \bv_j^{T}\bu_{\bx}\;\quad
   j=1\dots,p
\end{equation*}
and 
\begin{equation*} 
    w^{\case}=\psi_2\left(
   \frac{\rt}{\sigma_2}\right)
   /\frac{\rt}{\sigma_2}\;\quad \rt := \sqrt{\frac{1}{p}
  \sum_{j=1}^{p} \sigma_{1,j}^2\rho_1\left(\frac{
  r_j}{\sigma_{1,j}}\right)}.
\end{equation*}

For $j,l=1,\dots,p$ and $h,m=1,\dots,\rk$, we have that
\begin{align*}
\label{eq:D1}
    &-\frac{\partial}{\partial v_{lm}}\E_{H_0}\left[ w_j(x_j-\bv_j^T\bu_x)u_{x,h}\right]\nonumber=-\E_{H_0}\left[\frac{\partial w_j}{\partial v_{lm}} r_ju_{x,h}+w_j\frac{\partial r_j}{\partial v_{lm}}u_{x,h}+w_jr_j\frac{\partial u_{x,h}}{\partial v_{lm}}\right]
\end{align*}
where
\begin{align*}
  \frac{\partial r_j}{\partial v_{lm}}=-\left(\frac{\partial \bv_{j}^T}{\partial v_{lm}}\bu_{\bx}+\bv_{j}^T\frac{\partial \bu_{\bx}}{\partial v_{lm}}\right)=-\left(\delta_{j=l}u_{x,m}+\bv_{j}^T\frac{\partial \bu_{\bx}}{\partial v_{lm}}\right)
\end{align*}
and $\delta_{j=l}$ is the Kronecker delta. Then
\begin{align*}
  \frac{\partial w^{\cell}_j}{\partial v_{lm}}&=\left(\frac{\psi_1'\left(
   r_{j}/\sigma_{1,j}\right)}{\sigma_{1,j}}\left(\frac{\partial r_j}{\partial v_{lm}}\right)\frac{r_{j}}{\sigma_{1,j}}- \frac{\psi_1\left(
   r_{j}/\sigma_{1,j}\right)}{\sigma_{1,j}}\left(\frac{\partial r_j}{\partial v_{lm}}\right)\right)/\frac{r_{j}^2}{\sigma_{1,j}^2}\\
   &=\frac{\psi_1'\left(
   r_{j}/\sigma_{1,j}\right)}{r_j}\left(\frac{\partial r_j}{\partial v_{lm}}\right)- \frac{\sigma_{1,j}\psi_1\left(
   r_{j}/\sigma_{1,j}\right)}{r_j^2}\left(\frac{\partial r_j}{\partial v_{lm}}\right)\\
    &=(w^{\cell}_{j})'\left(\frac{\partial r_j}{\partial v_{lm}}\right)
\end{align*}
where $(w^{\cell}_{j})'=\psi_1'(r_{j}/\sigma_{1,j})/r_j
  -\sigma_{1,j}\psi_1(r_{j}/\sigma_{1,j})/r_j^2.$
Moreover
\begin{align*}
  \frac{\partial w^{\case}}{\partial v_{lm}}&=\left(\frac{\psi_2'\left(
   \rt/\sigma_{2}\right)}{\sigma_{2}}\frac{\partial \rt}{\partial v_{lm}}\frac{\rt}{\sigma_{2}}-\frac{\psi_2\left(
   \rt/\sigma_{2}\right)}{\sigma_{2}}\frac{\partial \rt}{\partial v_{lm}}\right)/\frac{\rt^2}{\sigma_{2}^2}\\
   &=\left(\frac{\psi_2'\left(
   \rt/\sigma_{2}\right)}{\rt}- \frac{\sigma_2\psi_2\left(
   \rt/\sigma_{2}\right)}{\rt^2}\right)\frac{\partial \rt}{\partial v_{lm}}\\
    &=(w^{\case})'\frac{\partial \rt}{\partial v_{lm}}
\end{align*}
where $(w^{\case})'=\psi_2'(\rt/\sigma_2)/r - 
  \sigma_2\psi_2(\rt/\sigma_2)/\rt^2$ and
\begin{align*}
  \frac{\partial \rt}{\partial v_{lm}}&=\frac{1}{2p\rt}\sum_{s=1}^{p}\sigma_{1,s}
  \psi_1\left(
   r_{s}/\sigma_{1,s}\right)\left(\frac{\partial r_s}{\partial v_{lm}}\right)\\
   &=\frac{1}{2p\rt}\sum_{s=1}^{p}
  w_s^{\cell}r_s\left(\frac{\partial r_s}{\partial v_{lm}}\right).
\end{align*}
Then 
\begin{align*}
  \frac{\partial w_j}{\partial v_{lm}}&=\frac{\partial w^{\cell}_j}{\partial v_{lm}}w^{\case}+w^{\cell}_j\frac{\partial w^{\case}}{\partial v_{lm}}\\
  &=w^{\case}(w^{\cell}_{j})'\left(\frac{\partial r_j}{\partial v_{lm}}\right)+w^{\cell}_j(w^{\case})'\frac{1}{2p\rt}\sum_{s=1}^{p}
  w_s^{\cell}r_s\left(\frac{\partial r_s}{\partial v_{lm}}\right)
  \\&=\sum_{s=1}^{p}\left(
  \frac{w^{\cell}_j(w^{\case})'w_s^{\cell}r_s}{2p\rt}+\delta_{s=j}w^{\case}(w^{\cell}_{j})'\right)\left(\frac{\partial r_s}{\partial v_{lm}}\right).
\end{align*}
Then $\bB(H_0,\vect(\bV(H)),\bsigma_0)$ is obtained from
$$\bB(H_0,\vect(\bV(H)),\bsigma_0)=\lbrace B_{\thh\tm}(H_0,\vect(\bV(H)),\bsigma_0)\rbrace,$$ where $B_{\thh\tm}(H_0,\vect(\bV(H)),\bsigma_0)=\left.-\frac{\partial}{\partial v_{lm}}\E_{H_0}\left[ w_j(x_j-\bv_j^T\bu_x)u_{x,h}\right]\right|_{\vect(\bV)=\vect(\bV(H))}$ 
with \linebreak
$\thh=(h-1)p+j$ and $\tm=(m-1)p+l$.
Note that $\frac{\partial \bu_x}{\partial v_{lm}}$ does not have a closed form, so it has to be computed numerically.

For $\bS(H_0,\vect(\bV(H)),\bsigma_0)$ we compute 
\begin{align*}
    &-\frac{\partial}{\partial \sigma_{1,j}}\E_{H_0}\left[ w_j(x_j-\bv_j^T\bu_x)u_{x,h}\right]\nonumber=-\E_{H_0}\left[\frac{\partial w_j}{\partial \sigma_{1,j}} r_ju_{x,h}+w_j\frac{\partial r_j}{\partial \sigma_{1,j}}u_{x,h}+w_jr_j\frac{\partial u_{x,h}}{\partial \sigma_{1,j}}\right]
\end{align*}
where
\begin{align*}
  \frac{\partial r_j}{\partial \sigma_{1,j}}=-\bv_{j}^T\frac{\partial \bu_{\bx}}{\partial \sigma_{1,j}}.
\end{align*}
Then 
\begin{align*}
  \frac{\partial w^{\cell}_j}{\partial \sigma_{1,j}}&=\left(\psi_1'\left(
   r_{j}/\sigma_{1,j}\right)\frac{\partial r_j/\sigma_{1,j}}{\partial \sigma_{1,j}}\frac{r_{j}}{\sigma_{1,j}}- \psi_1\left(
   r_{j}/\sigma_{1,j}\right)\frac{\partial r_j/\sigma_{1,j}}{\partial \sigma_{1,j}}\right)/\frac{r_{j}^2}{\sigma_{1,j}^2}\\
   &=\left(\frac{\psi_1'\left(
   r_{j}/\sigma_{1,j}\right)}{r_j}\sigma_{1,j}\frac{\partial r_j/\sigma_{1,j}}{\partial \sigma_{1,j}}- \sigma_{1,j}^2\frac{\psi_1\left(
   r_{j}/\sigma_{1,j}\right)}{r_j^2}\frac{\partial r_j/\sigma_{1,j}}{\partial \sigma_{1,j}}\right)\\
   &=\sigma_{1,j}(w^{\cell}_{j})'\frac{\partial r_j/\sigma_{1,j}}{\partial \sigma_{1,j}}.
\end{align*}
Moreover,
\begin{align*}
  \frac{\partial w^{\case}}{\partial \sigma_{1,j}}&=\left(\frac{\psi_2'\left(
   \rt/\sigma_{2}\right)}{\sigma_{2}}\frac{\partial \rt}{\partial \sigma_{1,j}}\frac{\rt}{\sigma_{2}}-\frac{\psi_2\left(
   \rt/\sigma_{2}\right)}{\sigma_{2}}\frac{\partial \rt}{\partial\sigma_{1,j}}\right)/\frac{\rt^2}{\sigma_{2}^2}\\
   &=\left(\frac{\psi_2'\left(
   \rt/\sigma_{2}\right)}{\rt}- \frac{\sigma_2\psi_2\left(
   \rt/\sigma_{2}\right)}{\rt^2}\right)\frac{\partial \rt}{\partial \sigma_{1,j}}\\
    &=(w^{\case})'\frac{\partial \rt}{\partial \sigma_{1,j}}
\end{align*}
where 
\begin{align*}
  \frac{\partial \rt}{\partial \sigma_{1,j}}&=\frac{1}{2p\rt}\sum_{s=1}^{p}\left[\delta_{s=j}2\sigma_{1,j}\rho_1\left(\frac{
  r_j}{\sigma_{1,j}}\right)+\sigma_{1,s}^2\psi_1\left(\frac{
  r_s}{\sigma_{1,s}}\right)\frac{\partial r_s/\sigma_{1,s}}{\partial \sigma_{1,j}}\right]\\
   &=\frac{1}{2p\rt}\sum_{s=1}^{p}\left[\delta_{s=j}2\sigma_{1,j}\rho_1\left(\frac{
  r_j}{\sigma_{1,j}}\right)+\sigma_{1,s}r_sw_s^{\cell}\frac{\partial r_s/\sigma_{1,s}}{\partial \sigma_{1,j}}\right],
\end{align*}
and
\begin{align*}
 \frac{\partial r_s/\sigma_{1,s}}{\partial \sigma_{1,j}}=\frac{\partial r_s}{\partial \sigma_{1,j}}\frac{1}{\sigma_{1,s}}-\delta_{s=j}\frac{r_j}{\sigma_{1,j}^2}\;.
\end{align*}
Then 
\begin{align*}
  \frac{\partial w_j}{\partial \sigma_{1,j}}&=\frac{\partial w^{\cell}_j}{\partial \sigma_{1,j}}w^{\case}+w^{\cell}_j\frac{\partial w^{\case}}{\partial \sigma_{1,j}}\\
  &=w^{\case}\sigma_{1,j}(w^{\cell}_{j})'\frac{\partial r_j/\sigma_{1,j}}{\partial \sigma_{1,j}}+\frac{w^{\cell}_j(w^{\case})'}{2p\rt}\sum_{s=1}^{p}\left[\delta_{s=j}2\sigma_{1,j}\rho_1\left(\frac{
  r_j}{\sigma_{1,j}}\right)+\sigma_{1,s}r_sw_s^{\cell}\frac{\partial r_s/\sigma_{1,s}}{\partial \sigma_{1,j}}\right].
\end{align*}
Moreover,
\begin{align*}
    &-\frac{\partial}{\partial \sigma_{2}}\E_{H_0}\left[ w_j(x_j-\bv_j^T\bu_x)u_{x,h}\right]\nonumber=-\E_{H_0}\left[\frac{\partial w_j}{\partial \sigma_{2}} r_ju_{x,h}+w_j\frac{\partial r_j}{\partial \sigma_{2}}u_{x,h}+w_jr_j\frac{\partial u_{x,h}}{\partial \sigma_{2}}\right]
\end{align*}
where
\begin{align*}
  \frac{\partial r_j}{\partial \sigma_{2}}=-\bv_{j}^T\frac{\partial \bu_{\bx}}{\partial \sigma_{2}}\;.
\end{align*}
Then
\begin{align*}
  \frac{\partial w^{\cell}_j}{\partial \sigma_{2}}&=\left(\psi_1'\left(
   r_{j}/\sigma_{1,j}\right)\frac{\partial r_j}{\partial \sigma_{2}}\frac{r_{j}}{\sigma_{1,j}^2}- \frac{\psi_1\left(
   r_{j}/\sigma_{1,j}\right)}{\sigma_{1,j}}\frac{\partial r_j}{\partial \sigma_{2}}\right)/\frac{r_{j}^2}{\sigma_{1,j}^2}\\
   &=\left(\frac{\psi_1'\left(
   r_{j}/\sigma_{1,j}\right)}{r_j}\frac{\partial r_j}{\partial \sigma_{2}}- \sigma_{1,j}\frac{\psi_1\left(
   r_{j}/\sigma_{1,j}\right)}{r_j^2}\frac{\partial r_j}{\partial \sigma_{2}}\right)\\
   &=(w^{\cell}_{j})'\frac{\partial r_j}{\partial \sigma_{2}}\;.
\end{align*}
Moreover,
\begin{align*}
  \frac{\partial w^{\case}}{\partial \sigma_{2}}&=\left(\psi_2'\left(
   \rt/\sigma_{2}\right)\frac{\partial \rt/\sigma_2}{\partial \sigma_{2}}\frac{\rt}{\sigma_{2}}-\psi_2\left(
   \rt/\sigma_{2}\right)\frac{\partial \rt/\sigma_2}{\partial \sigma_{2}}\right)/\frac{\rt^2}{\sigma_{2}^2}\\
   &=\left(\sigma_2\frac{\psi_2'\left(
   \rt/\sigma_{2}\right)}{\rt}- \sigma_2^2\frac{\psi_2\left(
   \rt/\sigma_{2}\right)}{\rt^2}\right)\frac{\partial \rt/\sigma_2}{\partial \sigma_{2}}\\
    &=\sigma_2(w^{\case})'\frac{\partial \rt/\sigma_2}{\partial \sigma_{2}}
\end{align*}
where 
\begin{align*}
  \frac{\partial \rt/\sigma_2}{\partial \sigma_{2}}&= \frac{\partial \rt}{\partial \sigma_{2}}\frac{1}{\sigma_{2}}-\frac{\rt}{\sigma_{2}^2}\\
  \end{align*}
and
  \begin{align*}
  \frac{\partial \rt}{\partial \sigma_{2}}&=\frac{1}{2p\rt}\sum_{s=1}^{p}\sigma_{1,s}\psi_1\left(\frac{
  r_j}{\sigma_{1,s}}\right) \frac{\partial r_s}{\partial \sigma_{2}}\\
   &=\frac{1}{2p\rt}\sum_{s=1}^{p}\sigma_{1,s}r_sw^{\cell}_s \frac{\partial r_s}{\partial \sigma_{2}}\;.
\end{align*}
Then 
\begin{align*}
  \frac{\partial w_j}{\partial \sigma_{2}}&=\frac{\partial w^{\cell}_j}{\partial \sigma_{2}}w^{\case}+w^{\cell}_j\frac{\partial w^{\case}}{\partial \sigma_{2}}\\
  &=w^{\case}(w^{\cell}_{j})'\frac{\partial r_j}{\partial \sigma_{2}}+w^{\cell}_j\sigma_2(w^{\case})'\left(\frac{1}{2p\rt}\sum_{s=1}^{p}\sigma_{1,s}r_sw^{\cell}_s \frac{\partial r_s}{\partial \sigma_{2}}\frac{1}{\sigma_{2}}-\frac{\rt}{\sigma_{2}^2}\right).
\end{align*}
Here $\frac{\partial \bu_x}{\partial \sigma_{1,j}}$ and $\frac{\partial \bu_x}{\partial \sigma_{2}}$ do not have not a closed form either, so they must be computed numerically as well.

\clearpage
\section{Pseudocode of the cellPCA algorithm}
\label{app:pseudocode}

Algorithm~\ref{alg:irls} provides the pseudocode of 
the steps in Sections~\ref{sec:obj}--\ref{sec:algo} 
of the main text.\\

\begin{algorithm}[H]
\caption{IRLS algorithm for cellPCA} \label{alg:irls}
\begin{algorithmic}[1]
  \STATE{Initialize $\bV_{(0)}$, $\bU_{(0)}$, $\bmu_{(0)}$ by
         the MacroPCA initial fit.}
  \STATE{Compute $\hsigma_{1,j}$ and 
     $\hsigma_2$ according to Section~\ref{sec:obj}.}
    \STATE{Compute the initial weight matrix $\bW_{(0)}$ using \eqref{eq:updateW}, \eqref{eq:cellweight}, and \eqref{eq:caseweight}.}
  \STATE{Set $\st = 0$.}
  \REPEAT  
     \STATE{(a) Update $\bV$:  
        \[
        (\bv_{(\st+1)})_j = 
        \big(\bU_{(\st)}^T\bW_{(\st)}^j\bU_{(\st)}
        \big)^{\dagger}\;
        \bU_{(\st)}^T\bW_{(\st)}^j\big(\bx^j-
        (\bmu_{(\st)})_j\bone_n\big) \quad j=1,\dots,p.
        \]
     }
     \STATE{(b) Update $\bU$:   
        \[
        (\bu_{(\st+1)})_i = \big(\bV_{(\st+1)}^T
        (\btW_{(\st)})_i\bV_{(\st+1)}\big)^{\dagger}\;
        \bV_{(\st+1)}^T(\btW_{(\st)})_i\big(\bx_i-\bmu_{(\st)}
        \big) \quad i=1,\dots,n.
        \]
     }
     \STATE{(c) Update $\bmu$:   
        \[
        \bmu_{(\st+1)} = \Big(\sum_{i=1}^n (\bW_{(\st)})_i\Big)^{-1}
        \sum_{i=1}^n (\bW_{(\st)})_i\big(\bx_i - \bV_{(\st+1)}
        (\bu_{(\st+1)})_i\big).
        \]
     }
     \STATE{(d) Update $\bW$: Compute $\bW_{(\st+1)}$ using \eqref{eq:updateW}, \eqref{eq:cellweight}, and \eqref{eq:caseweight}.}
     \STATE{Increment $\st$: $\st = \st + 1$.}
  \UNTIL{$\;||\bU_{(\st)}\bV_{(\st)}^T - \bU_{(\st-1)}\bV_{(\st-1)}^T||_F
         < \nu\, ||\bU_{(\st-1)}\bV_{(\st-1)}^T||_F$ for some tolerance $\nu$.}
\end{algorithmic}
\end{algorithm} 

\newpage
Algorithm~\ref{alg:outofsample} gives the pseudocode of 
the out-of-sample prediction in Section~\ref{sec:newx}.\\

\begin{algorithm}[H]
\caption{Out-of-Sample Prediction in cellPCA} \label{alg:outofsample}
\begin{algorithmic}[1]
  \STATE{Given a new data point $\bx^*$ with missingness indicator $\bm{m}^*$.}
  \STATE{Retrieve the $\bV$, $\bmu$ and $\hsigma_{1,j}$ 
         estimated in the training stage.}
  \STATE{Set $J = \{ j \;;\; m^*_j = 1 \}$.}
  \IF{$J = \varnothing$} 
      \STATE{Set $\bu^*$ and $\bhx^*$ to NA.} 
  \ELSE
      \STATE{Initialize $\bu^*_{(0)}$ as $\bV_J^T(\bx^*_J - \bmu_J)$
      where $\bx^*_J$ has the coordinates of $\bx^*$ in $J$.
      \STATE{Compute the initial weights  $\bw^*_{(0)}$ using \eqref{eq:cellweightnew}.}
      \STATE{Set $\st = 0$.}
      \REPEAT
         \STATE{Update $\bu^*$ by
          \[
          \bu^*_{(\st+1)} = \big(\bV^T
          \bW^*_{(\st)}\bV\big)^{\dagger}\;
          \bV^T\bW^*_{(\st)}(\bx^*-\bmu)
          \]
          \hspace{7mm} where $\bW^*_{(\st)}=\diag(\bw^*_{(\st)})$.}
          \STATE{Update cellwise weights by 
          \[
          (\bw^*_{(\st+1)})_j = \psi_1\!\left(
          \frac{(r_{(q+1)})_j}{\hsigma_{1,j}}\right)
          \Big/ \frac{(r_{(q+1)})_j}{\hsigma_{1,j}} 
          \quad \text{for } j \in J
          \] \hspace{7mm} and set $(\bw^*_{(\st+1)})_j = 0$ for $j \notin J$, with $(r_{(q+1)})_j=x^*_j-\mu_j-\bv_j^T\bu^*_{(q+1)}$\,.
          }
      \UNTIL{$\;||\bV\bu^*_{(\st)} - \bV\bu^*_{(\st-1)}||_F
         < \nu\, ||\bV\bu^*_{(\st-1)}||_F$ for some tolerance $\nu$.}
      \STATE{Compute the final estimate  $\bhx^*=(\hx^*_1,\dots, \hx^*_p)^T$ with $
      \hx^*_j := \mu_j + \bv_j^T\bu^*_{(\st+1)}$\,.}
  \ENDIF
  }
\end{algorithmic}
\end{algorithm}

\clearpage
\section{Complexity of the cellPCA algorithm}
\label{app:complexity}

We will first study the time complexity of the algorithm 
described in Sections~\ref{sec:obj}--\ref{sec:algo}
and~\ref{sec:est_pd}.

Starting from the initial estimates 
of $\bV$, $\bU$ and
$\bmu$ we first have to compute the robust M-scales 
$\hsigma_{1,j}$ for $j = 1,\ldots,p$ as described in
Section~\ref{sec:obj}. Each M-scale can be computed in
$O(n)$ time and we compute $p$ of them, so the 
complexity is $O(np)$.
The computation of the single M-scale $\hsigma_{2}$ is 
$O(n)$ which does not increase the overall complexity
$O(np)$ for all of these scale estimates.

Next we have to carry out the IRLS algorithm 
described in Section~\ref{sec:algo}, which initializes 
$\bW$ and then updates $\bV$, $\bU$, $\bmu$ and $\bW$.
Initializing the $n \times p$ matrix $\bW$ according 
to~\eqref{eq:updateW}--\eqref{eq:caseweight} is an
elementwise operation with complexity $O(np)$. 
The complexity for updating $\bV$ as in \eqref{eq:updateV} 
is obtained as follows.
\begin{enumerate}
\item \textbf{Compute $\bx^j - \bmu_j \bone_n$.}
     Subtracting a scalar from an $n$-dimensional 
		vector requires $O(n)$ operations.
\item \textbf{Compute $\bW^j \big(\bx^j - \bmu_j \bone_n\big)$.}
    Matrix-vector multiplication of $\bW^j$ ($n \times n$) and this  vector ($n \times 1$) would require $O(n^2)$ operations in general. However, as $\bW^j$ is diagonal the matrix by vector multiplication simplifies to elementwise multiplication, which requires $O(n)$ operations. 
\item \textbf{Compute $\bU^T \bW^j \big(\bx^j - \bmu_j \bone_n\big)$.}
       Matrix-vector multiplication of $\bU^T$ ($k \times n$) and  resulting vector ($n \times 1$) requires $O(nk)$ operations.
\item \textbf{Compute $\bU^T \bW^j \bU$.}
   The matrix multiplication $\bW^j \bU$  of $\bW^j$ ($n \times n$) and $\bU$ ($n \times k$) would require $O(n^2 k)$ operations in general, since each of the $nk$ entries of the product matrix is a sum of $n$ products of two scalars.  However, because $\bW^j$ is diagonal it is only the number of entries in $\bU$, so $O(n k)$ operations. To this we must add the complexity of multiplying $\bU^T$ ($k \times n$) and the resulting $\bW^j \bU$ ($n \times k$) which is $O(nk^2)$ which is bigger. Therefore the computation of $\bU^T \bW^j \bU$ requires $O(nk^2)$ operations.
\item \textbf{Compute the pseudoinverse $\big(\bU^T \bW^j \bU\big)^{\dagger}$.}
    Computing the pseudoinverse of a $k \times k$ matrix by a singular value decomposition requires $O(k^3)$ operations.
\item \textbf{Carry out the final multiplication.}
     Multiplying the $k \times k$ pseudoinverse with the $k \times 1$ vector requires $O(k^2)$ operations.
\end{enumerate}

\noindent Since \eqref{eq:updateV} is repeated 
$p$ times,  adding up all the steps gives
\begin{align*}
    &p\big( O(n)+O(n)+O(nk)+O(nk^2)+O(k^3)+O(k^2)\big)\\&=
O(npk)+O(npk^2)+O(pk^3) = O(npk^2 + pk^3) = O(npk^2)
\end{align*}
where the last equality follows from $npk^2 \geqslant pk^3$ due 
to $n \geqslant k$. 
Therefore, the time complexity for updating $\bV$ is
$O(npk^2)$.

For updating $\bU$ we have to repeat~\eqref{eq:updateU}
$n$ times. By a totally similar reasoning we obtain the
complexity 
\begin{align*}
    &n\big(O(p)+O(p)+O(pk)+O(pk^2)+O(k^3)+O(k^2)\big)\\&=
O(npk)+O(npk^2)+O(nk^3) = O(npk^2 + nk^3) =O(npk^2)
\end{align*}
where the last equality follows from $npk^2 \geqslant nk^3$ 
due to $p \geqslant k$.
So the time complexity for updating $\bU$ is
$O(npk^2)$ also.

Updating $\bmu$ consists of the following steps:
\begin{enumerate}
\item \textbf{Compute $\bx_i - \bV \bu_i$.}
     Multiplying $\bV$ ($p \times k$) with $\bu_i$ ($k \times 1$) requires $O(pk)$ operations. Subtracting the resulting $p \times 1$ vector from $\bx_i$ ($p \times 1$) requires $O(p)$ operations. Repeating this for all $i = 1,\ldots,n$ yields $O(np)+O(npk)=O(npk)$ operations.
\item \textbf{Compute $\bW_i \big(\bx_i - \bV \bu_i\big)$.}
    Since $\bW_i$ is a diagonal $p \times p$ matrix, the matrix by vector multiplication simplifies to elementwise multiplication, which requires $O(p)$ operations. For all $i$ this becomes $O(np)$ operations.
\item \textbf{Sum over $i = 1, \dots, n$.} Summing $n$ vectors
     ($p \times 1$) requires $O(np)$ operations.
\item \textbf{Compute $\sum_{i=1}^n \bW_i$.}
     Adding $n$ diagonal $p \times p$ matrices is performed by summing the diagonal elements. That is adding $n$ vectors ($p \times 1$), which requires $O(np)$ operations.
\item \textbf{Compute the inverse $\Big(\sum_{i=1}^n \bW_i\Big)^{-1}$.}
    Since $\sum_{i=1}^n \bW_i$ is diagonal, the inversion simplifies to inverting each diagonal entry, which requires $O(p)$ operations.
\item \textbf{Multiply the inverse with the vector sum.} Multiplying the $p \times p$ diagonal inverse matrix with the $p \times 1$ vector requires $O(p)$ operations.
\end{enumerate}
By adding up all the steps, the complexity for updating $\bmu$
becomes $O(npk) + O(np) + O(np) + O(np) + O(p) + O(p) = O(npk)$.

Finally, updating $\bW$ is again an elementwise 
operation with complexity $O(np)$.

Therefore, the total complexity of the algorithm in 
Sections~\ref{sec:obj}--\ref{sec:algo} is
$O(np) + O(npk^2) + O(npk^2) + O(npk)$
which equals $O(npk^2)$.

If we also want to estimate the principal directions as
in Section \ref{sec:est_pd} we need to carry out the DetMCD
method \citep{hubert2012deterministic} for $n$ points 
in $k$ dimensions. This requires 
$O(n \log(n)k^2)+ O(nk^2)=O(n \log(n)k^2)$ operations. 

Therefore, all the computations in
Sections~\ref{sec:obj}--\ref{sec:algo}
and~\ref{sec:est_pd} together have complexity 
$O((n\log(n) + np)k^2)$.

However, the cellPCA algorithm starts from the 
MacroPCA initial estimator, which has time complexity  
$O(np(\min(n, p) + \log(n) + \log(p)))$ as shown 
in \citep{hubert2019macropca}.
Its complexity does not depend on $k$ because MacroPCA
requires $k \leqslant k_{max}$ where by default
$k_{max} = 10$.
Therefore, the rank $k$ in the remainder of the cellPCA
algorithm also cannot increase beyond $k_{max}$\,,
so this part has complexity $O(n\log(n) + np)).$
The overall complexity of cellPCA is thus
\begin{align*}
  &O(np(\min(n, p)+\log(n)+\log(p)))+O(n\log(n)+np)\\
  &=O(np(\min(n, p)+\log(n)+\log(p)))
\end{align*}
so the additional steps of cellPCA do not increase the time
complexity beyond that of the initial estimator MacroPCA.
In robust estimation it is indeed typical that the overall 
complexity remains that of the initial estimator.

Figure~\ref{fig:results_time} shows the average 
computation times of the cellPCA components in seconds, 
as a function of $n$ and $p$ over 500 replications. 
These were measured on a workstation equipped with two 24-core 
sockets with an Intel\textsuperscript{\textregistered} 
Xeon\textsuperscript{\textregistered} Platinum 8160 processor 
with a clock frequency of 2.10GHz and 192GB of RAM.
The top row in Figure~\ref{fig:results_time} shows times
in function of $n$, for fixed choices of $(p,k).$ The
bottom row shows them in function of $p$, for fixed choices
of $(n,k)$.
The leftmost panels have the times for running the initial
MacroPCA. In the middle we see the times for 10 iterations
of the IRLS algorithm, which indeed look linear in both
$n$ and $p$. The rightmost panels show the total times of
cellPCA, which are the sums of the leftmost and the middle 
panels. The shapes of its curves are similar to those of
the initial MacroPCA method, as expected.

\begin{figure}[!ht]
\centering
\begin{tabular}{ccc}
\includegraphics[width=.3\textwidth] {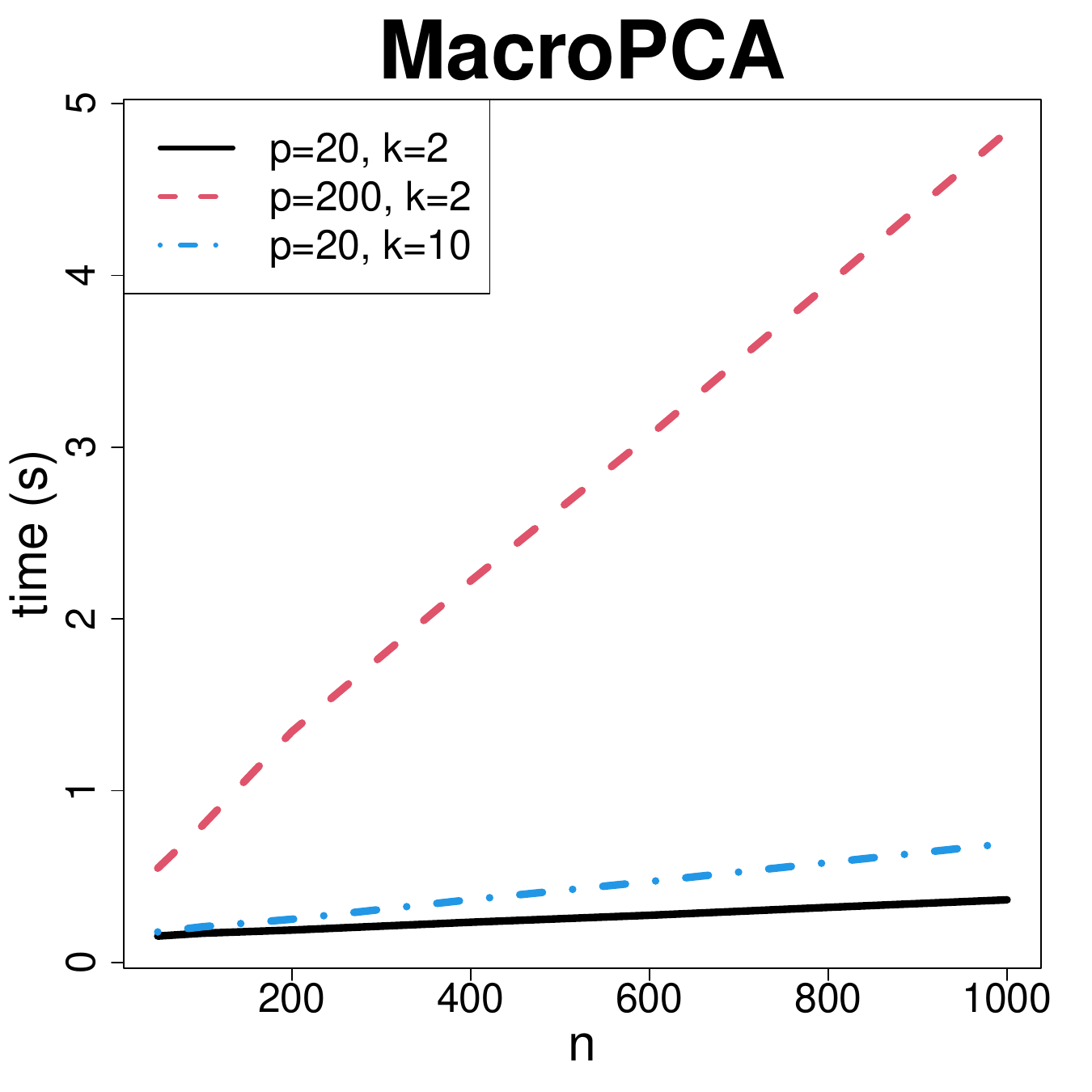} &
\includegraphics[width=.3\textwidth] {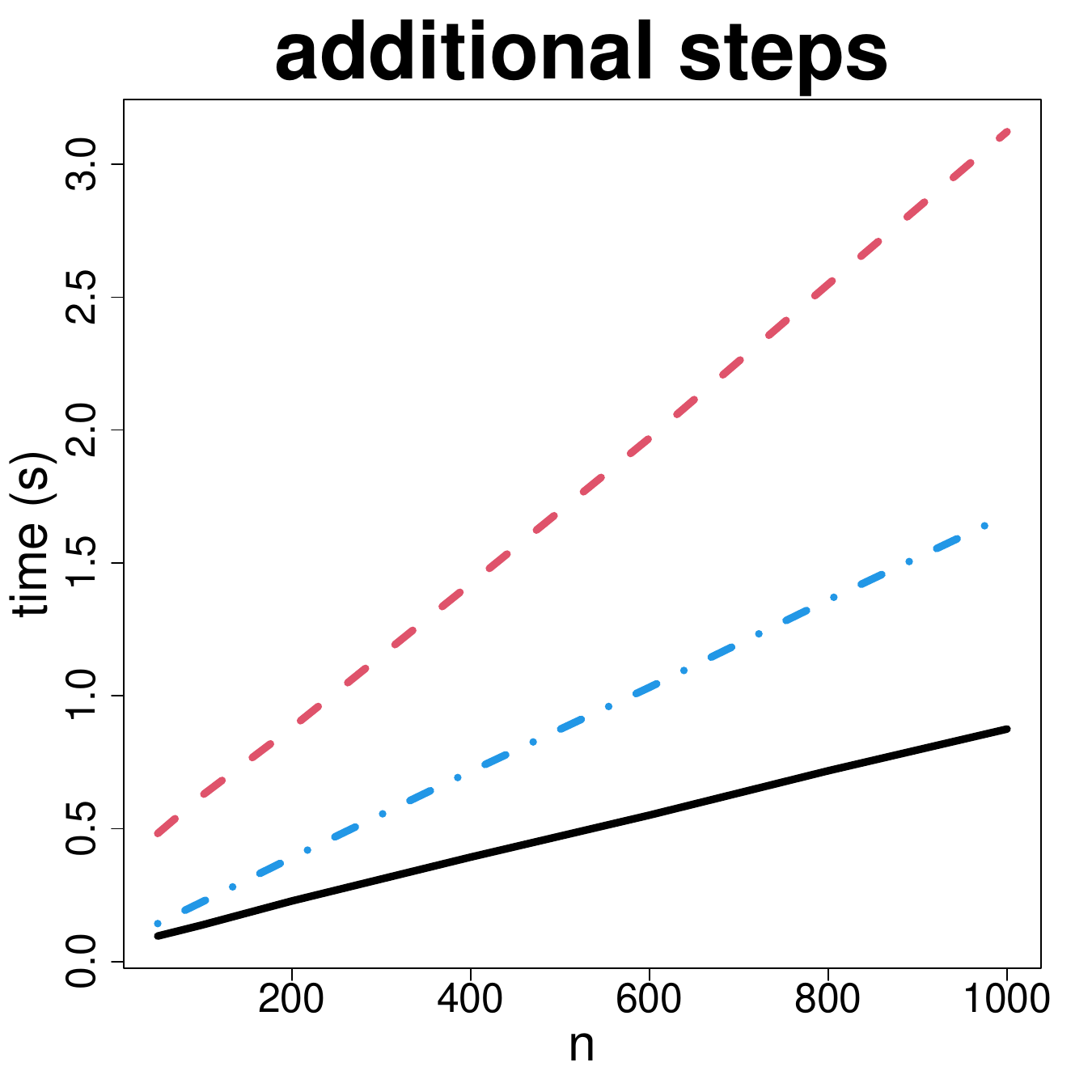} &
\includegraphics[width=.3\textwidth] {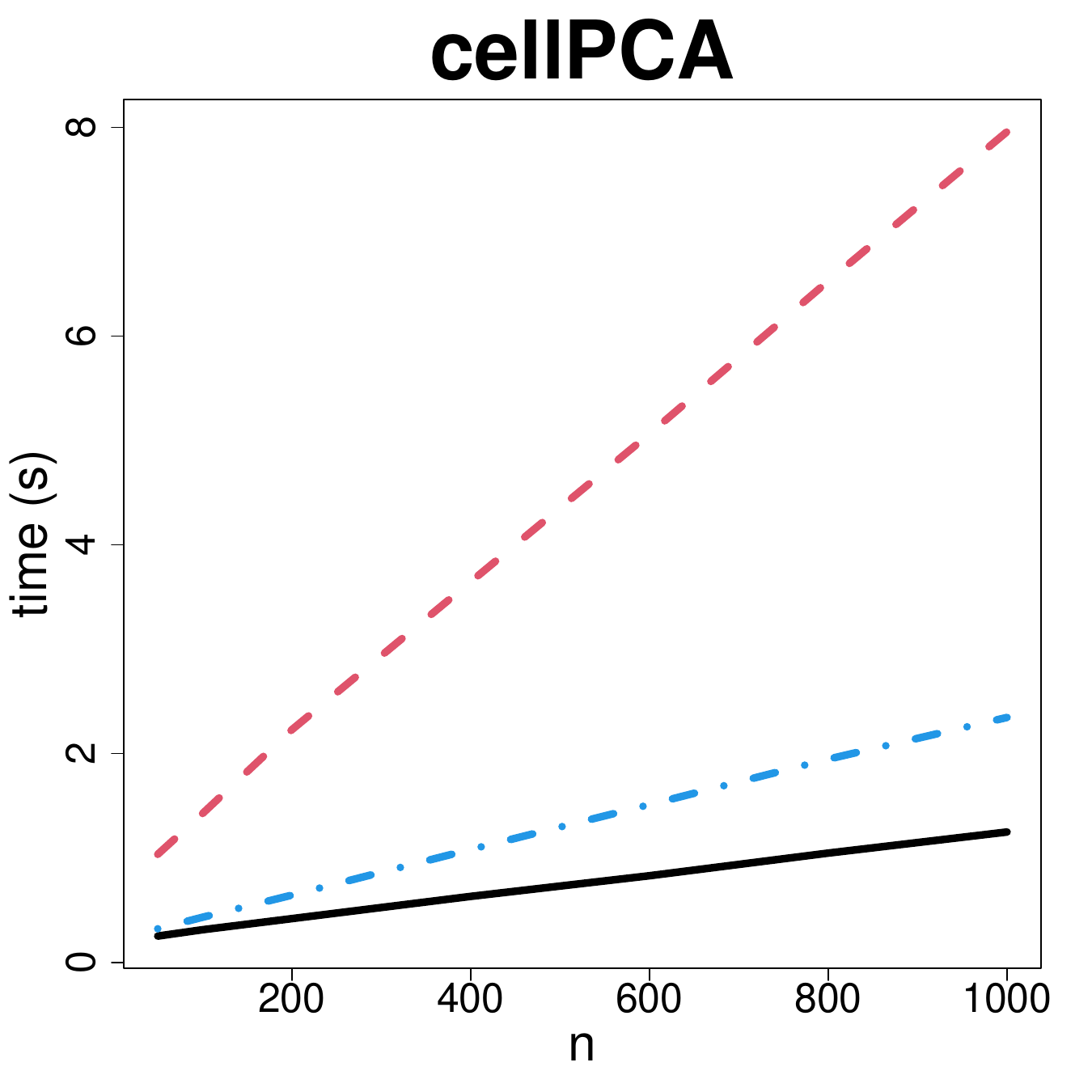}  \\
\includegraphics[width=.3\textwidth] {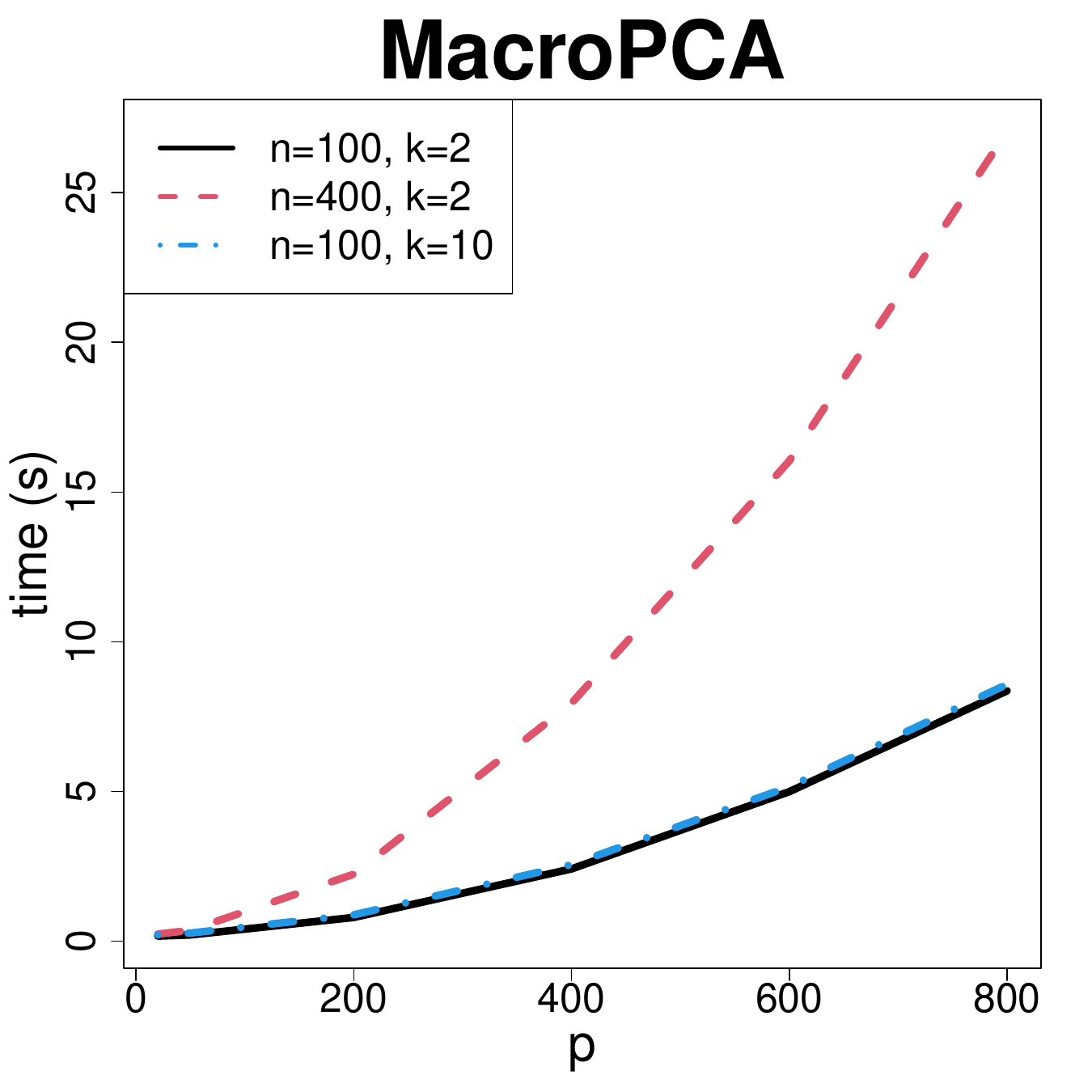} &
\includegraphics[width=.3\textwidth] {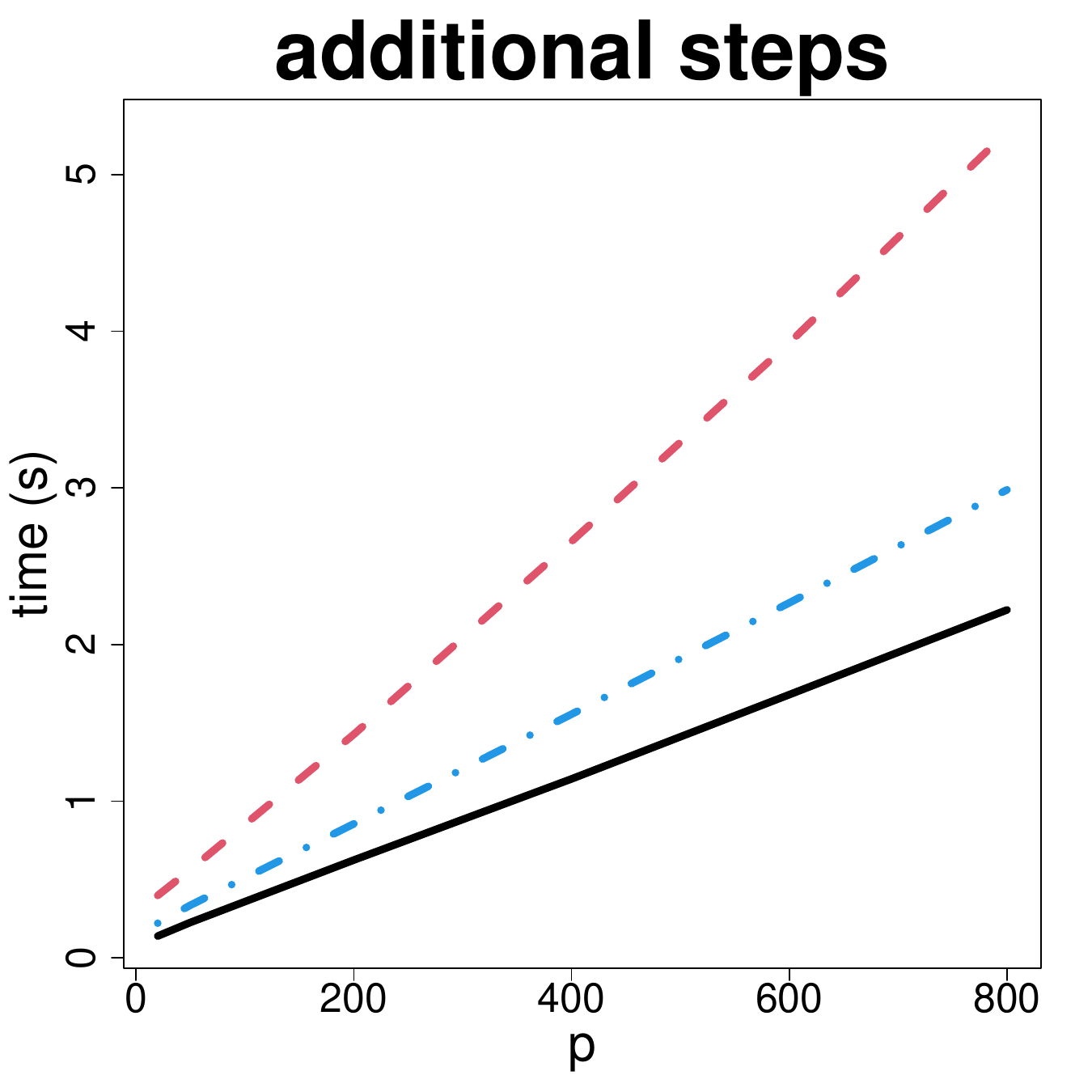} &
\includegraphics[width=.3\textwidth] {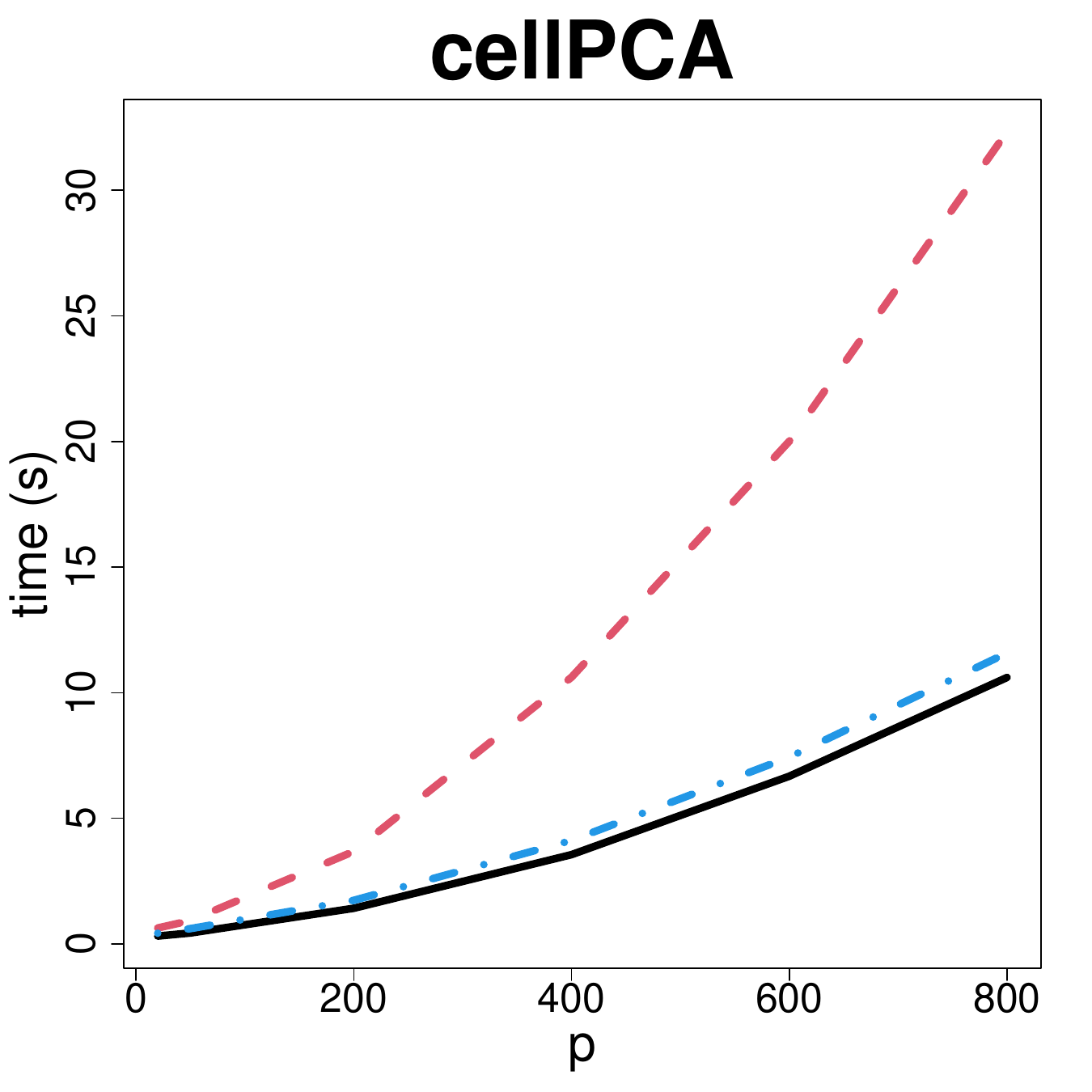} 
\end{tabular}
\caption{Average computation times in seconds 
of the different steps of cellPCA, as a function of 
the number of cases $n$ (top row) and the dimension 
$p$ (bottom row).}
\label{fig:results_time}
\end{figure}

For the space complexity of cellPCA we look
at each of its steps again. The first step of the
initial estimator MacroPCA is the DDC method, whose 
space complexity is $O(np)$ for large $n$ and $p$.
The other steps of MacroPCA also have space 
complexity $O(np)$. Its DetMCD component requires 
$O(nk + k^2 + k)$,
but since $k \leqslant \min(n, p)$ this is also at 
most $O(np)$. Therefore, the overall space 
complexity of MacroPCA remains $O(np)$. Similarly, 
updating $\bU, \bV, \bmu, \bW$ requires $O(np)$ 
space. Computing robust scales has a space 
complexity of $O(p)$, while estimating the center 
and principal directions requires $O(np)$. Hence, 
the total space complexity of cellPCA is $O(np)$. 
Since this is also the space complexity of the 
dataset, the algorithm does not add to it.

\clearpage
\section{Additional Results for the Ionosphere Data}
\label{app:ionos}

For cases 119 and 138 of the ionosphere data, Figure~\ref{fig:ionos_imp} illustrates that each imputed cell $\impx_i$ can be seen as a weighted average of the observed $x_{ij}$ and the fitted $\hx_{ij}$. Only a subset of the variables are shown. The plot also displays the cellwise weights $w^{\cell}_{ij}$\,. The smaller the weight, the more the imputed value becomes close to the fitted value. For $w^{\cell}_{ij} = 1$ the imputed cell is the observed value, and for $w^{\cell}_{ij} = 0$ the imputed cell is the fitted value. 

\vspace{-2mm}
\begin{figure}[!ht]
\centering
\includegraphics[width=0.63\textwidth]{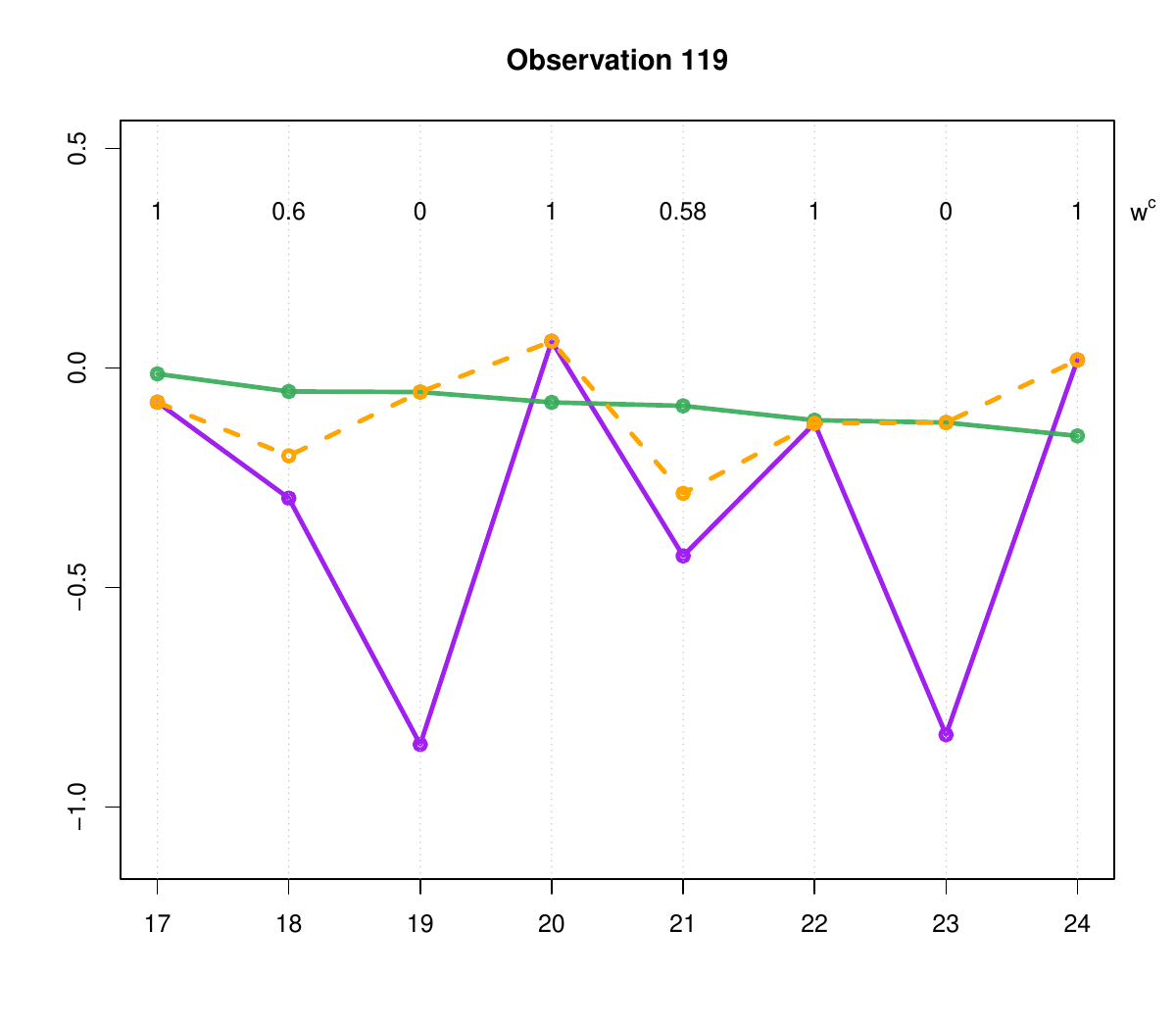}\\

\vspace{-6mm}
\includegraphics[width=0.63\textwidth]{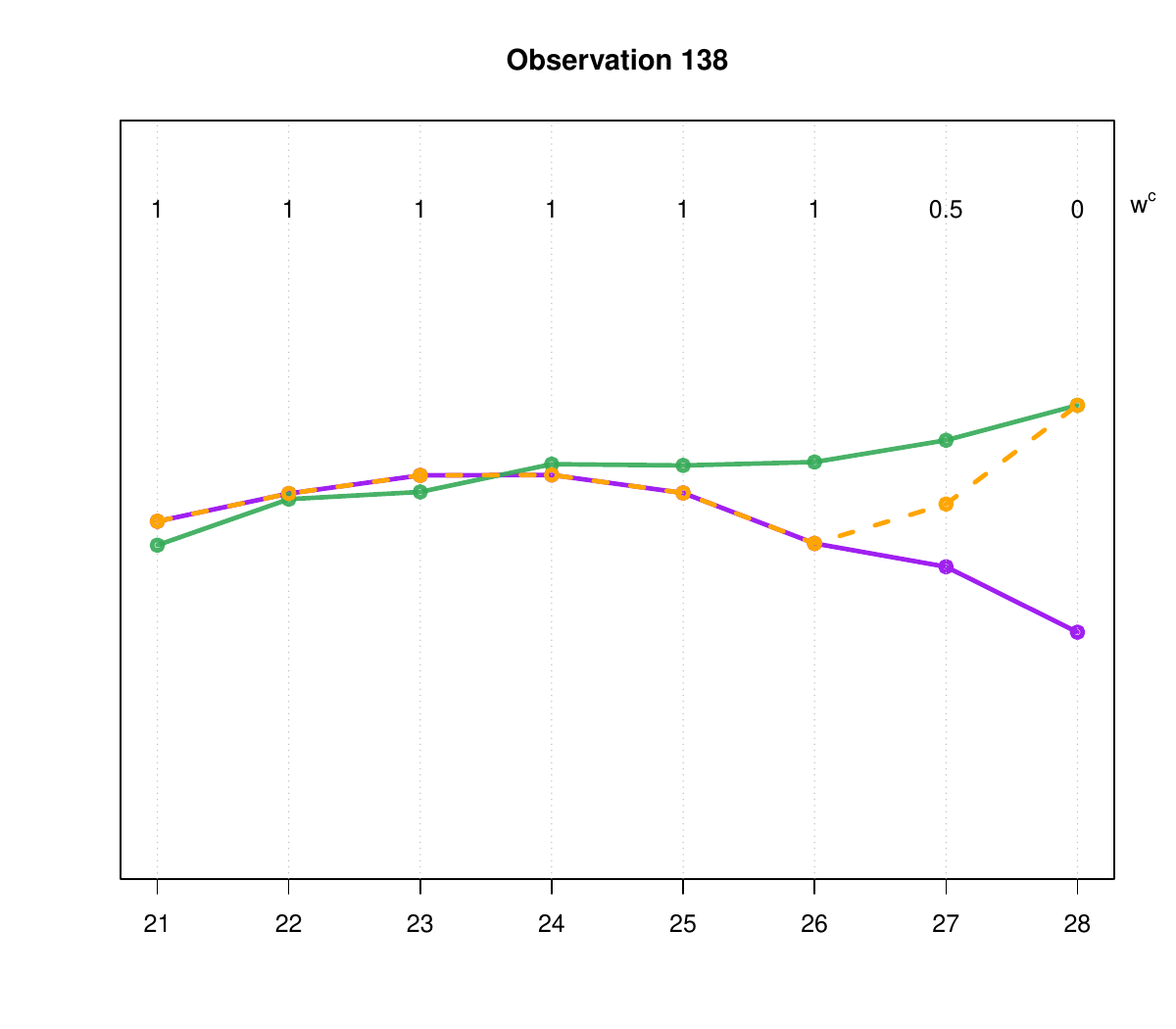}\\

\vspace{-6mm}
\caption{Observed (purple), fitted (green) and 
imputed (orange dashed) curves of cases 119 
and 138 of the Ionosphere dataset. Above each cell 
we see its cellwise weight.}
\label{fig:ionos_imp}
\end{figure}

\newpage
\section{Additional simulation results}
\label{app:addsim}

Figure~\ref{fig:results_p20_NA0} shows the median 
angle and MSE for the A09-type covariance model 
in the presence of either cellwise outliers, 
casewise outliers, or both, without NAs, this time 
for $p=20$. For $\gamma_{\cell} = 0$ no cellwise 
outliers were generated, and for 
$\gamma_{\case} = 0$ no casewise outliers.
Figure~\ref{fig:results_p20_NA0.2} shows the 
corresponding results when $\eps^{\obs}=20\%$ of 
randomly selected cells were made NA. All of 
these curves look a lot like those for $p=200$
in Section~\ref{sec:simulation} of the paper.\\

\vspace{5mm}

\begin{figure}[!ht]
\centering
\begin{tabular}{ccc}
   \large \textbf{Cellwise}  & \large \textbf{Casewise} &\large{\textbf{Casewise \& Cellwise}} \\
   [-4mm]
  \includegraphics[width=.3\textwidth]
  {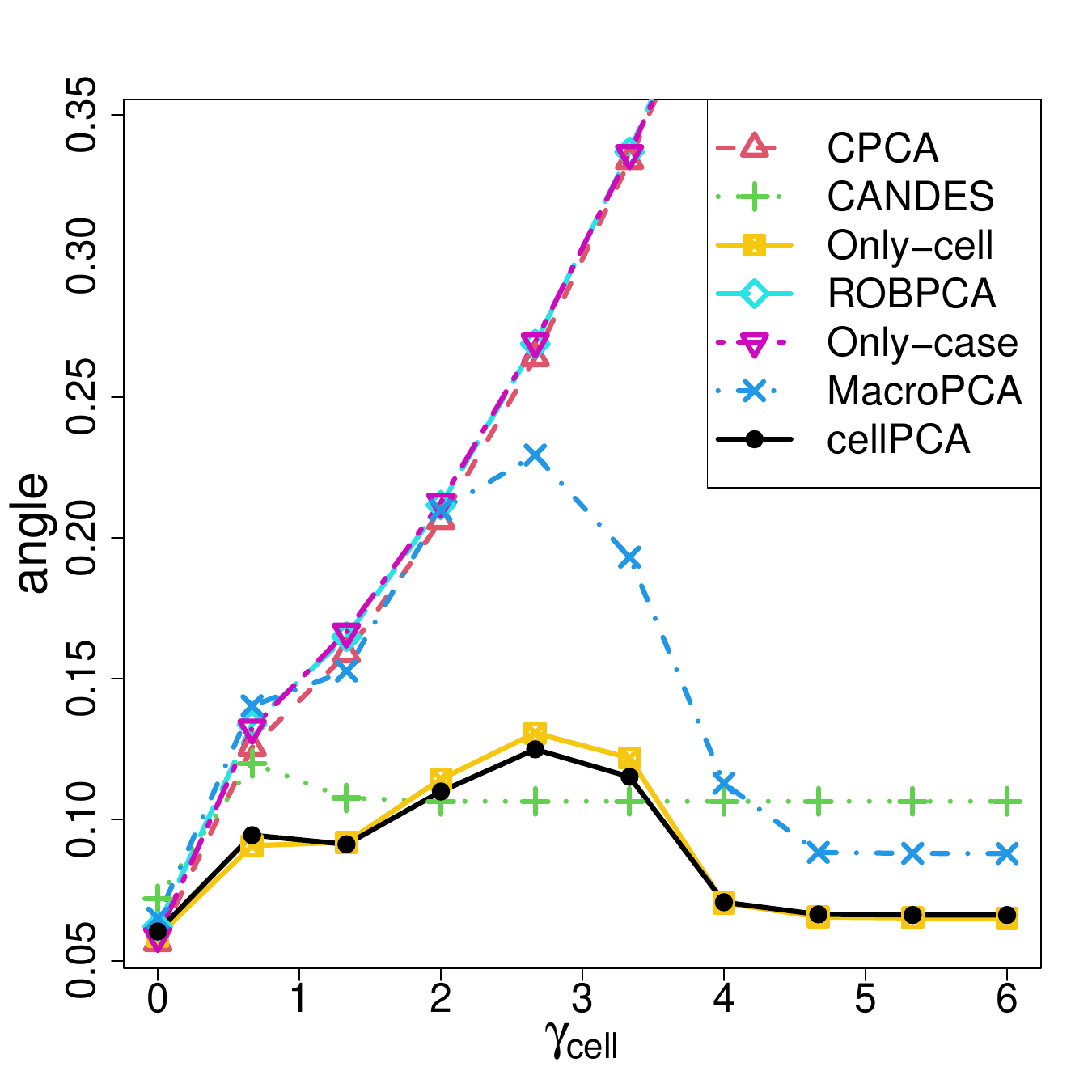} &\includegraphics[width=.3\textwidth]
  {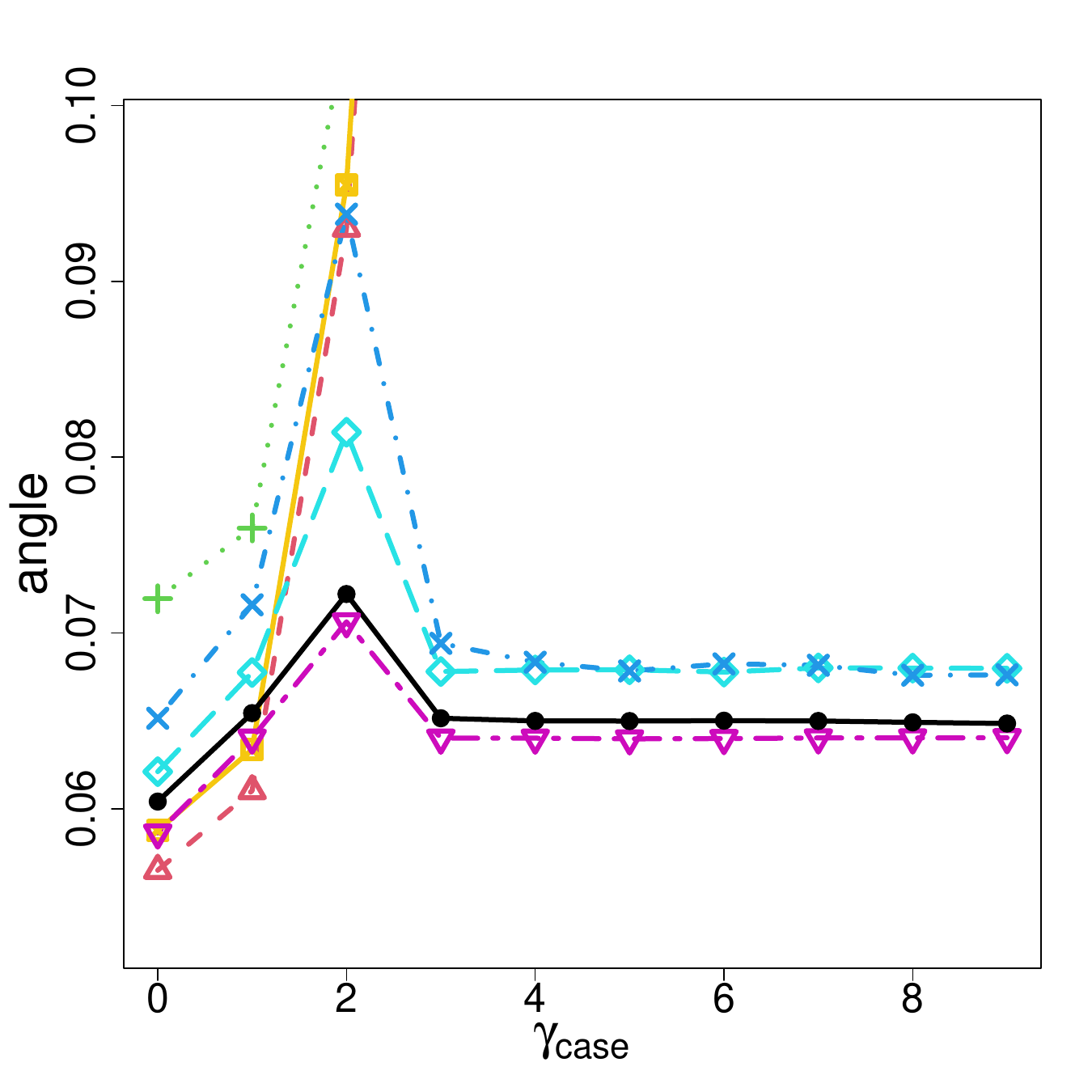} &\includegraphics[width=.3\textwidth]
  {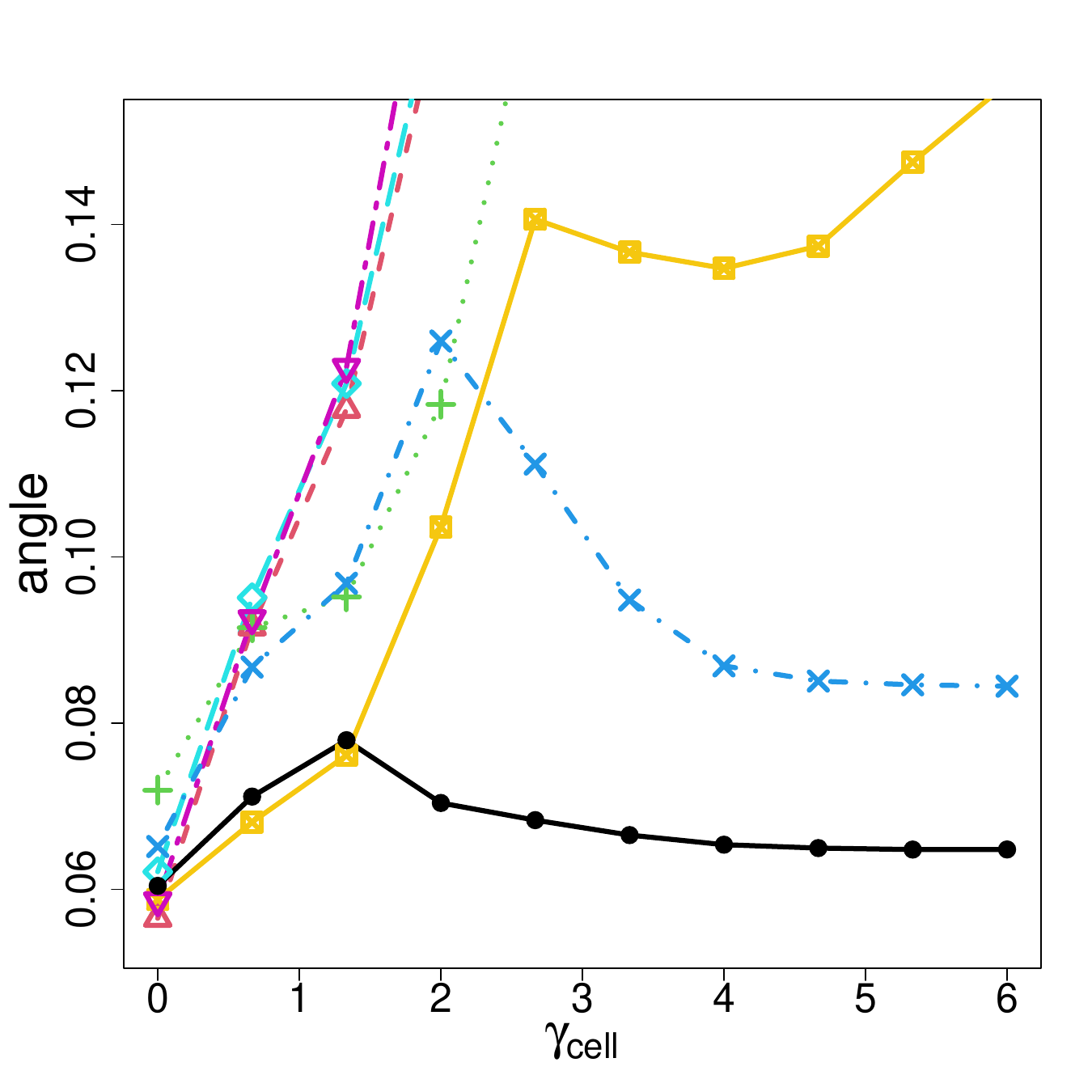}  \\
   [-4mm]
  \includegraphics[width=.3\textwidth]
  {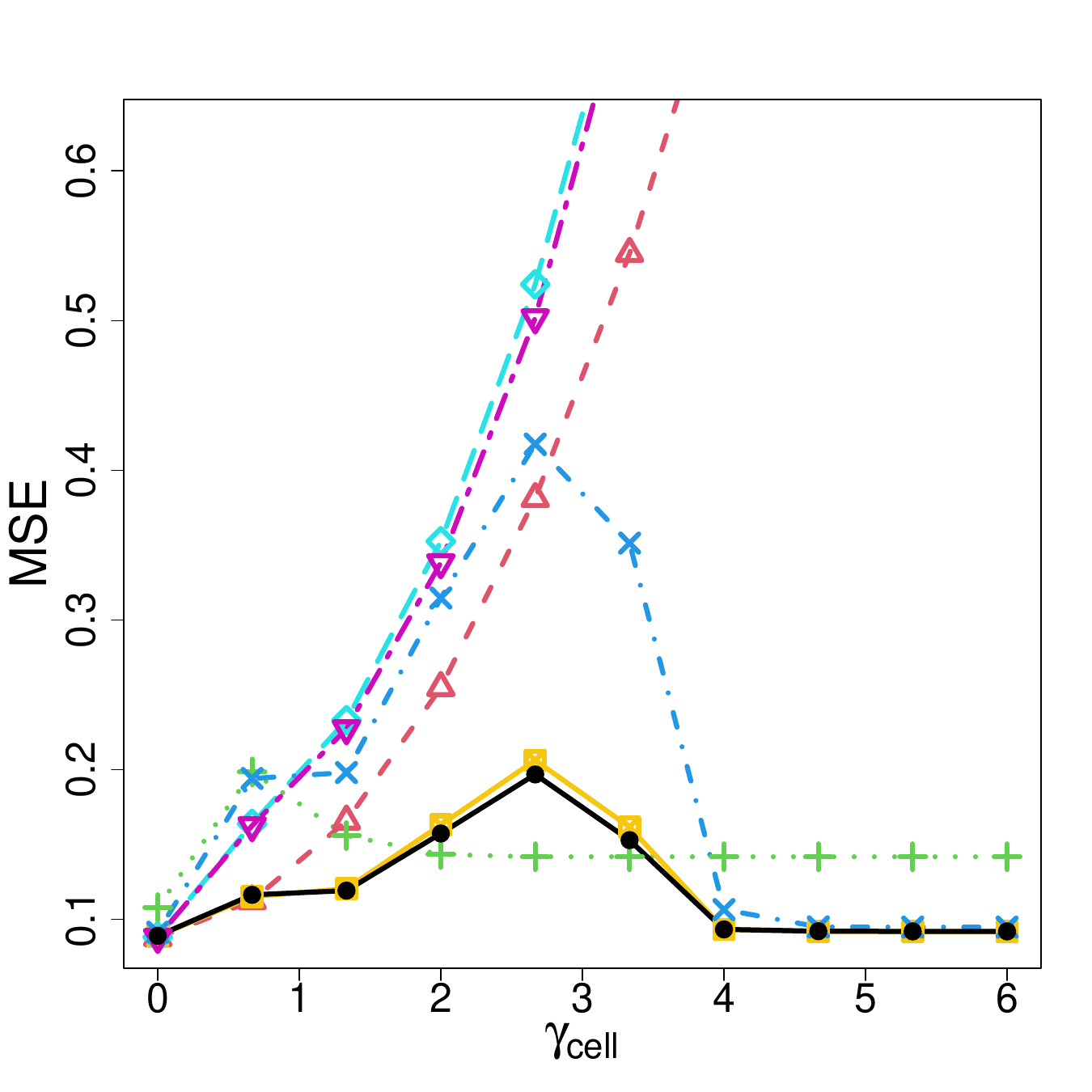} &\includegraphics[width=.3\textwidth]
  {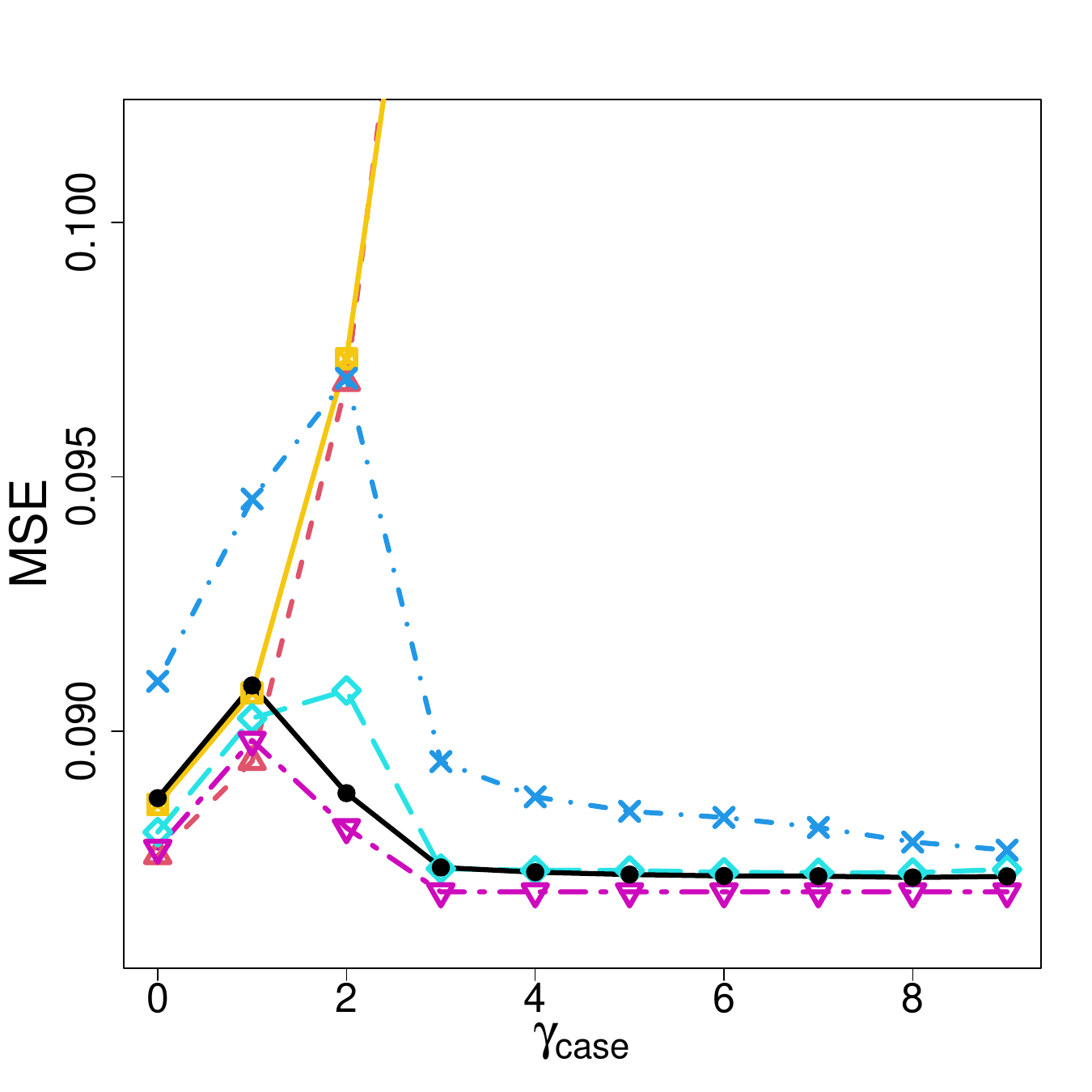} &\includegraphics[width=.3\textwidth]
  {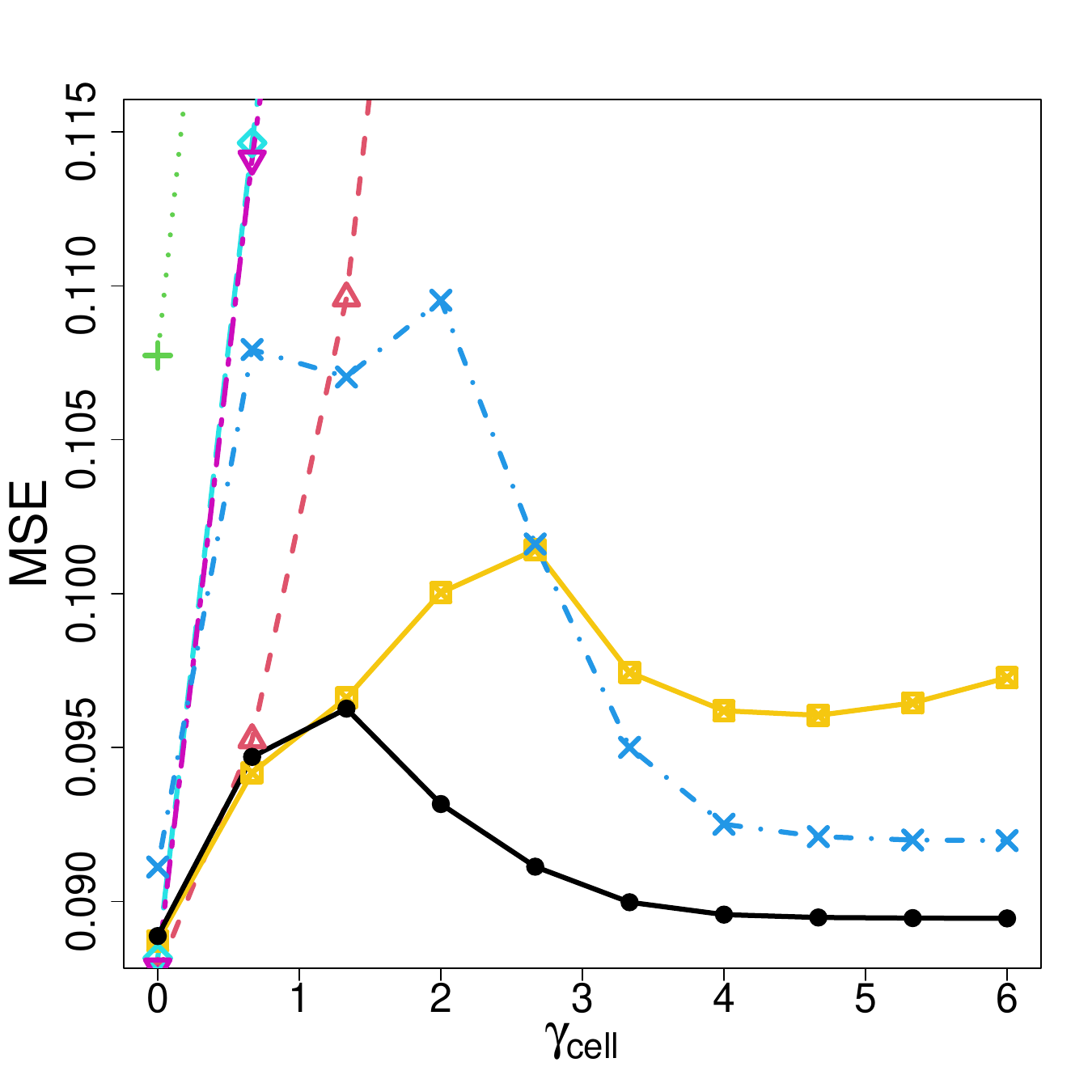} 
\end{tabular}
\caption{Median angle (top) and MSE (bottom) attained by CPCA, CANDES, Only-cell, ROBPCA, Only-case, MacroPCA, and cellPCA in the presence of either cellwise outliers, casewise outliers, or both. The covariance model was A09 with $n=100$ and $p=20$, without NAs.}
\label{fig:results_p20_NA0}
\end{figure}

\begin{figure}[!ht]
\centering
\begin{tabular}{ccc}
   \large \textbf{Cellwise}  & \large \textbf{Casewise} &\large{\textbf{Casewise \& Cellwise}} \\
   [-4mm]
  \includegraphics[width=.3\textwidth]
  {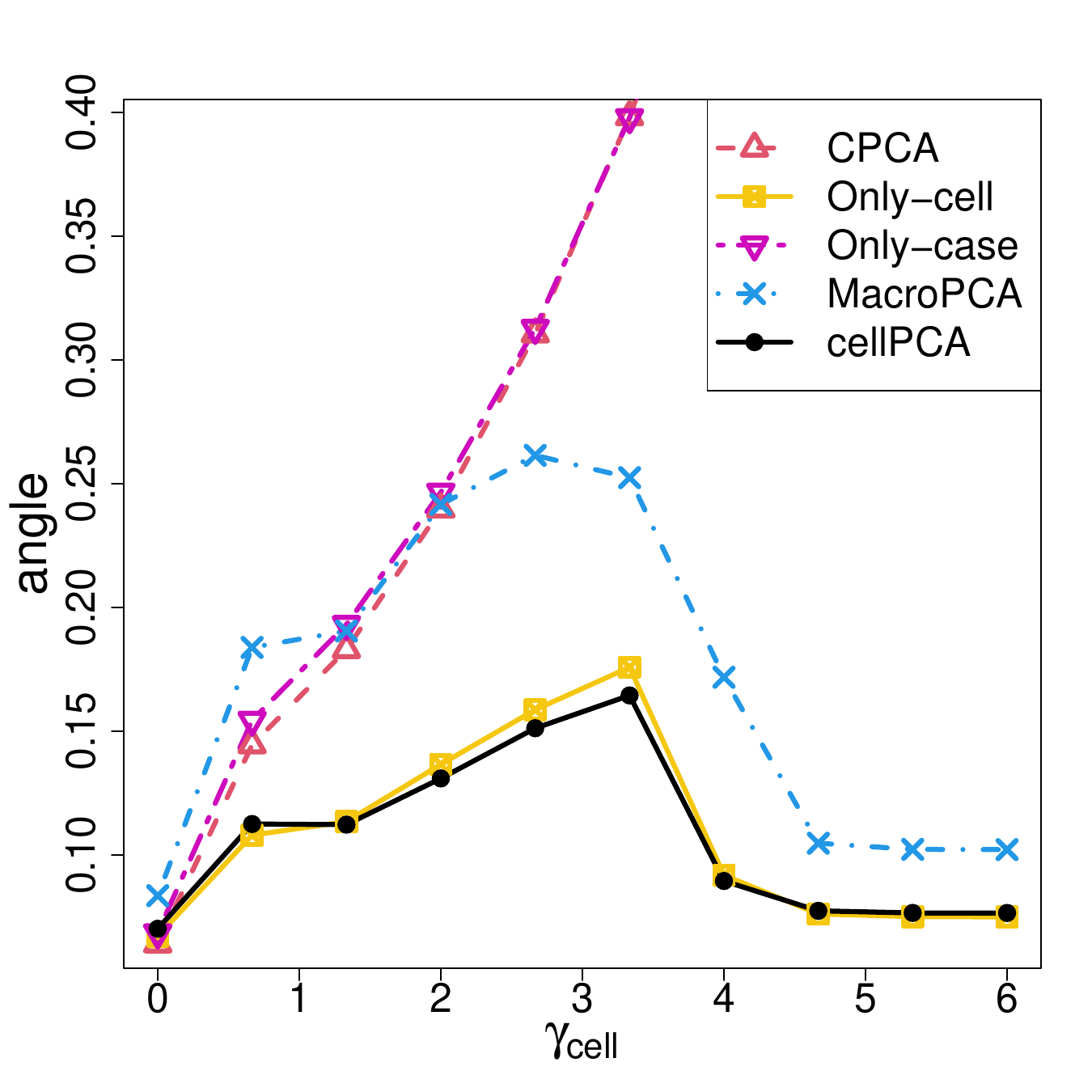} &\includegraphics[width=.3\textwidth]
  {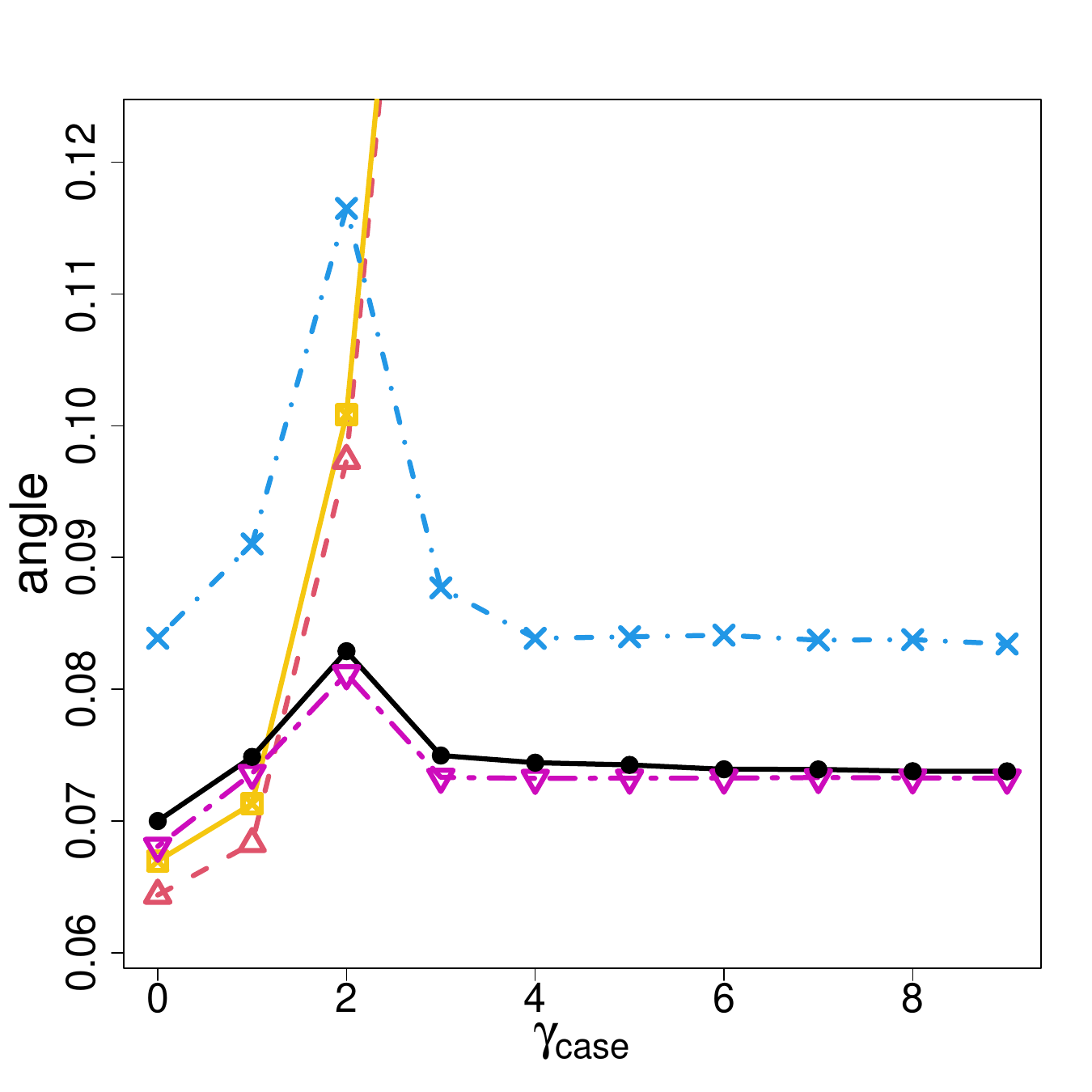} &\includegraphics[width=.3\textwidth]
  {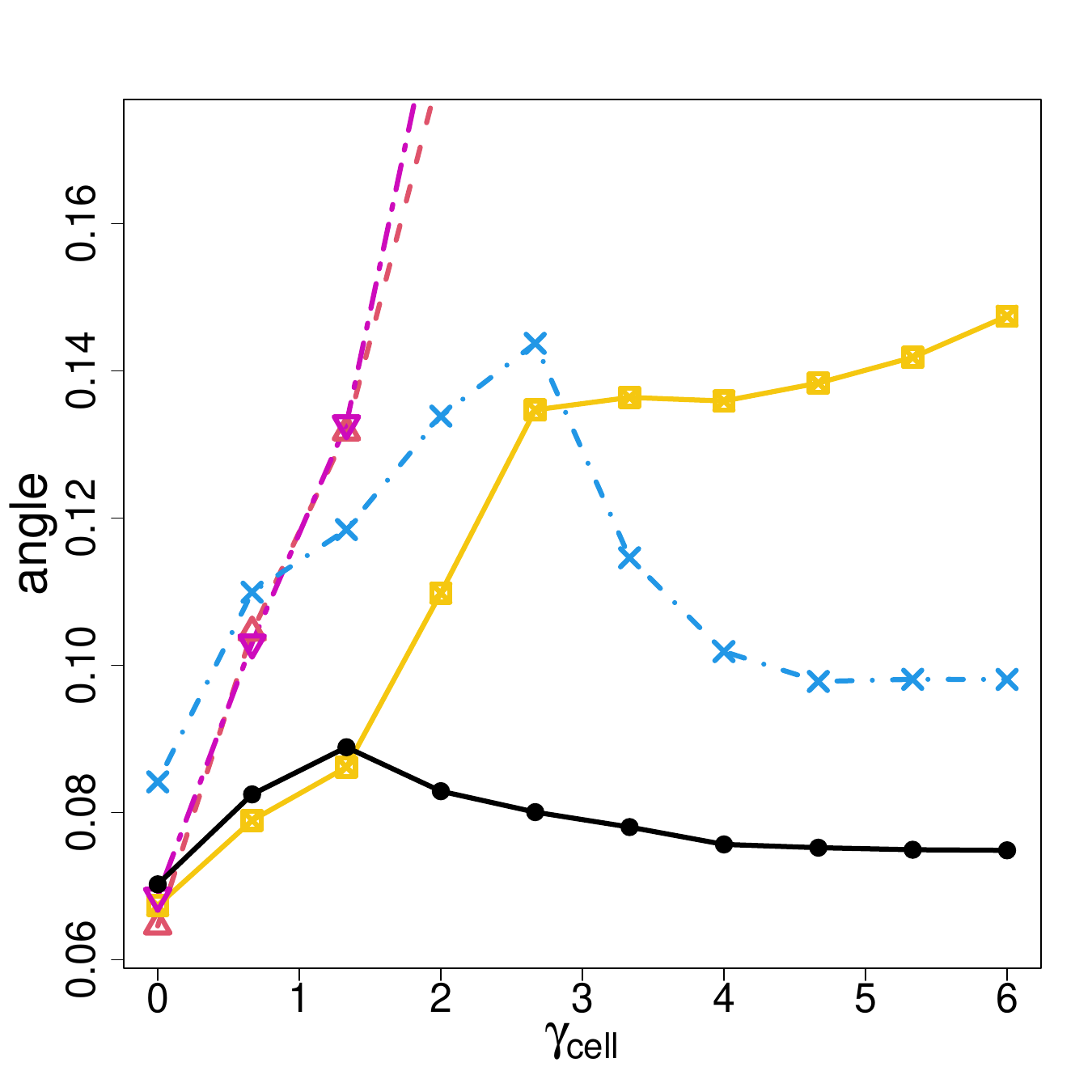}  \\
   [-4mm]
  \includegraphics[width=.3\textwidth]
  {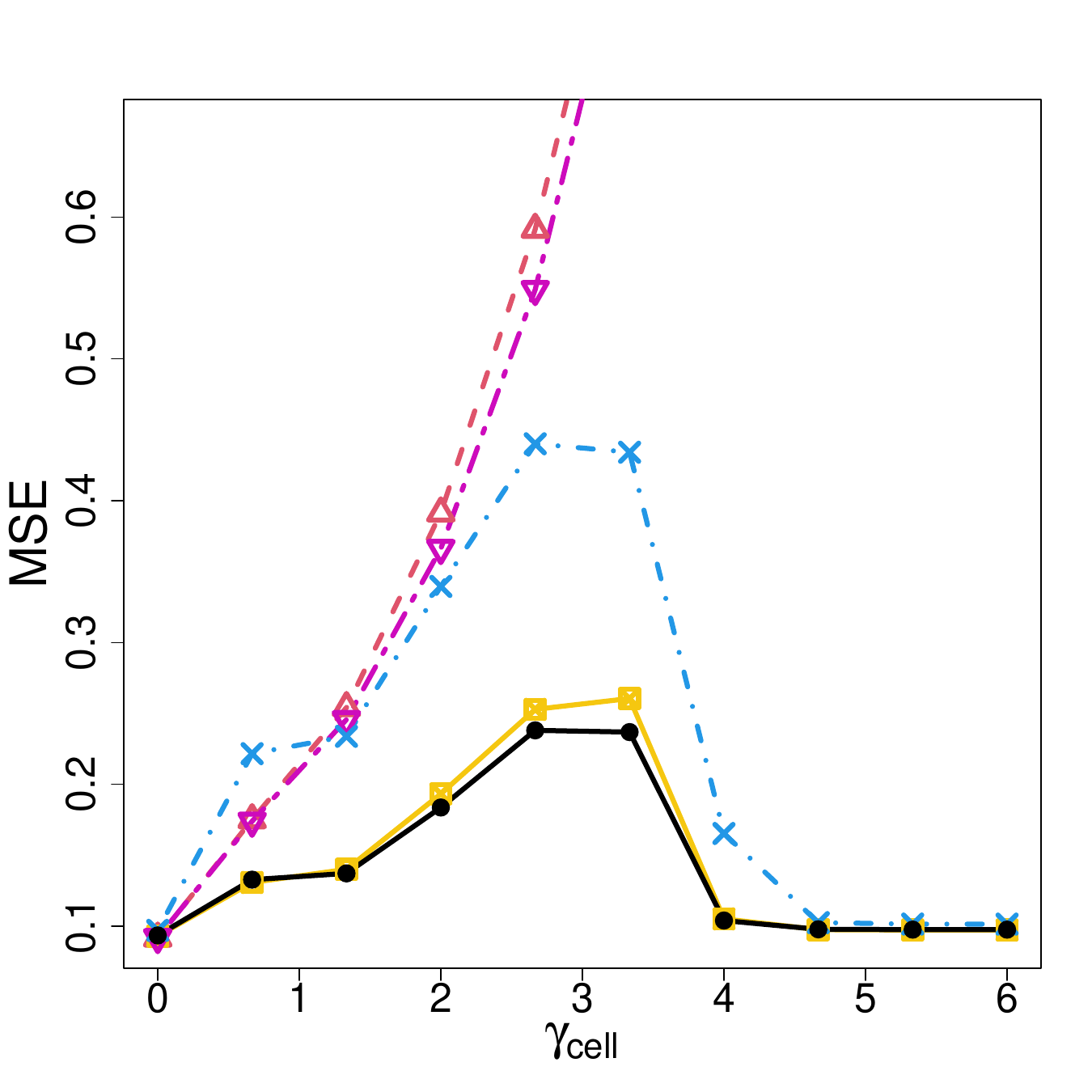} &\includegraphics[width=.3\textwidth]
  {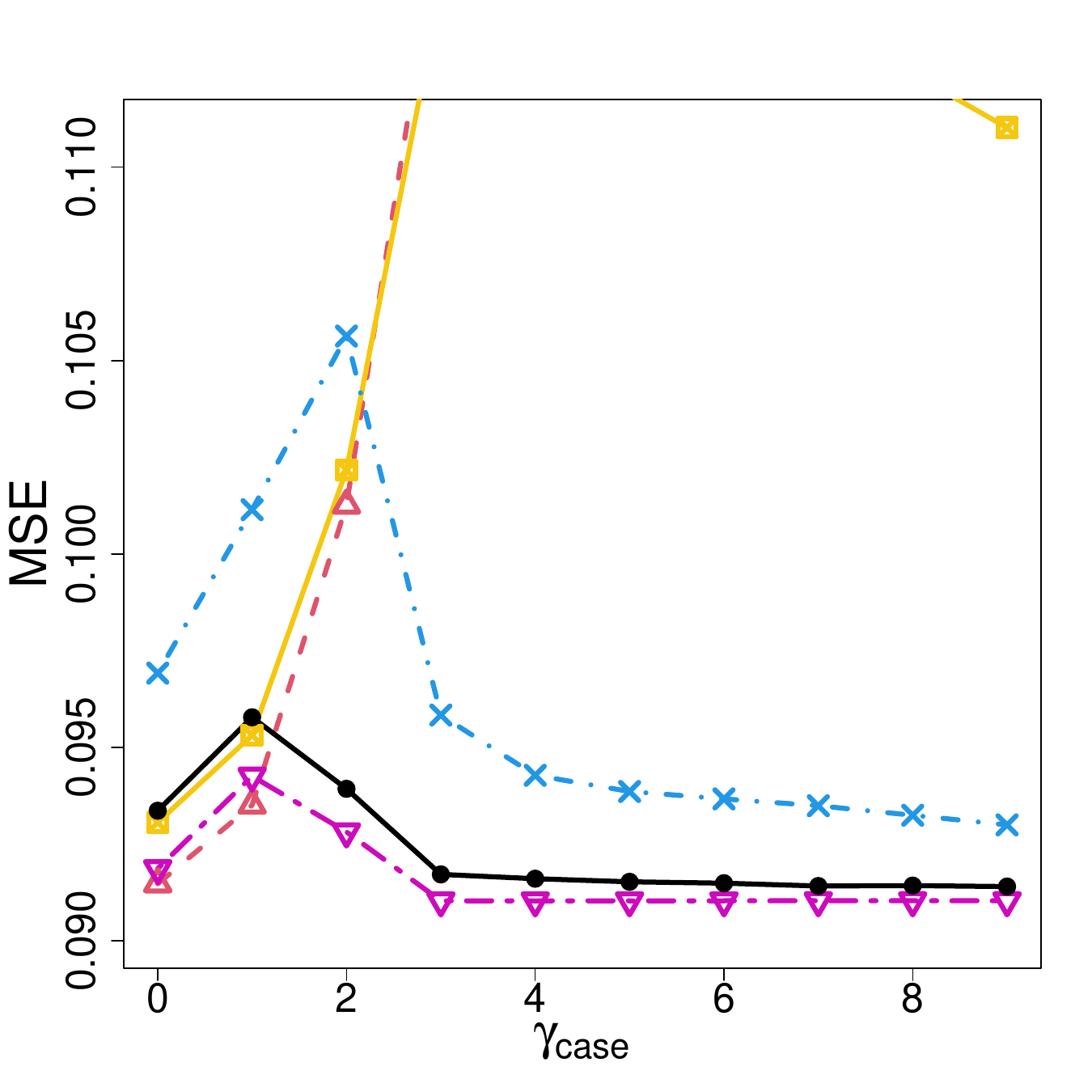} &\includegraphics[width=.3\textwidth]
  {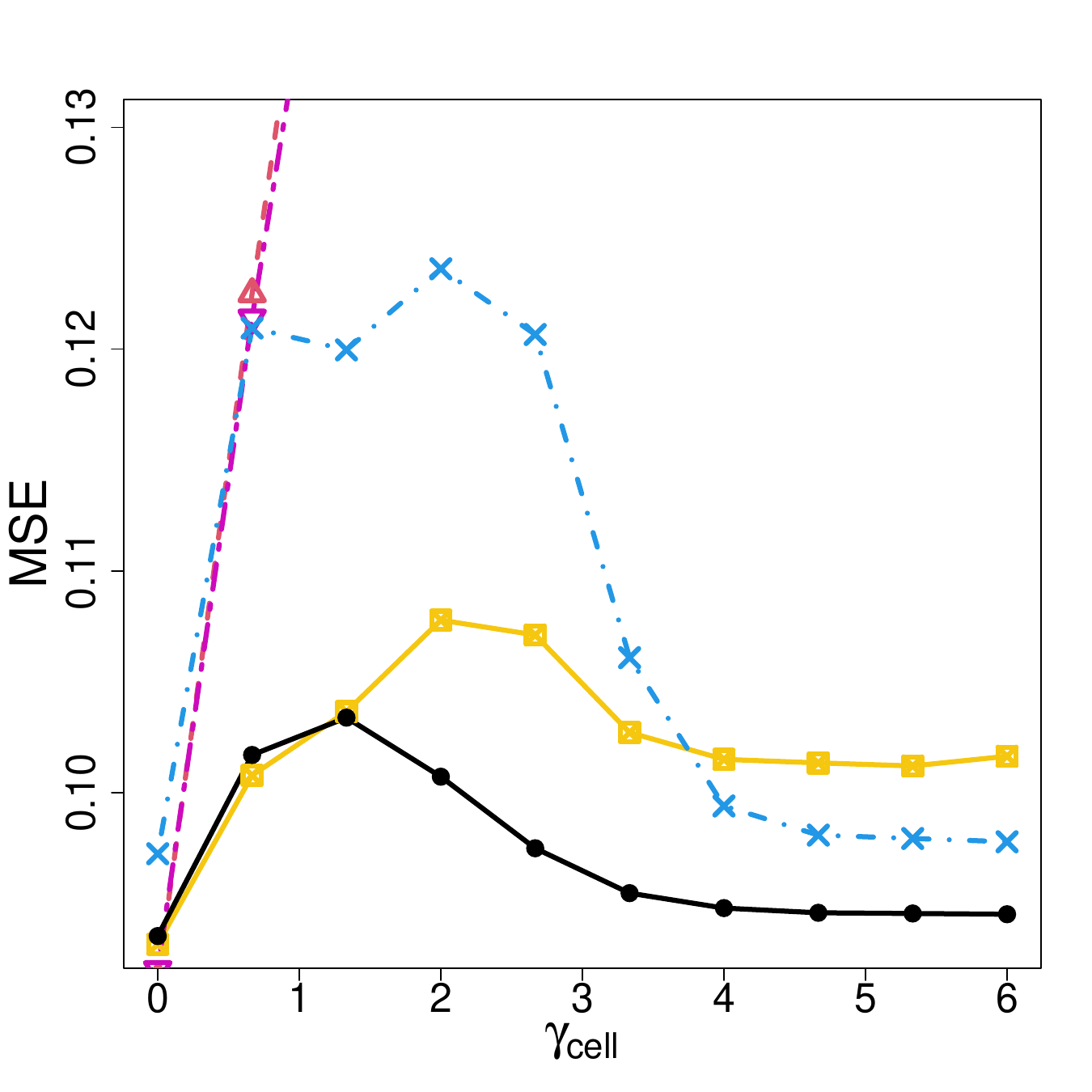} 
\end{tabular}
\caption{Median angle (top) and MSE (bottom) attained by CPCA, Only-cell, Only-case, MacroPCA, and cellPCA in the presence of either cellwise outliers, casewise outliers, or both. The covariance model was A09 with $n=100$ and $p=20$, and $20\%$ of randomly selected cells were set to NA.}
\label{fig:results_p20_NA0.2}
\end{figure}

\clearpage

The second type of covariance matrix is based on the 
random correlation matrices of 
\cite{agostinelli2015robust} and will be called ALYZ. 
These correlation matrices are turned into covariance 
matrices with other eigenvalues. More specifically, 
the matrix $\bL$ in the spectral decomposition 
$\bO \bL \bO^T$ of the correlation matrix is replaced 
by\linebreak 
$\diag(9.57, 6.70, 0.11, 0.10, \dots, 0.10)$ 
for $p=20$ and by
$\diag(104.86, 73.41, 0.11, 0.10, \dots, 0.10)$
for $p=200$. These are the same eigenvalues as
used in Figure~\ref{fig:results_p200_NA0}  
and Figures
\ref{fig:results_p20_NA0}-\ref{fig:results_p20_NA0.2}, 
but now the directions of the eigenvectors vary a lot. 
Again the first two components explain 90\% of 
the variance, and we set $\rk=2$ in the simulation.

Again three contamination types are considered. In 
the cellwise outlier scenario we randomly replace 
$\eps^{\cell}=20\%$ of the cells $x_{ij}$ with 
$\gamma_{\cell}\sigma_j$, where $\gamma_{\cell}$ 
varies from 0 to 6 and
$\sigma_j^2$ is the $j$th diagonal element of 
$\bSigma$. Note that the diagonal entries of
this $\bSigma$ are no longer~1 as in A09.
In the casewise outlier setting 
$\eps^{\case}=20\%$ of the cases are
generated from 
$N\left(\gamma_{\case} (\be_1+\be_{\rk+1}), 
\bSigma/1.5\right)$, where  
$\be_1$ and $\be_{\rk+1}$ 
are the first and $(\rk+1)$-th eigenvectors of 
$\bSigma$, and $\gamma_{\case}$ varies from 0 to 9 
when $p=20$, and from 0 to 24 when $p=200$.
In the third scenario, the data is contaminated by
$\eps^{\cell}=10\%$ of cellwise outliers as well 
as $\eps^{\case}=10\%$ of casewise outliers. Here 
$\gamma_{\case}=1.5\gamma_{\cell}$ when $p=20$ and 
$\gamma_{\case}=4\gamma_{\cell}$ when $p=200$, where 
$\gamma_{\cell}$ again varies from 0 to 6. When 
$\gamma_{\cell}=\gamma_{\case}=0$ we do not 
contaminate the data.

Repeating the entire simulation with the ALYZ 
covariance model instead of A09 yields 
Figure~\ref{fig:results_NA0_ALYZ} without NAs, and
Figure~\ref{fig:results_NA0.2_ALYZ} with NAs.
These figures are qualitatively similar to those
for A09, and yield the same conclusions.

\begin{figure}[!ht]
\centering
\begin{tabular}{ccc}
   \large \textbf{Cellwise}  & \large \textbf{Casewise} &\large{\textbf{Casewise \& Cellwise}} \\
   [-4mm]
  \includegraphics[width=.3\textwidth]
  {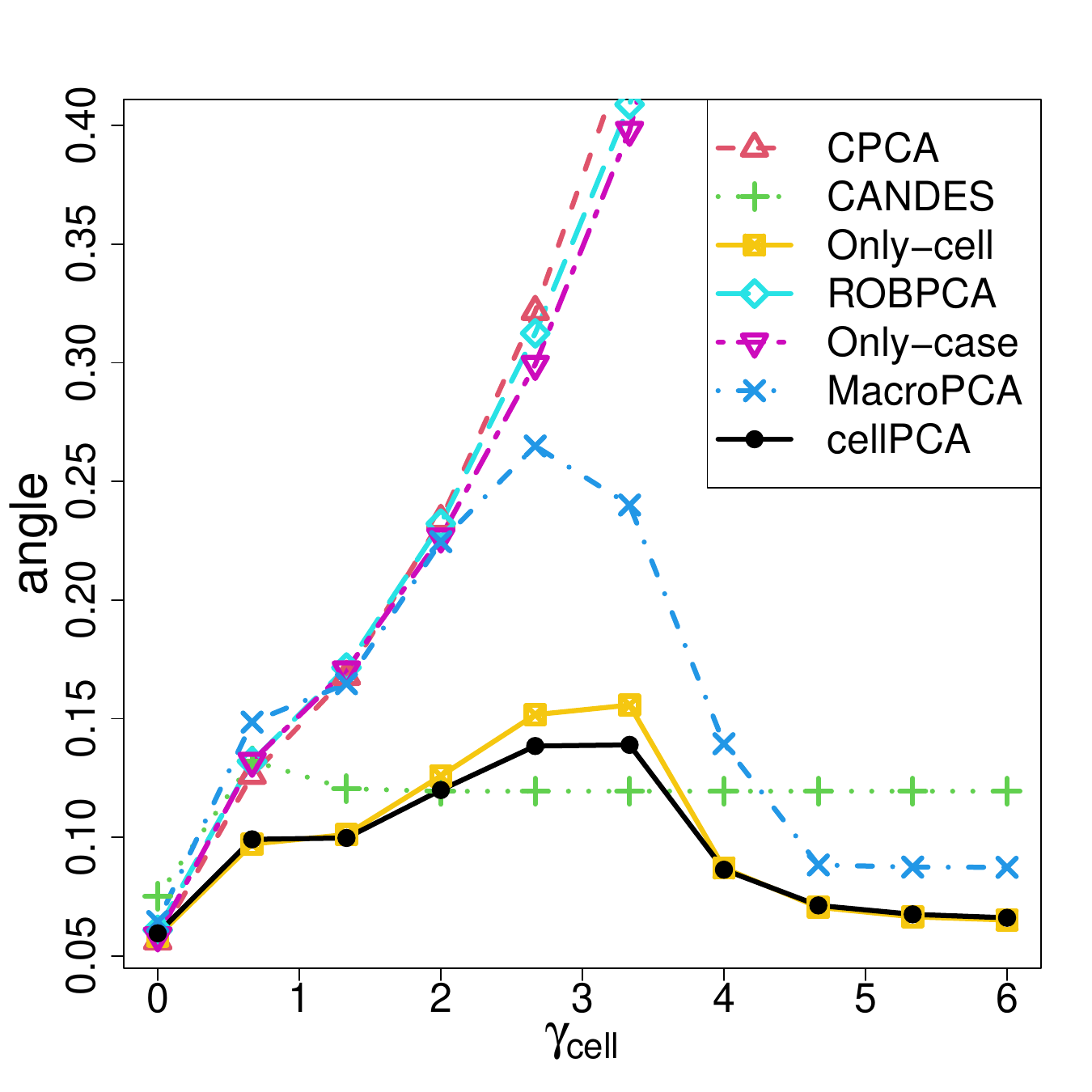} &\includegraphics[width=.3\textwidth]
  {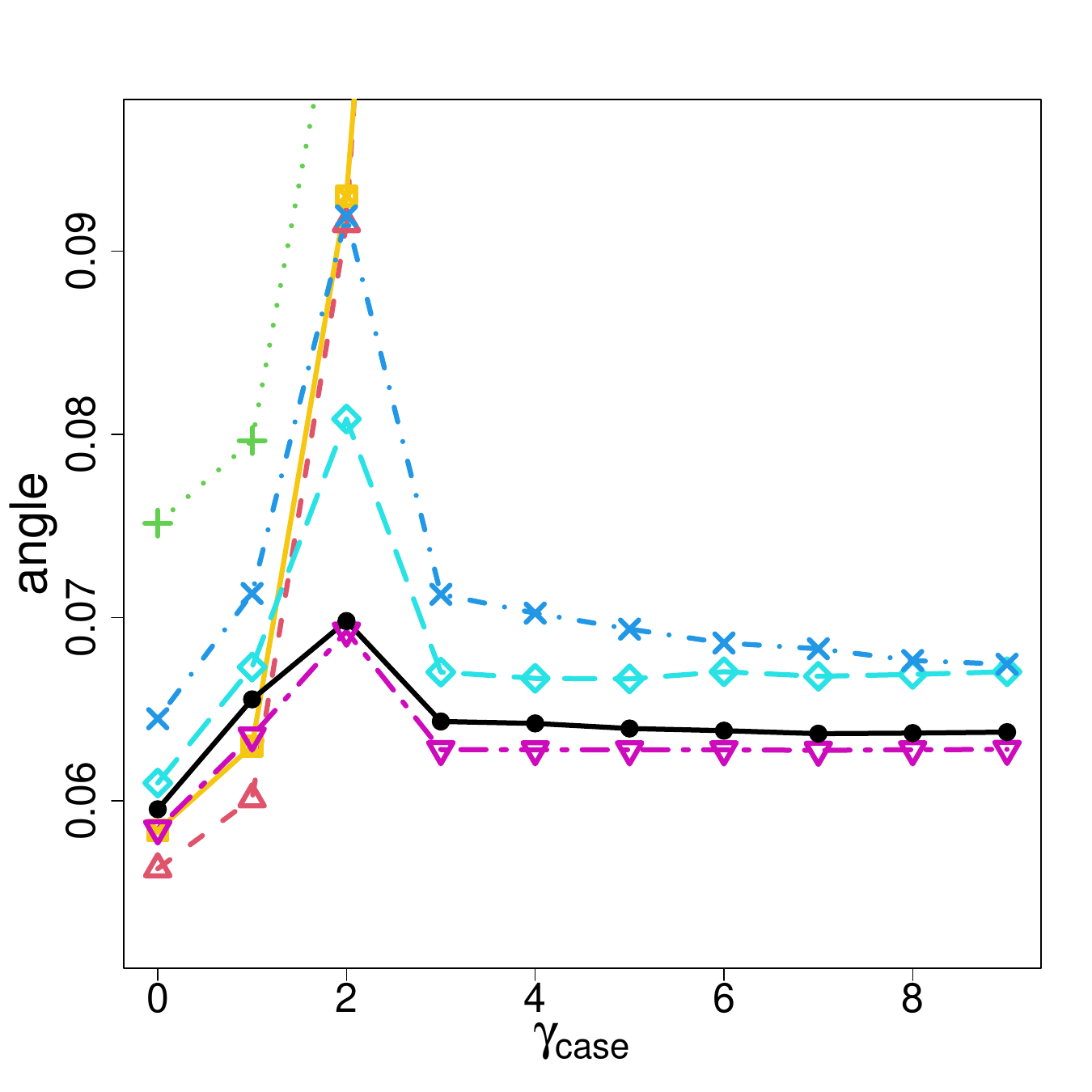} &\includegraphics[width=.3\textwidth]
  {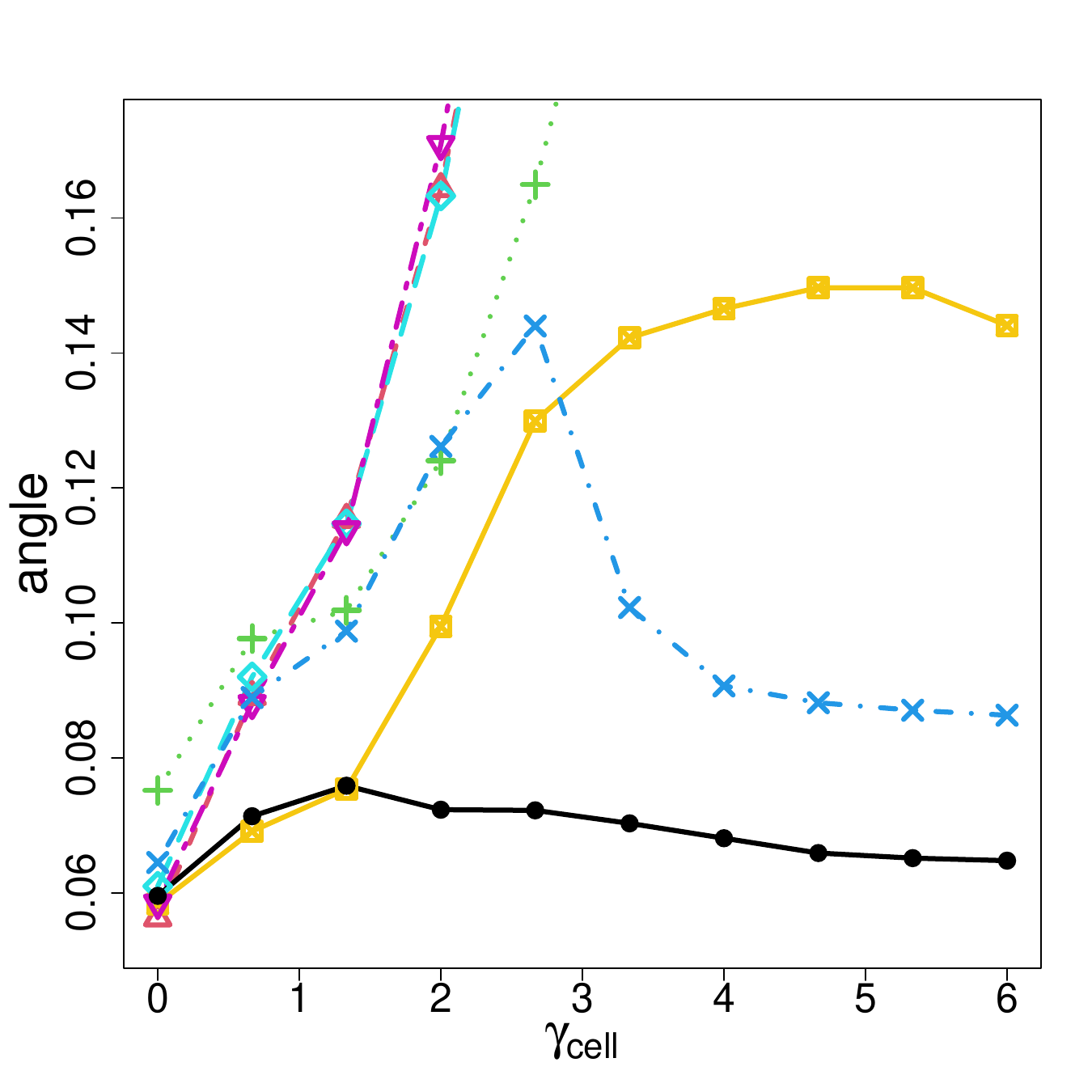}  \\
   [-4mm]
  \includegraphics[width=.3\textwidth]
  {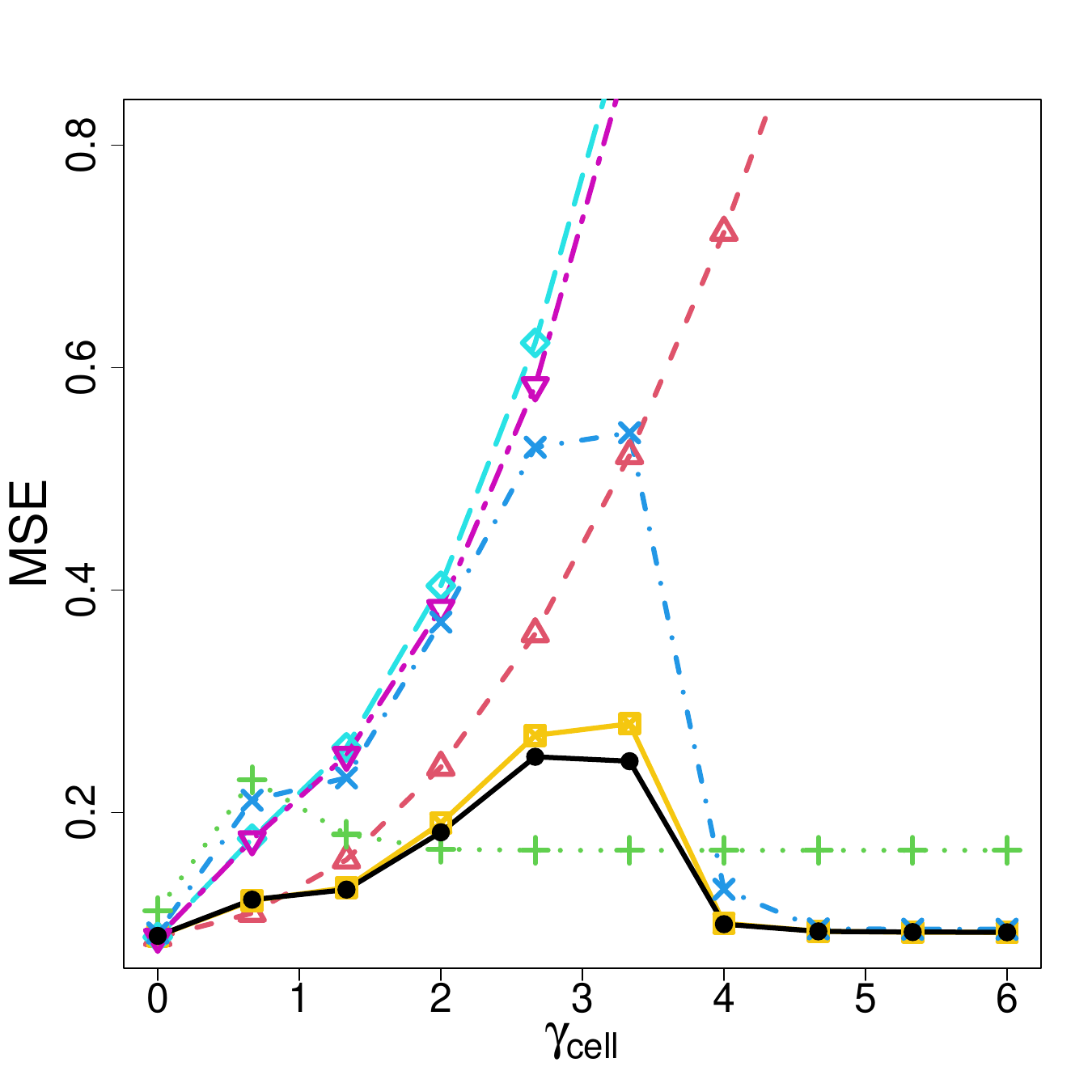} &\includegraphics[width=.3\textwidth]
  {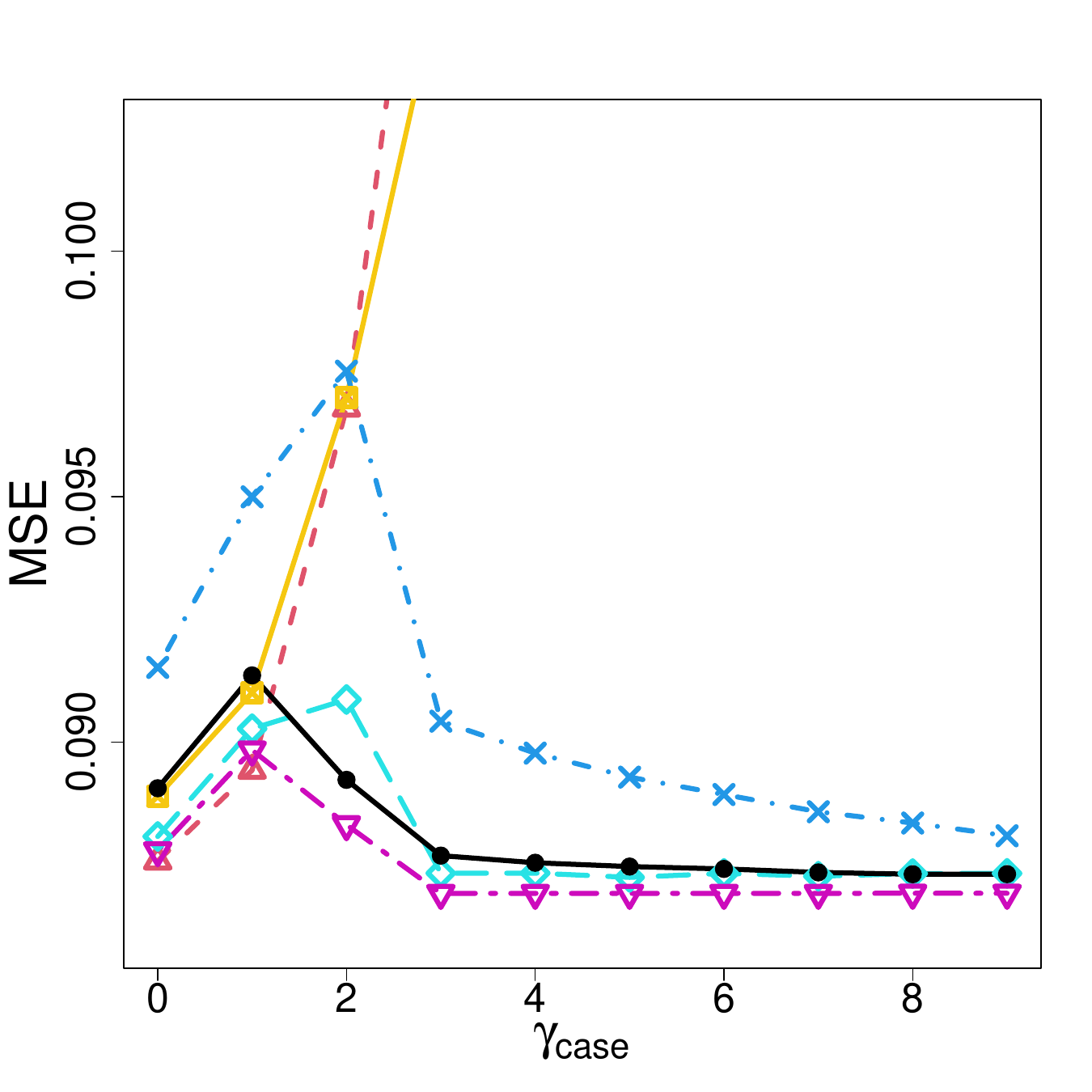} &\includegraphics[width=.3\textwidth]  {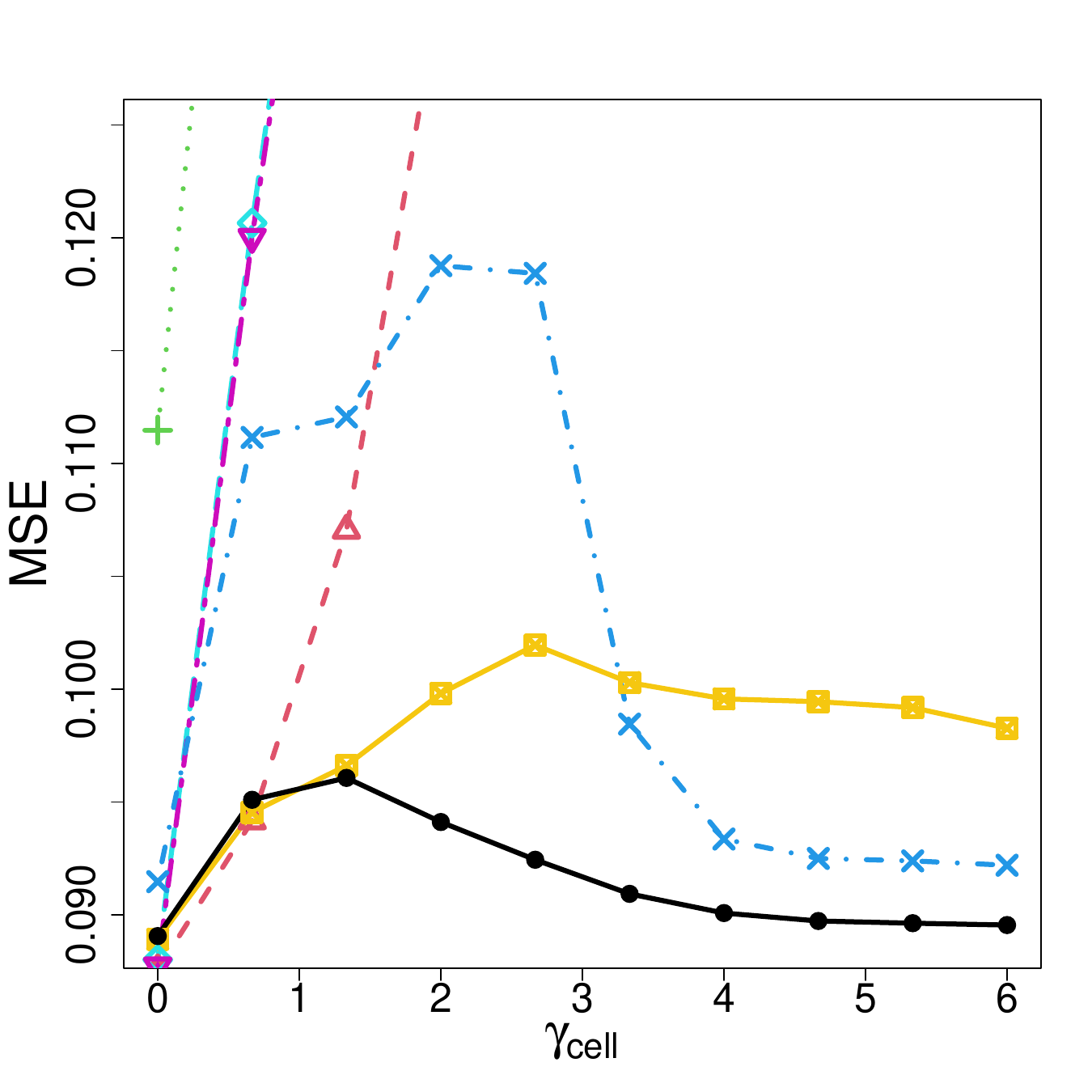} \\

  \includegraphics[width=.3\textwidth]
  {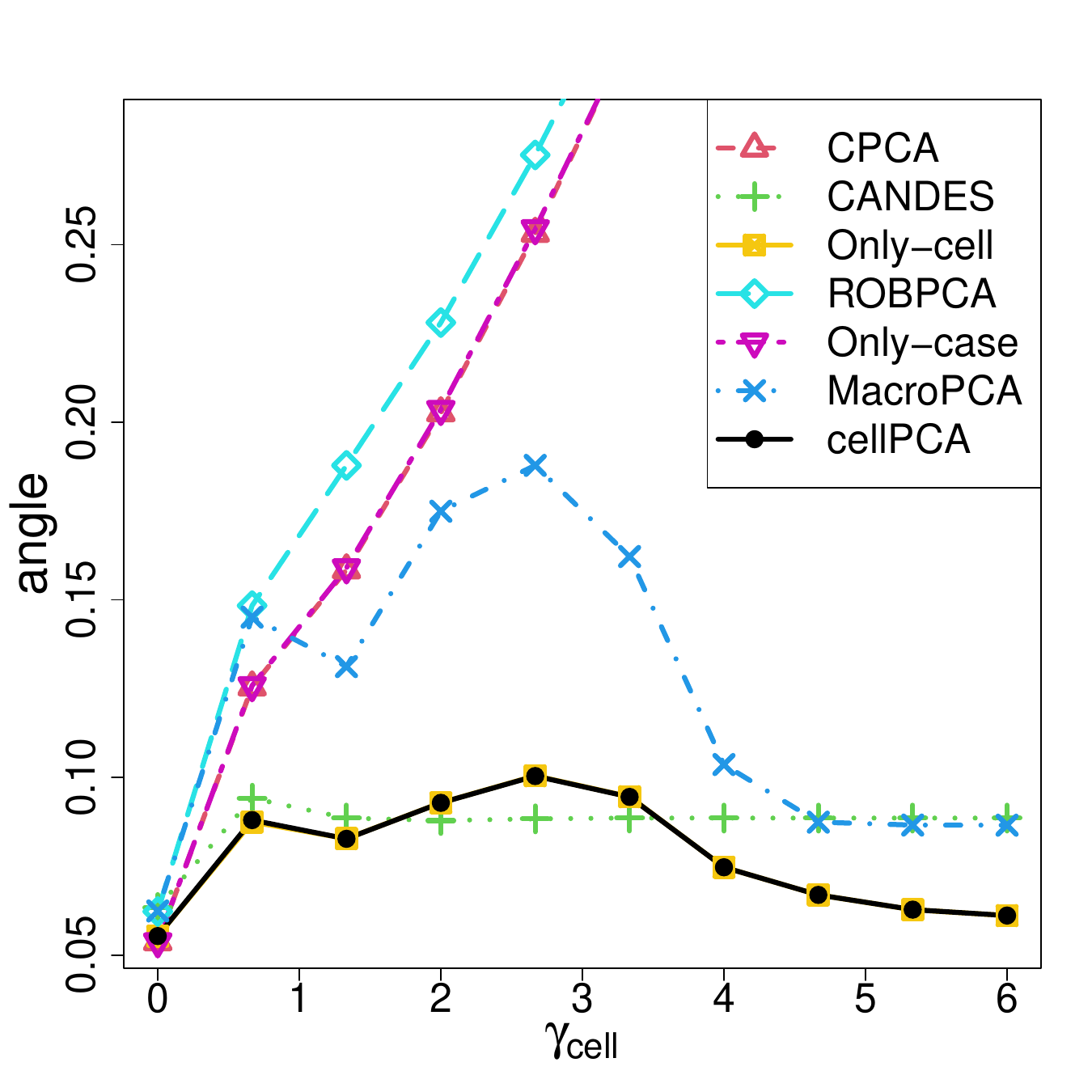} &\includegraphics[width=.3\textwidth]
  {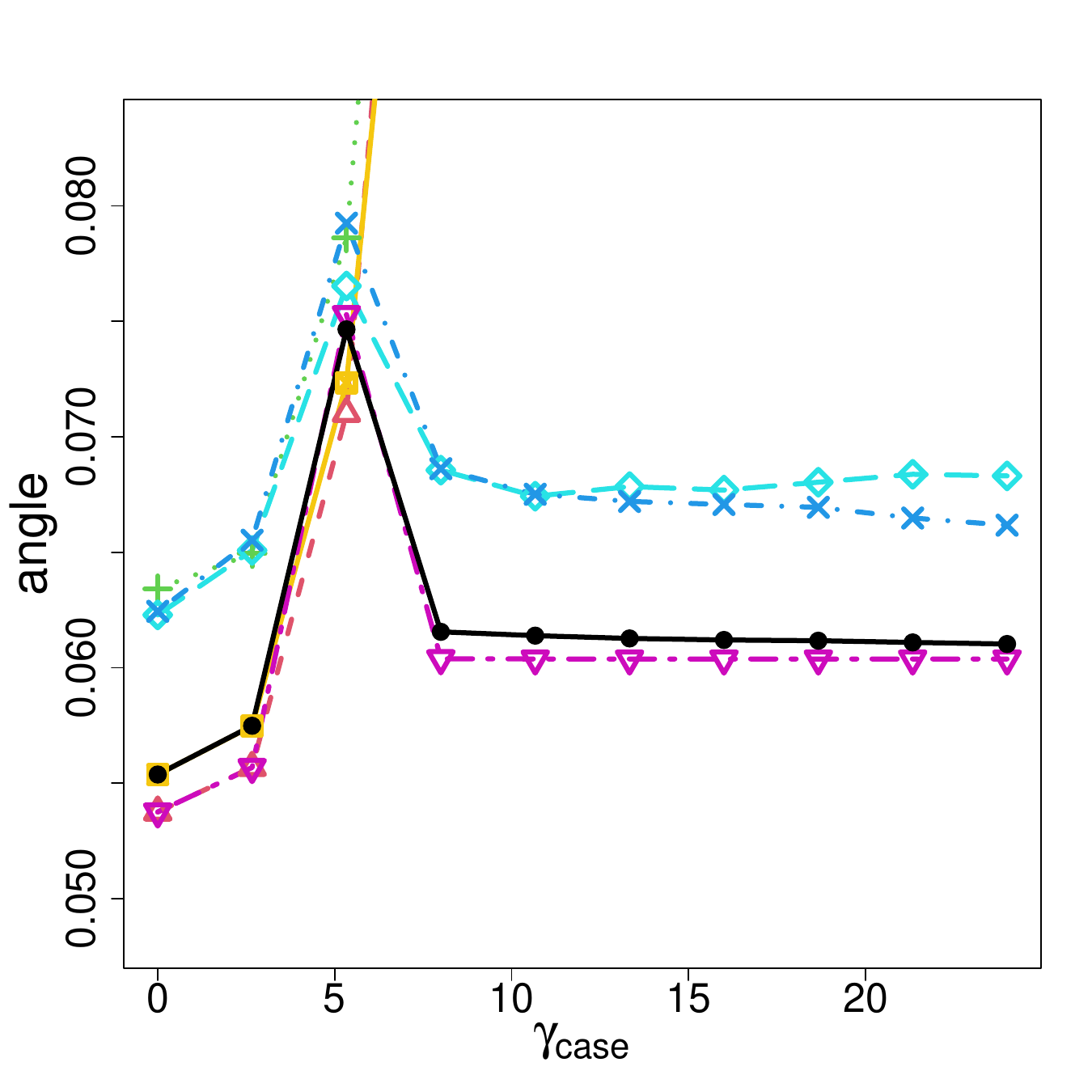} &\includegraphics[width=.3\textwidth]
  {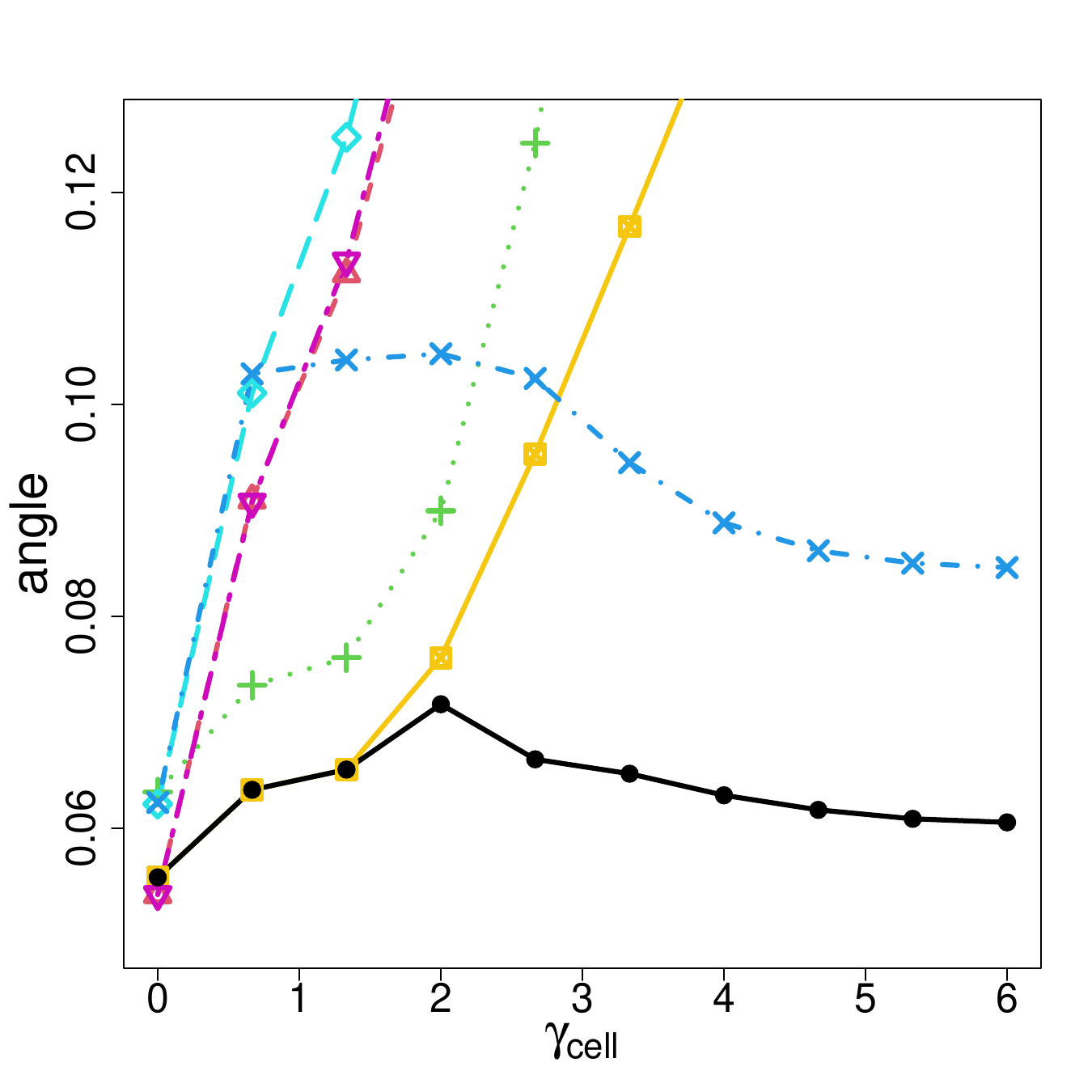}  \\
   [-4mm]
  \includegraphics[width=.3\textwidth]
  {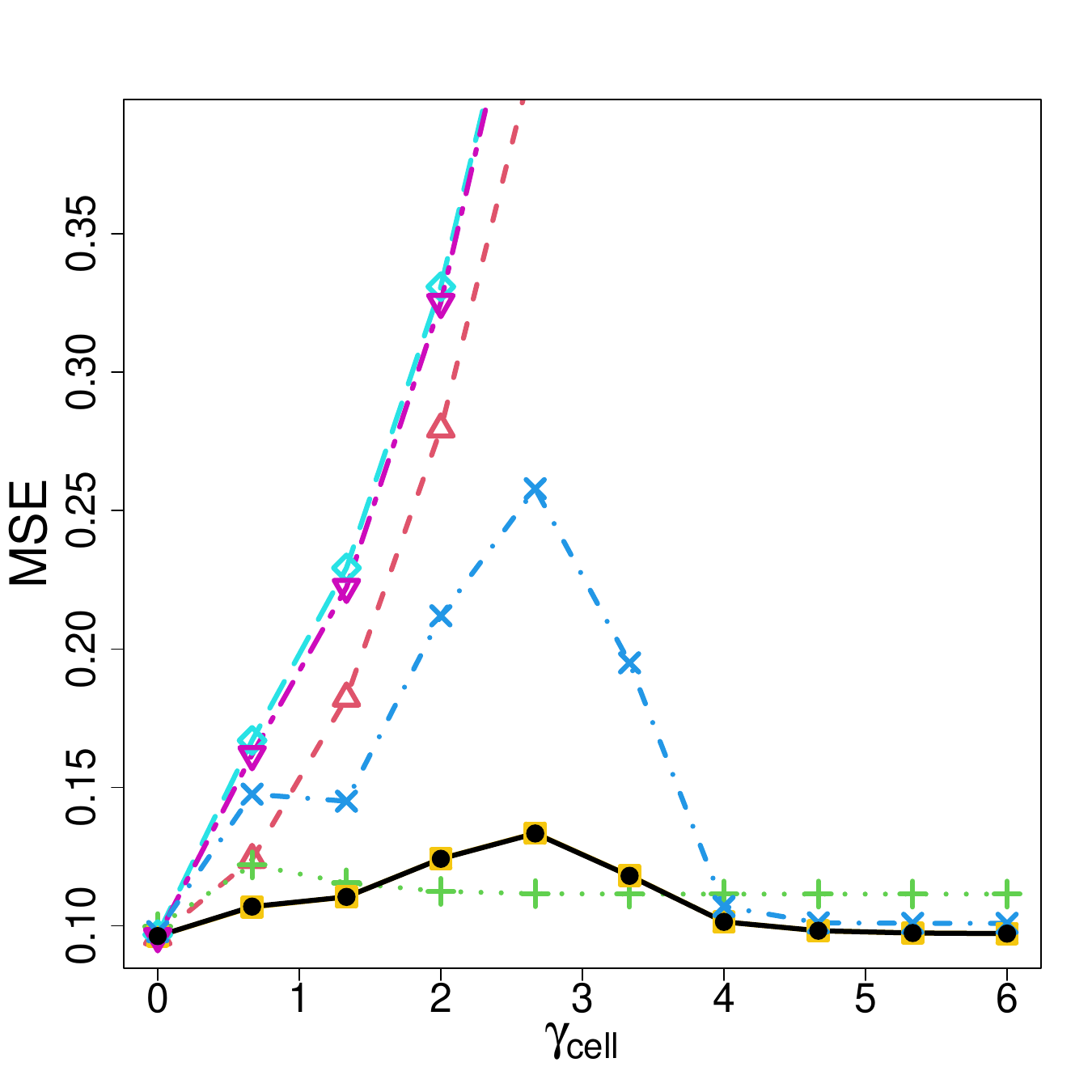} &\includegraphics[width=.3\textwidth]
  {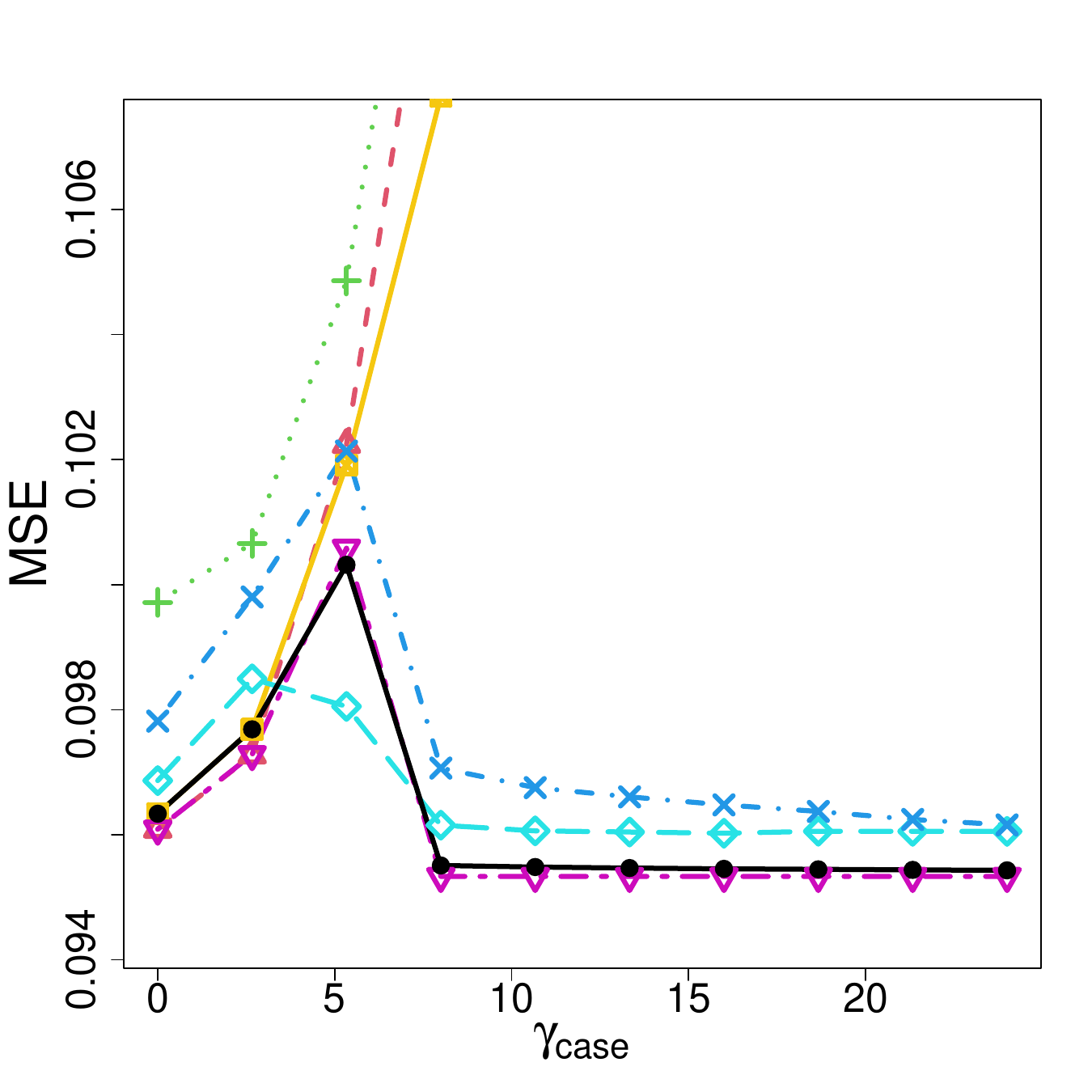} &\includegraphics[width=.3\textwidth]
  {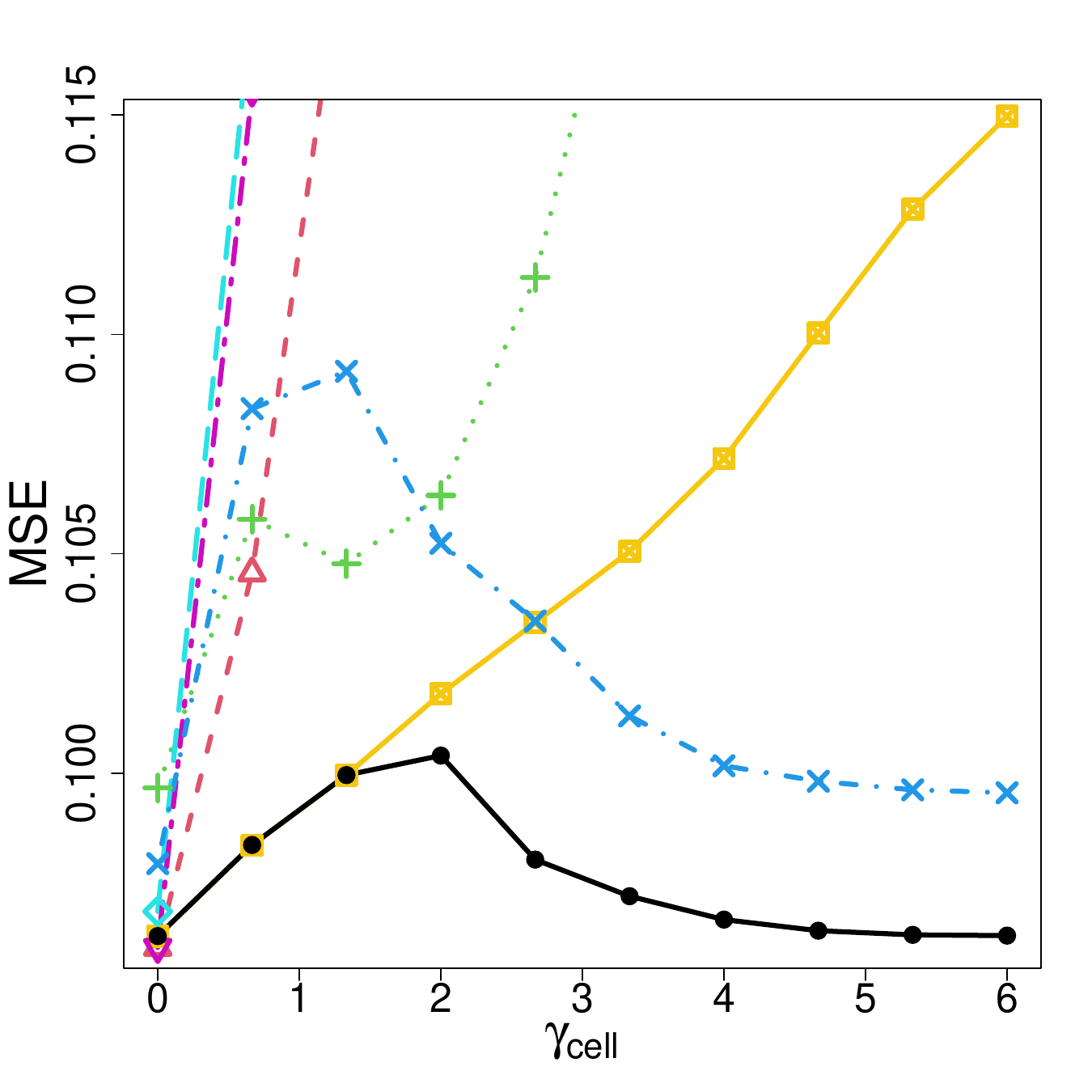} 
\end{tabular}
\caption{Median angle and MSE attained by CPCA, CANDES, Only-cell, ROBPCA, Only-case, MacroPCA, and cellPCA in the presence of either cellwise outliers, casewise outliers, or both. The covariance model was ALYZ with $n=100$, without NAs. The top two rows are for $p=20$, and the bottom two rows for $p=200$.}
\label{fig:results_NA0_ALYZ}
\end{figure}

\begin{figure}[!ht]
\centering
\begin{tabular}{ccc}
   \large \textbf{Cellwise}  & \large \textbf{Casewise} &\large{\textbf{Casewise \& Cellwise}} \\
   [-4mm]
  \includegraphics[width=.3\textwidth]
  {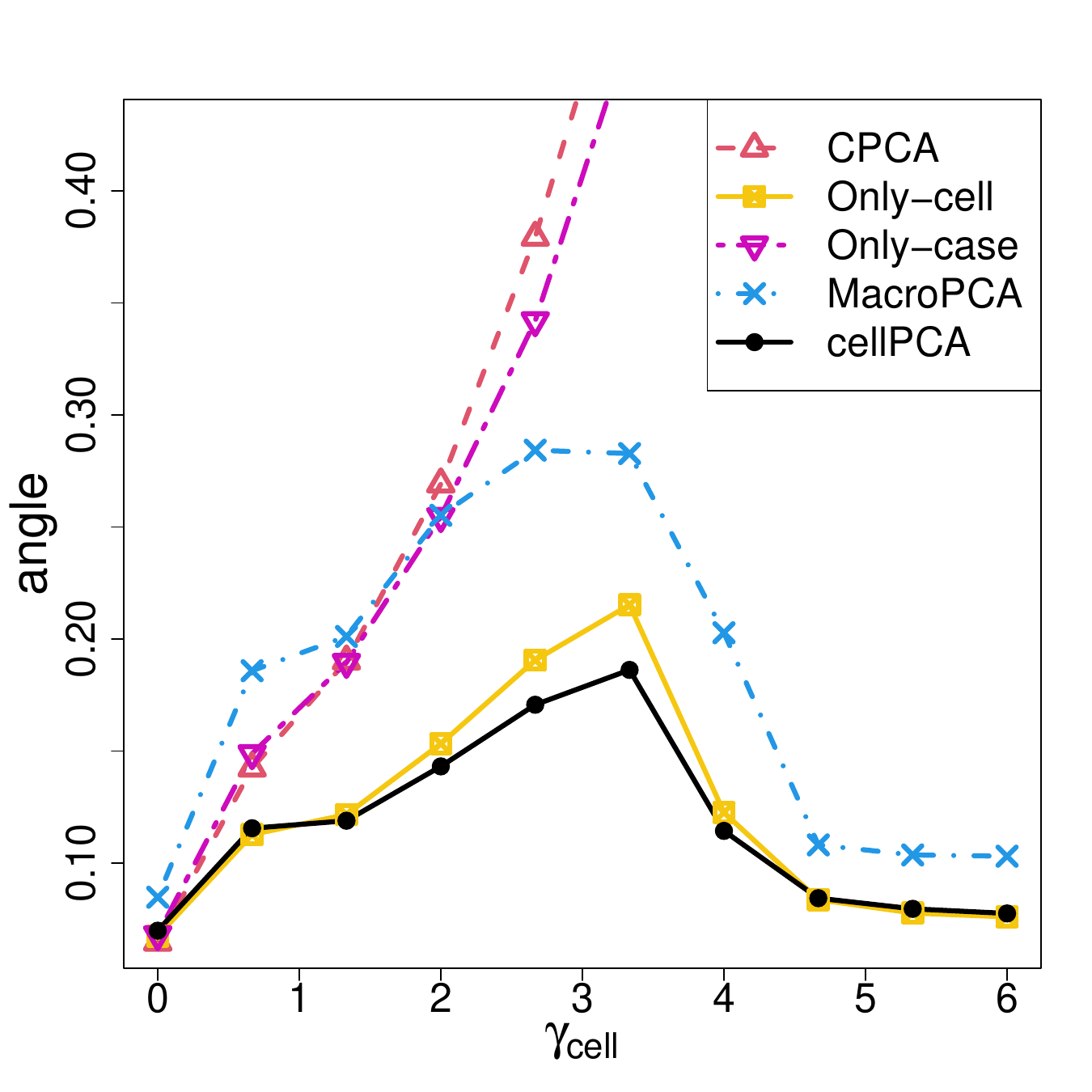} &\includegraphics[width=.3\textwidth]
  {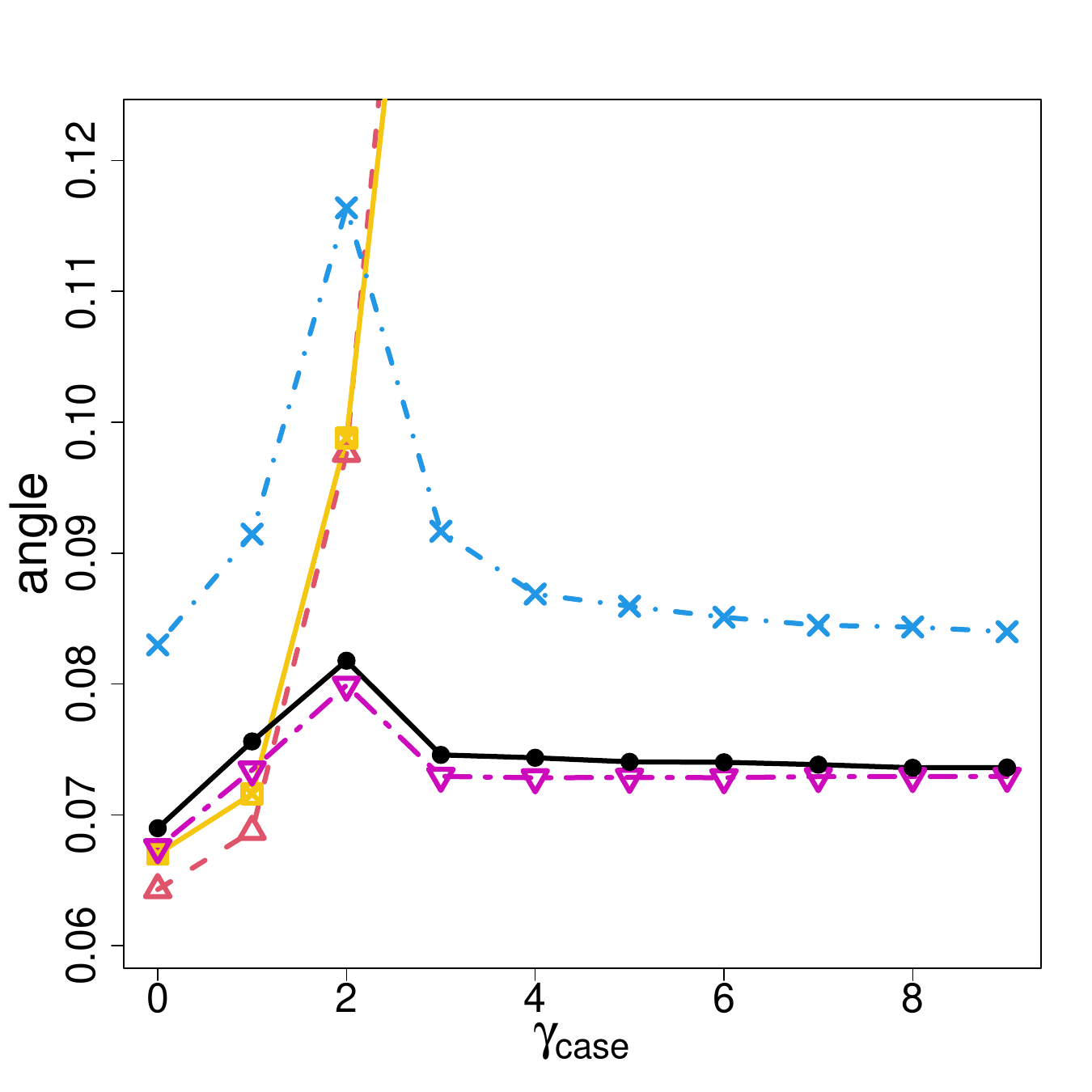} &\includegraphics[width=.3\textwidth]
  {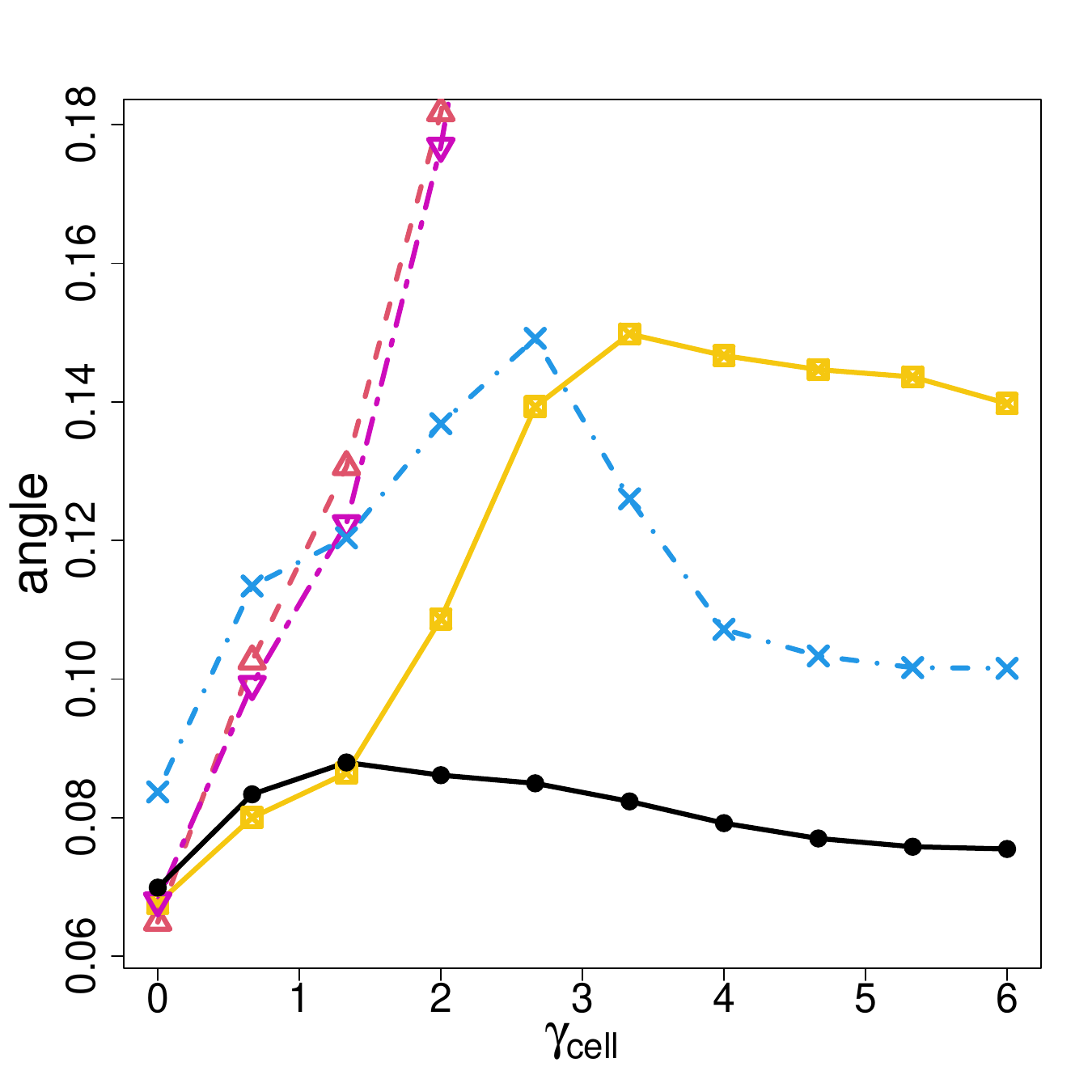}  \\
   [-4mm]
  \includegraphics[width=.3\textwidth]
  {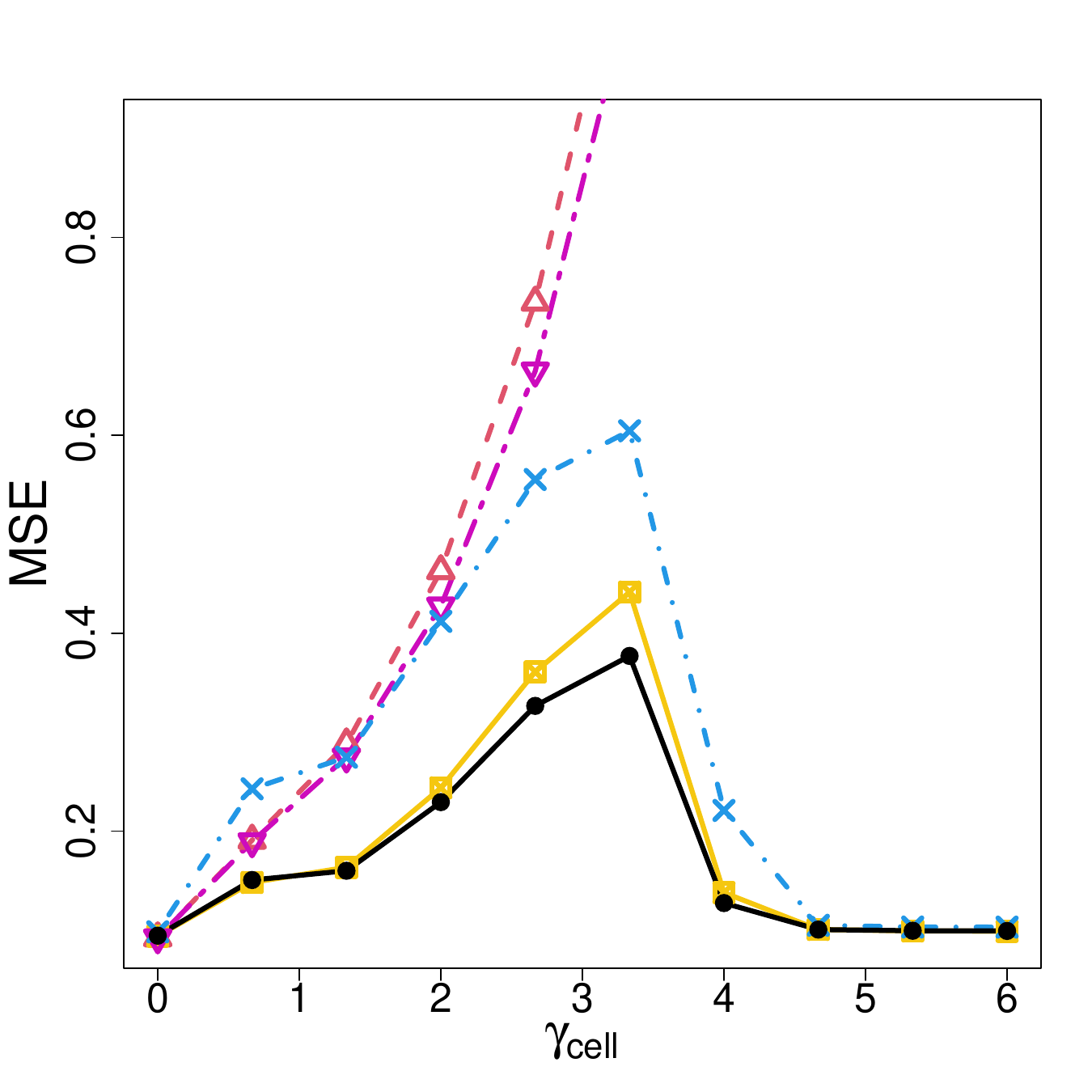} &\includegraphics[width=.3\textwidth]
  {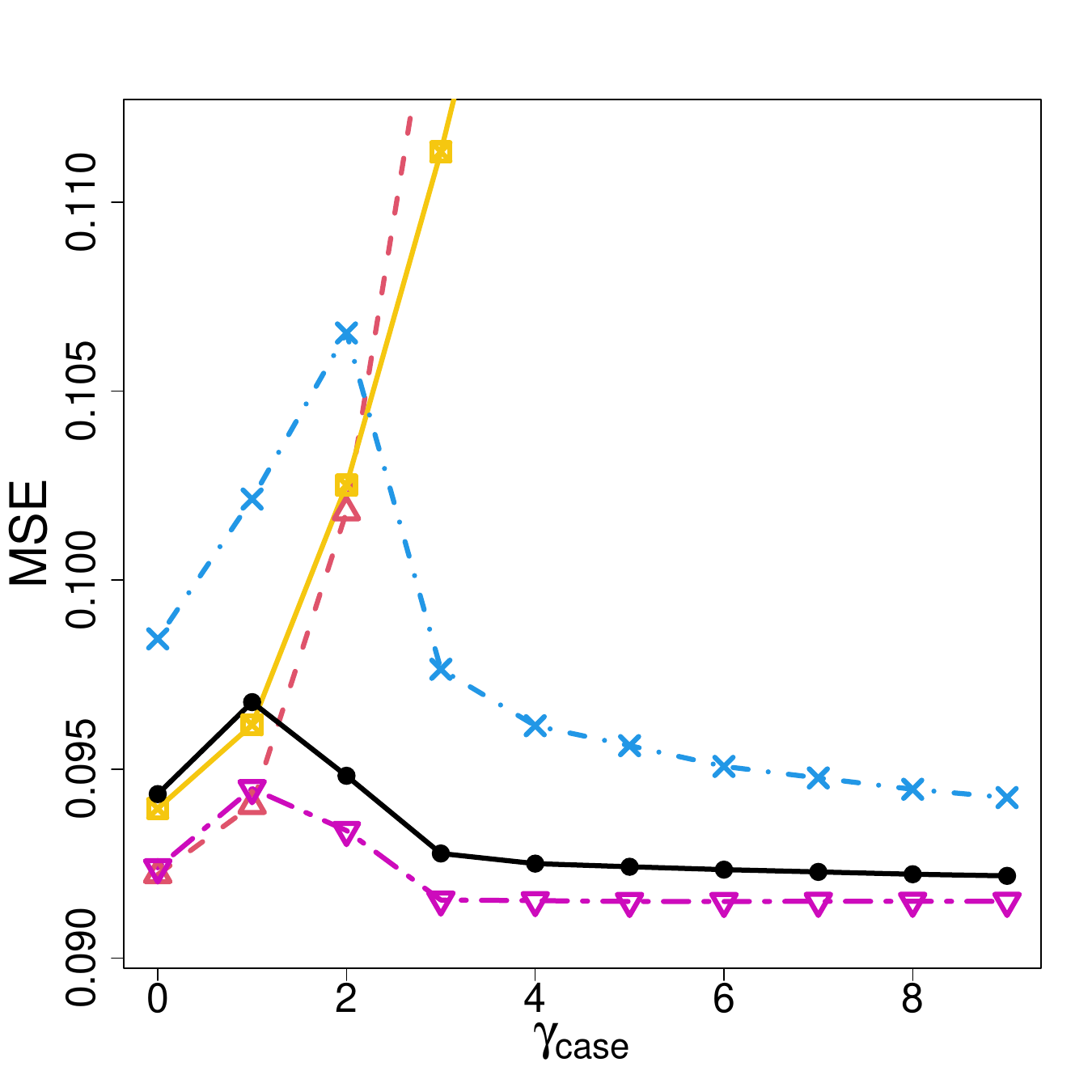} &\includegraphics[width=.3\textwidth]
  {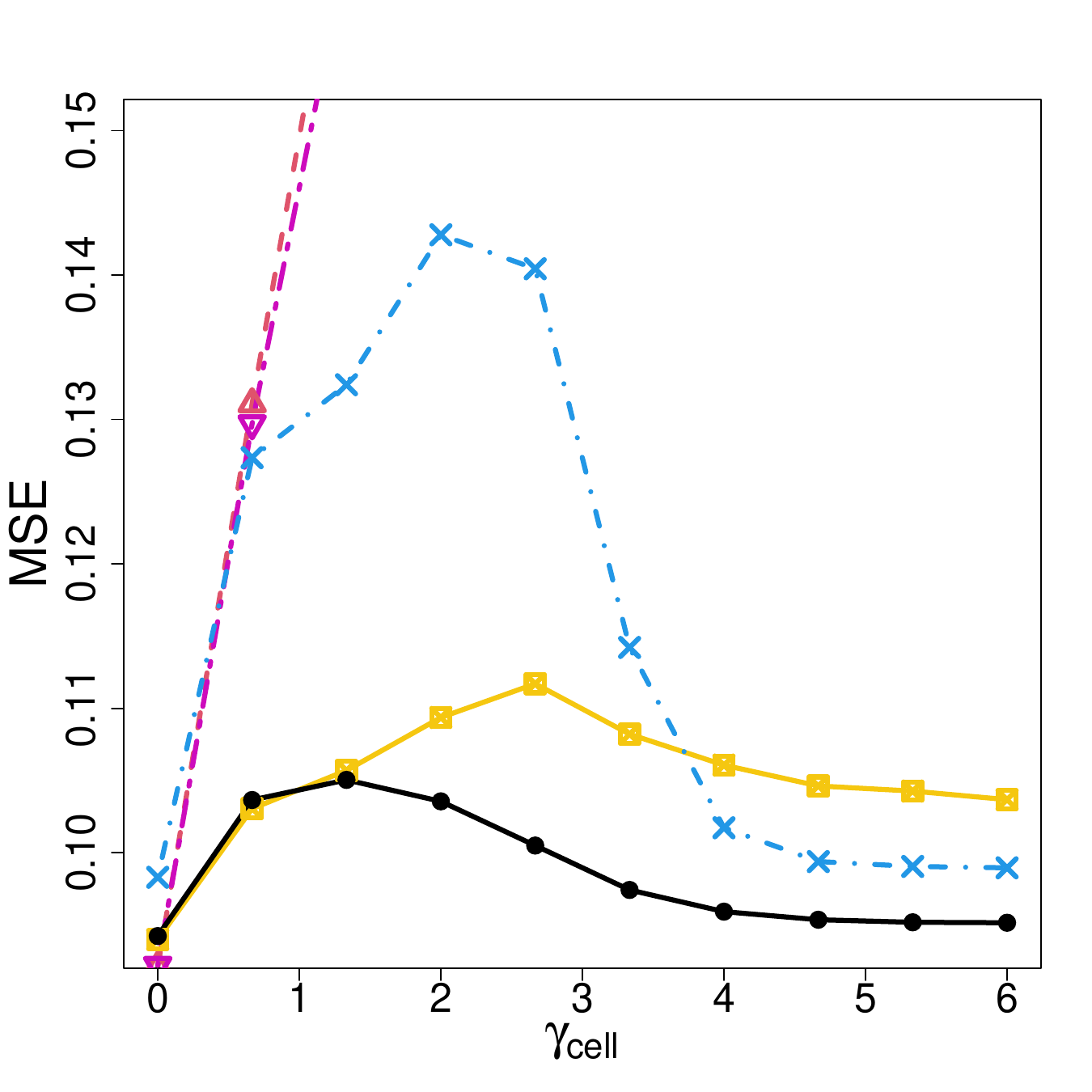} \\
\includegraphics[width=.3\textwidth]
  {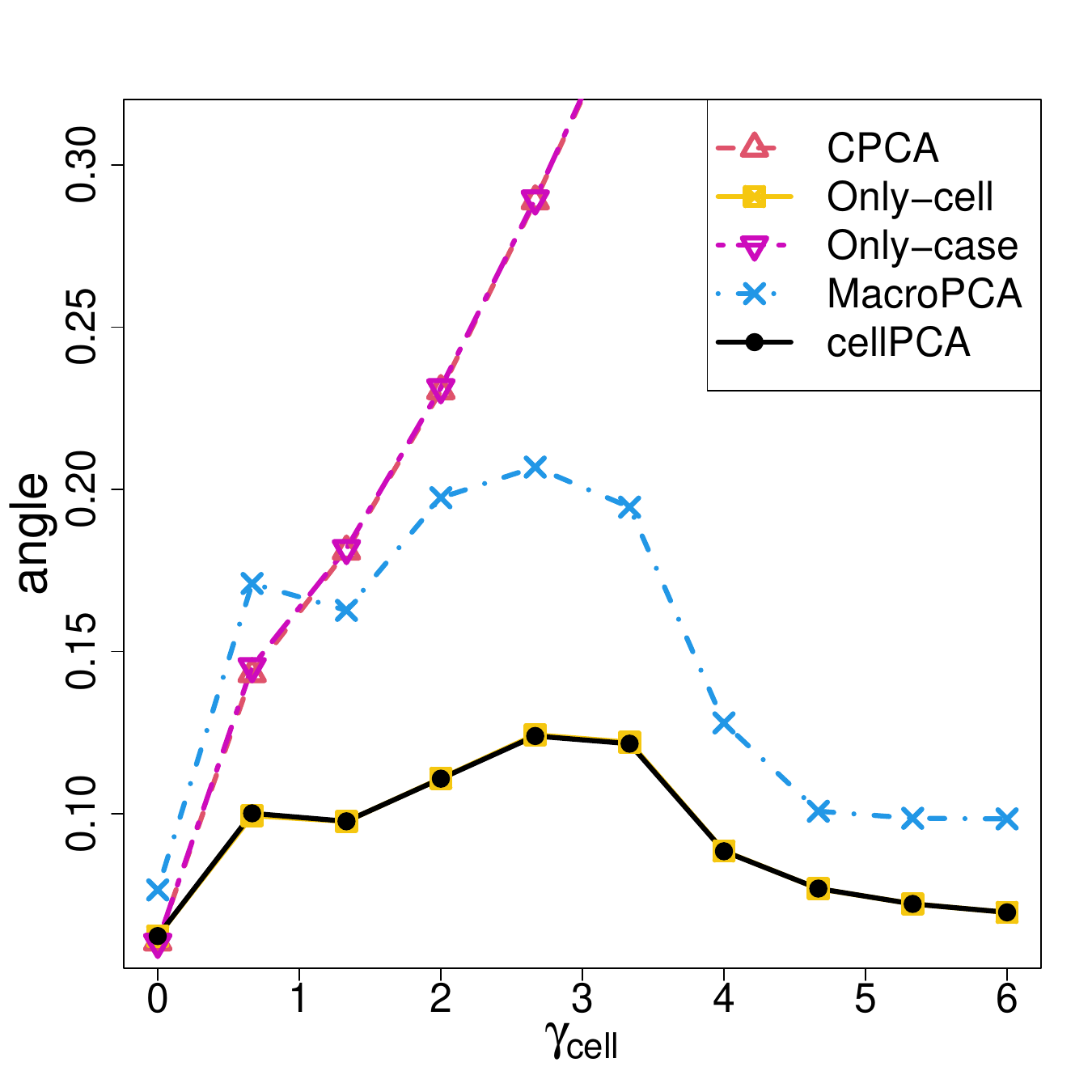} &\includegraphics[width=.3\textwidth]
  {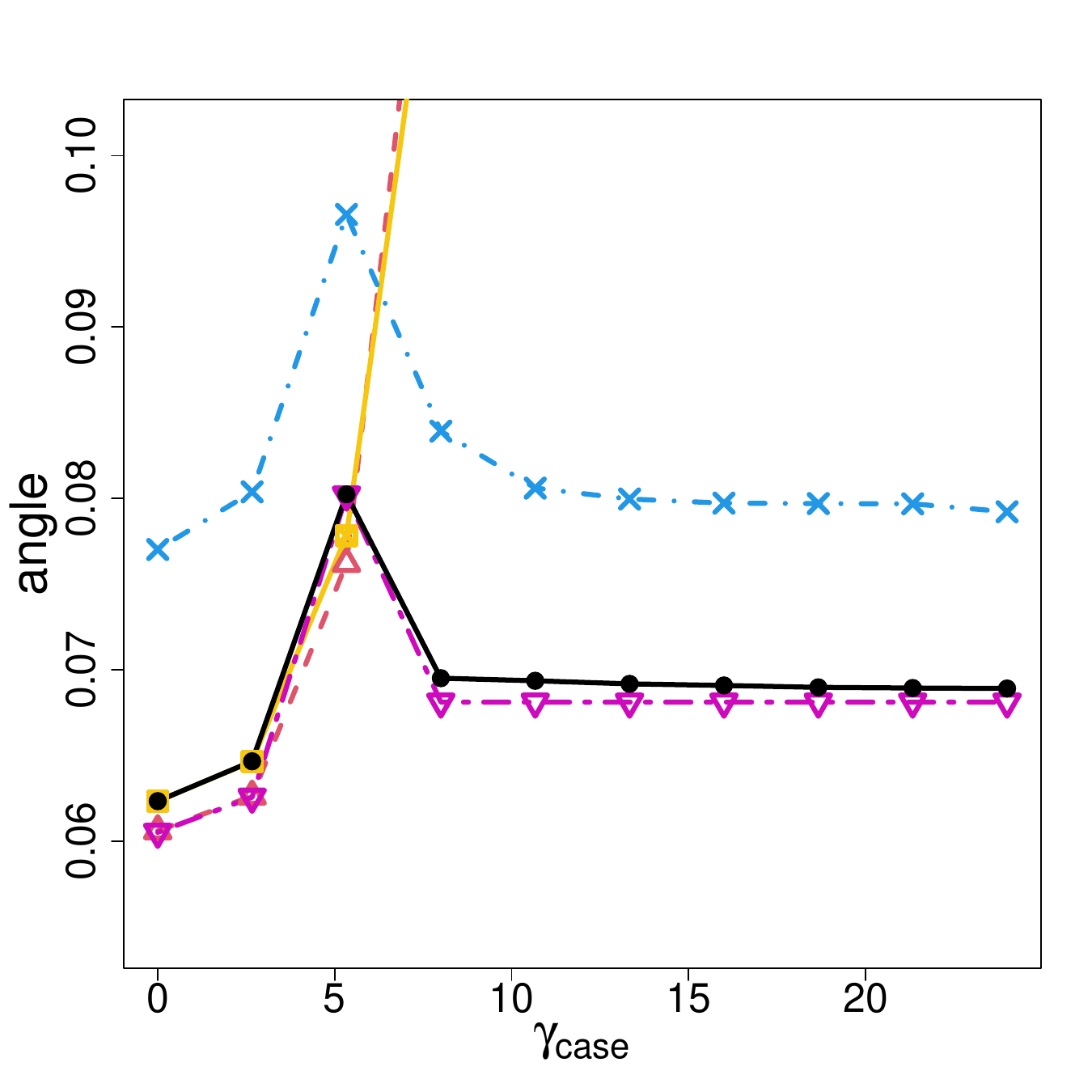} &\includegraphics[width=.3\textwidth]
  {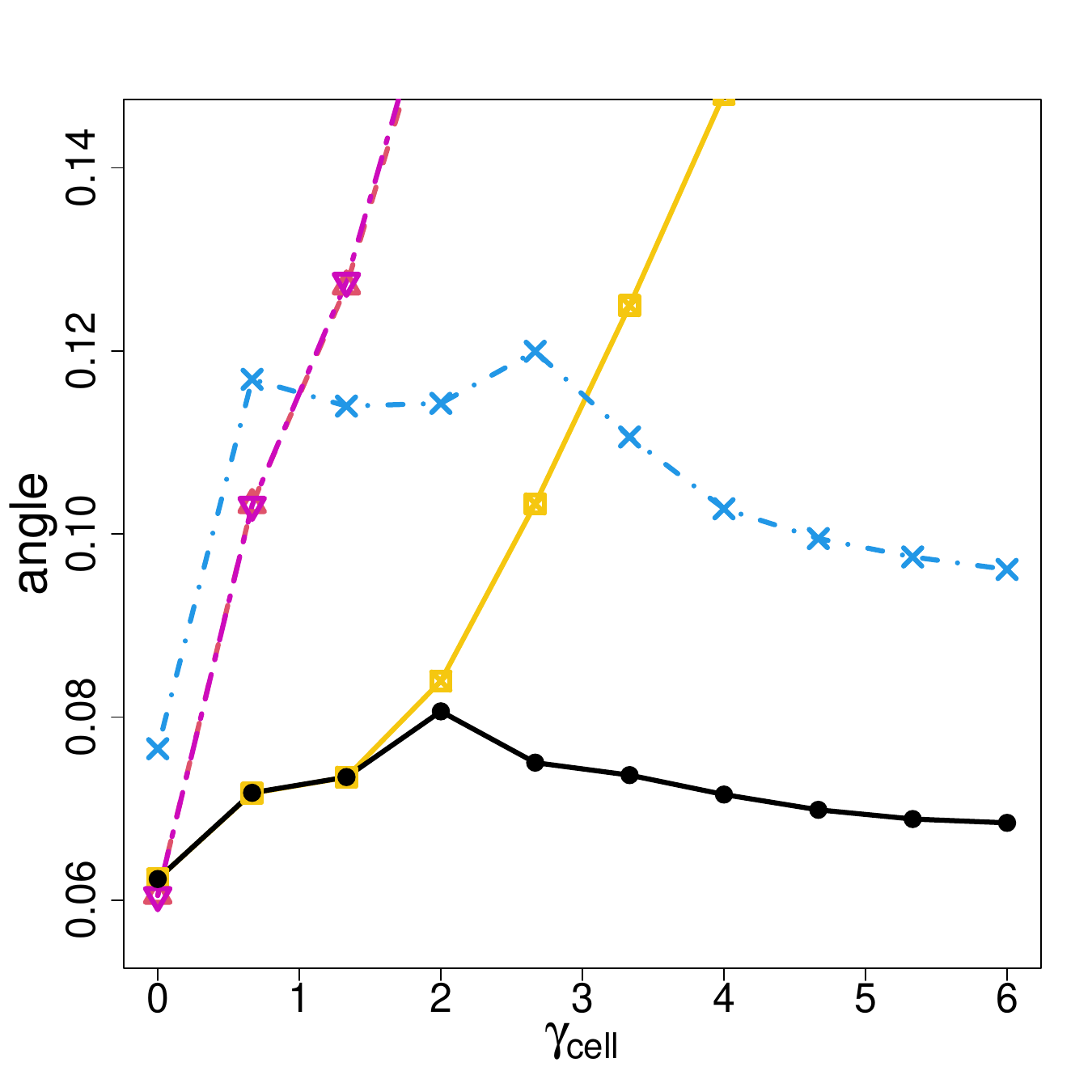}  \\
   [-4mm]
\includegraphics[width=.3\textwidth]
  {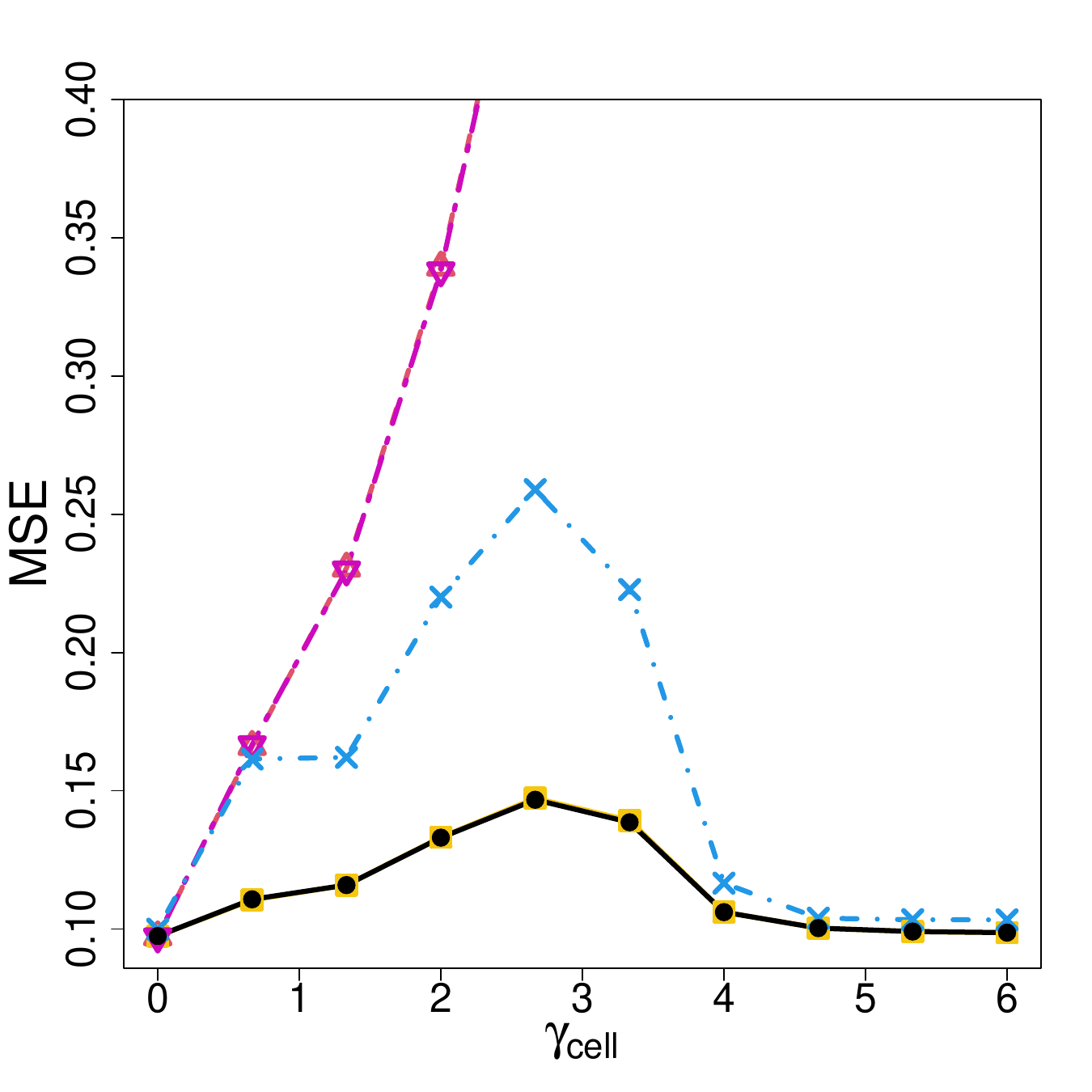} &\includegraphics[width=.3\textwidth]
  {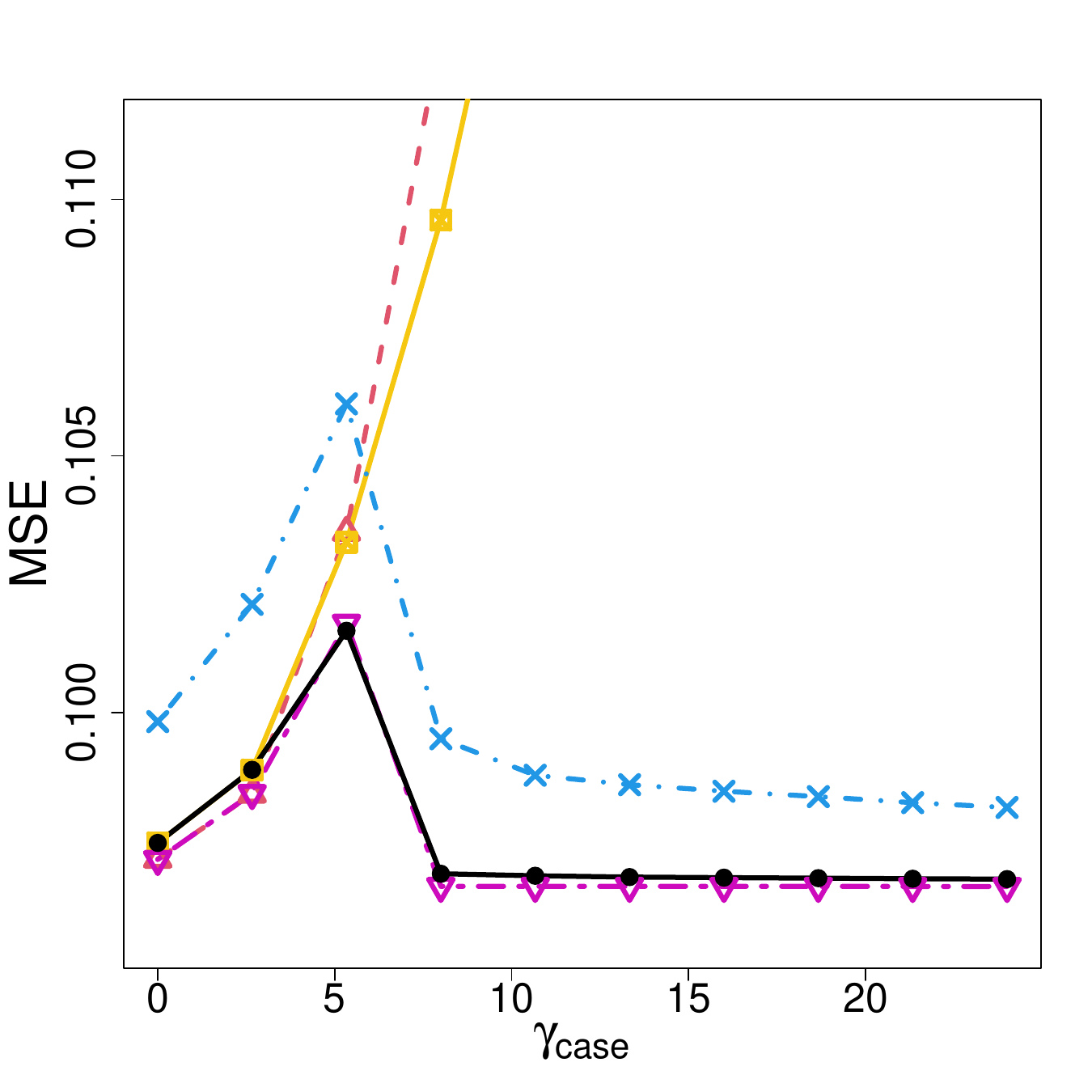} &\includegraphics[width=.3\textwidth]
  {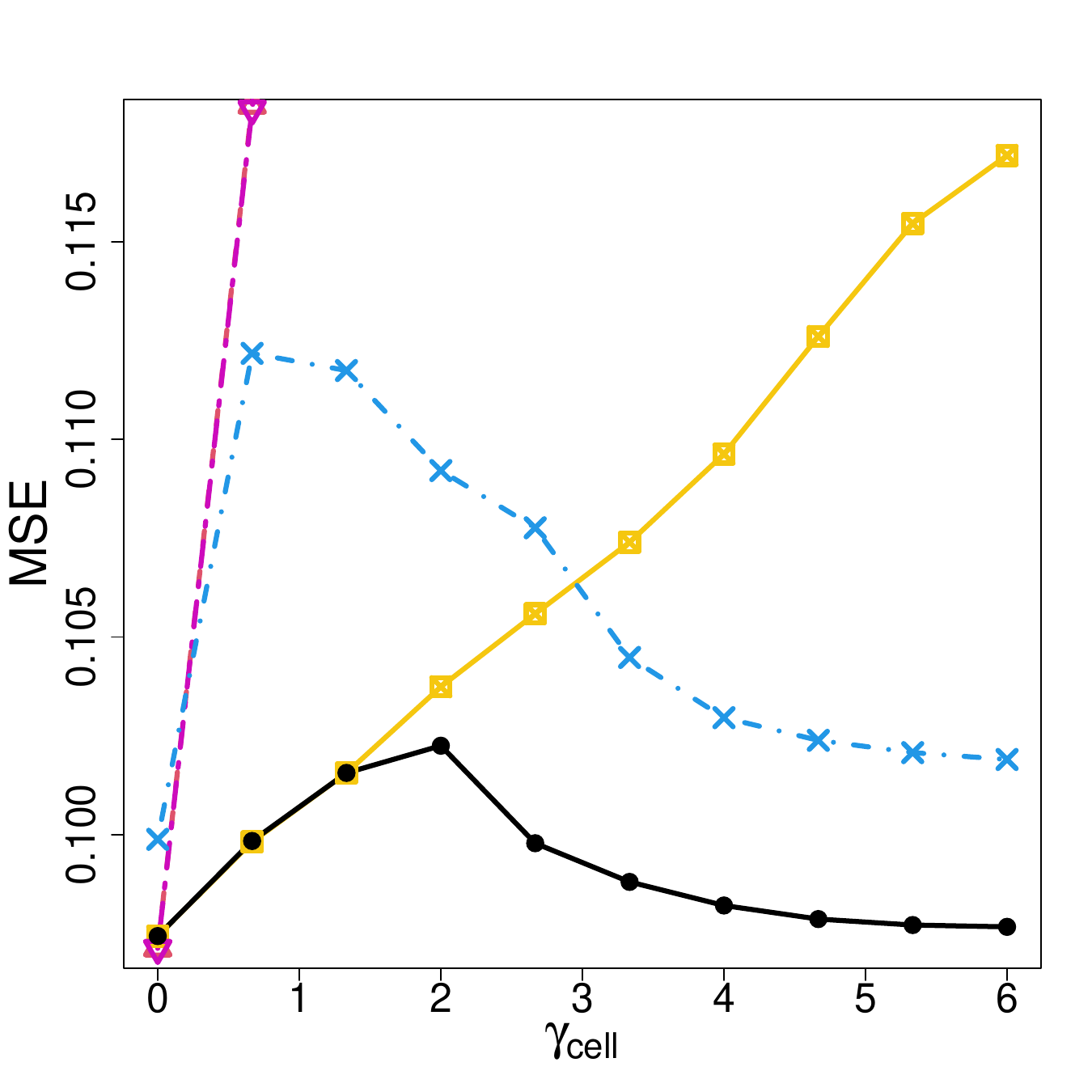} 
\end{tabular}
\caption{Median angle and MSE attained by CPCA, Only-cell, Only-case, MacroPCA, and cellPCA in the presence of either cellwise outliers, casewise outliers, or both. The covariance model was ALYZ with $n=100$, and $20\%$ of randomly selected cells were set to NA. The top two rows are for $p=20$, and the bottom two rows for $p=200$.}
\label{fig:results_NA0.2_ALYZ}
\end{figure}

The cellwise outliers considered in the simulation 
study in Section \ref{sec:simulation} are also referred 
to as \textit{elementwise} outliers, since they occur 
in random individual entries of the data matrix. To 
further evaluate the performance of cellPCA, we examine a more challenging \textit{structured cellwise} contamination scenario, where clean data, generated as in Section~\ref{sec:simulation}, are partially replaced by contaminated cells, with contamination percentages $\eps_j^{\cell}$ linearly varying across variables from $15\%$ to $25\%$. For each column of the data matrix, we randomly sample the corresponding percentage of cell indices to be replaced. In each row, say $(z_1, \ldots, z_p)$, we collect the indices of the cells selected for contamination. Denote this index set of size $q$ by $Q = \{j_1, \ldots, j_q\}$. We then replace the cells $(z_{j_1}, \ldots, z_{j_q})$ with the $q$-dimensional vector
$\gamma_{\cell} \, \sqrt{q}\,\bv_Q/(\bv_q^T\bSigma_Q^{-1}\bv_q)$,
where  $\bSigma_Q$ denotes the covariance matrix $\bSigma$ restricted to the indices in $Q$. The vector $\bv_Q$ is the normalized eigenvector of $\bSigma_Q$ corresponding to its smallest eigenvalue. The contamination level $\gamma_{\cell}$ varies from 0 to 24 for $p=20$ and from 0 to 12 for $p=200$.
In each row, the contaminated cells are structurally outlying in the subspace spanned by the variables in $Q$. Consequently, these cells are often not marginally outlying, especially when the size of $Q$ is large and $\gamma_{\cell}$ is relatively small, which makes them difficult to detect. Moreover, since the probability of observing cellwise outliers depends on the variable, this contamination scenario does not fall under the standard notion of elementwise contamination.
We also consider the scenario where the data is contaminated by a percentage of structured cellwise outliers that goes from $7.5\%$ to $12.5\%$ across variables, and by $\eps^{\case}=10\%$ of 
casewise outliers generated as described in Section~\ref{sec:simulation}. Here $\gamma_{\cell}$ again varies from 0 to 24 with
$\gamma_{\case}=2\gamma_{\cell}$ when $p=20$ and $\gamma_{\cell}$ again varies from 0 to 12 with
 $\gamma_{\case}=8\gamma_{\cell}$ when $p=200$.

Figure~\ref{fig:results_p20_NA0_struc} and Figure~\ref{fig:results_p200_NA0_struc} show the median 
angle and MSE in the presence of either structured cellwise
outliers, or structured cellwise and casewise outliers, for $p=20$ and $p=200$. The results are even more convincing than those obtained for elementwise contamination. In this setting, cellPCA significantly outperforms the competitors across all scenarios. In particular, the CANDES approach is adversely affected by the presence of structured outliers, since its development relies on the assumption of elementwise contamination. 

\begin{figure}[!ht]
\centering
\begin{tabular}{cc}
   \large \textbf{Structured Cellwise}  &\large{\textbf{Casewise \& Structured Cellwise}} \\
   [-4mm]
  \includegraphics[width=.3\textwidth]
  {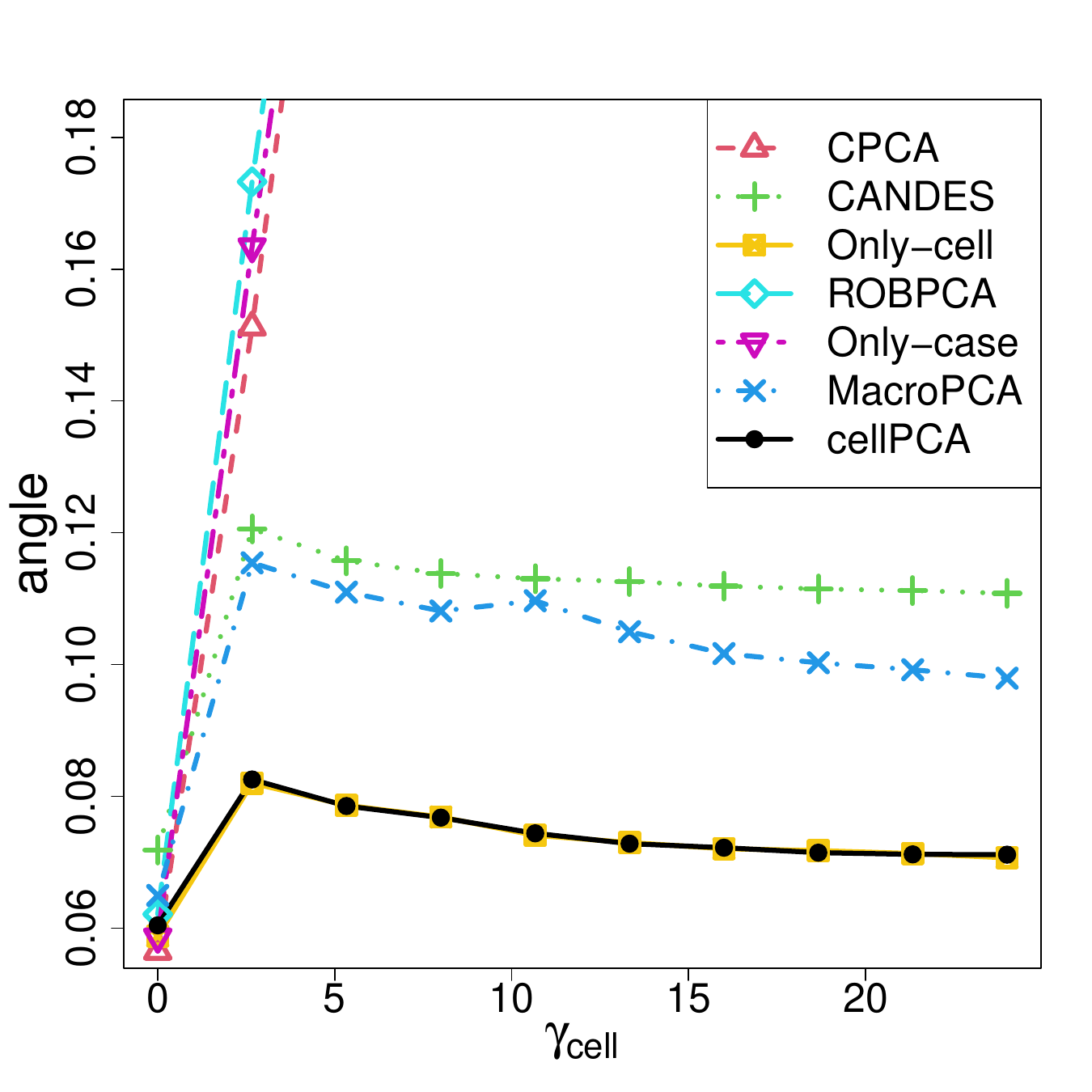}  &\includegraphics[width=.3\textwidth]
  {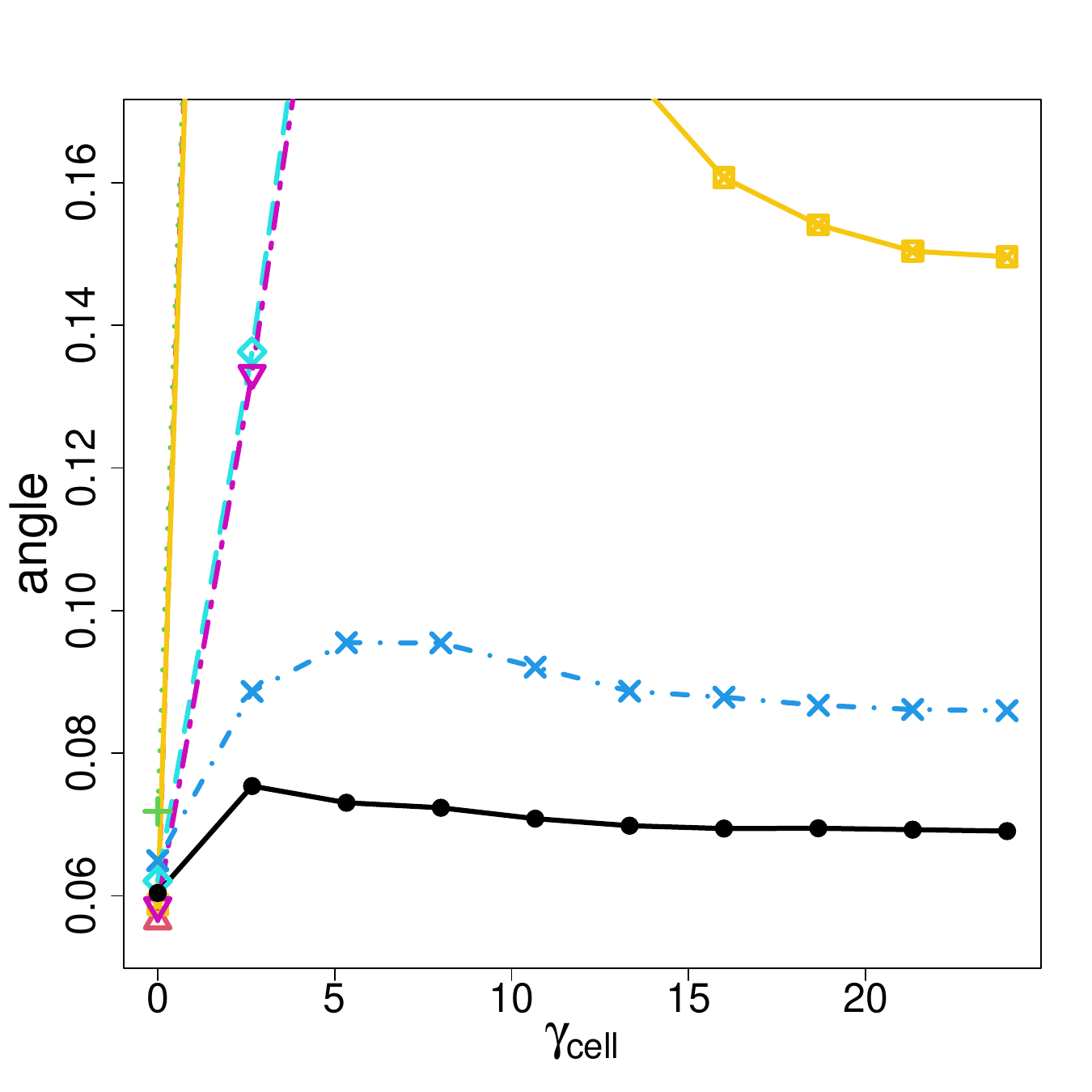}  \\
   [-4mm]
  \includegraphics[width=.3\textwidth]
  {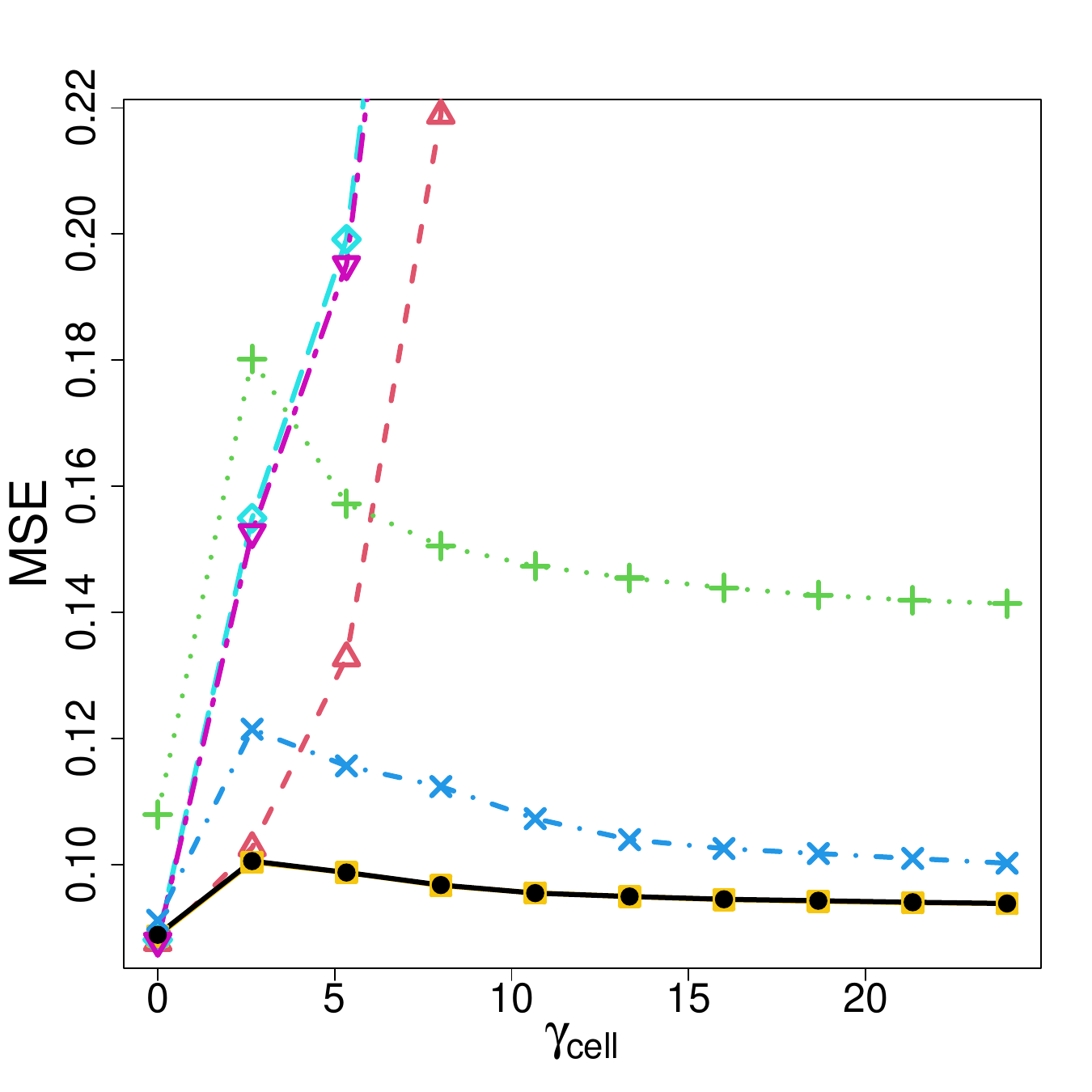} &\includegraphics[width=.3\textwidth]{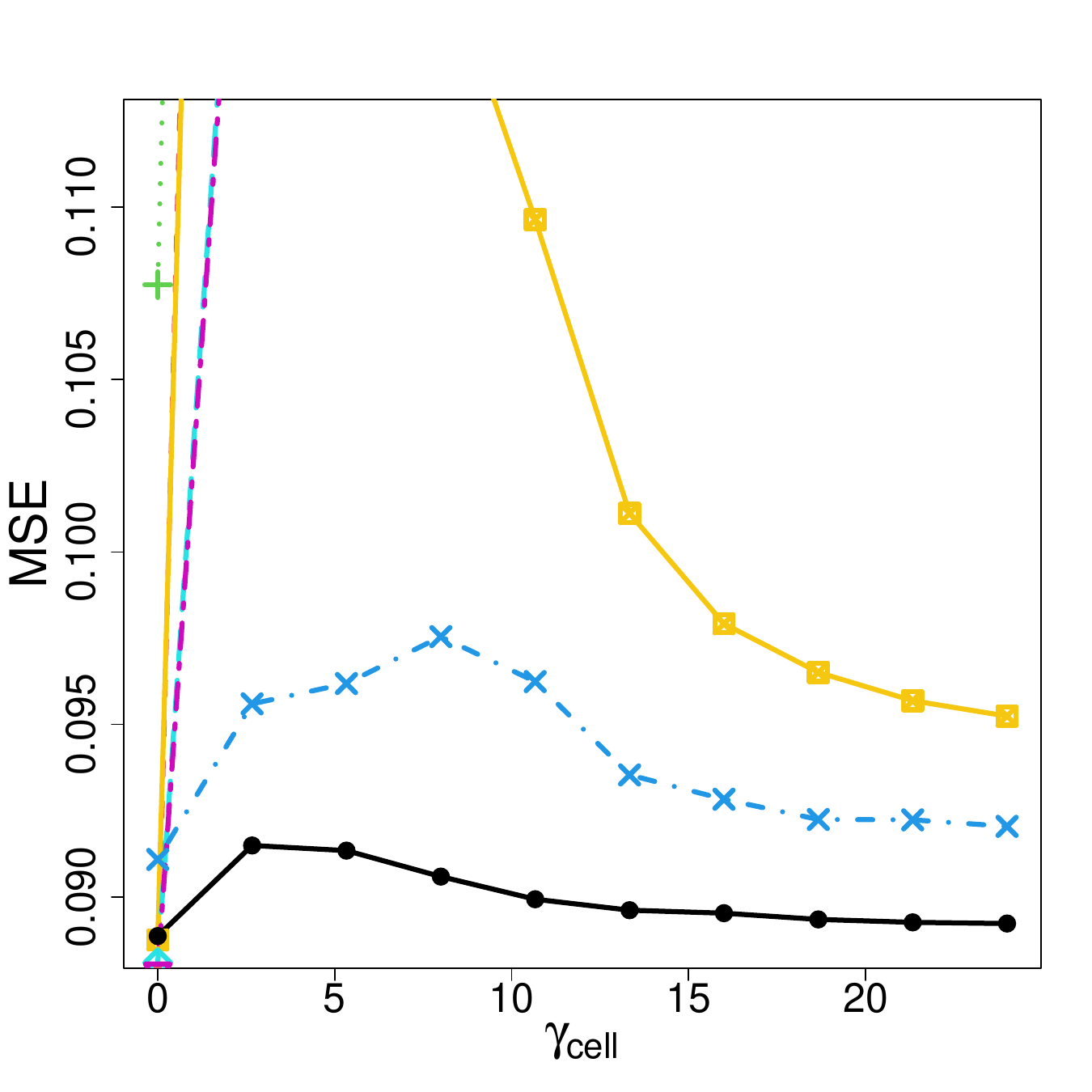} 
\end{tabular}
\caption{Median angle (top) and MSE (bottom) 
attained by CPCA, CANDES, Only-cell, ROBPCA, 
Only-case, MacroPCA, and cellPCA in the presence 
of structured cellwise outliers, and structured cellwise and casewise outliers. 
The covariance model was A09 with 
$n=100$ and $p=20$, without NAs.}
\label{fig:results_p20_NA0_struc}
\end{figure}\begin{figure}[!ht]
\centering
\begin{tabular}{cc}
   \large \textbf{Structured Cellwise}  &\large{\textbf{Casewise \& Structured Cellwise}} \\
   [-4mm]
  \includegraphics[width=.3\textwidth]
  {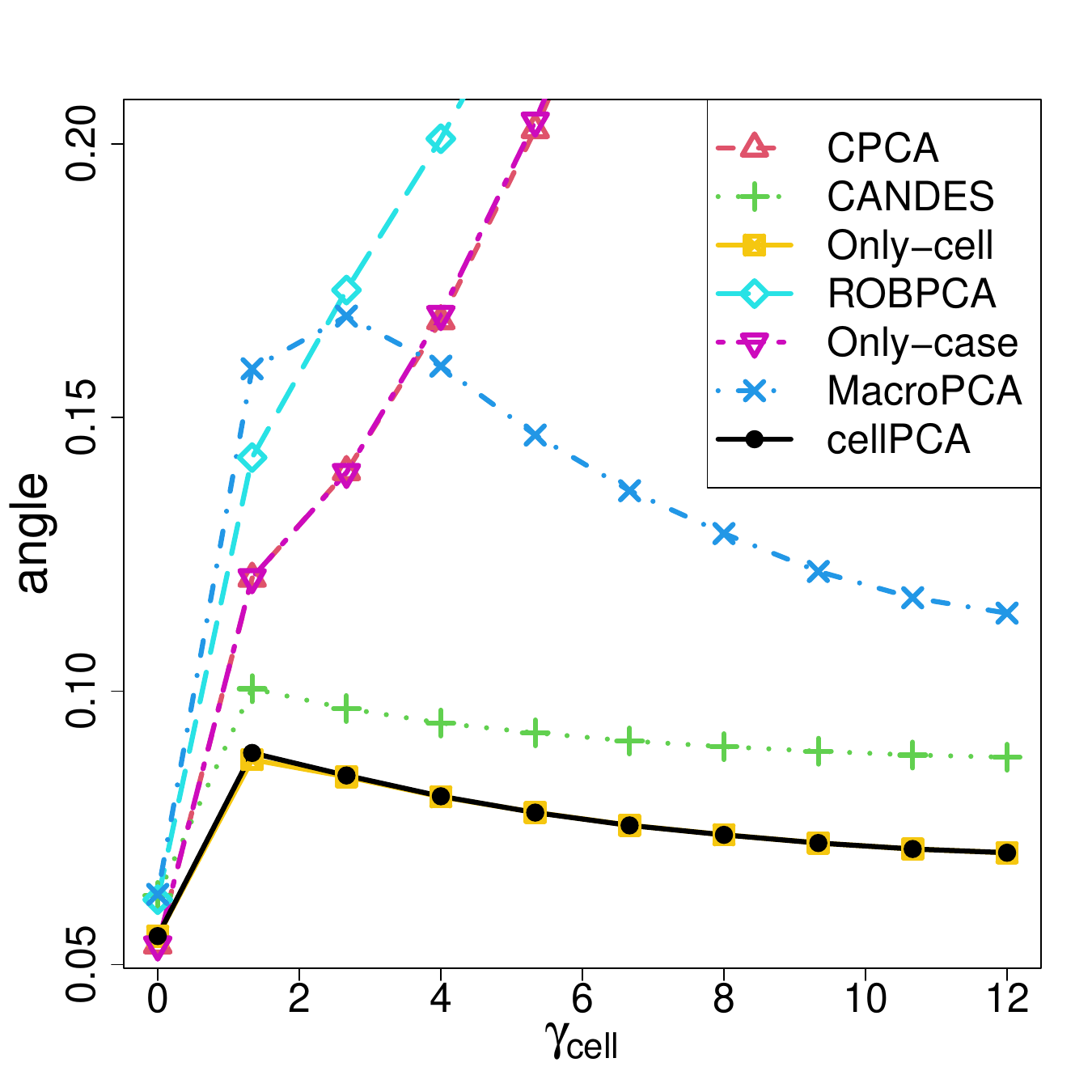}  &\includegraphics[width=.3\textwidth]
  {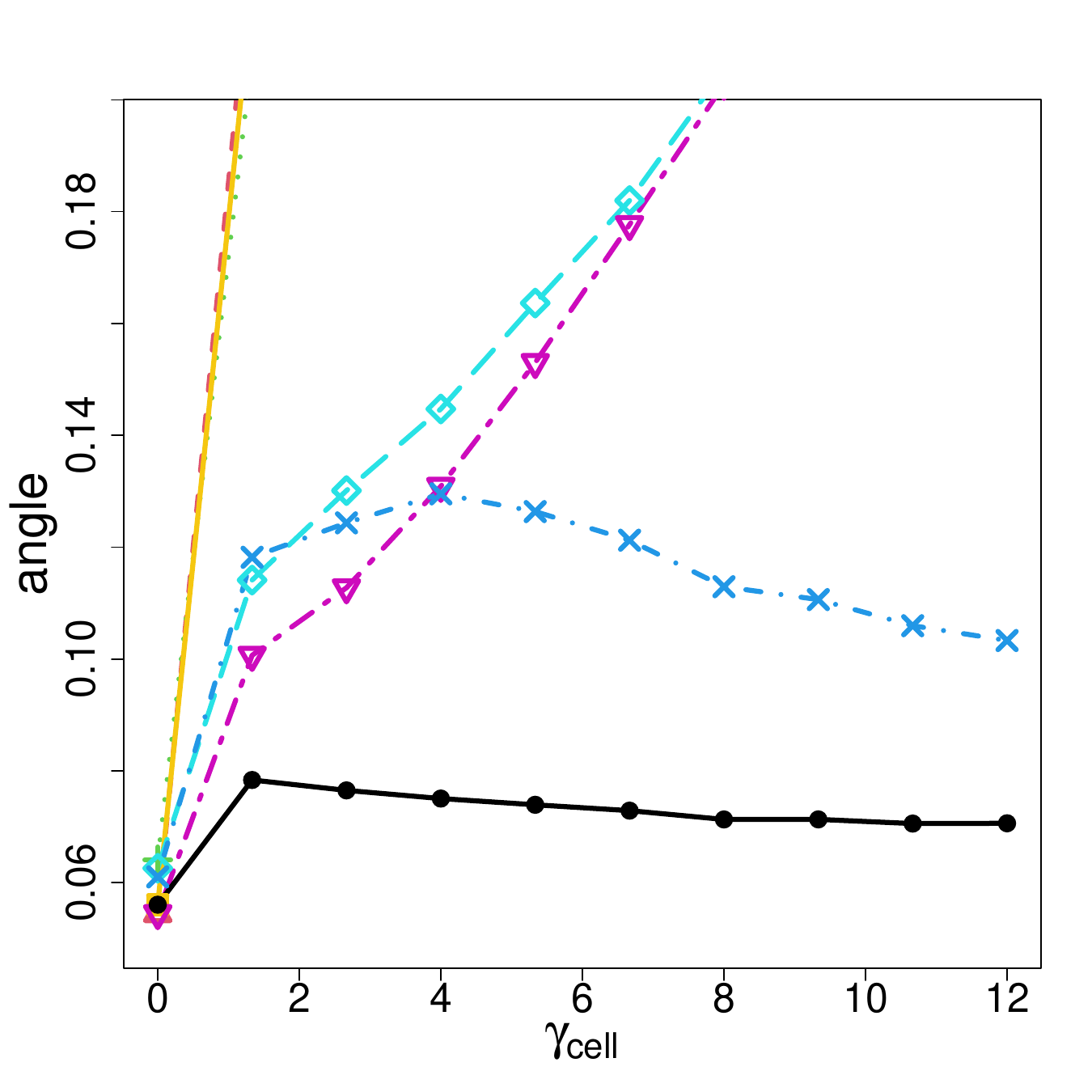}  \\
   [-4mm]
  \includegraphics[width=.3\textwidth]
  {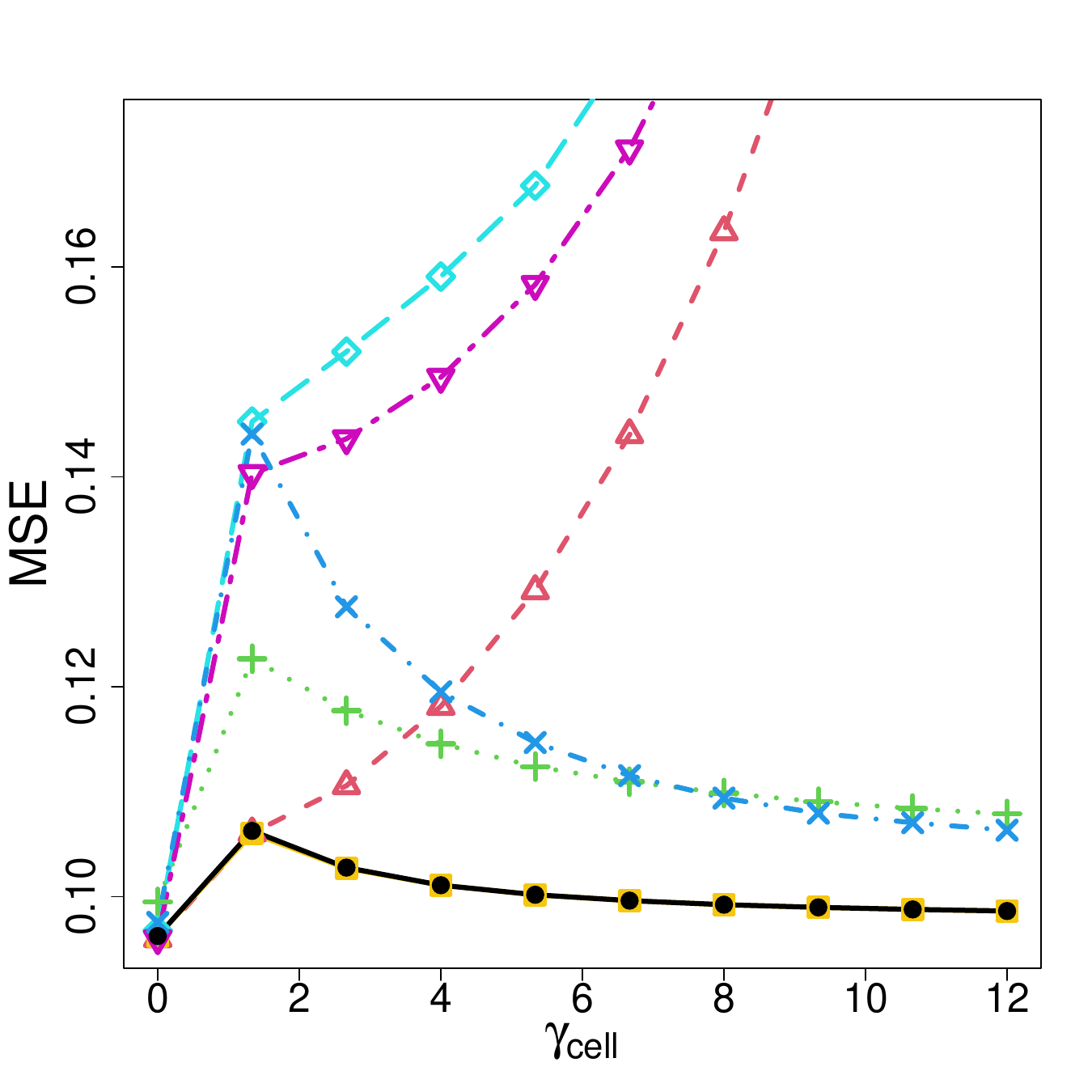} &\includegraphics[width=.3\textwidth]{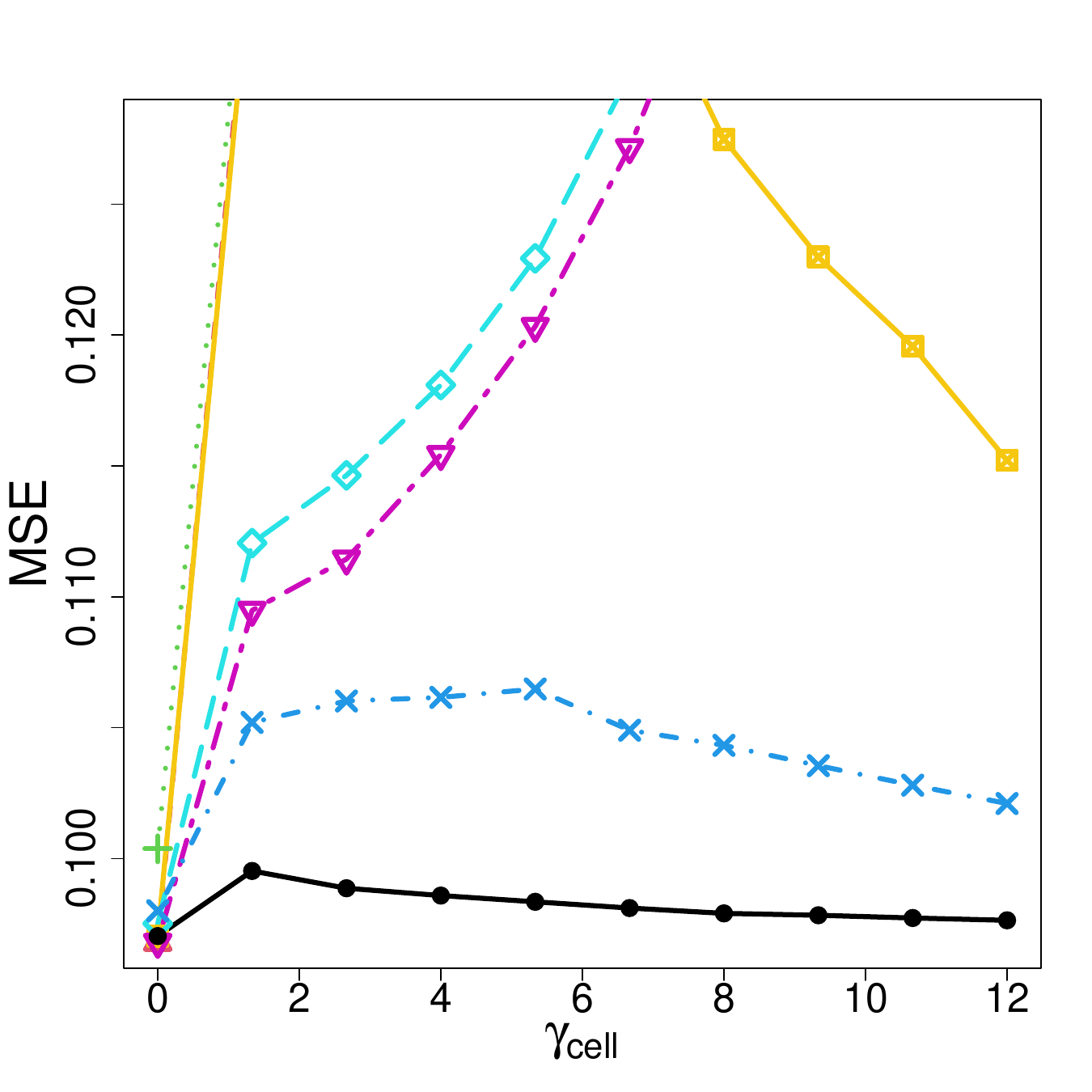} 
\end{tabular}
\caption{Median angle (top) and MSE (bottom) 
attained by CPCA, CANDES, Only-cell, ROBPCA, 
Only-case, MacroPCA, and cellPCA in the presence 
of structured cellwise outliers, and structured cellwise and casewise outliers. 
 The covariance model was A09 with 
$n=100$ and $p=200$, without NAs.}
\label{fig:results_p200_NA0_struc}
\end{figure}

\clearpage
To study the sensitivity of cellPCA and competing methods to the contamination fraction, Figures~\ref{fig:results_varying_pou1} and \ref{fig:results_varying_pou2} show the median 
angle and MSE as a function of $\varepsilon$ for data contaminated as in Section~\ref{sec:simulation} with cellwise
outliers generated with $\eps^{\cell}=\eps$ and $\gamma_{\cell}=\lbrace 3, 5\rbrace$, casewise outliers generated with $\eps^{\case}=\eps$ and $\gamma_{\case}=\lbrace 3,5 \rbrace$, and both together with $\eps^{\cell}=\eps^{\case}=\eps/2$, $\gamma_{\cell}=\lbrace 3, 5\rbrace$ and $\gamma_{\case}$ obtained from $\gamma_{\cell}$ as in Figures~\ref{fig:results_NA0_ALYZ} and~\ref{fig:results_NA0.2_ALYZ}.
The covariance model was A09 with $n=100$ and $p=20$, without NAs.
The two values of $\gamma$ represent intermediate and far outliers. 

\vspace{5mm}
\begin{figure}[!ht]
\centering
\begin{tabular}{ccc}
   \large \textbf{Cellwise}& \large \textbf{Casewise} &\large{\textbf{Casewise \& Cellwise}} \\
   [-4mm]
  \includegraphics[width=.3\textwidth]
  {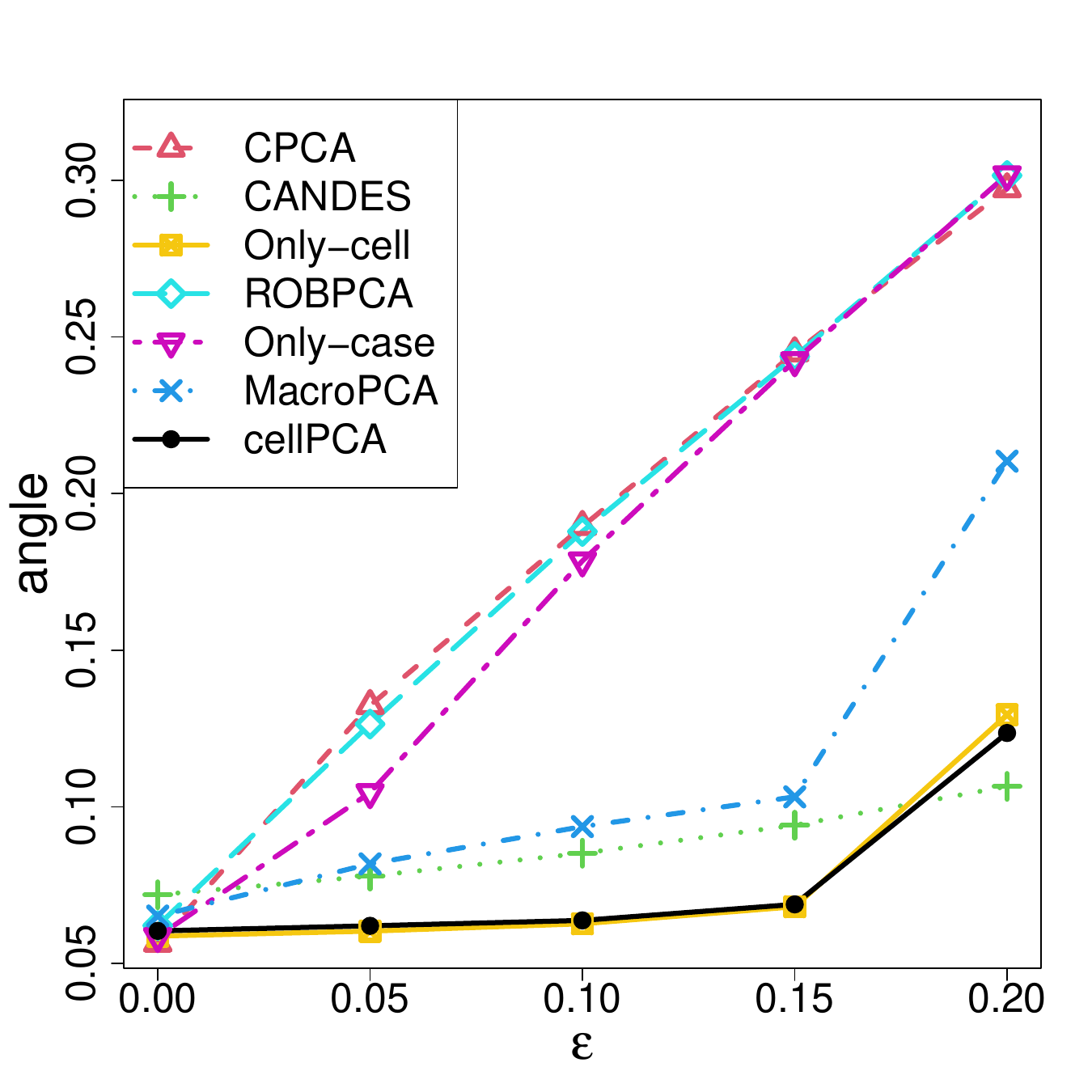} &\includegraphics[width=.3\textwidth]
  {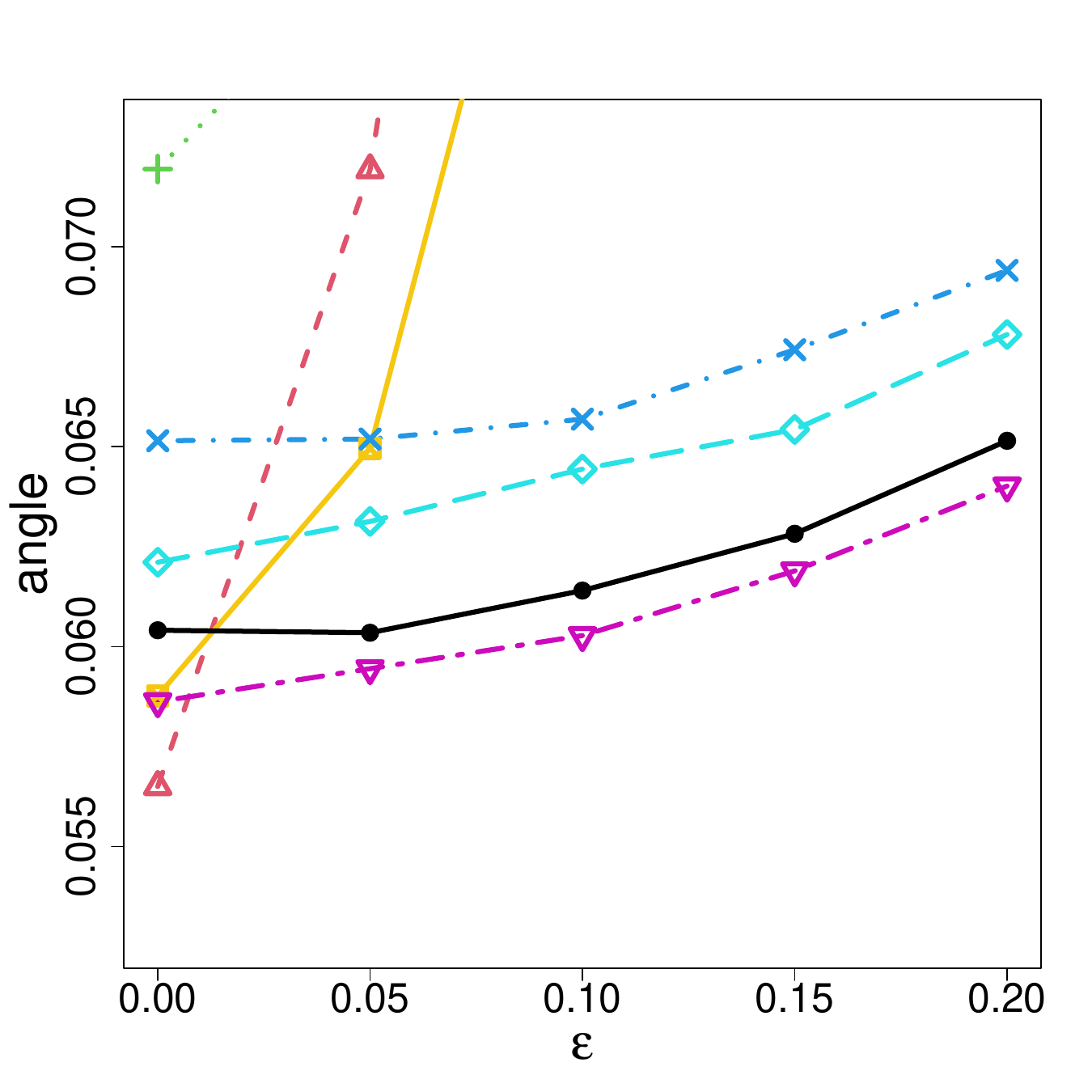} &\includegraphics[width=.3\textwidth]
  {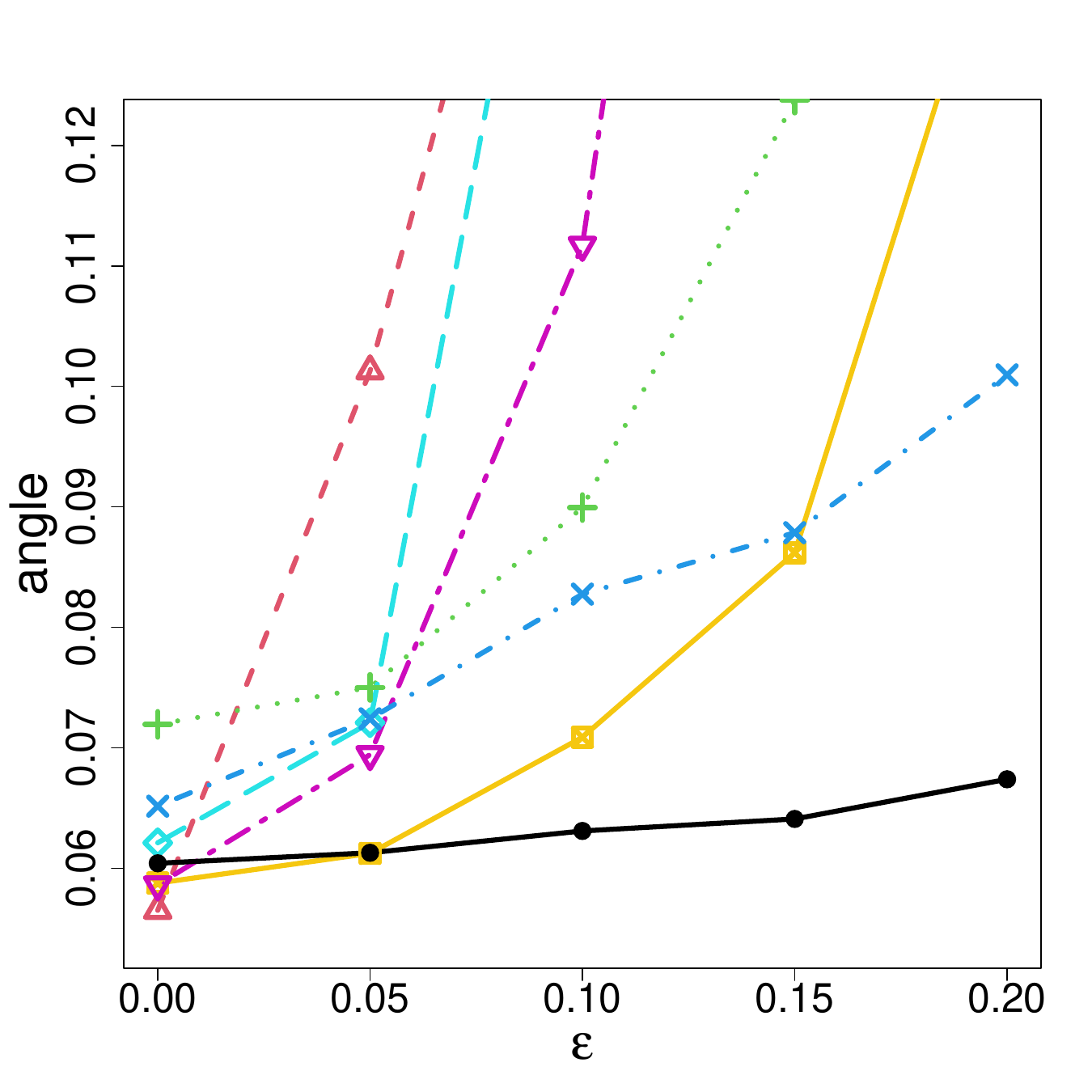}  \\
   [-4mm]
  
  \includegraphics[width=.3\textwidth]
  {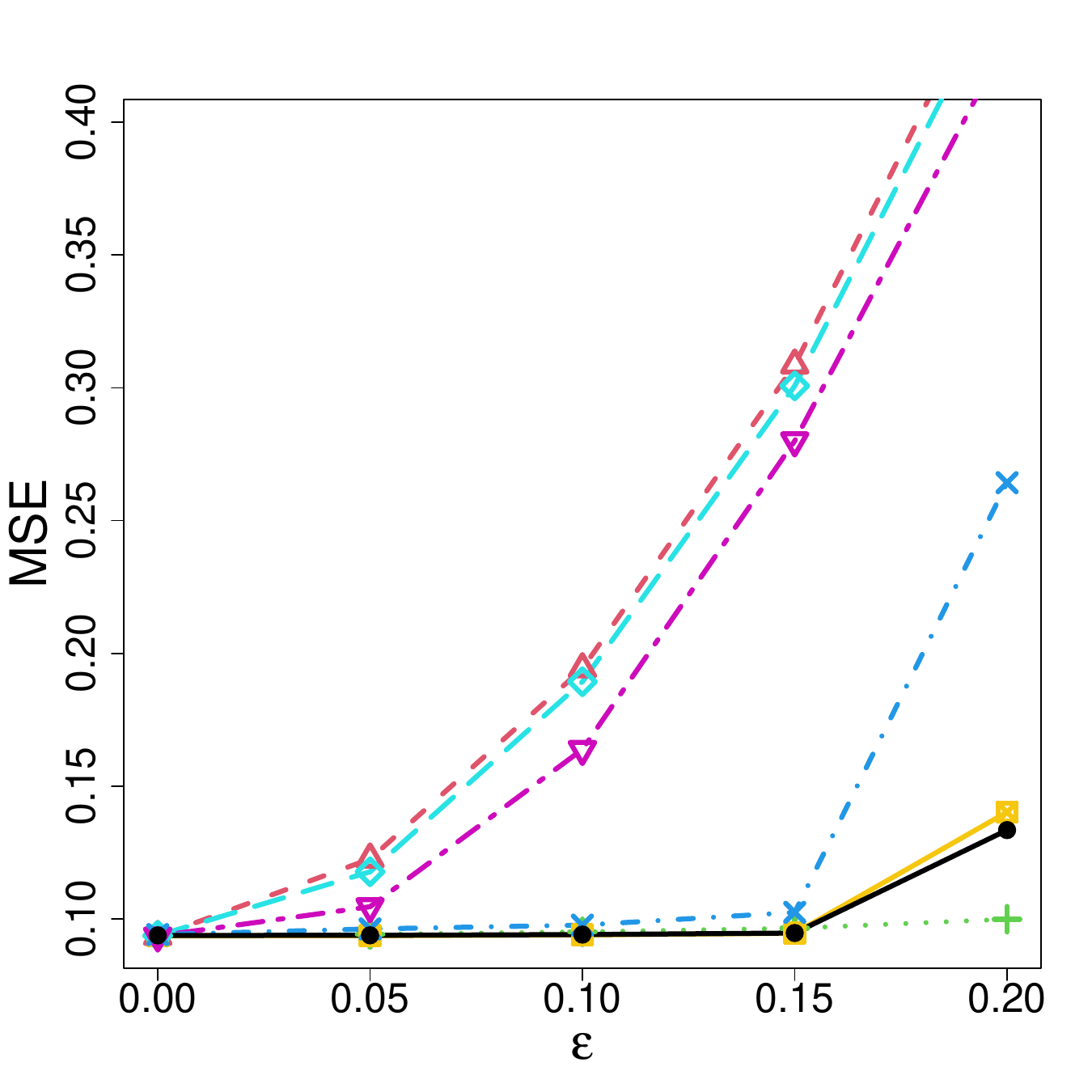} &\includegraphics[width=.3\textwidth]
  {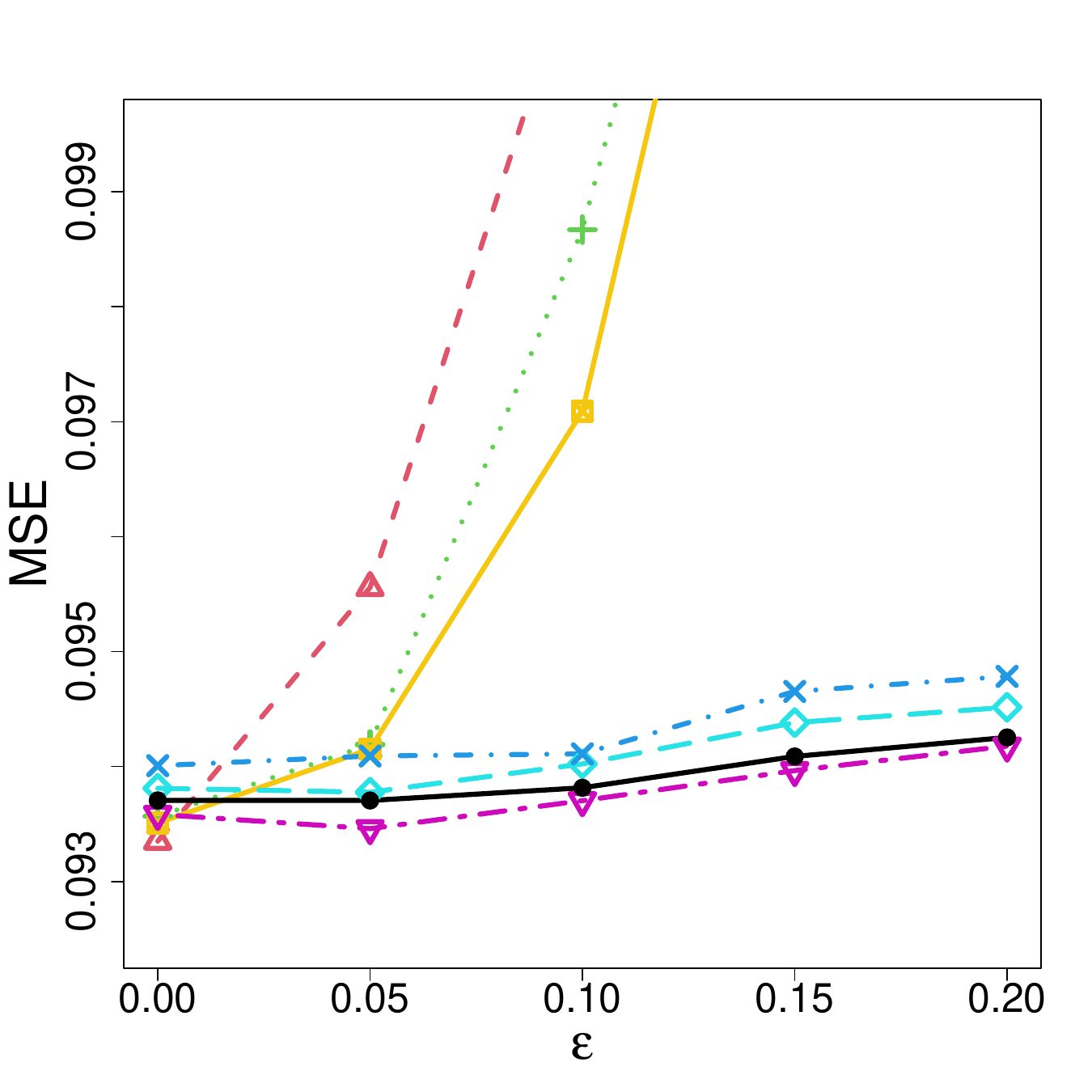} &\includegraphics[width=.3\textwidth]
  {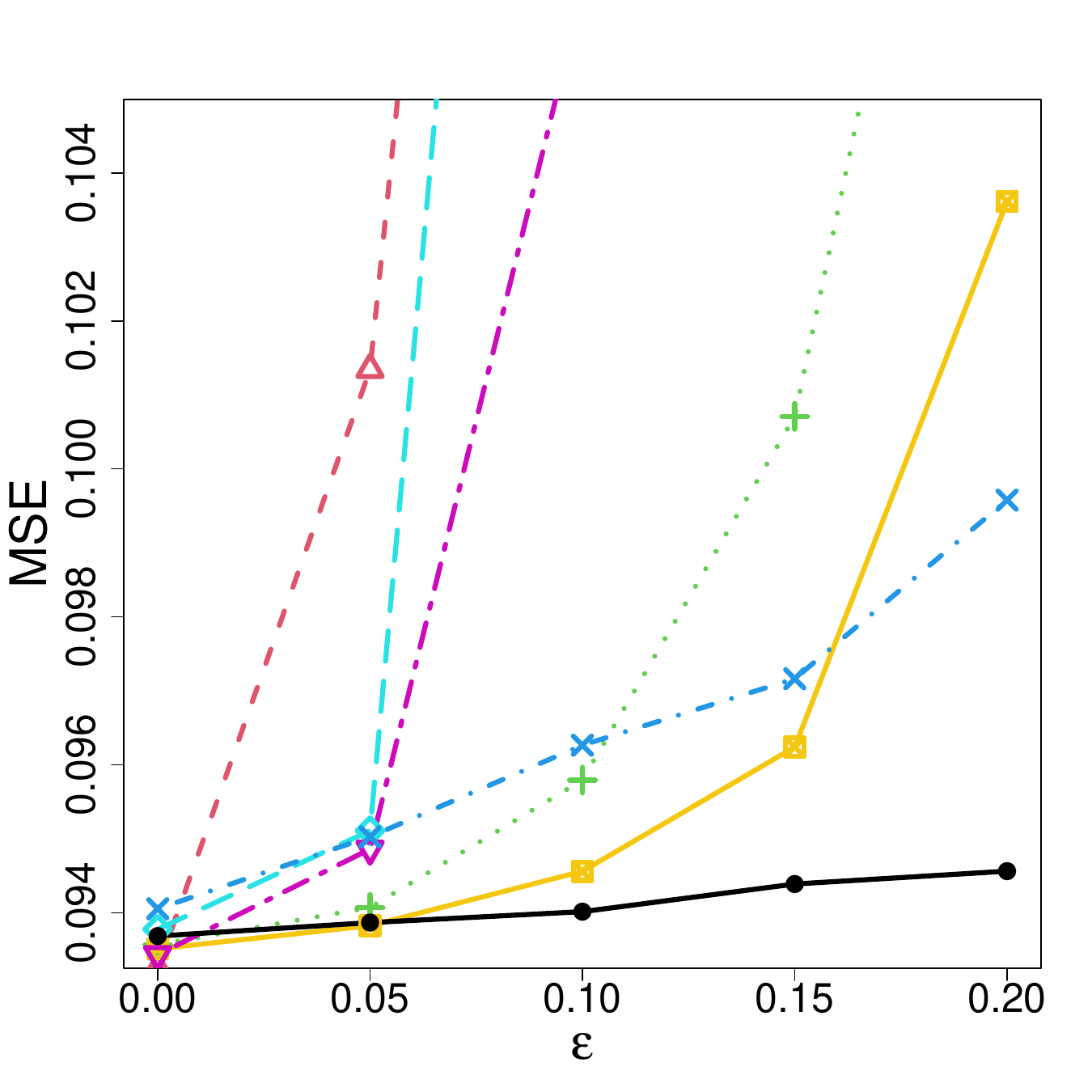}
\end{tabular}
\caption{Median angle (top) and MSE (bottom) attained by CPCA, CANDES, Only-cell, ROBPCA, Only-case, MacroPCA, and cellPCA as a function of $\varepsilon$ in the presence of cellwise outliers generated with $\gamma_{\cell}=3$, casewise outliers generated with $\gamma_{\case}=3$, and both generated with $\gamma_{\cell}=3$. The covariance model was A09 with $n=100$ and $p=20$, without NAs.}
\label{fig:results_varying_pou1}
\end{figure}

\begin{figure}[!ht]
\centering
\begin{tabular}{ccc}
   \large \textbf{Cellwise}& \large \textbf{Casewise} &\large{\textbf{Casewise \& Cellwise}} \\
   [-4mm]
  \includegraphics[width=.3\textwidth]
  {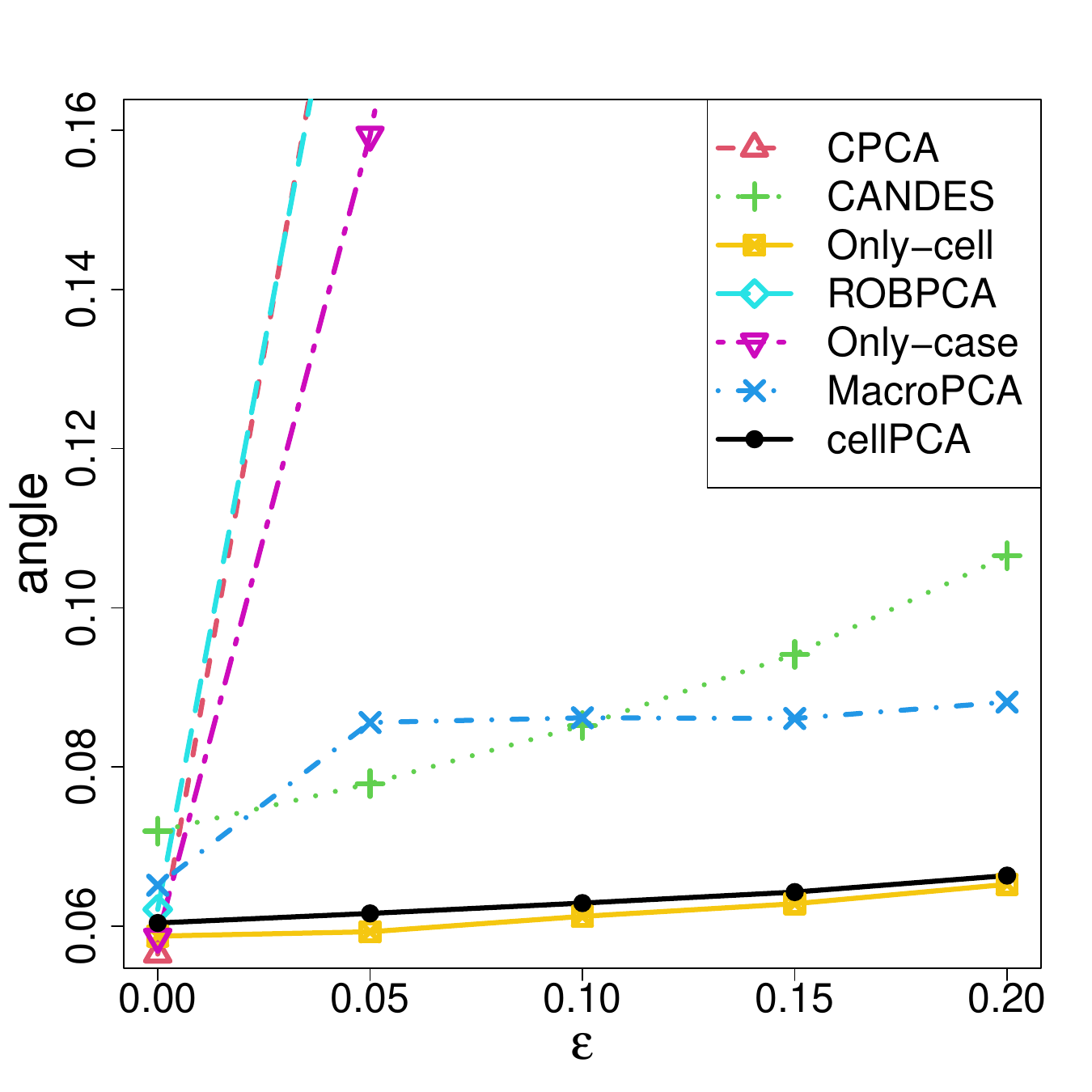} &\includegraphics[width=.3\textwidth]
  {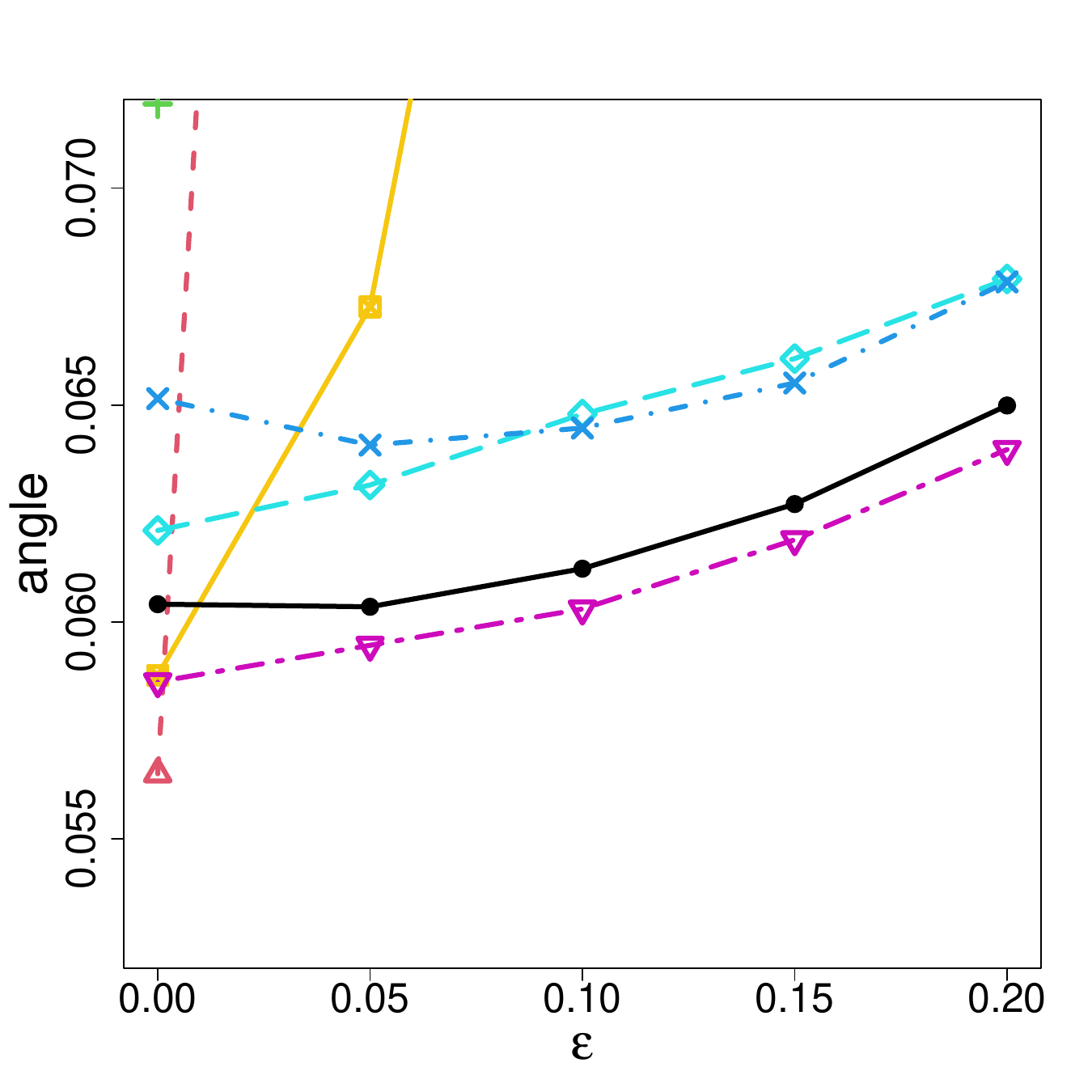} &\includegraphics[width=.3\textwidth]
  {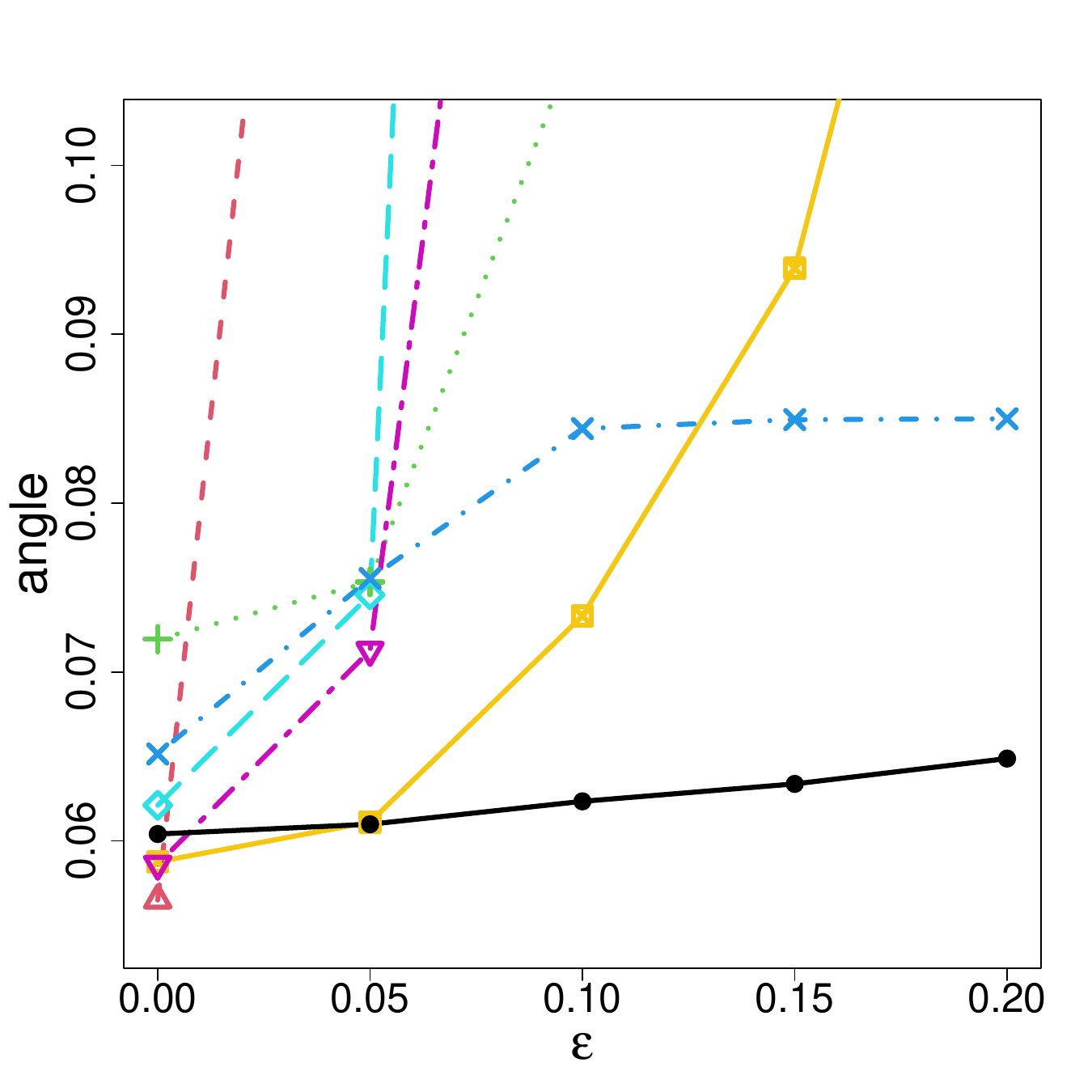}  \\
   [-4mm]

    \includegraphics[width=.3\textwidth]
  {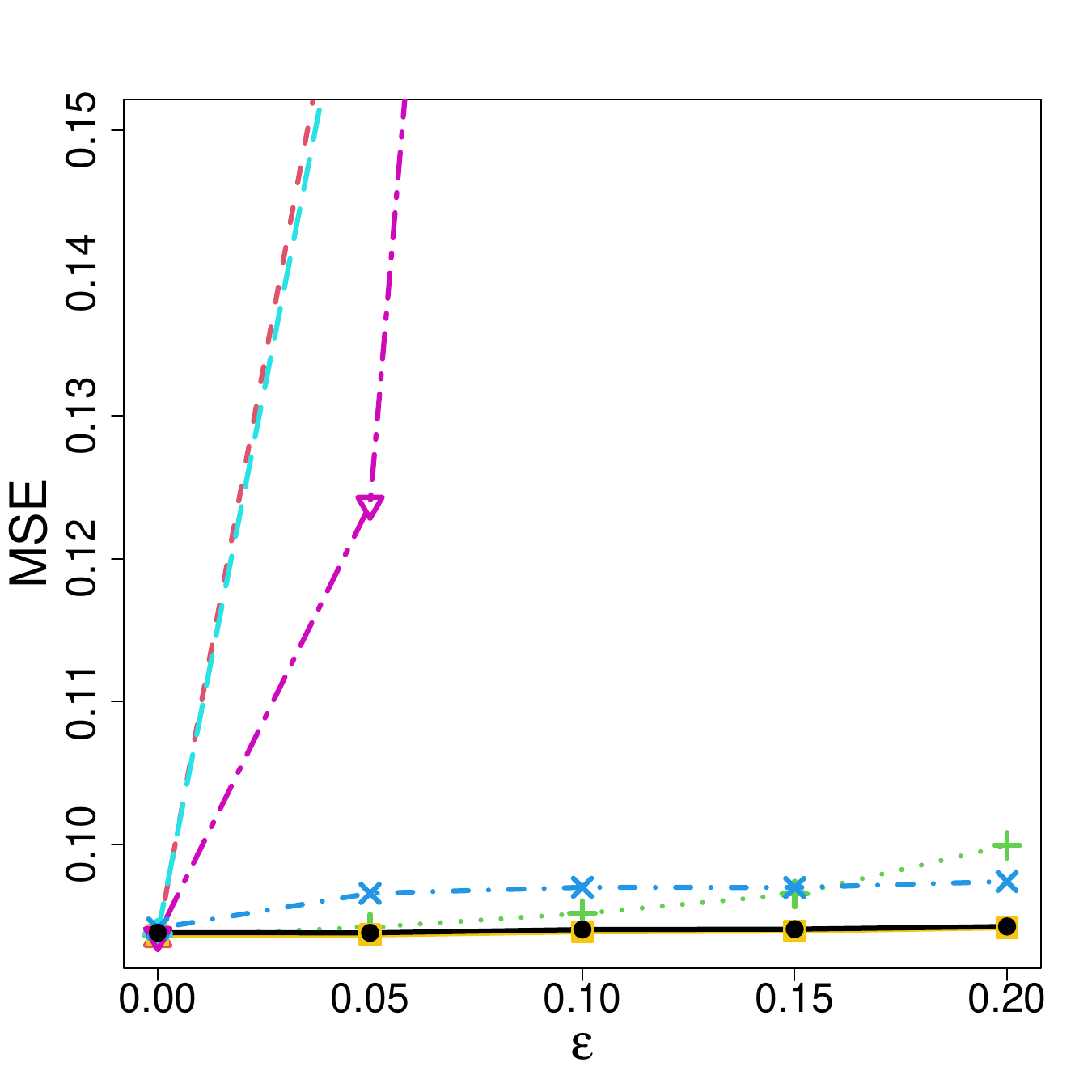} &\includegraphics[width=.3\textwidth]
  {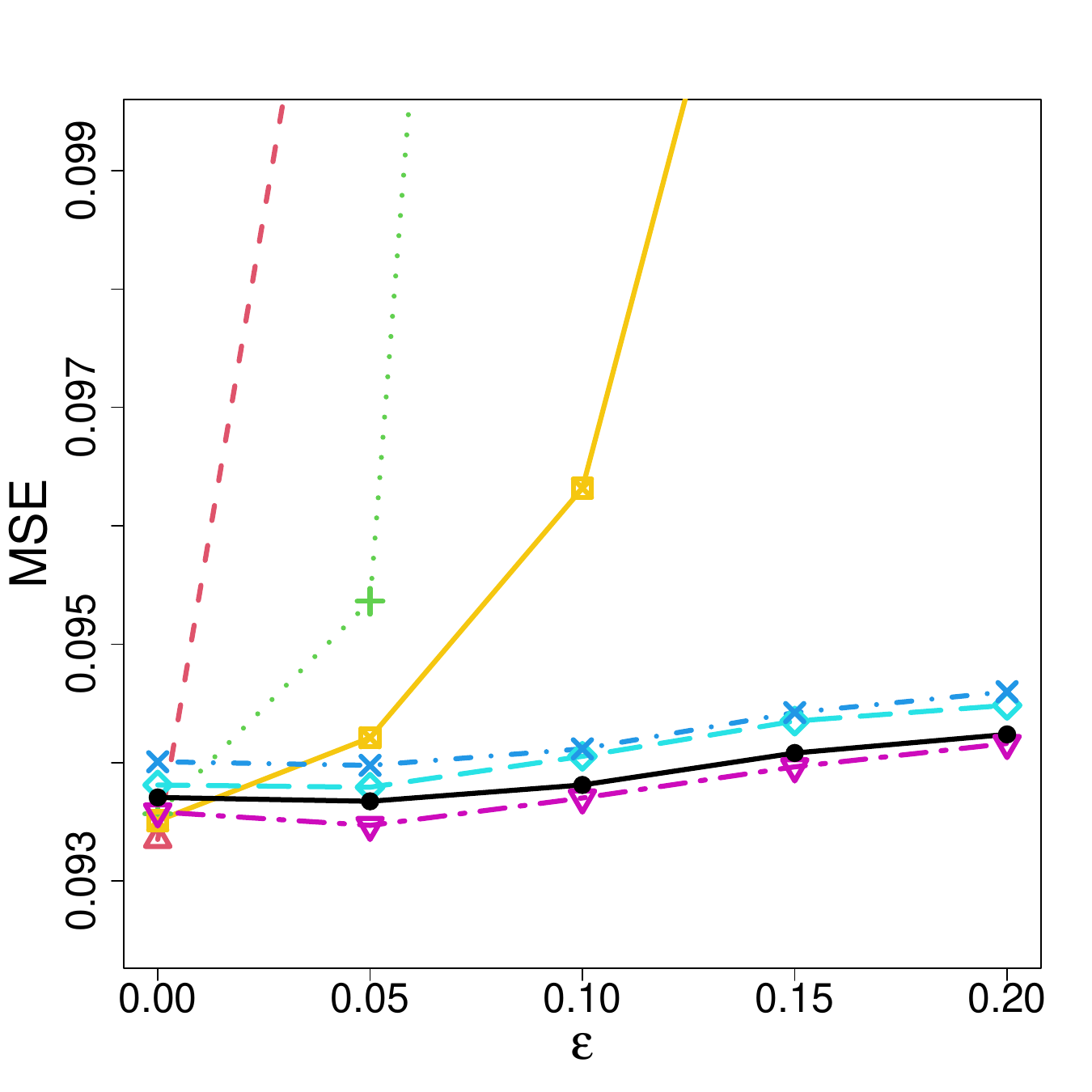} &\includegraphics[width=.3\textwidth]
  {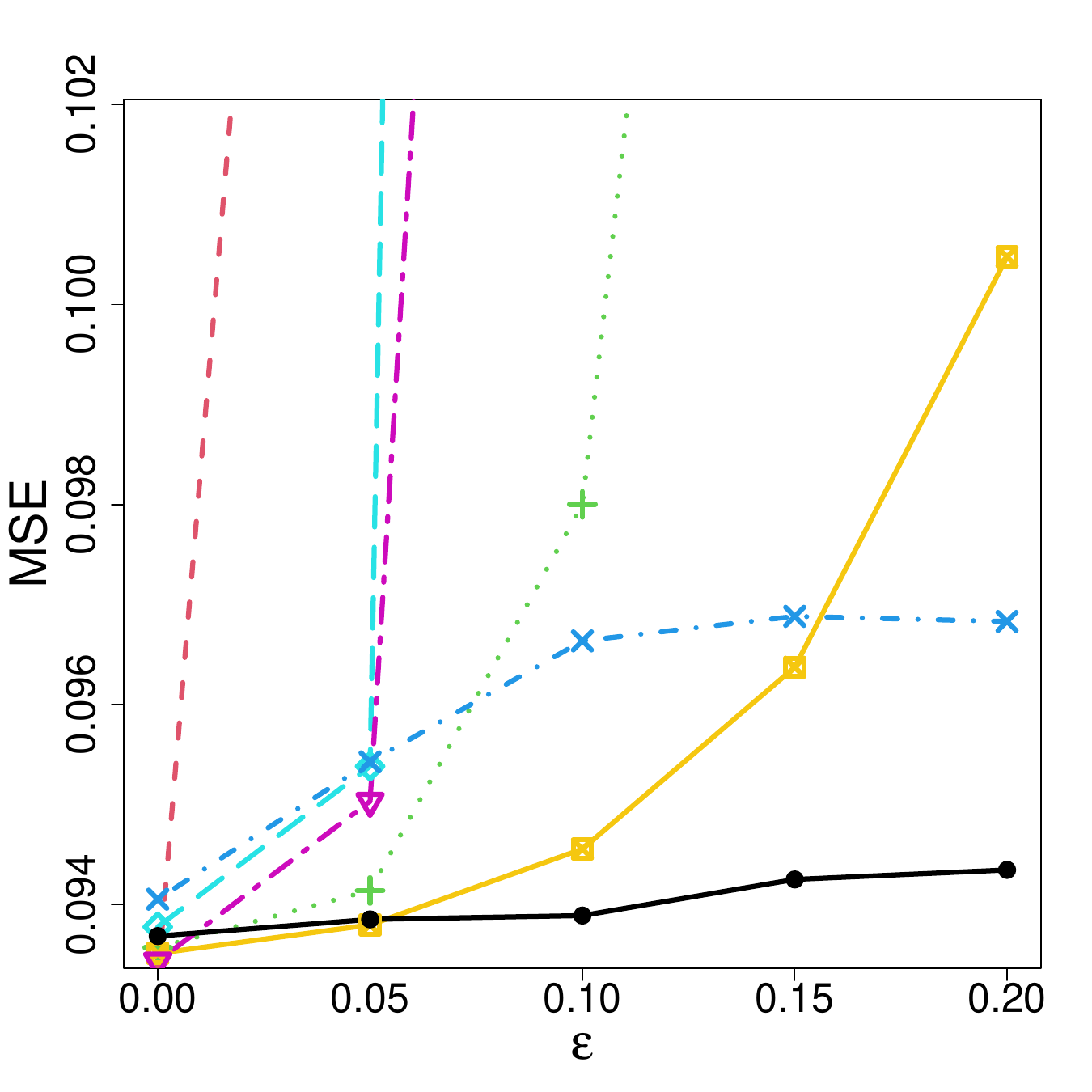}
\end{tabular}
\caption{Median angle (top) and MSE (bottom) attained by CPCA, CANDES, Only-cell, ROBPCA, Only-case, MacroPCA, and cellPCA as a function of $\varepsilon$ for data contaminated with cellwise outliers generated with $\gamma_{\cell}=5$, casewise outliers generated with $\gamma_{\case}=5$, and both generated with $\gamma_{\cell}= 5$.  The covariance model was A09 with $n=100$ and $p=20$, without NAs.}
\label{fig:results_varying_pou2}
\end{figure}

In Figures~\ref{fig:results_varying_pou1} 
and~\ref{fig:results_varying_pou2} we see that, as expected,
increasing $\varepsilon$ hurts all methods, but not to the same
extent. CellPCA did best in the presence of both cellwise
and casewise outliers, and did very well also in the other settings.

Indeed, in the cellwise contamination setting, cellPCA performs similarly to Only-cell and CANDES, which are specifically designed to tackle this type of contamination. However, in the casewise contamination setting the latter methods break down, whereas cellPCA still performs very well and is comparable to robust casewise methods such as Only-case and ROBPCA. In the presence of both types of contamination, cellPCA outperforms all other methods, including MacroPCA.
As expected, differences in performance become more pronounced as the contamination fraction increases, with higher levels representing more challenging scenarios. Nevertheless, even for a small fraction of contamination, e.g., $\eps = 0.05$, cellPCA is often the best, or among the best performing methods.
Note that for very small values of $\eps$, some methods are slightly more efficient than cellPCA. However, this reflects the common trade-off between robustness and efficiency.
In real applications, however, the amount and type of contamination are rarely known in advance. This uncertainty further supports the use of the proposed method, which overall performs best across a wide range of settings.

\clearpage
To study the effect of different fractions $\varepsilon^{\obs}$ of missing values, Figures~\ref{fig:results_varying_NA1},~\ref{fig:results_varying_NA2}, and~\ref{fig:results_varying_NA3}
show the median 
angle and MSE  as a function of $\varepsilon^{\obs}$ for uncontaminated data, data contaminated with cellwise
outliers generated with $\gamma_{\cell}=\lbrace 3, 5\rbrace$, casewise outliers generated with $\gamma_{\case}=\lbrace 2,5 \rbrace$, and both with $\gamma_{\cell}=\lbrace 3, 5\rbrace$. 
The covariance model was A09 with $n=100$ and $p=20$, and contamination was added as described in Section~\ref{sec:simulation}.

\begin{figure}[!ht]
\centering
\includegraphics[width=.3\textwidth]
  {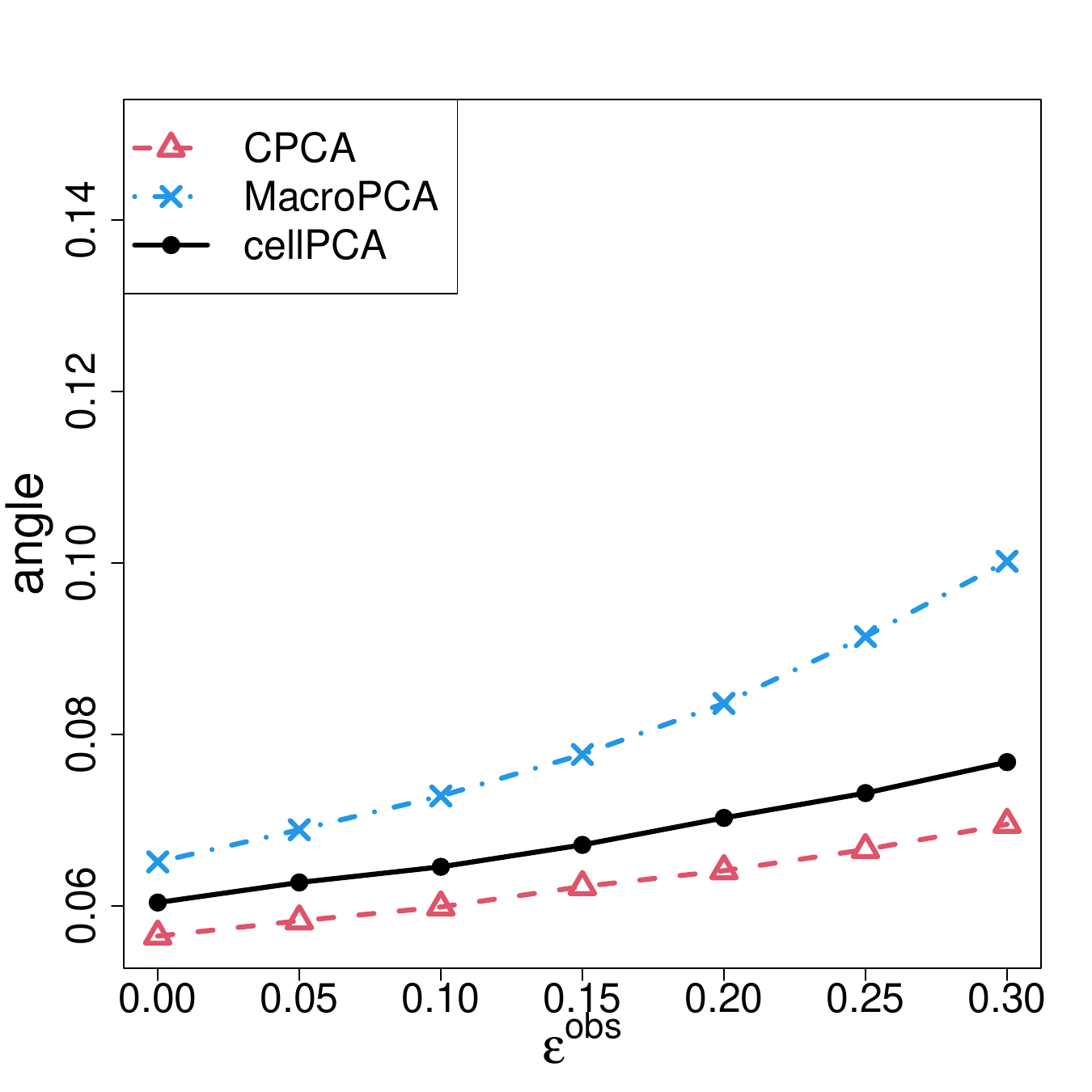} 
\includegraphics[width=.3\textwidth]
  {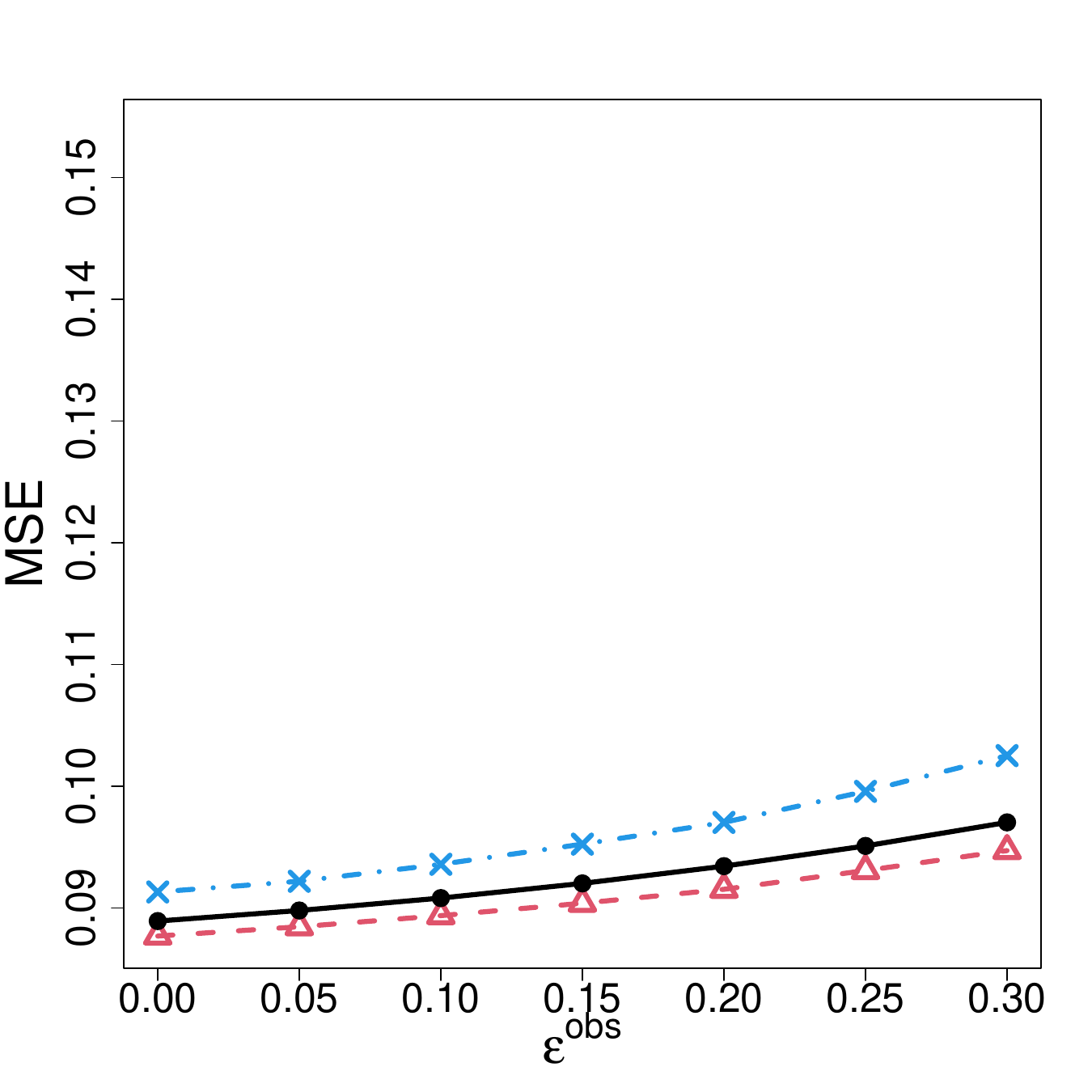} 
\caption{Median angle (left) and MSE (right) attained by CPCA, MacroPCA, and cellPCA as a function of $\varepsilon^{\obs}$ for uncontaminated data. The covariance model was A09 with $n=100$ and $p=20$.}
\label{fig:results_varying_NA1}
\end{figure}

In Figure~\ref{fig:results_varying_NA1} we see
that for uncontaminated data the classical method did
best, as expected. It is followed by cellPCA, which
outperforms its predecessor MacroPCA.\\

\clearpage

Figure~\ref{fig:results_varying_NA2} shows the
effect of a varying percentage of NAs in the presence
of intermediate outliers ($\gamma_{\cell} = 3$ and/or
$\gamma_{\case} = 3$). Now cellPCA performs best, followed
by MacroPCA, and CPCA is the most affected.

The situation is similar in the presence of far
outliers ($\gamma_{\cell} = 5$ and/or $\gamma_{\case} = 5$) in Figure~\ref{fig:results_varying_NA3}, except that 
CPCA now underperforms more substantially relative 
to the other methods.\\

\begin{figure}[!ht]
\centering
\begin{tabular}{ccc}
   \large \textbf{Cellwise}& \large \textbf{Casewise} &\large{\textbf{Casewise \& Cellwise}} \\
   [-4mm]
  \includegraphics[width=.3\textwidth]
  {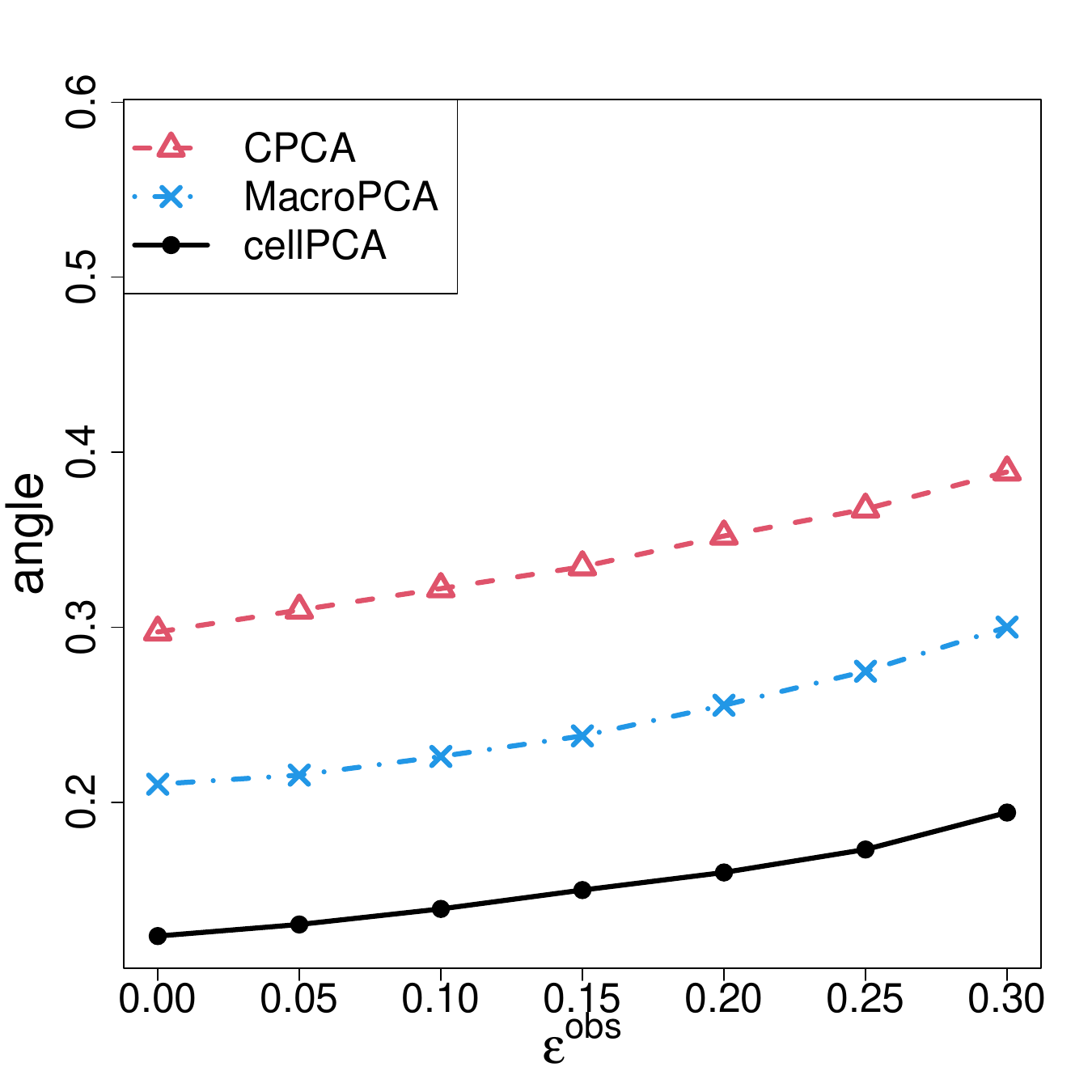} &\includegraphics[width=.3\textwidth]
  {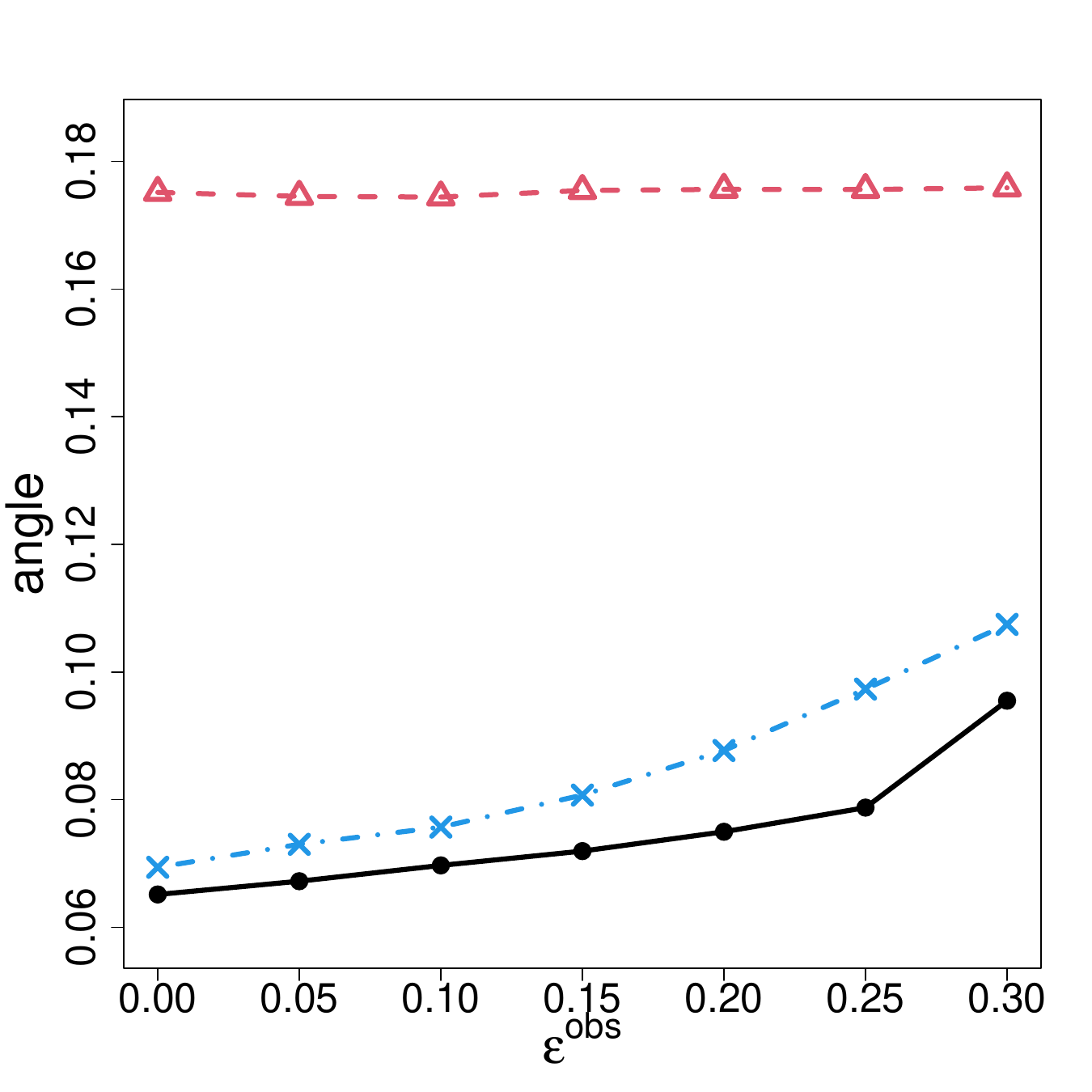} &\includegraphics[width=.3\textwidth]
  {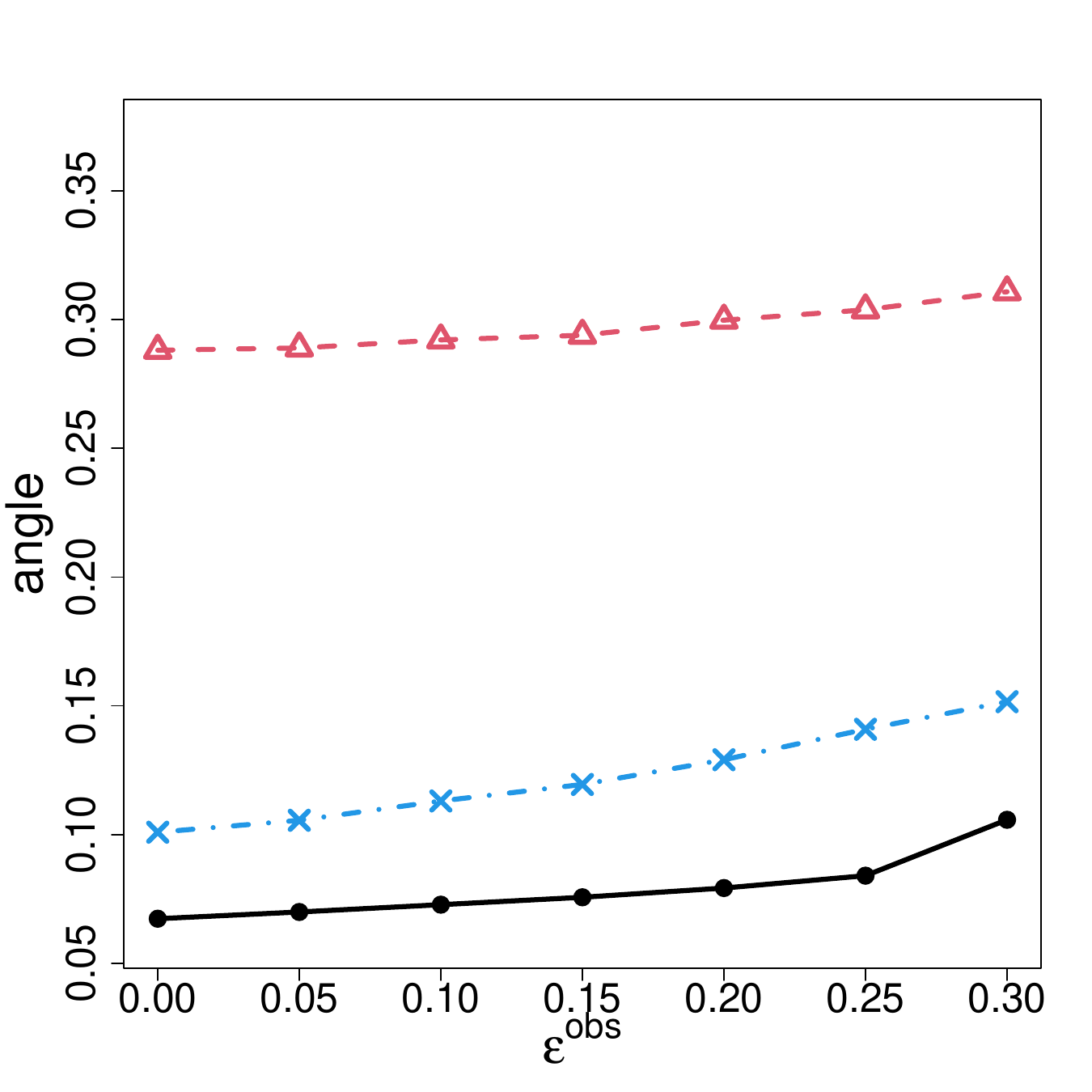}  \\
   [-4mm]
  \includegraphics[width=.3\textwidth]
  {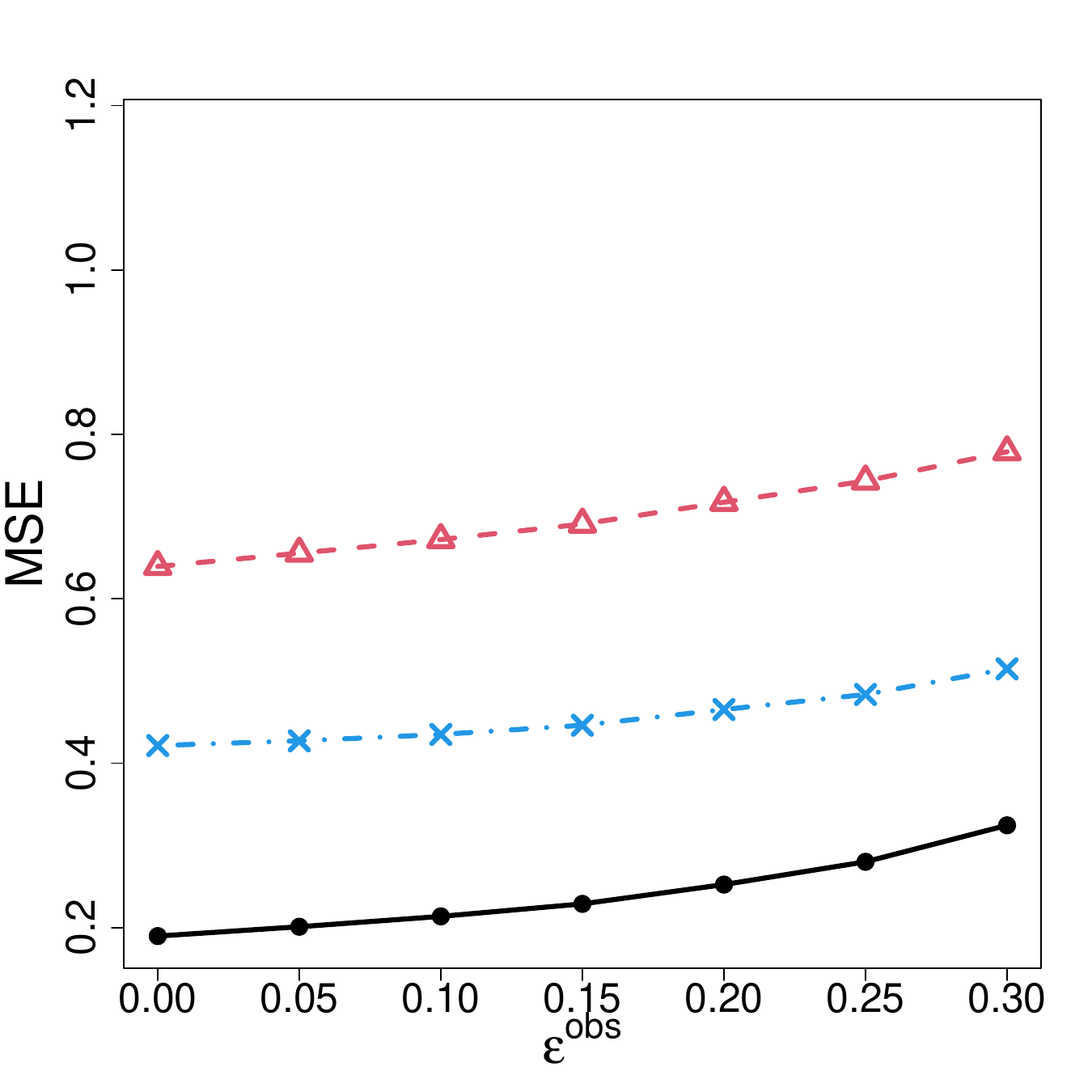} &\includegraphics[width=.3\textwidth]
  {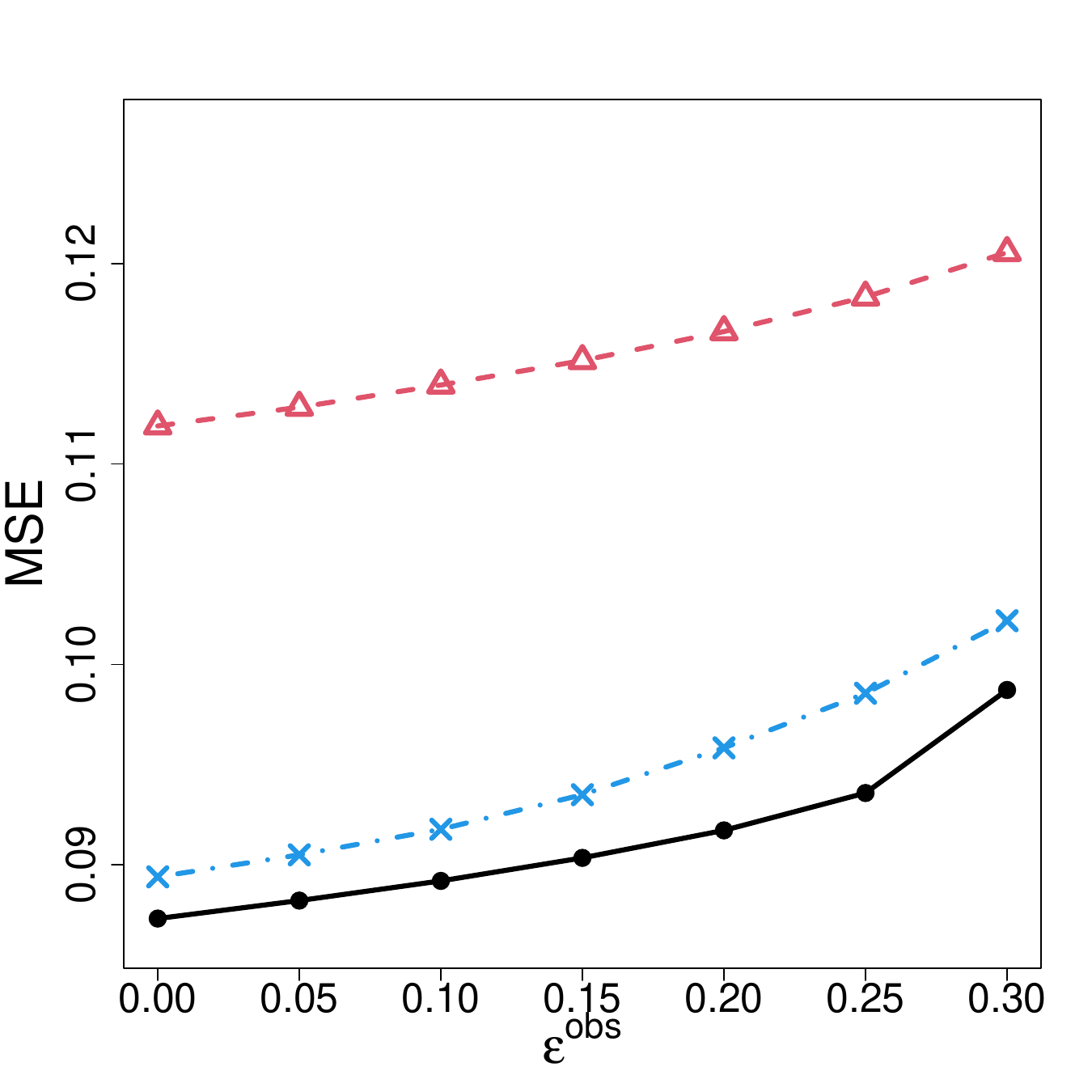} &\includegraphics[width=.3\textwidth]
  {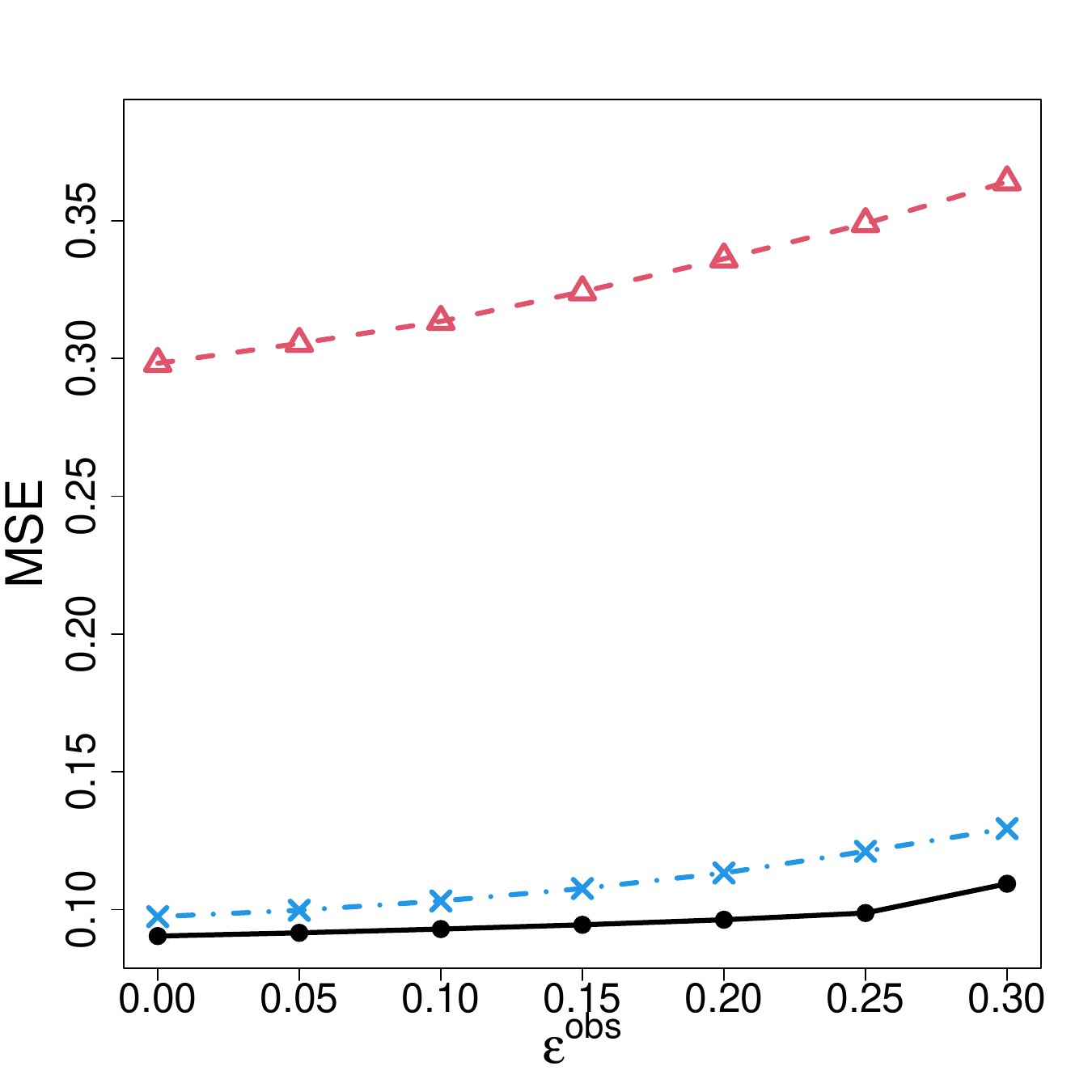}
\end{tabular}
\caption{Median angle (top) and MSE (bottom) attained by CPCA, MacroPCA, and cellPCA as a function of $\varepsilon^{\obs}$ for data contaminated with cellwise
outliers generated with $\gamma_{\cell}= 3$, casewise outliers generated with $\gamma_{\case}= 3$, and both with $\gamma_{\cell}= 3 $. The covariance model was A09 with $n=100$ and $p=20$.} 
\label{fig:results_varying_NA2}
\end{figure}

\begin{figure}[!ht]
\centering 
\begin{tabular}{ccc}
   \large \textbf{Cellwise}& \large \textbf{Casewise} &\large{\textbf{Casewise \& Cellwise}} \\
   [-4mm]
  \includegraphics[width=.3\textwidth]
  {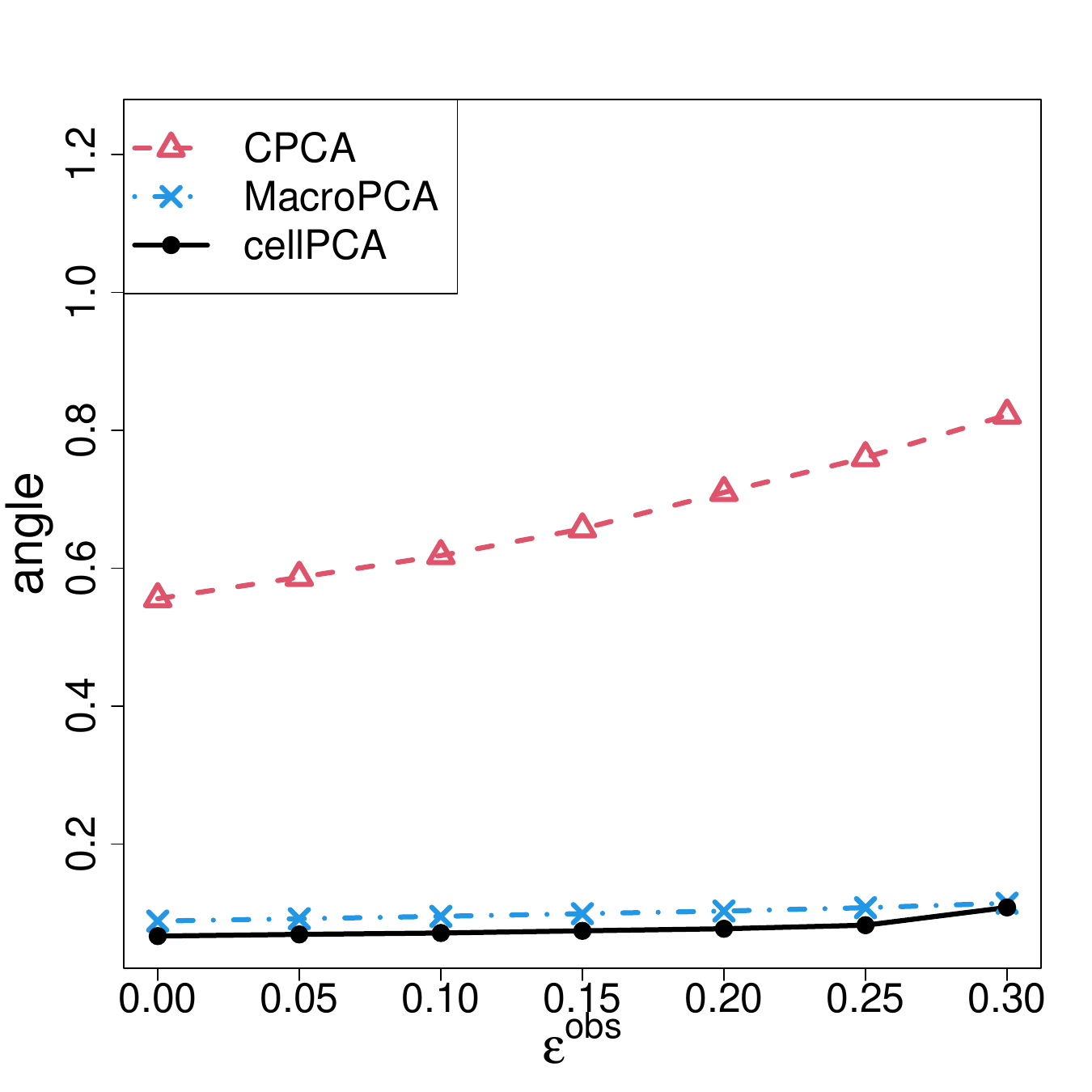} &\includegraphics[width=.3\textwidth]
  {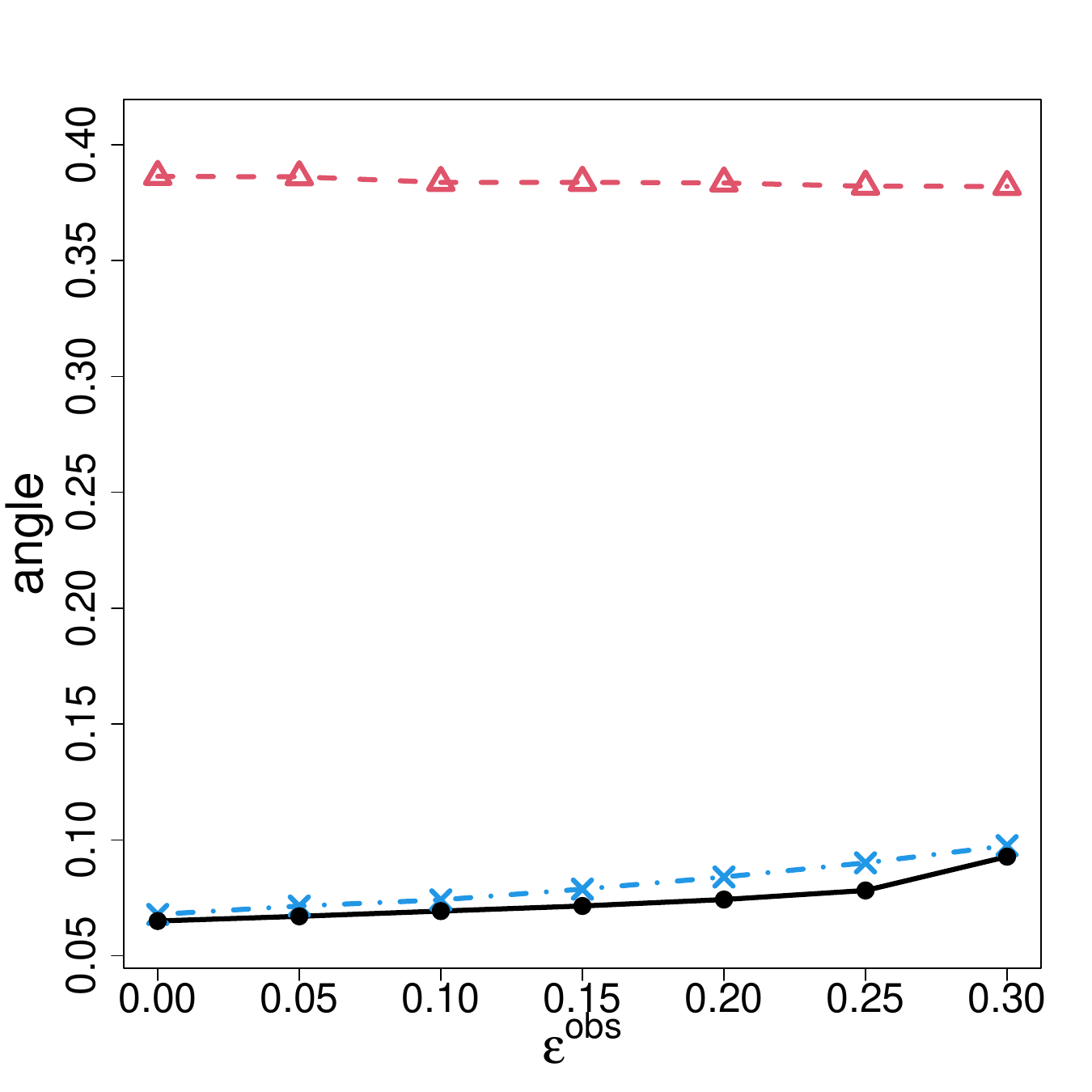} &\includegraphics[width=.3\textwidth]
  {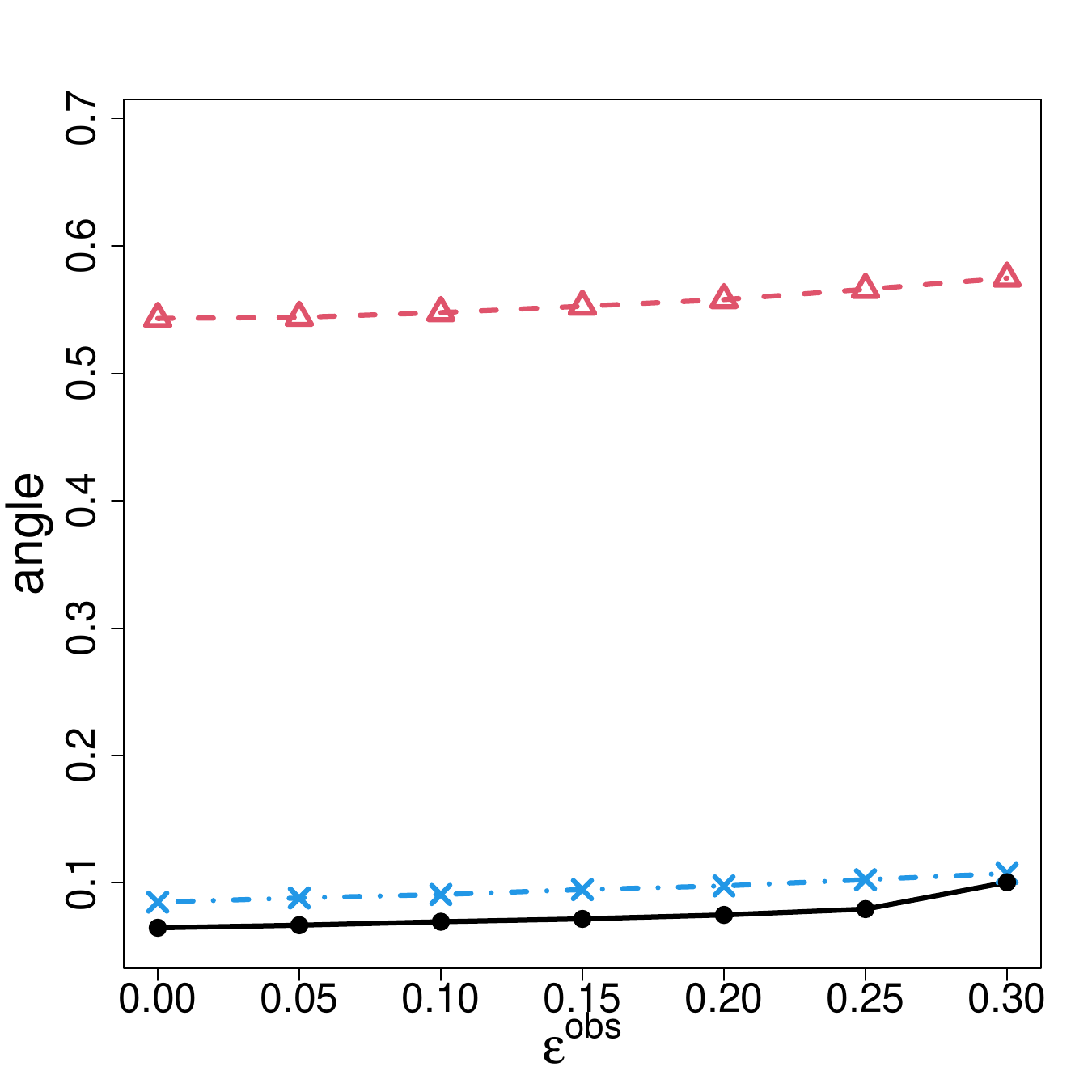}  \\
   [-4mm]
  \includegraphics[width=.3\textwidth]
  {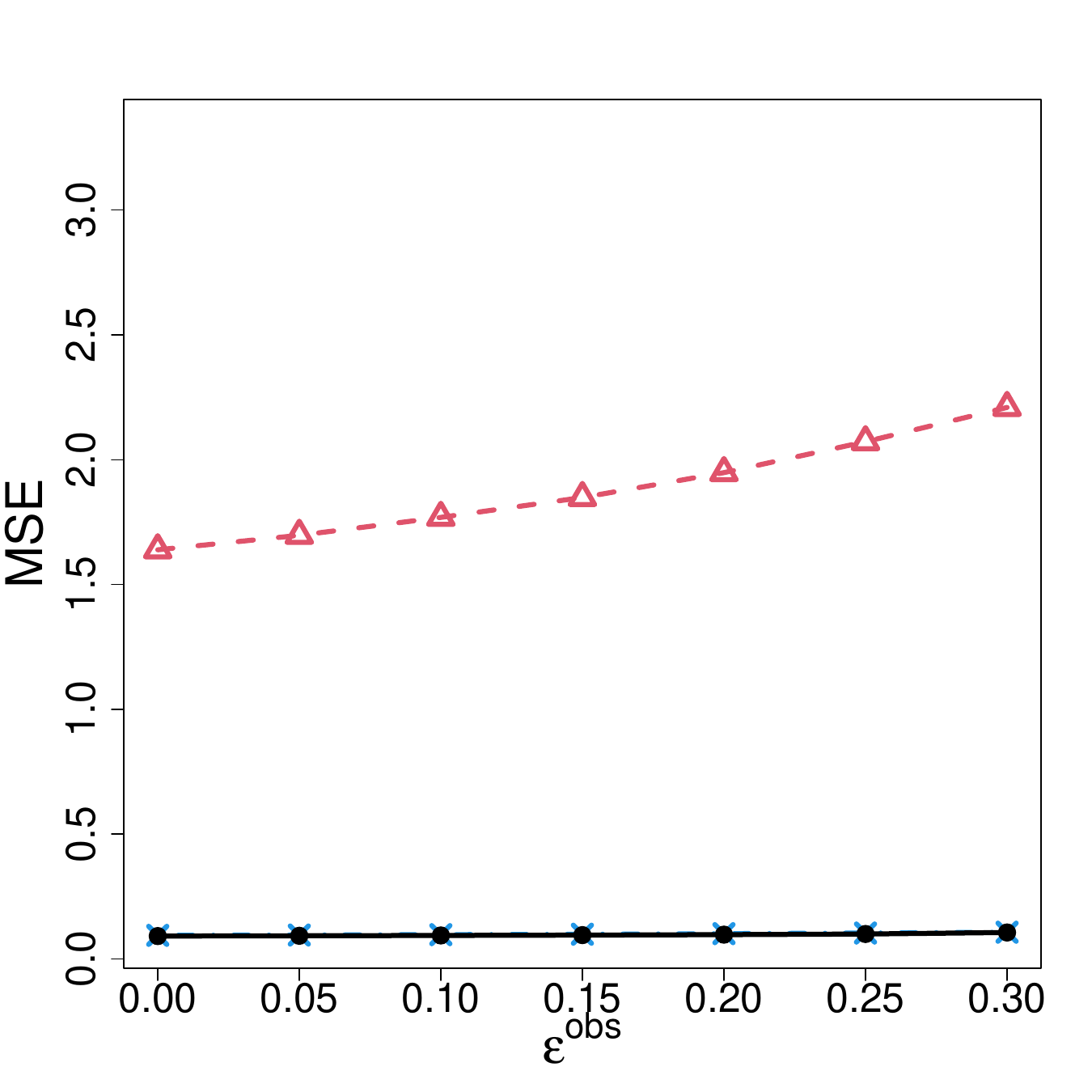} &\includegraphics[width=.3\textwidth]
  {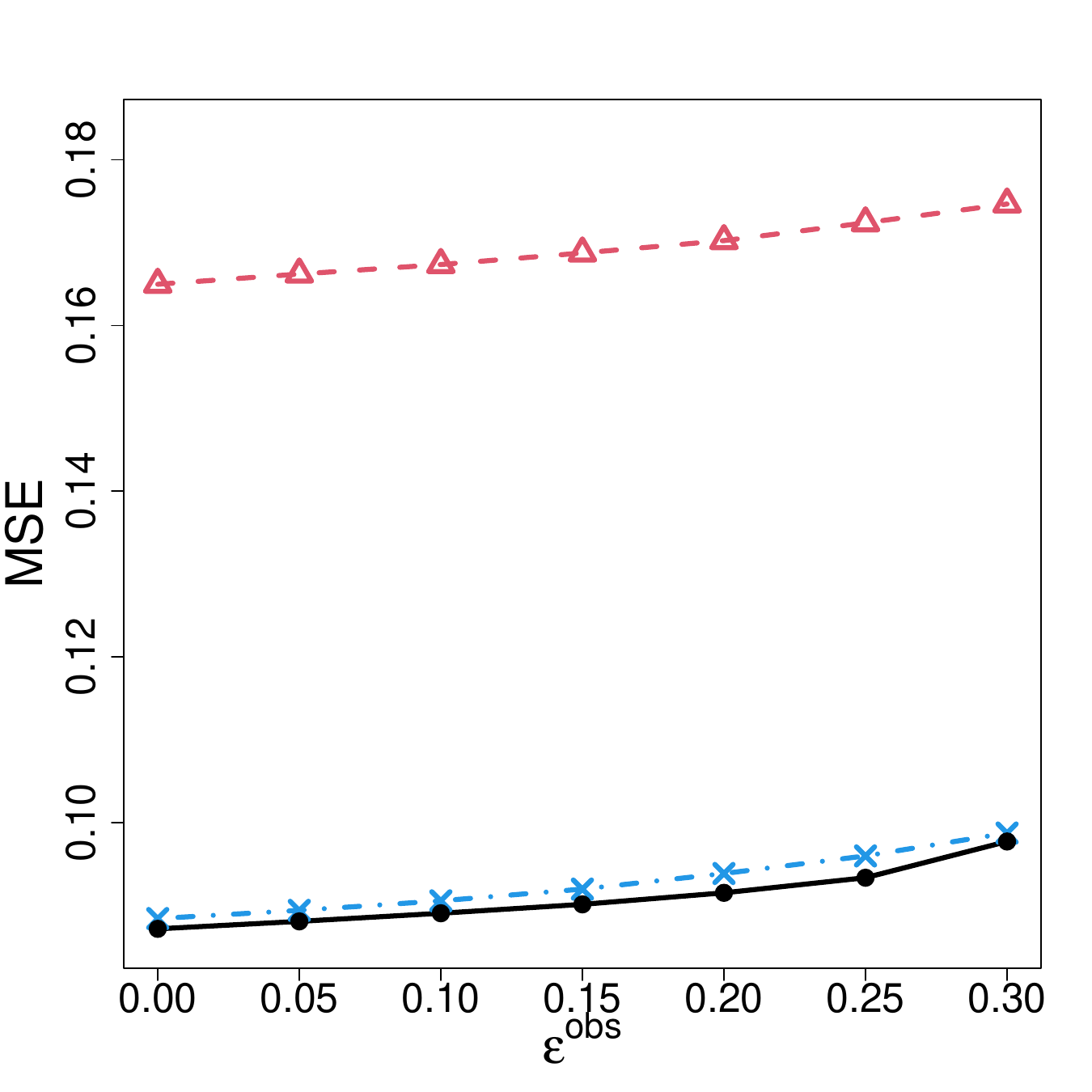} &\includegraphics[width=.3\textwidth]
  {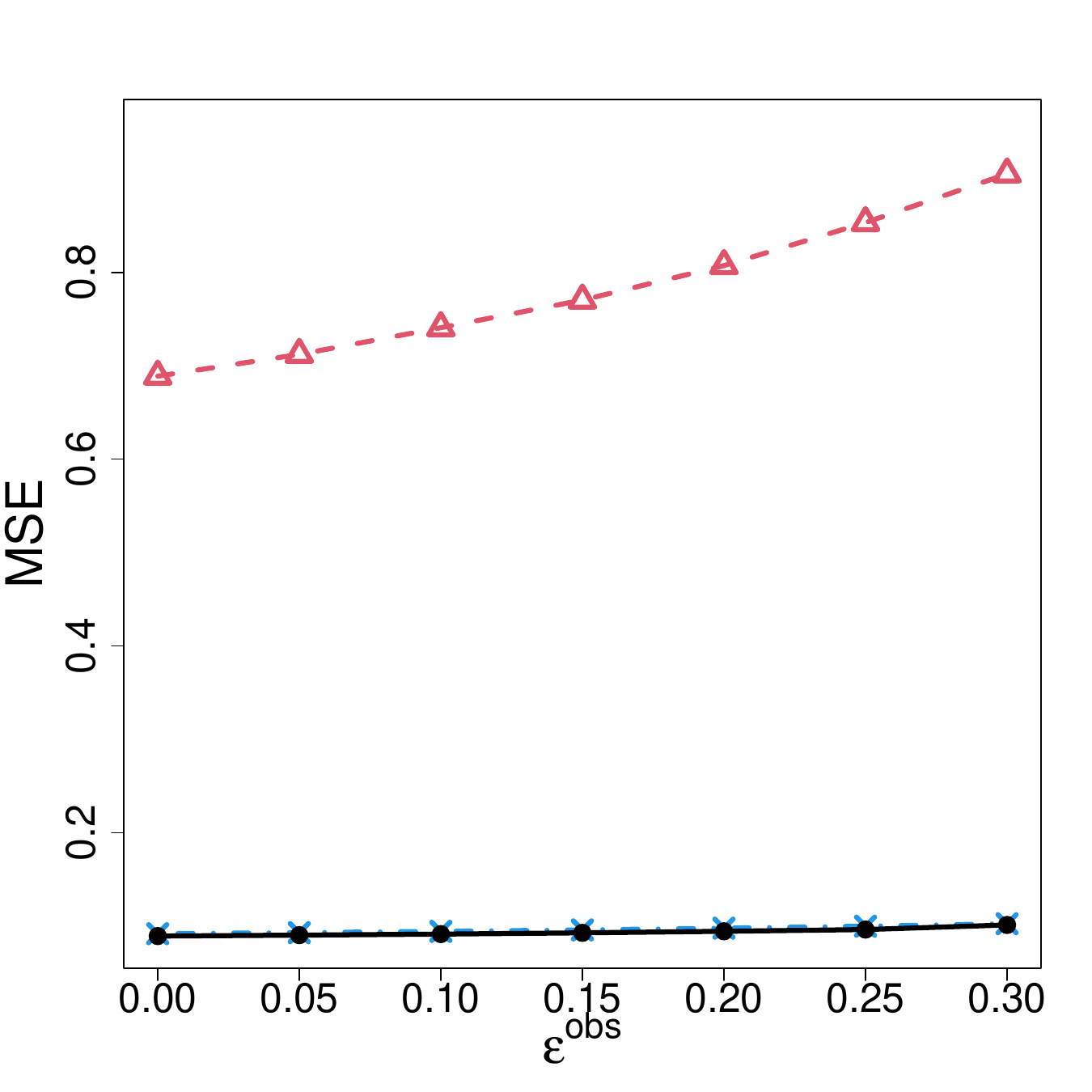}
\end{tabular}
\caption{Median angle (top) and MSE (bottom) attained by CPCA, MacroPCA, and cellPCA as a function of $\varepsilon^{\obs}$ for data contaminated with cellwise
outliers generated with $\gamma_{\cell}= 5$, casewise outliers generated with $\gamma_{\case}= 5$, and both with $\gamma_{\cell}= 5$. The covariance model was A09 with $n=100$ and $p=20$.}
\label{fig:results_varying_NA3}
\end{figure}

\clearpage
We now assess the performance of the out-of-sample 
prediction method described in Section~\ref{sec:newx}. 
The {\it training data} are generated as in Section~\ref{sec:simulation} by covariance model A09 with $n=100$ and $p=20$ without NAs, in the presence of cellwise outliers, casewise outliers, or both. The {\it test set}  follows the same covariance model with again $n=100$ and is either clean, contaminated with cellwise outliers, or simultaneously affected by cellwise outliers and missing values. Note that the test set does not contain casewise outliers because it is impossible to predict those, as they can be chosen arbitrarily.
Analogously to~\eqref{eq:MSE} we compute the out-of-sample MSE 
given by
\begin{equation} 
  \text{MSE}=\frac{1}{n}
  \sum_{i=1}^n\sum_{j=1}^p
  \left(\widehat{x^*_{ij}} - x^{*0}_{ij}\right)^2
\end{equation}
where $\widehat{x^*_{ij}}$ is the prediction of $x^*_{ij}$ 
and $x^{*0}_{ij}$ is the original value of that cell before 
any contamination took place. 

Figure \ref{fig:results_oos_pred} shows the median out-of-sample MSE of MacroPCA and cellPCA over 1000 replications. It has 9 panels. The three columns correspond to how the training data was generated, and the three rows reflect how the test set was generated. We clearly see that cellPCA has outperformed MacroPCA in all 9 situations.

\begin{figure}[!ht]
\centering
\begin{tabular}{cccc}
   &\large \textbf{Cellwise}  & \large \textbf{Casewise} &\large{\textbf{Casewise \& Cellwise}} \\
    [-4mm]
    \rotatebox{90}{\large \textbf{\parbox{5cm}{\centering Clean}}} &\includegraphics[width=.3\textwidth]
  {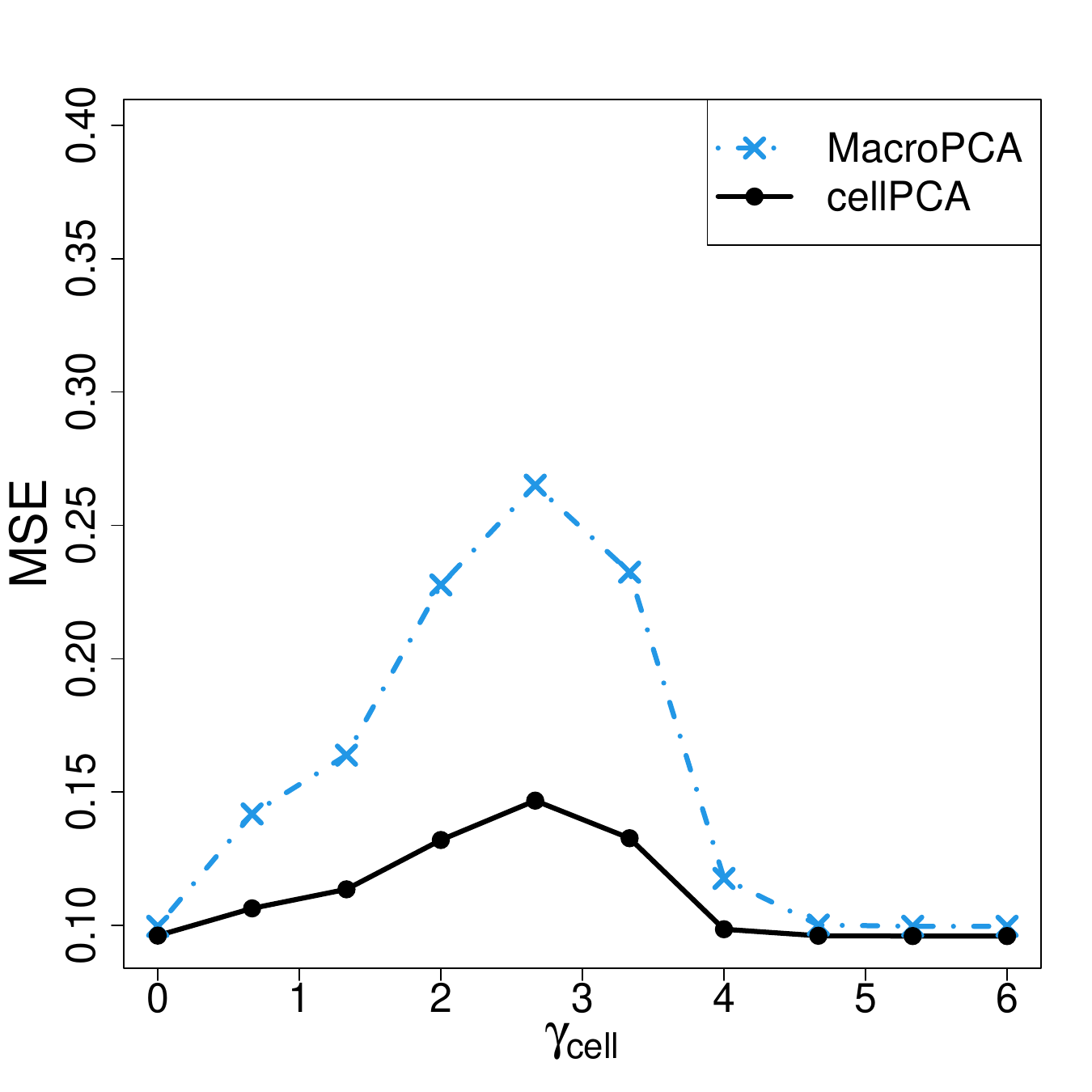} &\includegraphics[width=.3\textwidth]
  {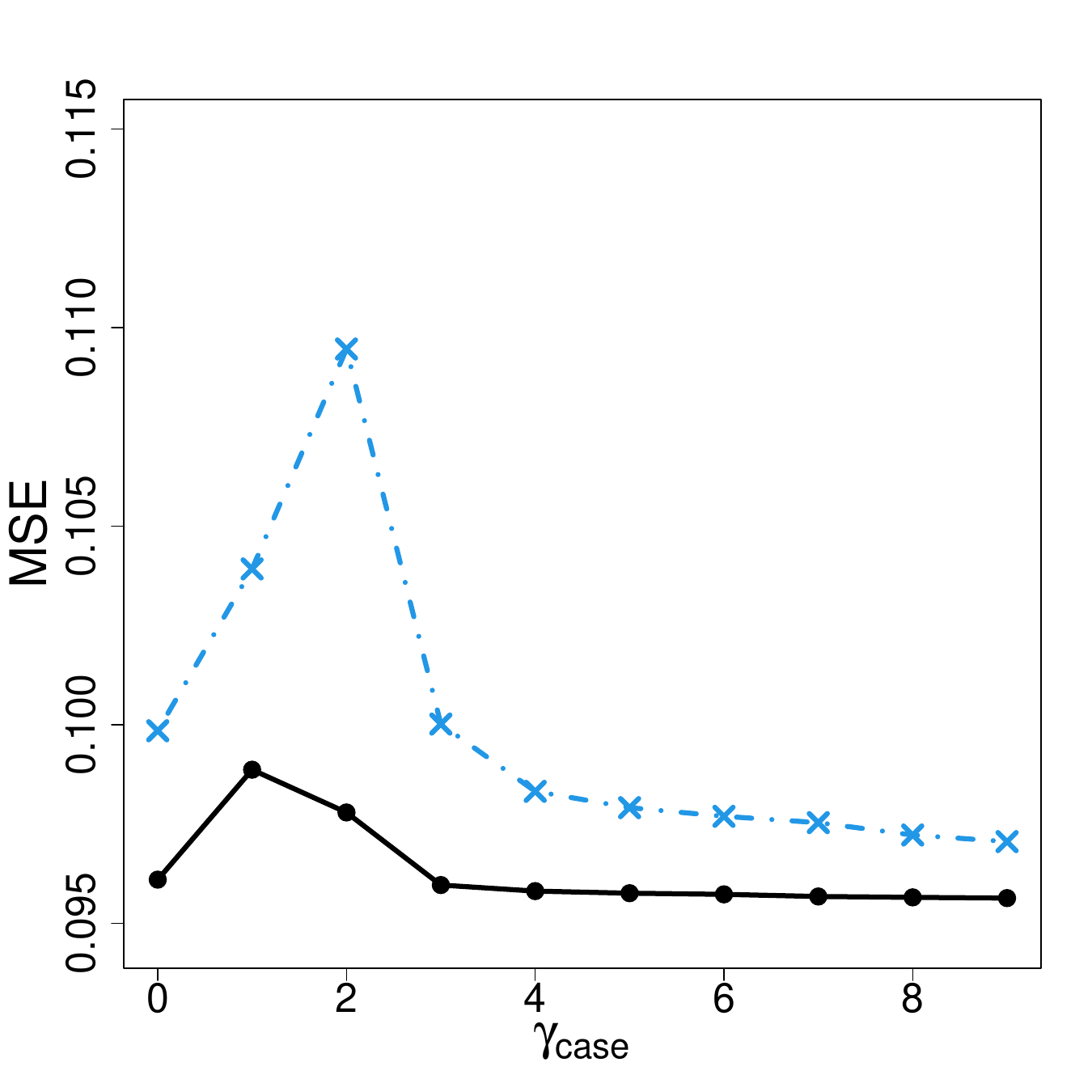} &\includegraphics[width=.3\textwidth]
  {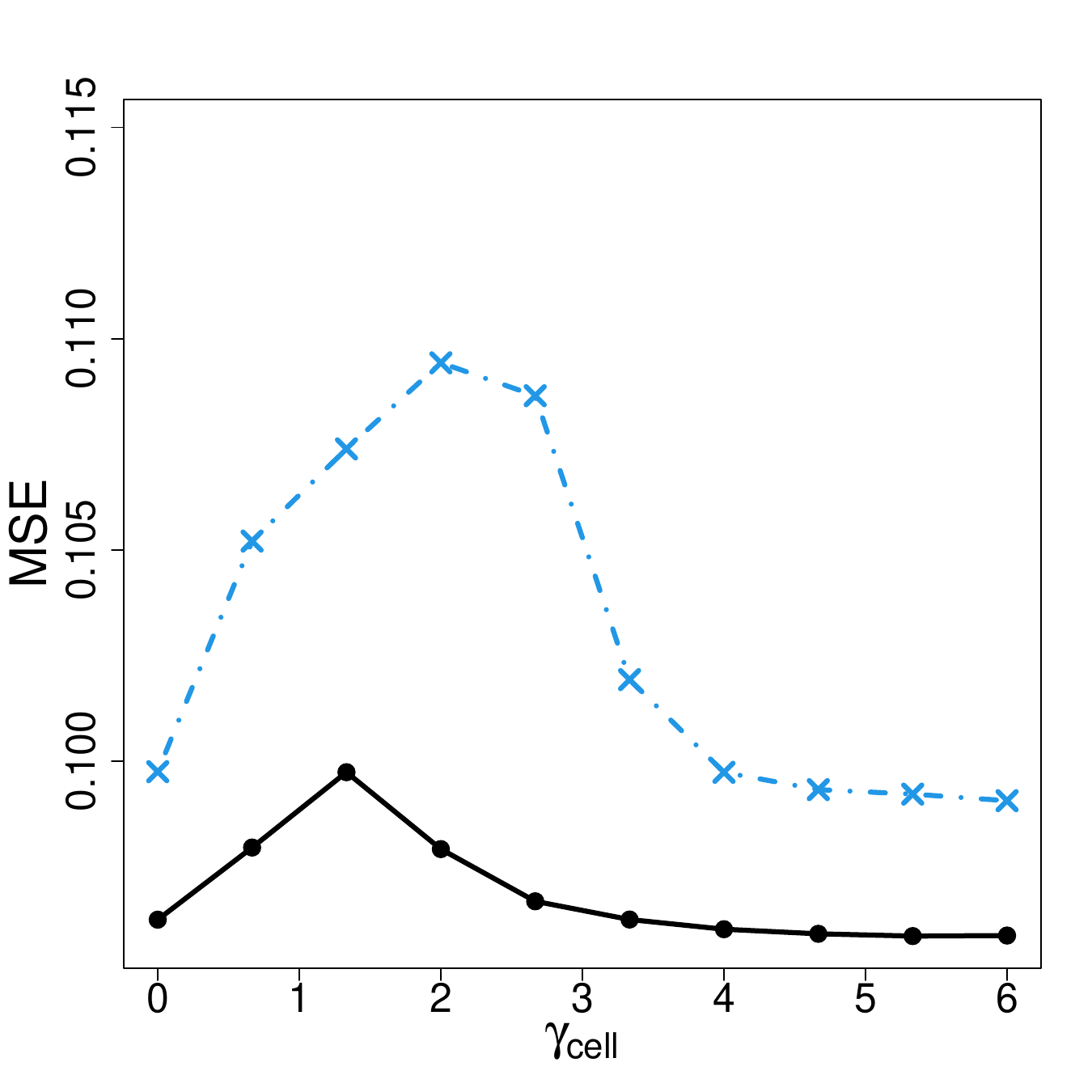}  \\
    [-4mm]
   \rotatebox{90}{\large \textbf{\parbox{5cm}{\centering Cellwise}}} &\includegraphics[width=.3\textwidth]
  {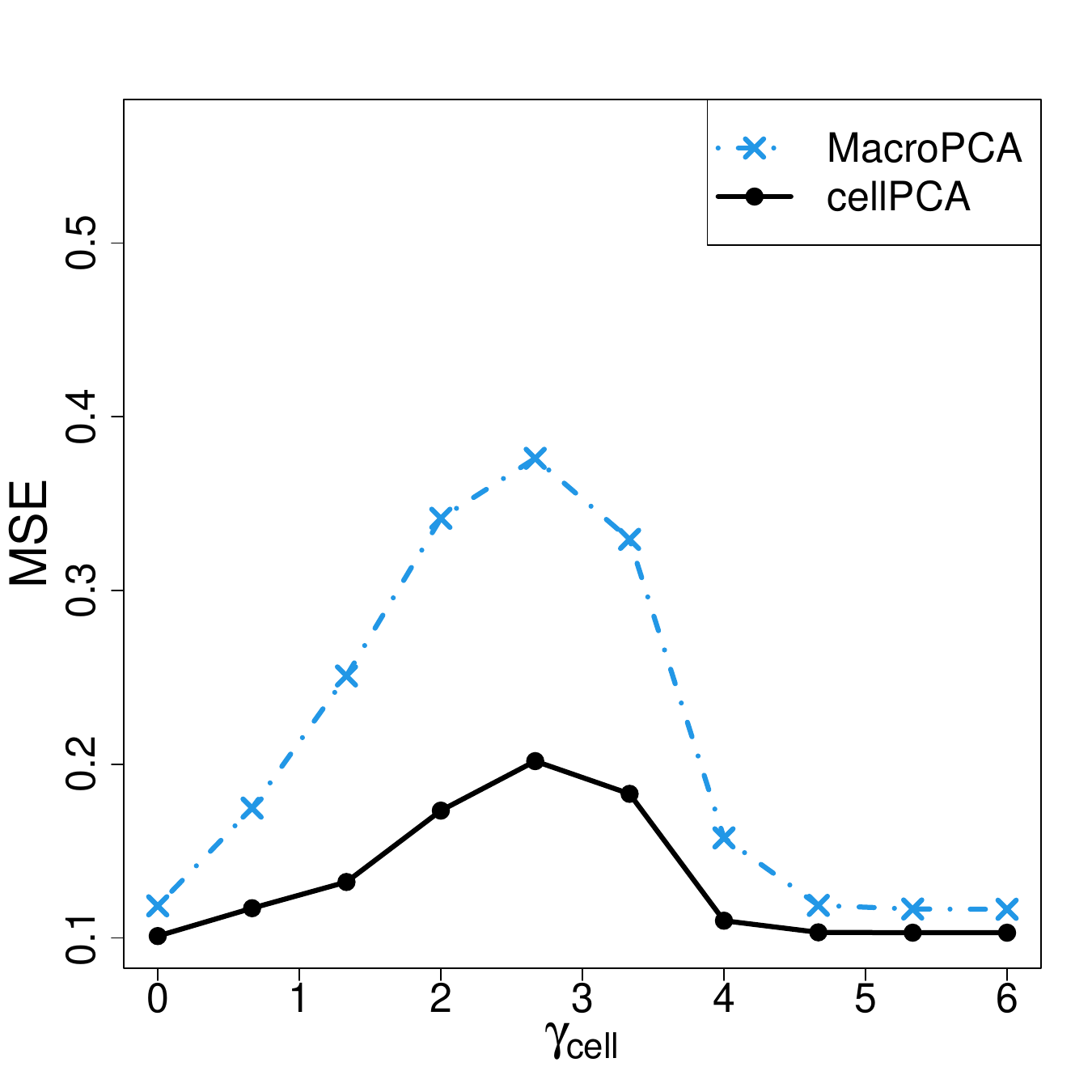} &\includegraphics[width=.3\textwidth]
  {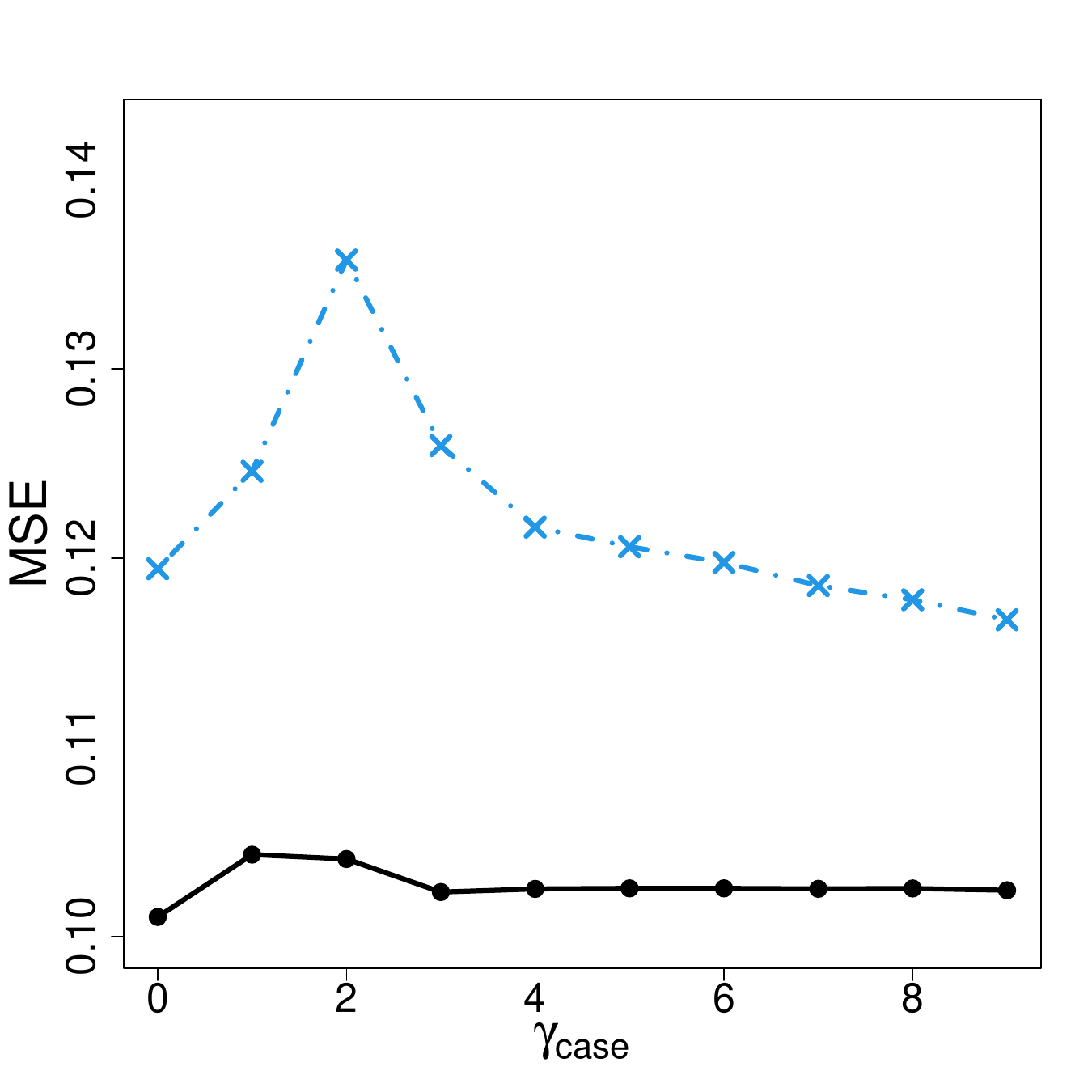} &\includegraphics[width=.3\textwidth]
  {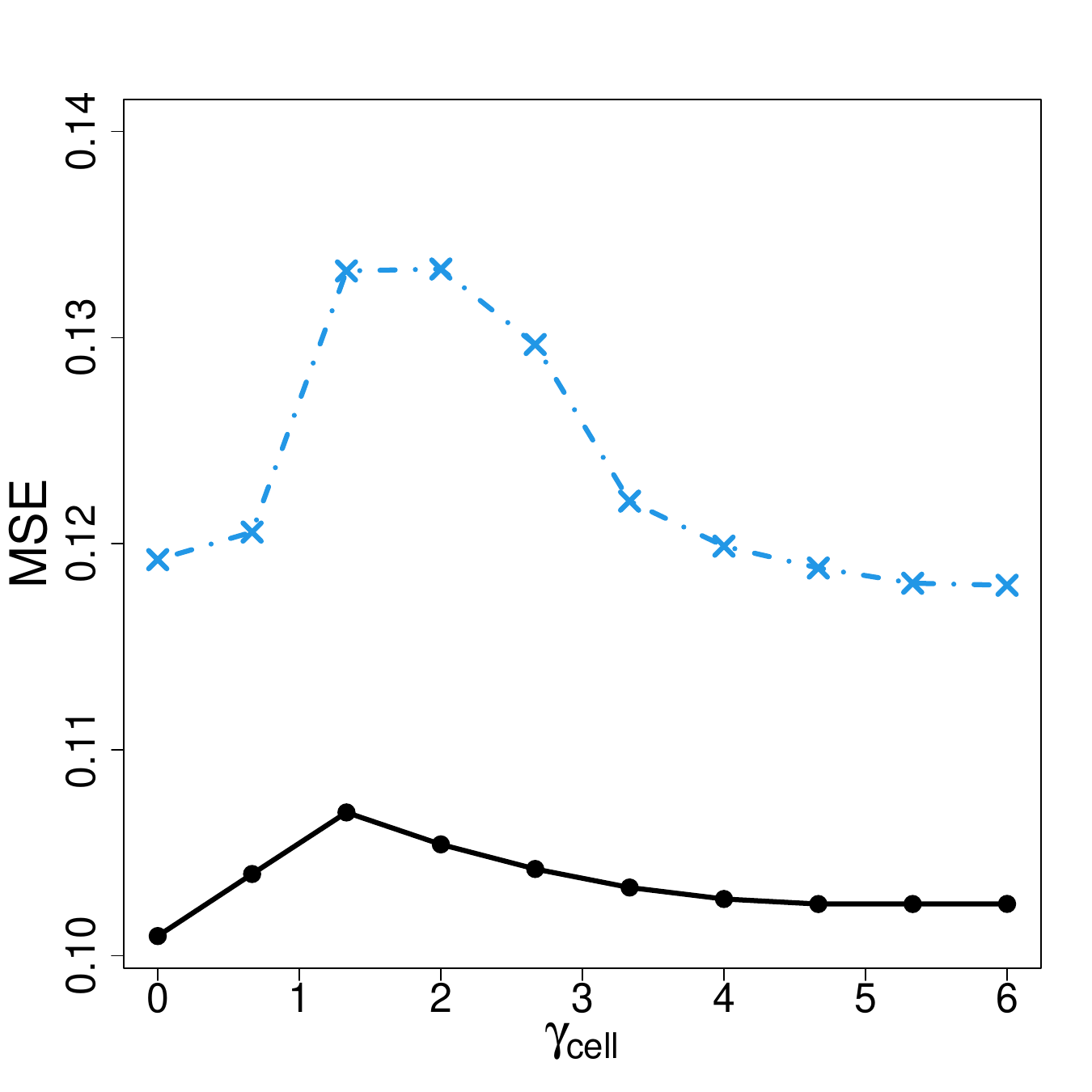}  \\
   [-4mm]
   \rotatebox{90}{\large \textbf{\parbox{5cm}{\centering Cellwise \& NAs}}}&\includegraphics[width=.3\textwidth]
  {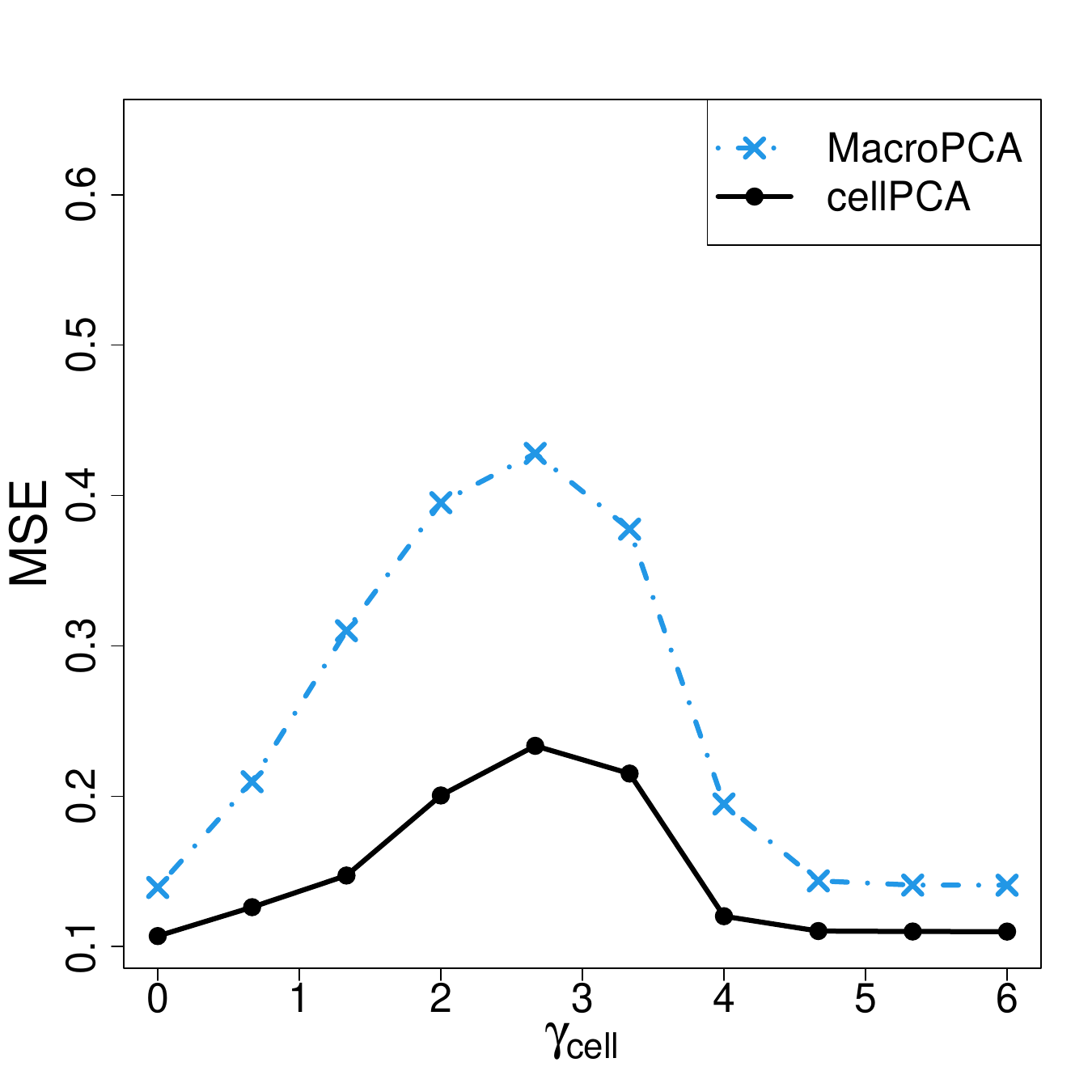} &\includegraphics[width=.3\textwidth]
  {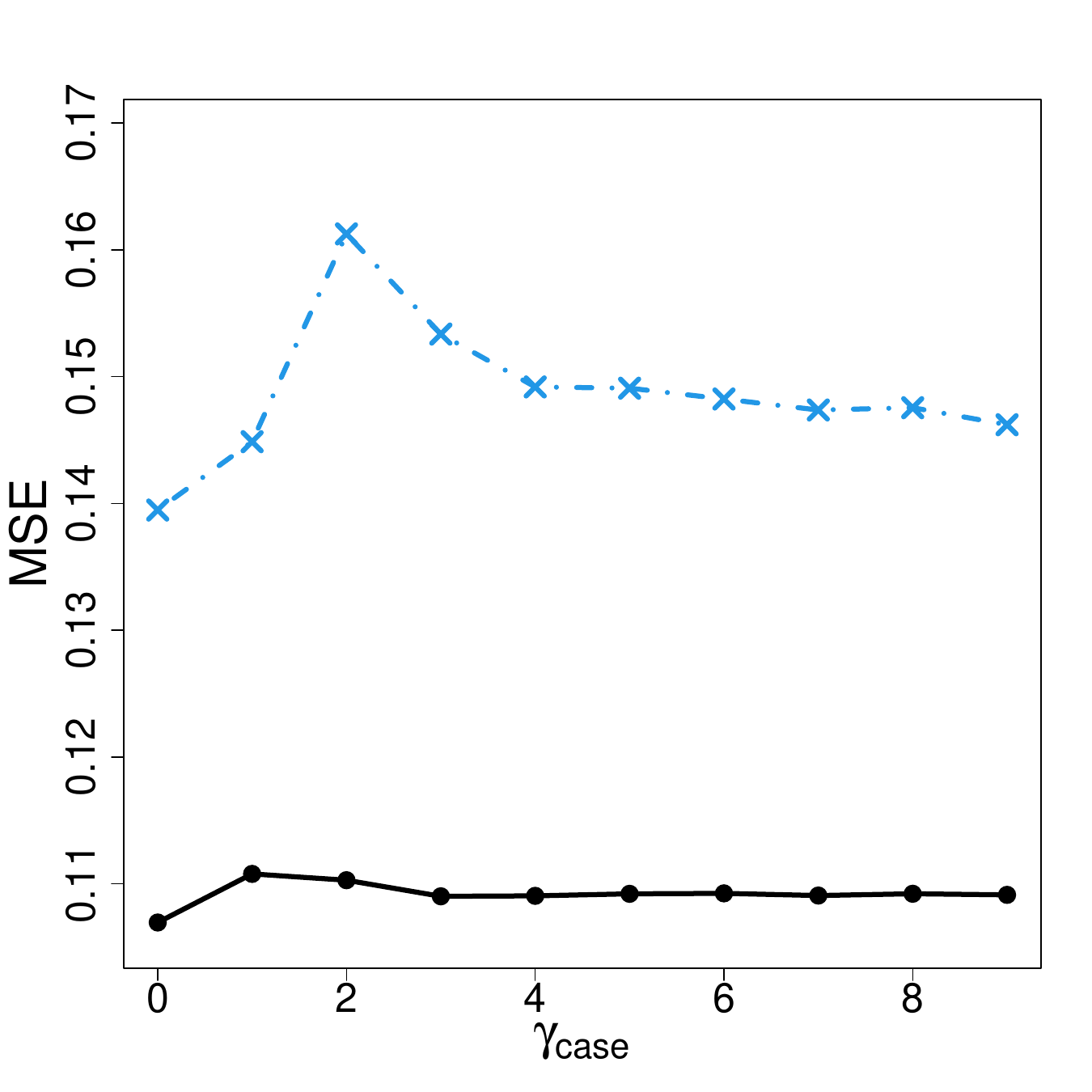} &\includegraphics[width=.3\textwidth]
  {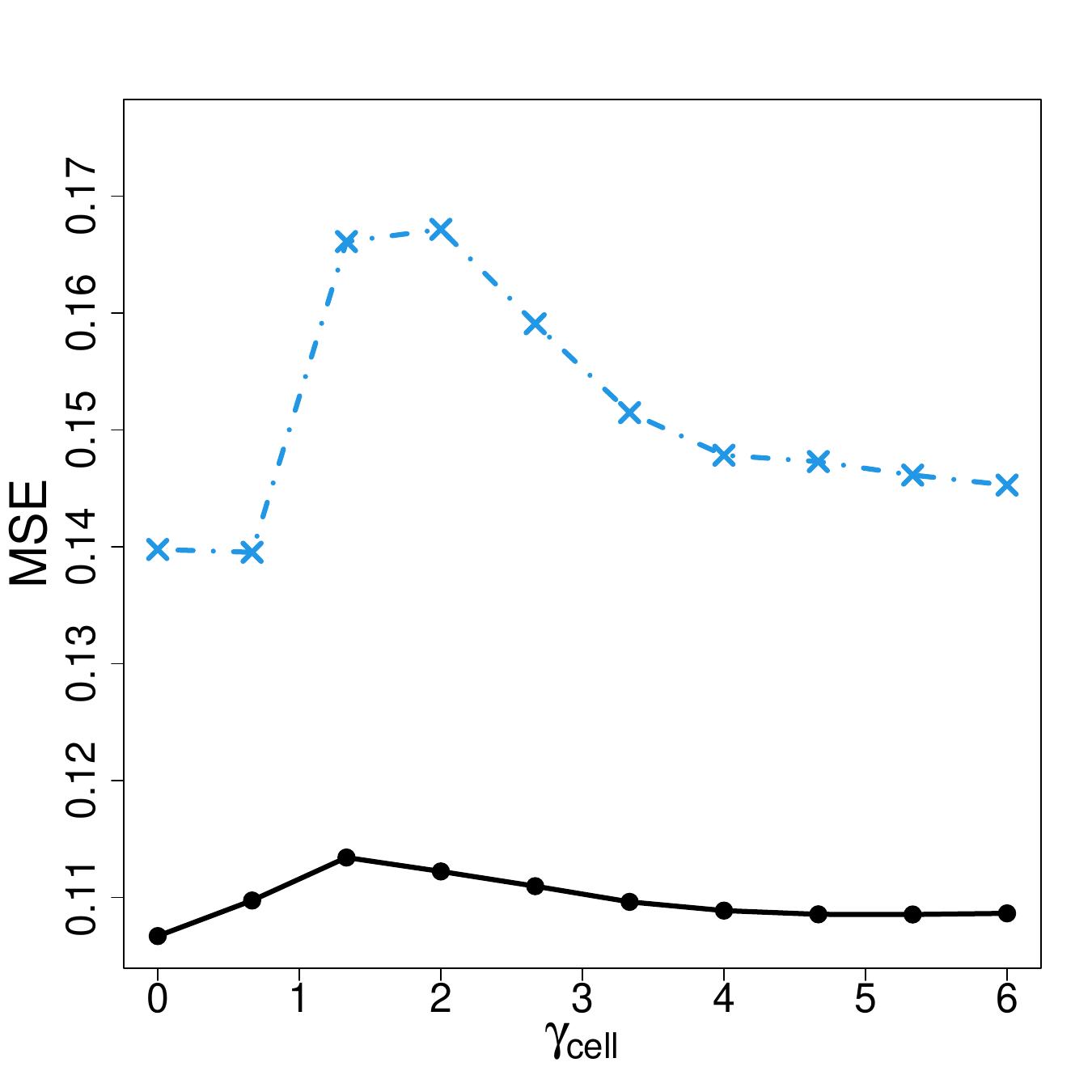}  
\end{tabular}
\caption{Median MSE of the out-of-sample prediction $\bhx$ obtained by MacroPCA and cellPCA where the training data are generated by the covariance model A09 with $n=100$ and $p=20$ with cellwise outliers, casewise outliers, or both. The test set is either clean (top row), contaminated with cellwise outliers (middle row), or simultaneously affected by cellwise outliers and missing values (bottom row).}
\label{fig:results_oos_pred}
\end{figure}

\clearpage
To investigate the effectiveness of the imputation method described in Section~\ref{sec:imput}, we compare cellPCA and MacroPCA based on the out-of-sample imputation MSE, which is given by
\begin{equation} 
  \text{MSE}=\frac{1}{n}
  \sum_{i=1}^n\sum_{j=1}^p
  \left(\impxsij - x^{*0}_{ij}\right)^2
\end{equation}
where $\impxsij$ is the imputation of
$x^*_{ij}$ and $x^{*0}_{ij}$ is the original value of 
that cell before any contamination took place. The data 
were generated as in Figure~\ref{fig:results_oos_pred}.

The resulting Figure~\ref{fig:results_imputation} 
has rows and columns corresponding to those of
Figure~\ref{fig:results_oos_pred}. Also here we see that
cellPCA has substantially outperformed MacroPCA.

\begin{figure}[!ht]
\centering
\begin{tabular}{cccc}
   &\large \textbf{Cellwise}  & \large \textbf{Casewise} &\large{\textbf{Casewise \& Cellwise}} \\
    [-4mm]
    \rotatebox{90}{\large \textbf{\parbox{5cm}{\centering Clean}}} &\includegraphics[width=.3\textwidth]
  {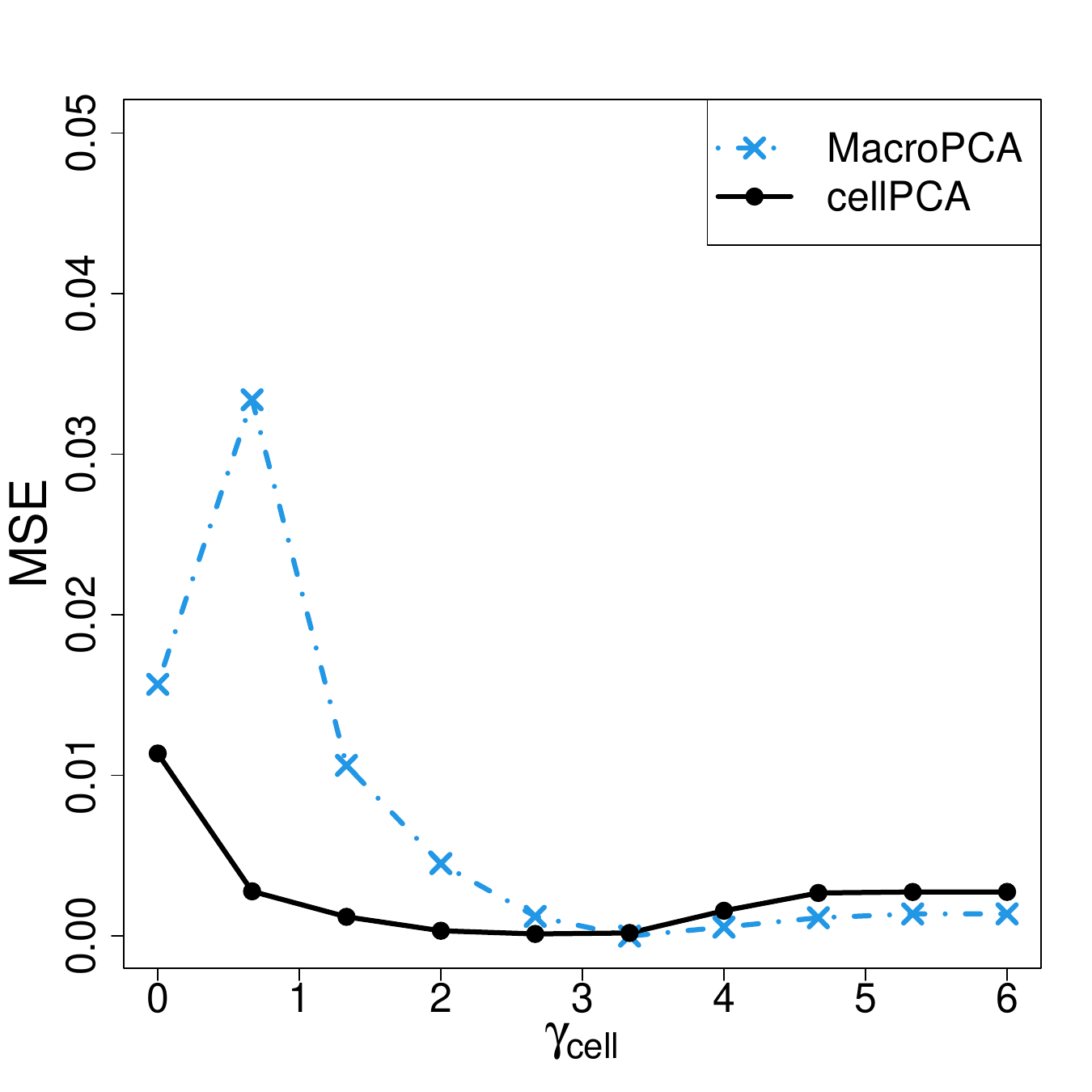} &\includegraphics[width=.3\textwidth]
  {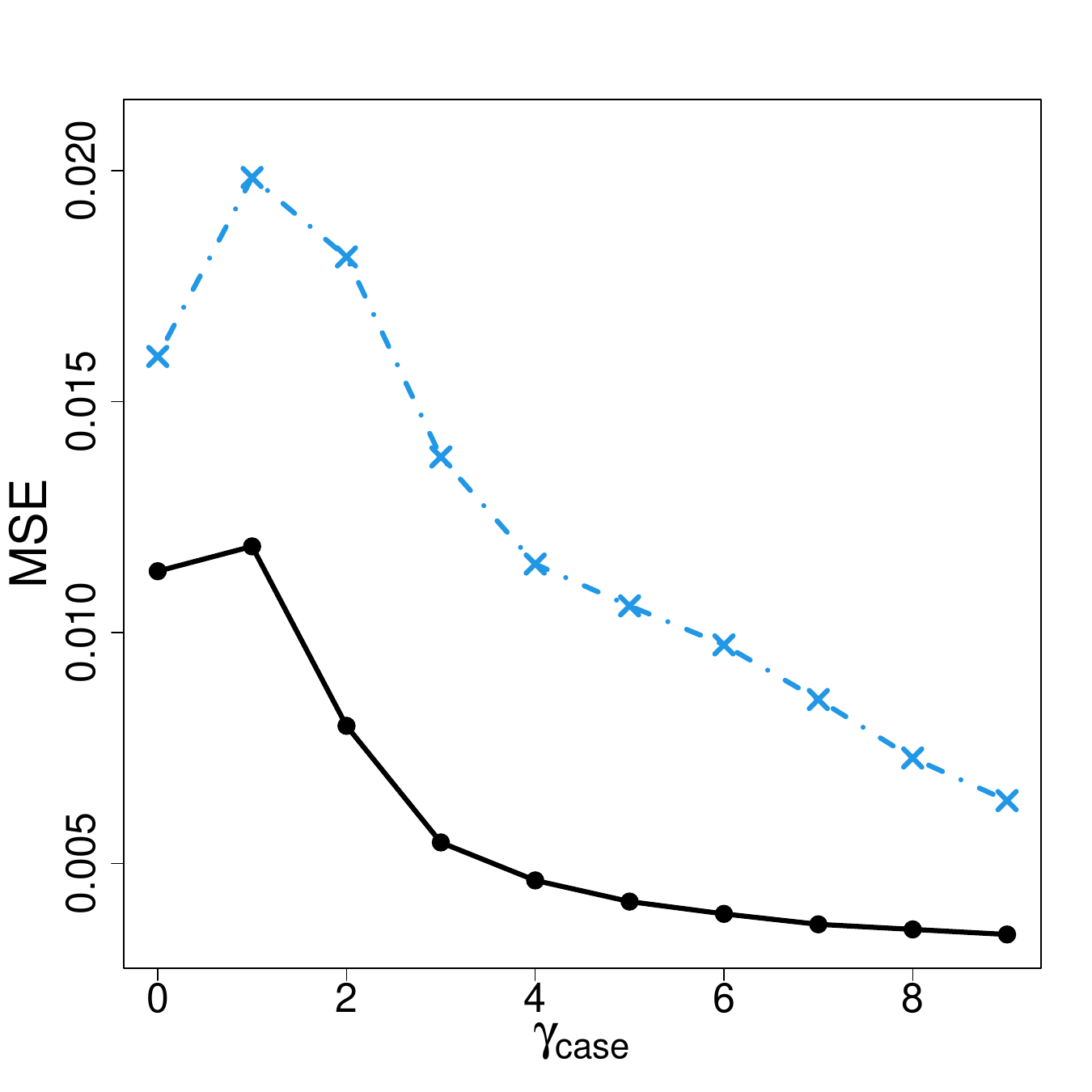} &\includegraphics[width=.3\textwidth]
  {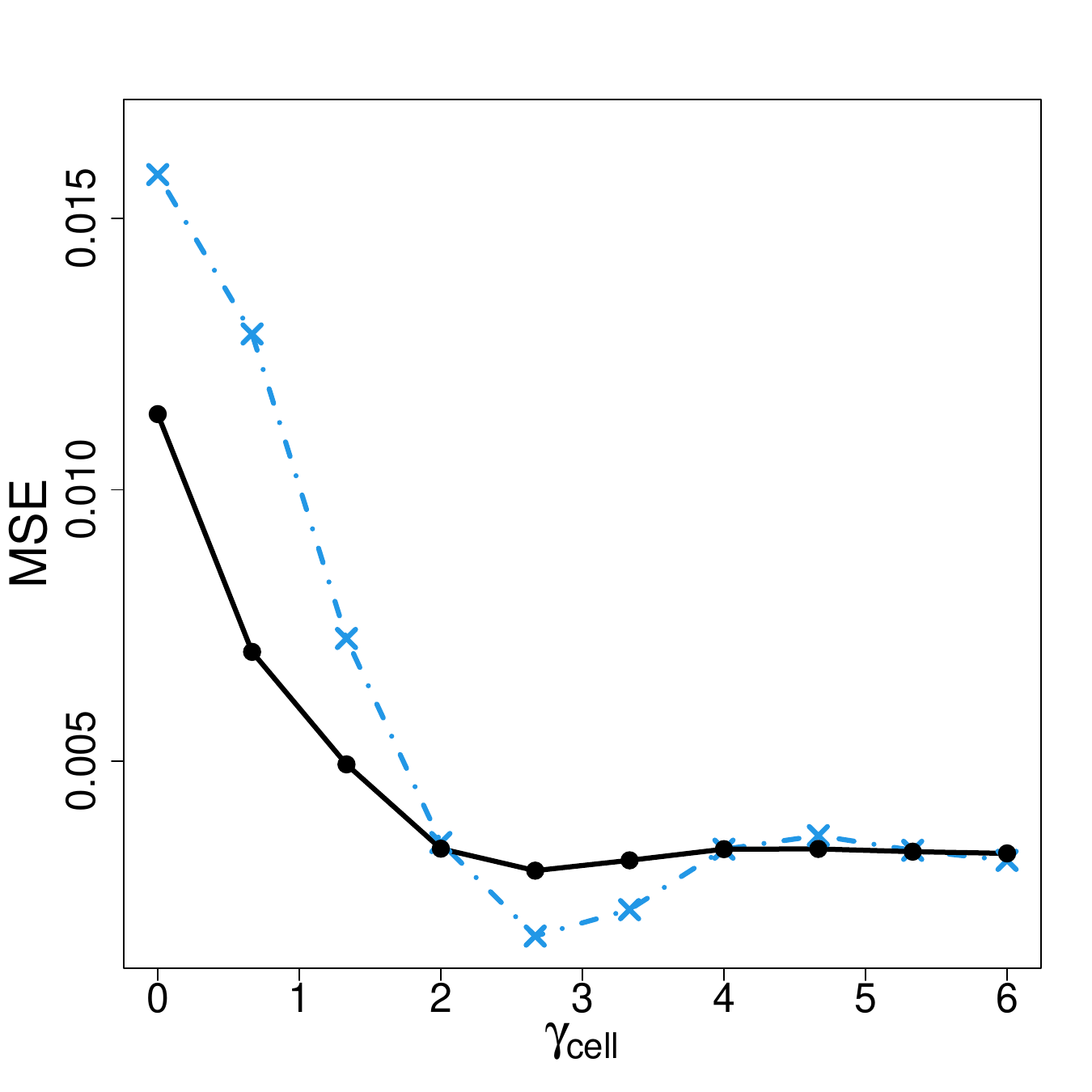}  \\
    [-4mm]
   \rotatebox{90}{\large \textbf{\parbox{5cm}{\centering Cellwise}}} &\includegraphics[width=.3\textwidth]
  {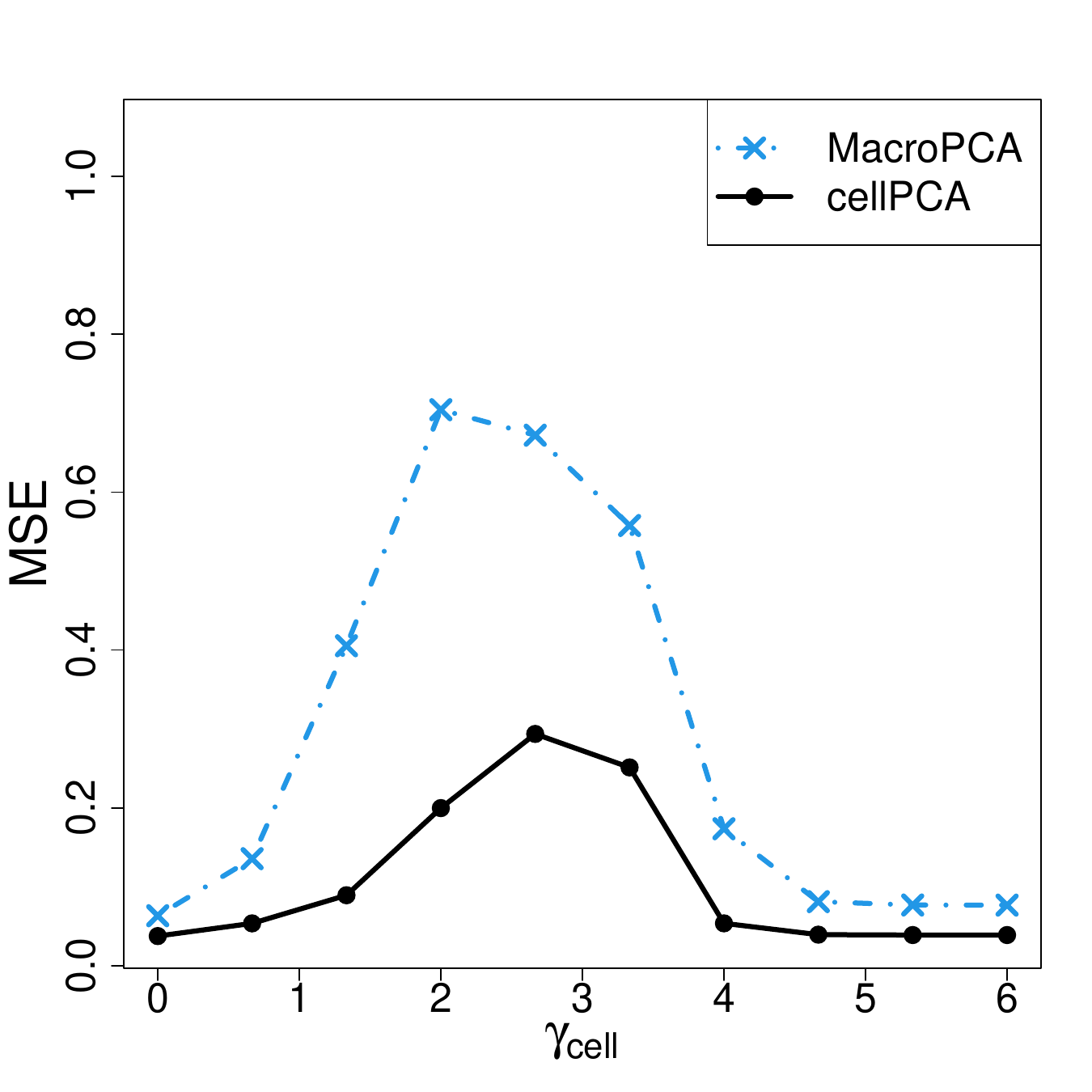} &\includegraphics[width=.3\textwidth]
  {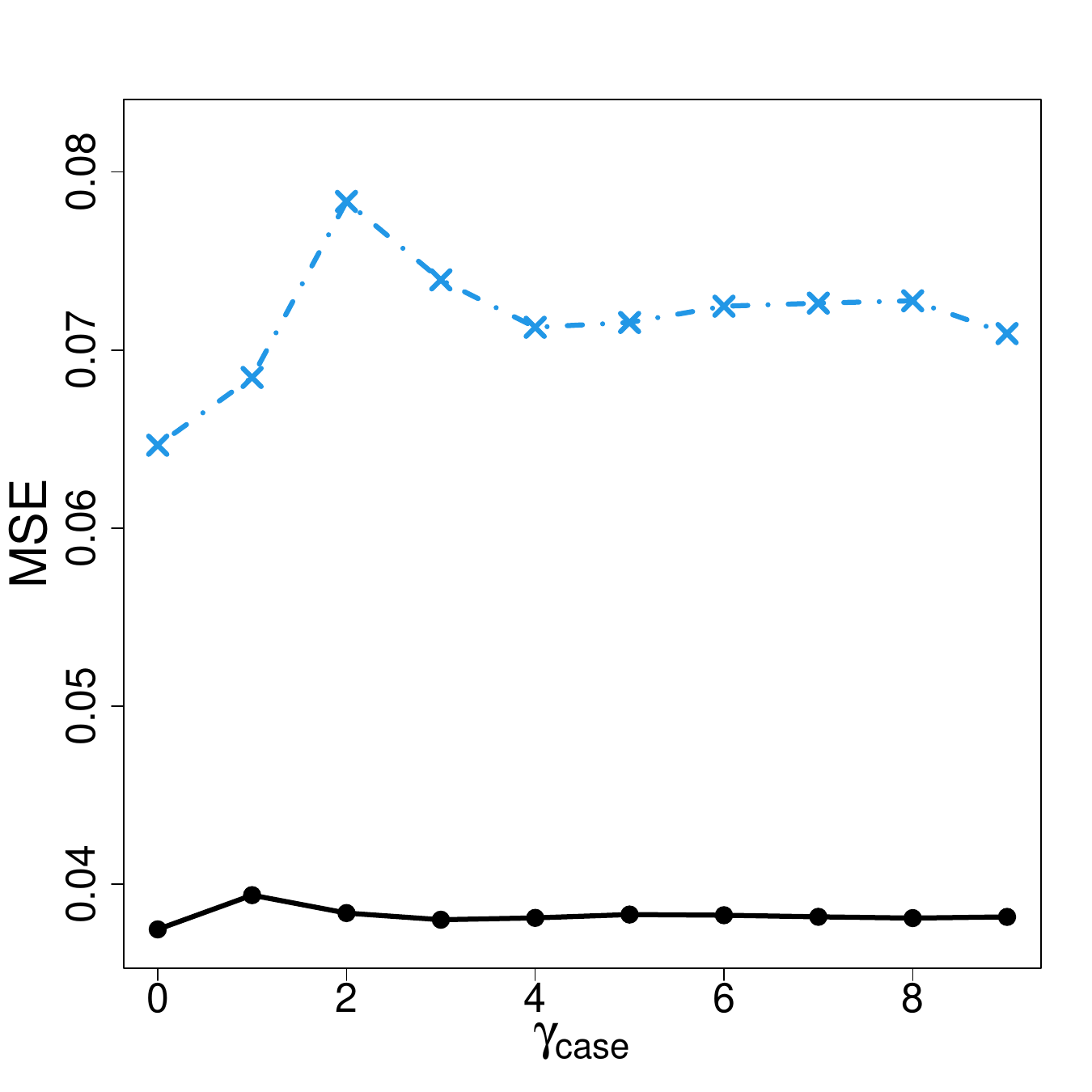} &\includegraphics[width=.3\textwidth]
  {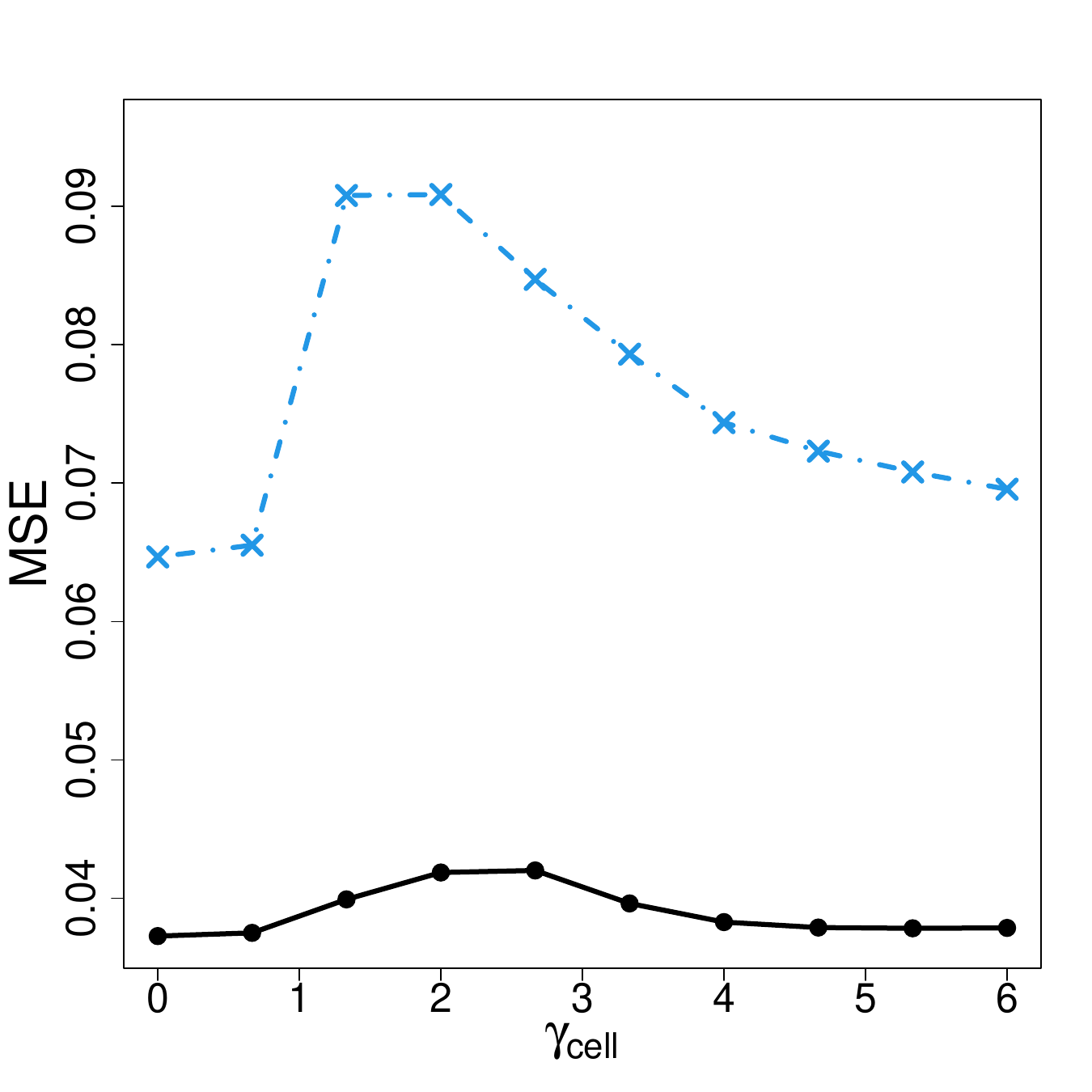}  \\
   [-4mm]
   \rotatebox{90}{\large \textbf{\parbox{5cm}{\centering Cellwise \& NAs}}}&\includegraphics[width=.3\textwidth]
  {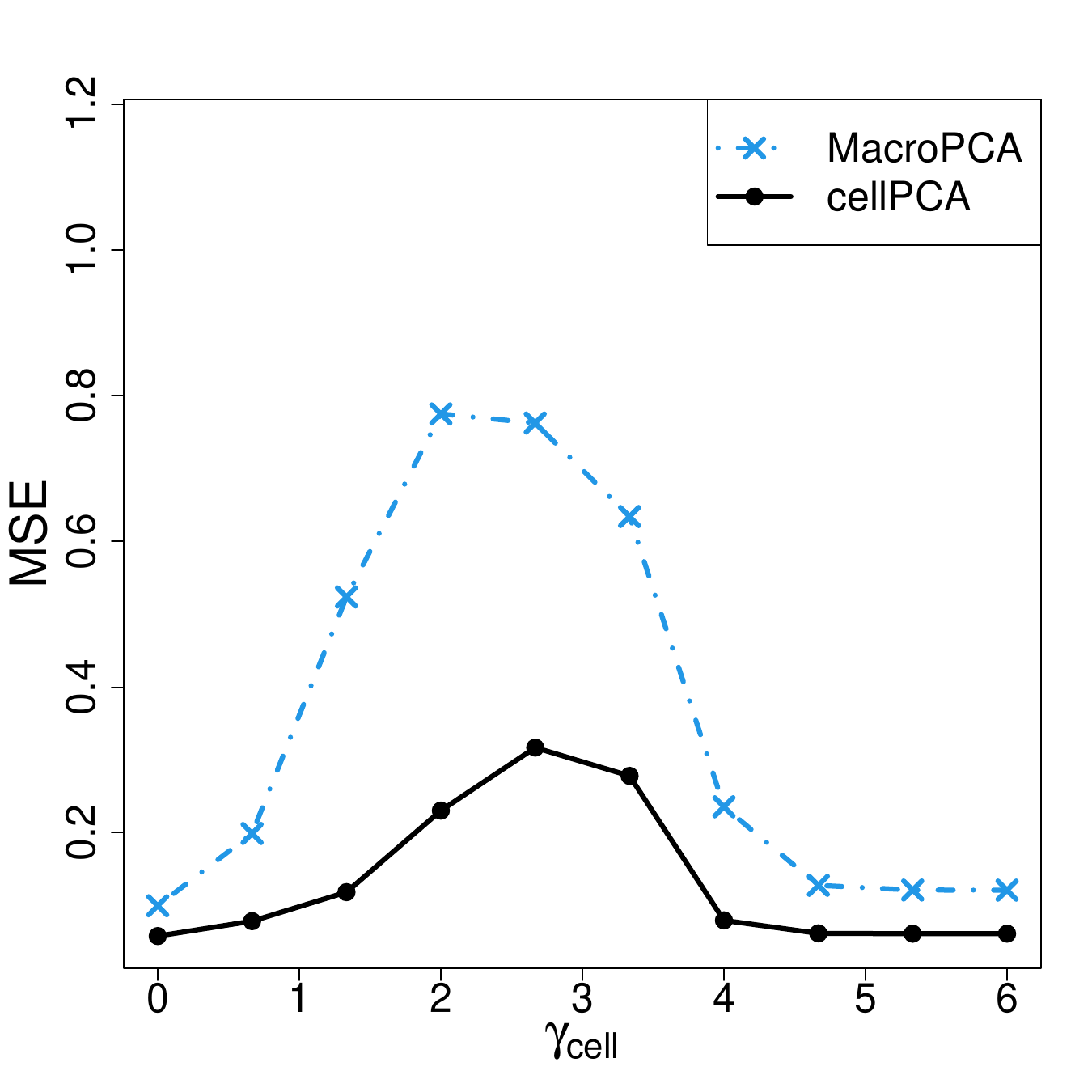} &\includegraphics[width=.3\textwidth]
  {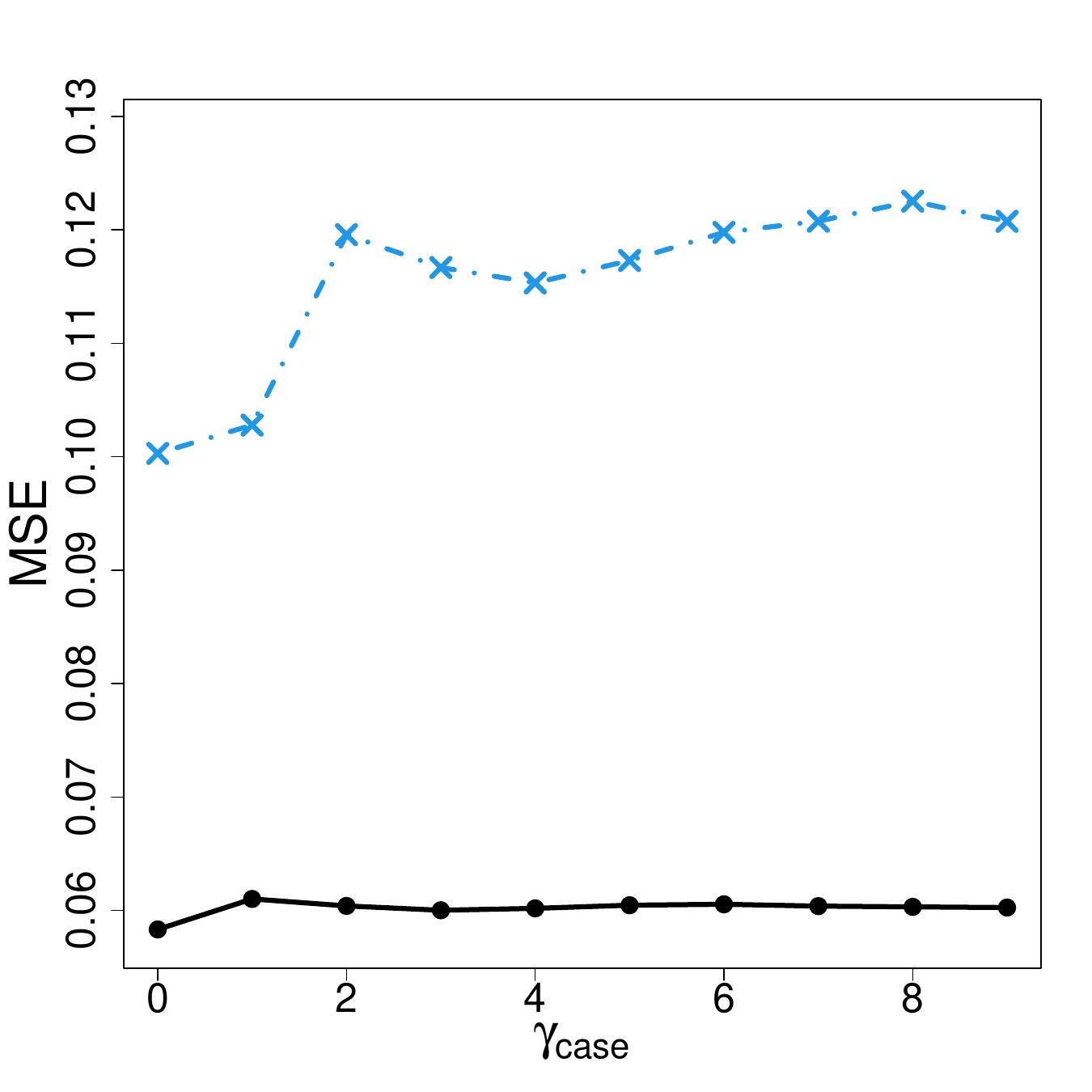} &\includegraphics[width=.3\textwidth]
  {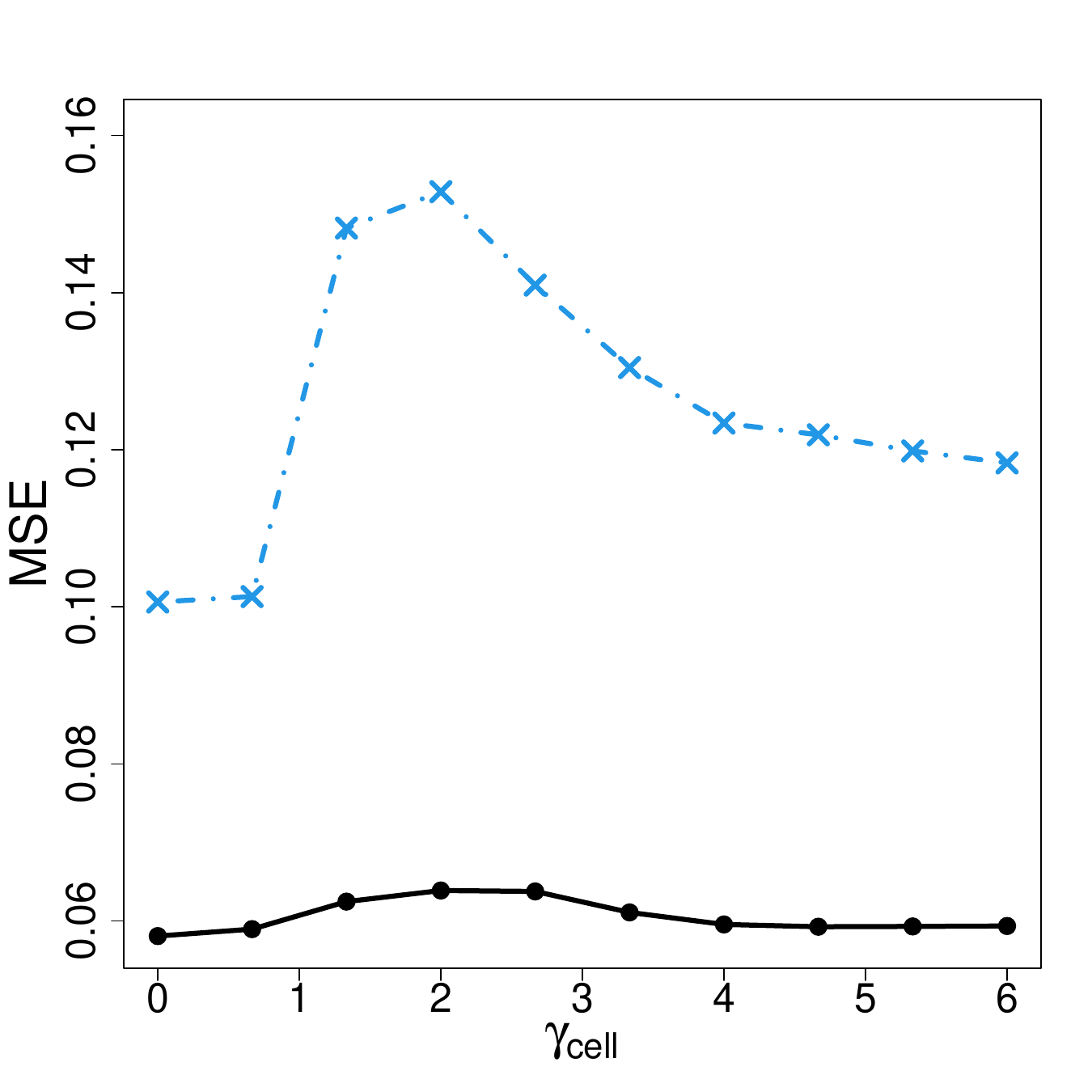}  
\end{tabular}
\caption{Median MSE of the imputation $\bimpxs$ obtained by MacroPCA and cellPCA where the training data are generated by the covariance model A09 with $n=100$ and $p=20$ with cellwise outliers, casewise outliers, or both. The test set is either clean (top row), contaminated with cellwise outliers (middle row), or simultaneously affected by cellwise outliers and missing values (bottom row).}
\label{fig:results_imputation}
\end{figure}

\clearpage
To evaluate the performance of the rank selection method
discussed in Section~\ref{sec:rank}, we run both MacroPCA
and cellPCA and apply the Kneedle algorithm 
\citep{satopaa2011finding} to both.
The data were generated from the covariance model A09 with 
$n=100$ and $p=20$, without NAs, in the presence of 
cellwise outliers, casewise outliers, or both. In each 
setting we ran 100 replications.
From the way the A09 covariance was constructed, it is
clear that the natural number of components is $k=2$.

Figure~\ref{fig:results_rank} displays the average 
of the selected ranks $\widehat{k}$ obtained by MacroPCA 
and cellPCA, over the 100 replications. 
We see that MacroPCA already did quite well, with only
small deviations from 2 when there is contamination.
CellPCA did not even deviate at all.

\begin{figure}[!ht]
\centering 
\vspace{5mm}
\begin{tabular}{ccc}
   \large \textbf{Cellwise}  & \large \textbf{Casewise} &\large{\textbf{Casewise \& Cellwise}} \\
    [-4mm]
    \includegraphics[width=.3\textwidth]
  {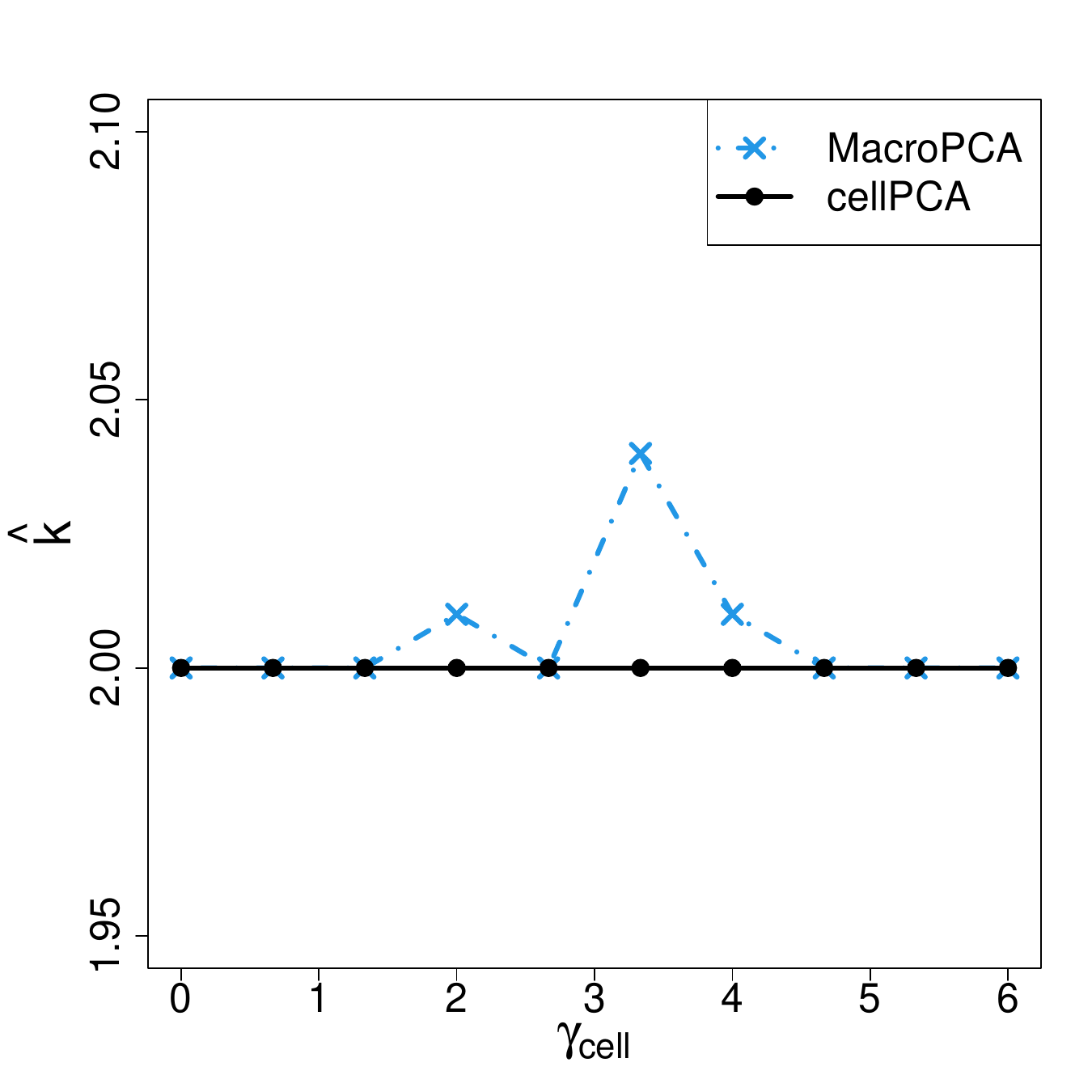} &\includegraphics[width=.3\textwidth]
  {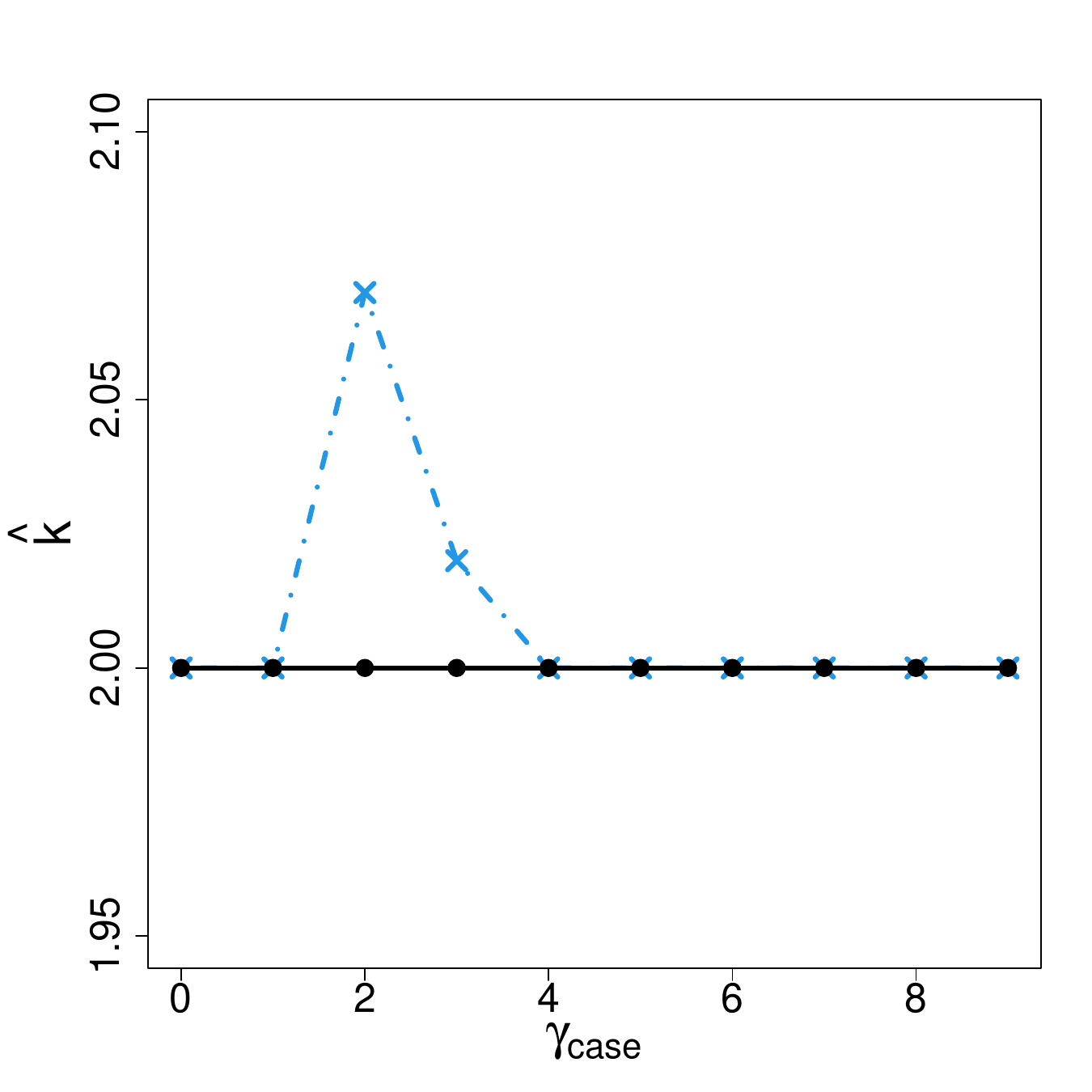} &\includegraphics[width=.3\textwidth]
  {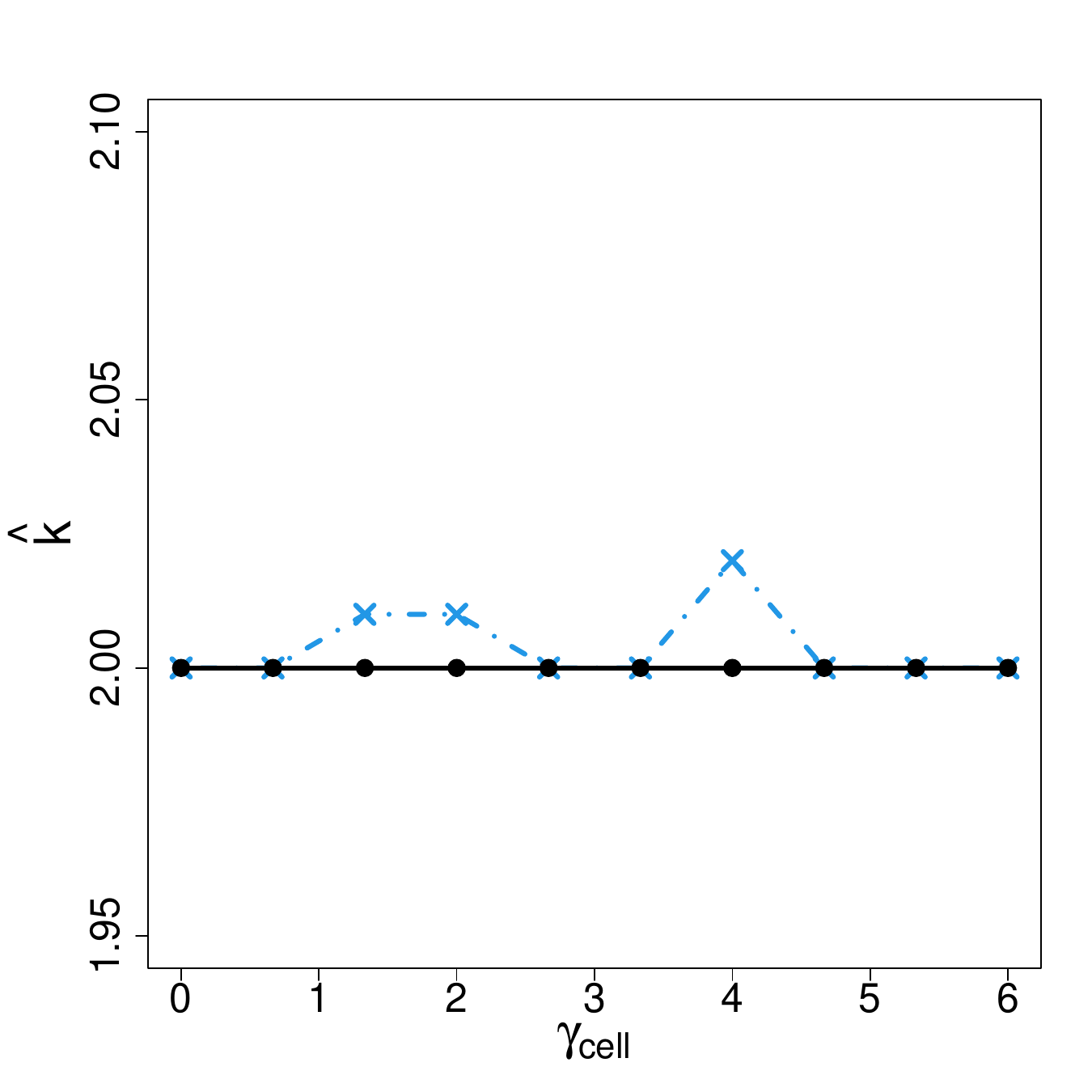}   
\end{tabular}
\caption{Average selected rank $\widehat{k}$ obtained by MacroPCA and cellPCA, over 100 replications. The data were generated by covariance model A09 with $n=100$ and $p=20$, without NAs, in the presence of cellwise outliers, casewise outliers, or both. The ranks were selected by applying the Kneedle algorithm to the scree plots.}
\label{fig:results_rank}
\end{figure}

\clearpage
\section{More on the Solfatara Data}
\label{app:addreal}

Figure~\ref{fig_frame203} shows frame 203 of the Solfatara
data, taken in the Fall. The interpretation is similar to
that in Section 7.2 for frame 14 taken in the Spring, only
the shape of the condensation cloud is a bit different.\\

\begin{figure}[hb]
\centering
\includegraphics[width=0.49\textwidth]{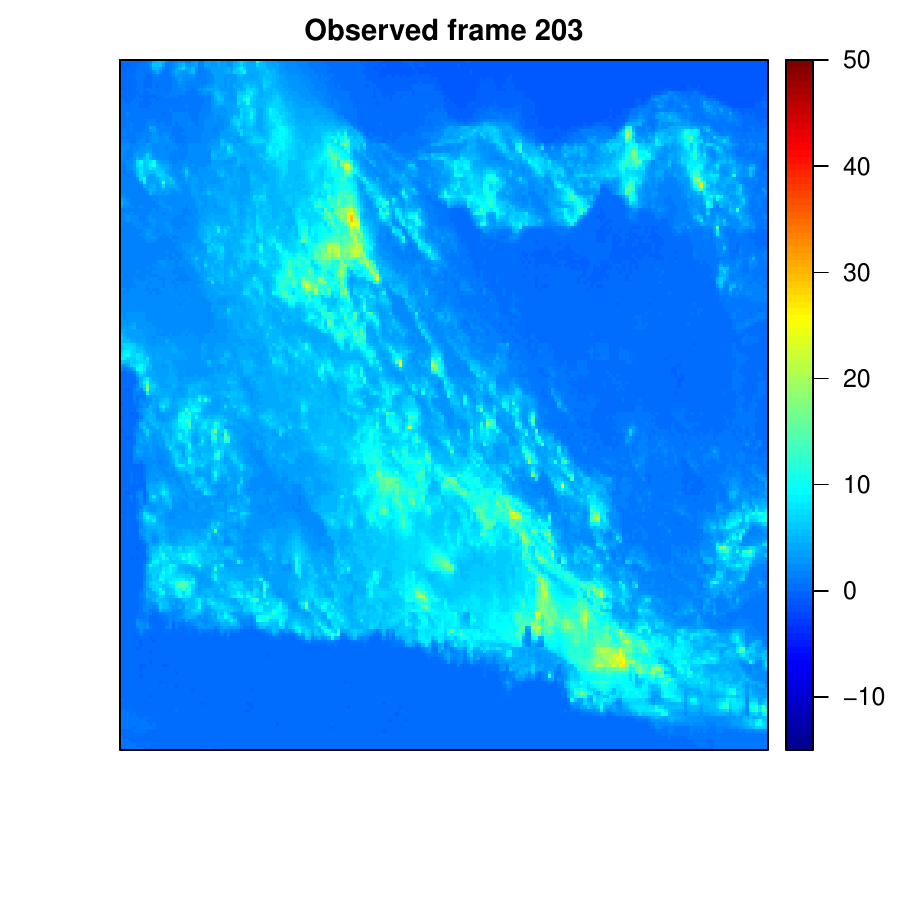}
\includegraphics[width=0.49\textwidth]{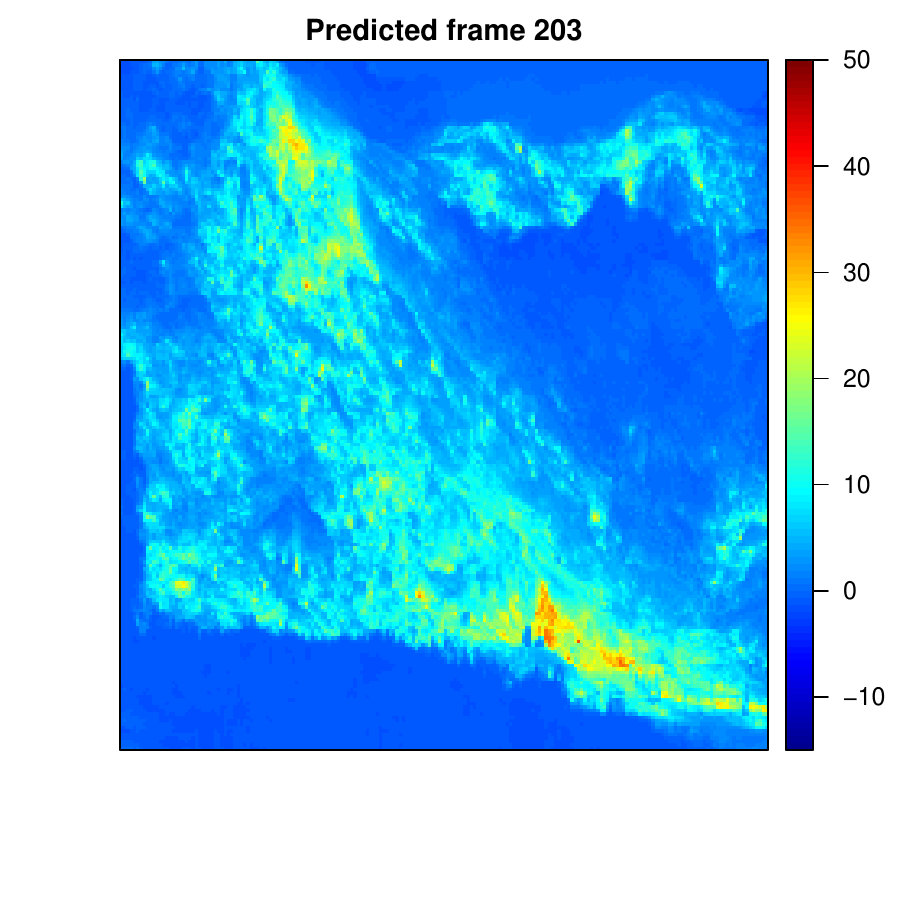}\\

\vspace{-7mm}
\includegraphics[width=0.5\textwidth]
  {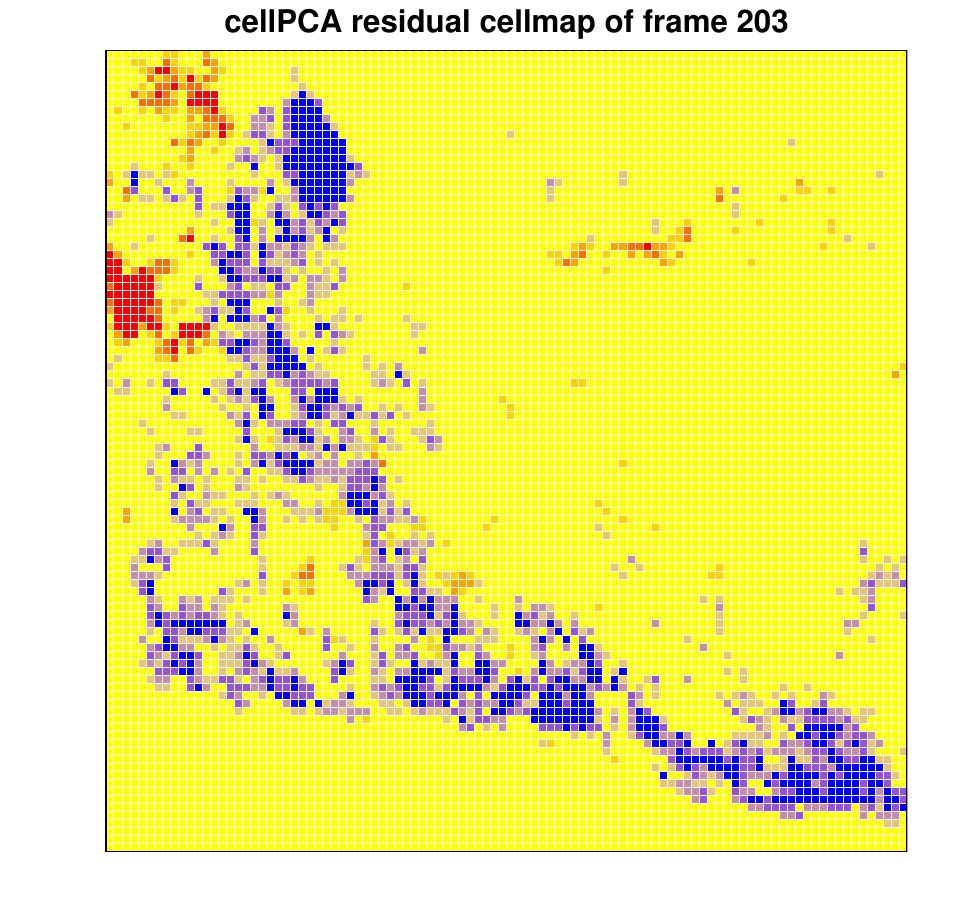}\\

\caption{Solfatara data: observed 
frame 203, its prediction, and
its residual cellmap.}
\label{fig_frame203}
\end{figure}
\clearpage

\renewcommand{\refname}{Additional References}

\end{document}